\documentclass[11pt]{amsart}
\usepackage{amssymb}
\usepackage{MnSymbol}
\usepackage{mathrsfs}
\usepackage{eufrak,verbatim}
\usepackage{amsthm}
\usepackage{fancyhdr}
\usepackage{setspace}
\usepackage{color}
\usepackage[pdftex]{graphicx}
\usepackage{wasysym}
\usepackage{upgreek}
\usepackage{hyperref}

\numberwithin{equation}{subsection}

\newtheorem{theorem}{Theorem}
\newtheorem{proposition}{Proposition}[subsection]
\newtheorem{lemma}[proposition]{Lemma}
\newtheorem{corollary}[proposition]{Corollary}

\theoremstyle{definition}
\newtheorem{definition}{Definition}[subsection]
\newtheorem{remark}{Remark}[subsection]

\newcommand{\eqdef}{\overset{\mbox{\tiny{def}}}{=}}

\newcommand{\ux}{\underline{x}}
\newcommand{\uy}{\underline{y}}

\newcommand{\Ldual}{{^{\star \hspace{-.03in}} \mathscr{L}}}
\newcommand{\Far}{\mathcal{F}}
\newcommand{\Gar}{\mathcal{G}}
\newcommand{\Max}{\mathcal{M}}
\newcommand{\Fardual}{{^{\star \hspace{-.03in}}\mathcal{F}}}

\newcommand{\Maxdual}{{^{\star \hspace{-.05in}}\mathcal{M}}}
\newcommand{\Stress}{\dot{Q}}
\newcommand{\EMT}{Q}

\newcommand{\ualpha}{\underline{\alpha}}
\newcommand{\uL}{\underline{L}}

\newcommand{\ualphadot}{{^{\odot} \underline{\alpha}}}
\newcommand{\alphadot}{{^{\odot} \alpha}}
\newcommand{\rhodot}{{^{\odot} \rho}}
\newcommand{\sigmadot}{{^{\odot} \sigma}}

\newcommand{\emet}{e}
\newcommand{\Sigmafirstfund}{\underline{g}}
\newcommand{\SigmafirstfundNabla}{\underline{\nabla}}

\newcommand{\angn}{\not \nabla}
\newcommand{\angdiv}{{\not {\hspace{-.03in} \mbox{div}}} \ }
\newcommand{\angcurl}{{\not {\hspace{-.03in} \mbox{curl}}} \ }

\newcommand{\angepsilon}{{\not \epsilon}}
\newcommand{\uepsilon}{\underline{\epsilon}}

\newcommand{\ucurcl}{\underline{\mbox{curl}} \ }
\newcommand{\udiv}{\underline{\mbox{div}} \ }

\newcommand{\angg}{\not g}

\newcommand{\Lie}{\mathcal{L}}
\newcommand{\Liemod}{\hat{\mathcal{L}}}

\newcommand{\Knorm}{\nupdownline}
\newcommand{\Kintnorm}{\nUpdownline}

\newcommand{\Electricfield}{E}
\newcommand{\Displacement}{D}
\newcommand{\Magneticinduction}{B}
\newcommand{\Magneticfield}{H}
\newcommand{\SFar}{\mathfrak{E}}
\newcommand{\SFardual}{\mathfrak{B}}

\newcommand{\Farinvariant}{\mbox{\lightning}}

\newcommand{\Vmult}{X_{local}}

\begin{document}
\pagestyle{fancy}
\title{The Nonlinear Stability of the Trivial Solution to the Maxwell-Born-Infeld System}

\author{Jared Speck}
\thanks{The author was supported by the Commission of the European Communities, ERC Grant Agreement No 208007. He was also funded in parts by NSF through grants DMS-0406951 and DMS-0807705.}
\address{University of Cambridge, Department of Pure Mathematics \& Mathematical Statistics,
Wilberforce Road, Cambridge, CB3 0WB, United Kingdom}
\email{jspeck@math.princeton.edu}

\begin{abstract}
In this article, we use an electromagnetic gauge-free framework to establish the existence of small-data global solutions to the Maxwell-Born-Infeld (MBI) system on the Minkowski space background in $1 + 3$ dimensions. Because the nonlinearities in the system satisfy a version of the null condition, we are also able to show that these solutions decay at exactly the same rates as solutions to the linear Maxwell-Maxwell system. In addition, we show that on any Lorentzian manifold, the MBI system is hyperbolic in the interior of the field-strength regime in which its Lagrangian is real-valued.
\end{abstract}



\date{Version of \today}
\maketitle

\tableofcontents

\section{Introduction} \label{S:Introduction}

The Maxwell-Born-Infeld (MBI) system is a nonlinear model of classical electromagnetism that was introduced in
the 1930's by Born and Infeld \cite{mBlI1934}, with \cite{mB1933} a precursor by Born. In this article, we
study the source-free (i.e. the right-hand sides of \eqref{E:dFis0intro} - \eqref{E:dMis0intro} are $0$) MBI system in the fixed spacetime\footnote{By spacetime, we mean a $4-$dimensional time-oriented manifold $M$ together with a Lorentzian metric $g$ of signature $(-,+,+,+)$.} $(M,g).$ We will assume throughout the article that $(M,g)$ is equal to $\mathbb{R}^{1+3}$ equipped with the usual Minkowski metric, which has components $g_{\mu \nu} = \mbox{diag}(-1,1,1,1)$ in an inertial coordinate system
$(x^0,(x^1,x^2,x^3)) \eqdef (t,\ux).$ Nonetheless, much of our discussion regarding the structure of the MBI system remains valid in an arbitrary spacetime. As is explained in detail in Section \ref{S:MBI}, the MBI equations can be expressed as 

\begin{subequations}
\begin{align} 
	d \Far & = 0, \label{E:dFis0intro} \\
	d \Max & = 0, \label{E:dMis0intro}
\end{align}
\end{subequations}
where $d$ denotes the exterior derivative operator, the \emph{Faraday tensor} $\Far,$ which is a two-form, is the fundamental unknown, the \emph{Maxwell tensor} $\Max,$ which is also a two-form, is defined by

\begin{align} \label{E:Maxsdefintro}
	\Max & = \ell_{(MBI)}^{-1}\big( \Fardual + \Farinvariant_{(2)} \Far \big),
\end{align}
$\star$ denotes the Hodge dual, $\Farinvariant_{(1)} \eqdef \frac{1}{2} (g^{-1})^{\zeta \kappa} (g^{-1})^{\eta \lambda} \Far_{\zeta \eta} \Far_{\kappa \lambda},$ $\Farinvariant_{(2)} \eqdef \frac{1}{4} (g^{-1})^{\zeta \kappa} (g^{-1})^{\eta \lambda} \Far_{\zeta \eta} \Fardual_{\kappa \lambda}$ are the \emph{electromagnetic invariants}, $(g^{-1})^{\mu \nu}$ are the components of the inverse of the \emph{spacetime metric} $g,$ and $\ell_{(MBI)} \eqdef \big(1 + \Farinvariant_{(1)} - \Farinvariant_{(2)}^2 \big)^{1/2}.$ Born and Infeld's contribution to the above system was their provision of the constitutive relation \eqref{E:Maxsdefintro}, while equations \eqref{E:dFis0intro} - \eqref{E:dMis0intro} were postulated\footnote{Maxwell's formulation of electromagnetism was not presented using the framework of the Faraday tensor, nor that of the familiar electric field $\Electricfield$ and magnetic induction $\Magneticinduction;$ rather, he used the structure of quaternions to write down a system of 20 equations in 20 unknowns. The familiar ``vector'' formulation in terms of $\Electricfield$ and $\Magneticinduction$ was developed by Heaviside \cite{pN1988}.} in the 1860's by Maxwell \cite{jM1891a}, \cite{jM1891b}. We recall that in contrast to \eqref{E:Maxsdefintro}, Maxwell adopted the linear constitutive law $\Max = \Fardual.$ Hence, we refer to the nonlinear system \eqref{E:dFis0intro}, \eqref{E:dMis0intro}, \eqref{E:Maxsdefintro} as the ``Maxwell-Born-Infeld'' equations, and the linear system \eqref{E:dFis0intro}, \eqref{E:dMis0intro}, $\Max = \Fardual$ as the ``Maxwell-Maxwell'' equations. We summarize our main results here; they are thoroughly stated and proved in Sections \ref{S:IVP} and \ref{S:GlobalExistence}.

\begin{changemargin}{.25in}{.25in} 
\textbf{Main Results.} 
The trivial solution to the MBI system on the $1 + 3$ dimensional Minkowski space background is globally stable. More specifically, if the initial data for the MBI system are sufficiently small as measured by the weighted Sobolev norm defined in \eqref{E:HNdeltanorm}, then these data launch a unique classical solution to the MBI system existing in all of Minkowski space. Moreover, these small-data solutions decay at the same rate as solutions to the linear Maxwell-Maxwell equations. In addition, the MBI system is hyperbolic\footnote{By ``hyperbolic,'' we mean that there is a local energy estimate available that can be used to prove that initial data have a non-trivial development, and that furthermore, the system has finite speed of propagation.} in the interior of the field-strength regime in which its Lagrangian is real-valued. In particular, the system is locally well-posed in the aforementioned weighted Sobolev space.
\end{changemargin}

\begin{remark}
	Certain large fields can cause $\ell_{(MBI)}$ to become complex-valued. For such fields, MBI theory is
	not even well-defined. However, as is discussed in Remark \ref{R:BDHyperbolicity}, the MBI equations are 
	well-defined \emph{and hyperbolic} for all finite values of the state-space variables $(\Magneticinduction,\Displacement),$ 
	which are introduced in Section \ref{SS:electromagneticdecomposition}. In particular, the MBI system is well-posed
	for sufficiently regular initial data belonging to the interior of the region of state-space in which the equations are 
	well-defined.
\end{remark}

\begin{remark}
	Although we only discuss global existence to the future, our results apply to just as well to the past. Our notion of ``future'' is determined by assumption that the vectorfield $\partial_t$ is future-directed.
\end{remark}

Recently, several scientific communities have expressed renewed interest in the MBI system for a variety of reasons. 
As an interesting example, we cite the works \cite{mK2004a}, \cite{mK2004b}, in which Kiessling has proposed a model of classical electrodynamics with point charges that has the promise of self-consistency\footnote{The motion of point charges in linear Maxwell-Maxwell theory is \emph{mathematically ill-defined} under the usual Lorentz force law.}: it is expected that
the theory is well-defined without truncation, regularization, or renormalization. This theory couples a first-order guiding law for the point charges, whose corresponding relativistic guiding field satisfies a Hamilton-Jacobi type PDE, to the MBI field equations. In contrast to the case of the linear Maxwell-Maxwell equations, the electromagnetic potentials\footnote{Recall that an electromagnetic potential is a one form $A$ such that $\Far = d A.$} of the solutions to the MBI system \emph{with non-accelerating point charge sources} in Minkowski space can be chosen to be \emph{Lipschitz continuous}; it is expected that this continuity property should remain true even for accelerating point charges, which would then allow for a well-defined coupling to the Hamilton-Jacobi theory. As a second example, we note that the MBI system has mathematical connections to string theory, for its Lagrangian (see \eqref{E:LMBI}) appears in connection
with the motion gauge fields (arising in the study of attached, open strings) on a D-brane (see e.g. \cite{gG2003}).

As is true for Kiessling, our interest in the MBI system is further motivated by results contained in \cite{gB1969} and \cite{jP1970}, which show that it is the unique\footnote{More precisely, there is a one-parameter family of such theories indexed by $\upbeta > 0,$ where $\upbeta$ is Born's ``aether'' constant.} theory of classical electromagnetism that is derivable from an action principle and that satisfies the following $5$ postulates (see also the discussion in \cite{iBB1983}, \cite{mK2004a}):

\begin{enumerate}
	\item The field equations transform covariantly under the Poincar\'e group.
	\item The field equations are covariant under a Weyl (gauge) group.
	\item The electromagnetic energy associated to a stationary point charge is finite.
	\item The field equations reduce to the linear Maxwell-Maxwell equations in the weak field limit.
	\item The solutions to the field equations are not birefringent (we will soon elaborate upon this notion).
\end{enumerate}
We remark that the Maxwell-Maxwell system satisfies all of the above postulates except for $(3),$ and that
the MBI system was shown to satisfy $(3)$ by Born in \cite{mB1933}.
 
We would now like to further discuss postulate $(5).$ Physically, it is equivalent to the statement that the
``speed of light propagation'' is independent of the polarization of the wave fields. Mathematically, it can be recast as a statement about the \emph{characteristic subset} of the field equations. To flesh out this notion, we
need to discuss some technical details. First, we remark that equation \eqref{E:HmodifieddMis0summary}, which reads $H^{\mu \nu \kappa \lambda} \nabla_{\mu} \Far_{\kappa \lambda} = 0,$ is equivalent to \eqref{E:dMis0intro} - \eqref{E:Maxsdefintro} modulo equation \eqref{E:dFis0intro}, where the tensorfield $H^{\mu \nu \kappa \lambda}$ is defined in \eqref{E:Hdef}. Now for each covector $\xi \in T_p^* M,$ the cotangent space of $M$ at $p,$ we consider the quadratic form $\chi^{\mu \nu}(\xi) \eqdef H^{\mu \kappa \nu \lambda}\xi_{\kappa} \xi_{\lambda}.$ Because of the properties \eqref{E:hminussignproperty1} - \eqref{E:hsymmetryproperty}, which are also possessed by $H^{\mu \kappa \nu \lambda},$ it follows that $\xi$ is an element of $N\big(\chi(\xi)\big),$  the null space of $\chi(\xi).$ The characteristic subset of $T_p^* M,$ which we denote by $C_p^*,$ is defined
to be the set of all non-zero $\xi$ such that $N\big(\chi(\xi)\big)$ is strictly larger than span($\xi$); i.e.,

\begin{align}
	C_p^* \eqdef \lbrace \xi \neq 0 \in T_p^* M \mid N\big(\chi(\xi)\big) \slash \mbox{span}(\xi) \neq \emptyset \rbrace.
\end{align}
As discussed in detail in \cite[Chapter 6]{dC2000}, the set $C_p^*$ governs the local speeds of propagation of solutions to the MBI system\footnote{More precisely, it is $C_p,$ the characteristic subset of $T_p M,$ the tangent space of $M$ at $p,$  that corresponds to the local speeds of propagation. $C_p$ is dual to $C_p^*$ in a sense defined in \cite{dC2000}.}. It is easy to see that $C_p^*$ is a conical set in the sense that if $\xi \in C_p^*,$ then any non-zero multiple of $\xi$ is also an element of $C_p^*.$ In general, this conical subset may have several different ``sheets.'' However, in the case of the MBI system, there is a degeneracy resulting in the presence of \emph{only a single sheet} (i.e., there is only one ``null cone'' associated to the MBI system); this is the mathematical characterization of ``no birefringence.'' As we alluded to above, the Maxwell-Maxwell system also possesses this property. Moreover, in the case of the Maxwell-Maxwell system, $C_p^*$ exactly coincides with the gravitational null cone $\lbrace \xi \in T_p^* M \mid (g^{-1})^{\kappa \lambda} \xi_{\kappa} \xi_{\lambda} = 0 \rbrace.$ However, in a general nonlinear\footnote{More specifically, we mean quasilinear.} theory, and specifically in the case of the MBI system, $C_p^*$ does \emph{not} coincide with the gravitational null cone. Although we do not directly prove this fact in this article, we plan to discuss this issue in detail in a future publication, in which we will give a complete discussion of the geometry of the MBI system; see also the discussion in the proof of Proposition \ref{P:LocalExistenceCurrent} 
and in \cite{iBB1983}. We do however, as an aside, investigate a related issue that would be relevant if one wanted to couple the MBI system to the equations of general relativity. Namely, we prove that MBI system's energy-momentum tensor satisfies the \emph{dominant energy condition}; see Lemma \ref{L:dominantenergycondition}. Physically, this means that the speeds of propagation associated to the MBI system are no larger than the speeds associated to the gravitational null cone; i.e., the ``speed of MBI light is less than or equal to the speed of gravity.'' Mathematically, this means that $C_p^*$ lies outside\footnote{The dual picture perhaps more intuitively corresponds to the notion of MBI light traveling ``more slowly than gravity:'' $C_p$ lies inside of the gravitational null cone in $T_p M,$ which is defied to be $\lbrace X \in T_p M \mid g_{\kappa \lambda} X^{\kappa} X^{\lambda} = 0 \rbrace.$} of the gravitational null cone.

\subsection{Comparison with related work}

The core of our proof is based on a blend of ideas presented in \cite{dCsK1990} and \cite{dC2000}; in \cite{dCsK1990}, Christodoulou and Klainerman use methods similar to the ones used in this paper to analyze the decay properties of solutions to the linear Maxwell-Maxwell equations in Minkowski space, while in \cite{dC2000}, Christodoulou provides a framework for deriving positive ``almost conserved'' quantities for nonlinear, hyperbolic PDEs that are derivable from a Lagrangian, of which the MBI system is an example. In short, using the methods of \cite{dC2000}, we are able to construct certain almost-conserved (in the small-solution regime) energies that have coercive properties nearly identical to those of the conserved quantities constructed in \cite{dCsK1990}. 

The aforementioned works and the present one are applications of a collection of geometric-analytic techniques that are applicable to a large class of hyperbolic PDEs derivable from a Lagrangian. These techniques are often collectively referred to as the \emph{vectorfield method}. The term ``vectorfield'' refers to the fact that in typical applications, coercive quantities are constructed with the help of special vectorfields connected to the symmetries (or approximate symmetries) of the system. Originally introduced by Klainerman \cite{sK1985}, \cite{sK1986} in his analysis of small-data global solutions to nonlinear wave equations, the vectorfield method has grown into its own industry. As examples, we provide a non-exhaustive list of topics for which the vectorfield method has proven fruitful: 

\begin{itemize}
	\item Global nonlinear stability results for the Einstein equations \cite{lBnZ2009},
	\cite{dCsK1993}, \cite{mDgH2006}, \cite{sKfN2003}, \cite{hLiR2005}, \cite{hLiR2010}, \cite{iRjS2009}, \cite{jS2010b}.
	\item Small-data global existence for nonlinear elastic waves \cite{tS1996}.
	\item The formation of shocks in solutions to the relativistic Euler equations \cite{dC2007}.
	\item Decay results for linear equations on curved backgrounds \cite{lApB2009}, \cite{pB2008}, \cite{mDiR2005a}, 
		\cite{mDiR2008b}, \cite{mDiR2008a}, \cite{gH2010}.
	\item The formation of trapped surfaces in vacuum solutions to the Einstein equations \cite{dC2009}, \cite{sKiR2009a}.
	\item Local existence and non-relativistic limits for the relativistic Euler equations
		without the use of symmetrizing variables \cite{jS2008b}, \cite{jS2008a}, \cite{jSrS2010}.
\end{itemize}

\subsection{Difficulties in working with a four-potential}

In various contexts during the study of linear Maxwell-Maxwell theory, authors commonly analyze the components of a four-potential $A$  and its derivatives, rather than the Faraday tensor itself  (\cite[chapter $6$]{jJ1999} is a classic reference, and \cite{jL2008}, \cite{jL2009}  
are examples in the context of the Einstein-Maxwell system). Recall that a four-potential is a one-form $A$ such that $\Far = dA;$ the existence of such a one-form is guaranteed by \eqref{E:dFis0intro} and Poincar\'e's Lemma. $A$ is not unique, for any ``gauge'' transformation of the form $A \rightarrow A + \nabla \gamma,$ where $\gamma$ is a scalar-valued function, preserves the relation $\Far = dA.$ A well-known method of capitalizing on this gauge freedom is to work in the \emph{Lorenz gauge}, which is the added condition 

\begin{align} \label{E:Lorenz}
	\nabla_{\kappa} A^{\kappa} = 0. 
\end{align}	
The advantages of the Lorenz gauge are discussed below. Of course, the viability of the gauge condition \eqref{E:Lorenz}, which can be arranged to hold initially, depends on the fact that it is preserved by the flow of an appropriate version of the
Maxwell-Maxwell equations (e.g. the system \eqref{E:Maxwelldecoupledintro} below).

We would now like to discuss some subtle issues concerning the difficulties that may arise in an attempt to work with the Lorenz gauge when the electromagnetic equations are quasilinear. As in the remainder of the article, we assume in this section that  $(M,g)$ is Minkowski spacetime, and furthermore, that we are working in an inertial coordinate system. However, these assumptions have no substantial bearing on the issues at hand, for the same issues arise in any other spacetime $(M,g)$ 
equipped with any coordinate system. We begin with a brief summary of the framework used for discussing an arbitrary nonlinear covariant theory of electromagnetism that is derivable from a Lagrangian\footnote{We are slightly departing from the usual convention by referring to the Hodge dual of the Lagrangian, which we denote by $\Ldual,$ as the Lagrangian; $\mathscr{L}$ is a four-form, while $\Ldual$ is scalar-valued.}. If we choose to describe such a theory through the use of four-potentials $A,$ then the Lagrangian $\Ldual = \Ldual[\nabla A]$ can be written as a function of $\nabla A.$ A very detailed elaboration of this discussion can be found in \cite{dC2000}; here, we only introduce the facts that are relevant to the issues at hand. The Euler-Lagrange equations (equations of motion) for such a theory can be written as

\begin{align} \label{E:EquationsforA}
	h_{\mu \nu}^{\zeta \eta} \nabla_{\zeta} \nabla_{\eta} A^{\mu} = 0, && (\nu = 0,1,2,3),
\end{align}
where 

\begin{align} 
	h_{\mu \nu}^{\zeta \eta} \eqdef \frac{\partial^2 \Ldual}
	{\partial(\nabla_{\zeta} A^{\mu})\partial(\nabla_{\eta} A^{\nu})}.
\end{align}
Note that $h$ has a symmetry property that will be important for the construction of energies; it is invariant under the following simultaneous exchange of indices:

\begin{align} \label{E:hpotentialsymmetries}
	\mu & \leftrightarrow \nu,  && \zeta \leftrightarrow \eta.
\end{align}

Before discussing the difficulties that arise in the quasilinear case, let us recall the advantages of using the Lorenz gauge in the linear Maxwell-Maxwell theory. In this case, there is a well-known, remarkable simplification that occurs in Lorenz gauge: the equations \eqref{E:EquationsforA} can be written as a system of completely decoupled wave equations for the components of $A.$ That is, in Lorenz gauge, the components of $A$ are solutions to the following system:
%Consequently, it is ``natural'' to construct positive energies for the components of $\Far,$ but in contrast, there is no reason to expect that positive energies for the components of $\nabla A$ should exist in general. 

\begin{align} \label{E:Maxwelldecoupledintro}
	g_{\mu \nu} (g^{-1})^{\zeta \eta} \nabla_{\zeta} \nabla_{\eta} A^{\mu} = 0, && (\nu = 0,1,2,3).
\end{align}
Consequently, we have that

\begin{align} \label{E:SeparableCase}
		h_{\mu \nu}^{\zeta \eta} = g_{\mu \nu} (g^{-1})^{\zeta \eta}.
\end{align}
Because of the full decoupling at the quasilinear level, we can multiply both sides of \eqref{E:Maxwelldecoupledintro} by the ``seemingly non-geometric'' quantity $\nabla_0 A_{\nu}$ (with the index $\nu$ downstairs!), integrate over $\mathbb{R}^3,$ and integrate by parts\footnote{These steps can alternatively be carried out using an energy current framework, similar to the energy current estimate \eqref{E:Indefiniteenergytimederivative} described below.} to show that the following energy $\mathcal{E}$ is conserved for solutions to \eqref{E:Maxwelldecoupledintro}:

\begin{align} \label{E:Maxwellconservedquantity}
	\mathcal{E}^2(t) \eqdef \frac{1}{2} 
		\sum_{\zeta, \eta = 0}^4 \int_{\mathbb{R}^3} \big(\nabla_{\zeta} A_{\eta}(t,\ux) \big)^2 \, d^3 \ux.
\end{align}
In the language of \cite{dC2000}, the special structure of $h_{\mu \nu}^{\zeta \eta}$ in \eqref{E:SeparableCase} is called
\emph{separability}; the existence of the conserved \emph{coercive} quantity \eqref{E:Maxwellconservedquantity} is because of this additional structure, \emph{which is not typically present in the equations of a quasilinear theory}.

Let us contrast this to the case of the MBI system (or any other quasilinear perturbation of 
linear Maxwell-Maxwell theory derivable from a covariant Lagrangian). In the case of the MBI system in Lorenz gauge, it can be shown using \eqref{E:Lorenz} that the MBI equations can be written in such a way that

\begin{align}  \label{E:MBIhforAsplitting}
	h_{\mu \nu}^{\zeta \eta} = g_{\mu \nu} (g^{-1})^{\zeta \eta} + \widetilde{h}_{\mu \nu}^{\zeta \eta},
\end{align}
where $\widetilde{h}_{\mu \nu}^{\zeta \eta},$ which has the symmetry property \eqref{E:hpotentialsymmetries},
is a term that is of quadratic order in $\nabla A$ in the small-solution regime. 
The corresponding system of PDEs is therefore

\begin{align} \label{E:MBIintegrationbypartsproblem}
	g_{\mu \nu} (g^{-1})^{\zeta \eta} \nabla_{\zeta} \nabla_{\eta} A^{\mu} 
	+ \widetilde{h}_{\mu \nu}^{\zeta \eta} \nabla_{\zeta} \nabla_{\eta} A^{\mu}
	= 0, && (\nu = 0,1,2,3).
\end{align}
Unfortunately, in general, it is not possible to simply multiply both sides of \eqref{E:MBIintegrationbypartsproblem} by 
$\nabla_0 A_{\nu},$ integrate over $\mathbb{R}^3,$ and integrate by parts; the tensorfield in \eqref{E:MBIhforAsplitting} 
is not separable in general, nor in the particular case of the MBI system \textbf{Note that this difficulty does not arise in the study of a single quasilinear wave equation}; e.g., small quasilinear perturbations of the linear wave equation in Minkowski space preserve hyperbolicity and the availability of a basic $L^2$ energy estimate. 

%Multiplying by the more ``natural'' quantity $\nabla_t A^{\nu}$ and integrating by parts \emph{with the help of the symmetry property} \eqref{E:hpotentialsymmetries}, leads to the quantity \eqref{E:Indefiniteenergy}, which is indefinite in $\nabla A$ and which is discussed further just below. 

One may \emph{attempt} to resolve this difficulty by using the framework of \emph{energy currents} developed in \cite{dC2000}. 
A natural quantity that arises from an application of this framework is $I(t),$ which is defined by

\begin{align} \label{E:Indefiniteenergy}
	I(t) \eqdef \int_{\mathbb{R}^3} J_{(MBI+Lorenz)}^{0}(t,\ux) \, d^3 \ux,
\end{align}
where the \emph{energy current} $J_{(MBI+Lorenz)}^{\mu}$ is defined by

\begin{align} \label{E:JMBIintro}
	J_{(MBI+Lorenz)}^{\mu} & \eqdef - h_{\kappa \lambda}^{\mu \eta}(\nabla_{\eta} A^{\kappa})(\nabla_0 A^{\lambda})
		+ \frac{1}{2} \delta_0^{\mu} h_{\kappa \lambda}^{\zeta \eta} (\nabla_{\zeta} A^{\kappa})(\nabla_{\eta} 
		A^{\lambda}), && (\mu = 0,1,2,3).
\end{align}
As explained in detail in \cite{dC2000} and in Section \ref{S:CanonicalStress}, the current \eqref{E:JMBIintro} can be constructed by contracting a certain tensor, namely the canonical stress, against the vectorfield $\partial_t$. 
The details of this construction do not concern us here. We remark only that the quantity $I(t)$ is what one first tries to construct in an effort control solutions to the MBI system, i.e., during a proof of local well-posedness. On the one hand, it can be shown that $\frac{d}{dt} I(t)$ has one of the properties that is essential in order for it to be of use in analyzing the solution $\nabla A,$ namely that its time derivative can be bounded in terms of the $L^2$ norm of $\nabla A$ itself. More specifically, it can be shown that

\begin{align} \label{E:Indefiniteenergytimederivative}
	\frac{d}{dt} I(t) \leq C(\|\nabla h \|_{L^{\infty}}) \| \nabla A \|_{L^2}^2.
\end{align}
We remark that a quick way to see \eqref{E:Indefiniteenergytimederivative} is to use the equations \eqref{E:EquationsforA} and the symmetry property \eqref{E:hpotentialsymmetries} to show that $|\nabla_{\mu} J_{(MBI+Lorenz)}^{\mu}| \leq C(\|\nabla h \|_{L^{\infty}}) |\nabla A|^2;$ \eqref{E:Indefiniteenergytimederivative} then follows from the divergence theorem. Alternatively, one may multiply both sides of \eqref{E:MBIintegrationbypartsproblem} by the ``geometric'' quantity $\nabla_0 A^{\nu}$ (with the $\nu$ index upstairs!) and integrate by parts \emph{with the help of the symmetry property} \eqref{E:hpotentialsymmetries}, arriving at \eqref{E:Indefiniteenergytimederivative}.

However, we quickly run into a difficulty: \eqref{E:MBIhforAsplitting} and \eqref{E:JMBIintro} imply that in the small-solution regime, $J_{(MBI+Lorenz)}^{0}$ is \emph{indefinite} in $\nabla A:$

\begin{align} \label{E:J0MBIintro}
	J_{(MBI+Lorenz)}^{0} & = \frac{1}{2}\sum_{\zeta = 0}^3 g_{\kappa \lambda} (\nabla_{\zeta}A^{\kappa}) (\nabla_{\zeta}A^{\lambda})
		+ O(|\nabla A|^4).
\end{align}
Therefore, $I(t)$ is not a coercive quantity, and in particular, it is of no use in controlling the $L^2$ norms of solutions to \eqref{E:MBIintegrationbypartsproblem}.

These difficulties are not fatal in the sense that the fundamental unknown is the Faraday tensor $\Far = dA,$ and as explained in Section \ref{S:NormsandEnergies}, we can construct suitable positive energies by working directly with $\Far.$ 
More specifically, our energies control the combinations $\nabla_{\mu} A_{\nu} - \nabla_{\nu} A_{\mu}$ of any four-potential
satisfying $\Far = dA,$ but they do not control the individual components $\nabla_{\mu} A_{\nu}.$ Furthermore, there is an important advantage to working directly with $\Far:$ \emph{our conditions for global existence depend only on physical quantities, and not on auxiliary mathematical quantities such as the values of a four-potential} $A.$ We remark that we do not know whether or not an alternate argument\footnote{The authors in \cite{dChH2003} claim to have overcome these difficulties, but their chain of reasoning in going from equation \cite[Eqn. (3.8)]{dChH2003} to equation \cite[Eqn. (3.11)]{dChH2003} is difficult to follow; in particular, in equation \cite[Eqn. (3.10)]{dChH2003}, it is not clear whether the $\nu$ index for the four-potential $A$ is supposed to be ``upstairs'' or ``downstairs,'' a distinction which is essential for establishing the positivity of their energies.} could produce a positive quantity that is suitable for controlling the $L^2$ norm of the components $\nabla_{\mu} A_{\nu}$ of solutions $A$ to the MBI system (or, for that matter, any other covariant, quasilinear, non-separable system of electromagnetic equations derivable from a Lagrangian) in Lorenz gauge. That is to say, it is not clear whether or not the hyperbolicity of the equations is visible at the level of the components of $\nabla A$ in Lorenz gauge. If the answer turns out to be ``no,'' then this would mean that the Lorenz gauge is not viable. We believe that this is an interesting question worthy of further investigation.

\subsection{Comments on the analysis} \label{SS:Commentsonanalysis}
In this section, we summarize the main ideas of our proof. We first remark that all of the discussion in this section assumes that we have fixed an inertial coordinate system $(t,\ux)$ on $M,$ which is a global coordinate system in which the spacetime metric has the components $g_{\mu \nu} = \mbox{diag}(-1,1,1,1).$ Throughout this article, we work directly with the Faraday tensor $\Far,$ and thus avoid the aforementioned difficulties associated with choosing a gauge for the four-potential. To analyze $\Far,$ we use the framework of \cite{dCsK1990} and decompose it into its \emph{Minkowski null components}. Before discussing the notion of null components, we first introduce the following foliations of Minkowski space: the family of \emph{ingoing Minkowski null cones} $C_{s}^- \eqdef \lbrace (\tau,\uy) \mid |\uy| + \tau = s \rbrace;$ the family of \emph{outgoing Minkowski null cones} $C_{q}^+ \eqdef \lbrace (\tau,\uy) \mid |\uy| - \tau = q \rbrace;$ and the spacelike hypersurfaces $\Sigma_t \eqdef \lbrace (\tau,\uy) \mid \tau = t \rbrace,$ which intersect the null cones in spheres $S_{r,t} \eqdef \lbrace (\tau,\uy) \mid \tau = t, |\uy| = r \rbrace.$ All of these families of surfaces, will play an important role in this article. 

Now in a neighborhood of each spacetime point, there exists a \emph{null frame} $\lbrace \uL, L, e_1, e_2 \rbrace,$ where  $\uL \eqdef \partial_t - \partial_r$ is an ingoing geodesic vectorfield tangent to the corresponding cone $C_{s}^-,$ 
$L \eqdef \partial_t + \partial_r$ is an outgoing geodesic vectorfield tangent to the corresponding cone $C_{q}^+$ normalized by the condition $g(\uL,L) = - 2,$ and the orthonormal vectorfields $e_1,$ $e_2$ are tangent to the corresponding sphere $S_{r,t}$ and normal to both $\uL,$ $L.$ At each point $p$ where it is defined, the null frame forms a basis for $T_p M.$ The null components of the two-form $\Far$ are then defined to be the following pair of one-forms $\ualpha= \ualpha[\Far],$ $\alpha = \alpha[\Far]$ tangent to the $S_{r,t},$ and the following two scalar quantities 
$\rho = \rho[\Far],$ $\sigma = \sigma[\Far]:$

\begin{align*}
	\ualpha_A & = \Far_{A \uL}, \\
	\alpha_A & = \Far_{AL}, \\
	\rho & = \frac{1}{2} \Far_{\uL L}, \\
	\sigma & = \Far_{12},
\end{align*}
where we have abbreviated $\Far_{A \uL} = e_A^{\kappa} \uL^{\lambda}\Far_{\kappa \lambda}, \ \Far_{12} = e_1^{\kappa} e_2^{\lambda} \Far_{\kappa \lambda},$ etc; see Section \ref{S:Decompositions} for more details. 

\subsubsection{Linear analysis}

The following decay properties, which can be expressed with the help of the \emph{null coordinates} $q \eqdef r - t,$ $s \eqdef r + t,$ where $r \eqdef |\ux|,$ were shown in \cite{dCsK1990} for solutions to the linear Maxwell-Maxwell system
arising from data with suitable decay properties at infinity\footnote{The finiteness of $\| (\mathring{\Magneticinduction}, \mathring{\Displacement}) \|_{H_1^3}$ is sufficient, where $(\mathring{\Magneticinduction}, \mathring{\Displacement})$ is the 
electromagnetic decomposition of the data for $\Far$ described in Section \ref{SS:electromagneticdecomposition},
and the weighed Sobolev norm $H_1^3$ is defined in Definition \ref{D:HNdeltanorm}.}

\begin{itemize}
	\item  The worst decaying component is $\ualpha,$ which decays like $(1 + s)^{-1} (1 + |q|)^{-3/2}.$
	\item The fastest decaying component is $\alpha,$ which decays like $(1 + s)^{-5/2}.$
	\item $\rho$ and $\sigma$ each decay at the intermediate rate $(1 + s)^{-2} (1 + |q|)^{-1/2}.$
	\item Any derivative tangential to the outgoing cones $C_{q}^+$ (i.e. $\nabla_L, \nabla_{e_A}$) 
	creates additional decay of order $(1 + s)^{-1},$ while the transversal derivative $\nabla_{\uL}$ creates additional
	decay of order $(1 + |q|)^{-1},$ which is weaker than $(1 + s)^{-1}.$
\end{itemize}

In Section \ref{S:GlobalExistence} (see also Proposition \ref{P:GlobalSobolev}), we show that small-data solutions to the 
MBI system have exactly the same decay properties. Since the analysis of the linear theory also plays a key role in our
analysis of the MBI system, we first discuss the basic strategy for establishing the aforementioned decay of solutions to the linear Maxwell-Maxwell system

\begin{align} \label{E:Maxwellintro}
	d \Far & = 0, && d \Fardual = 0.
\end{align}
We recall that for any two-form $\Far,$ the corresponding Maxwell-Maxwell energy momentum tensor is 
\begin{align}
	\EMT_{\mu \nu}^{(Maxwell)} & \eqdef \Far_{\mu}^{\ \kappa} \Far_{\nu \kappa} - \frac{1}{4} g_{\mu \nu}
	\Far_{\kappa \lambda}\Far^{\kappa \lambda},
\end{align}
and that if $\Far$ is a solution of \eqref{E:Maxwellintro}, then $\nabla_{\mu} \EMT_{\ \nu}^{\mu(Maxwell)} = 0,$
$(\nu = 0,1,2,3).$
Furthermore, using the timelike conformal Killing field\footnote{Recall that a conformal Killing field is a vectorfield $X$ that satisfies $\nabla_{\mu} X_{\nu} + \nabla_{\nu} X_{\mu} = \phi_X g_{\mu \nu}$ for some scalar-valued function $\phi_X.$} $\overline{K},$ which has components $\overline{K}^{0} = 1 + t^2 + (x^1)^2 + (x^2)^2 + (x^3)^2, \overline{K}^j = 2t x^j,$ $(j = 1,2,3),$ we can construct the \emph{energy current} $J_{(Maxwell)}^{\mu} = - \EMT_{\ \kappa}^{\mu (Maxwell)} \overline{K}^{\kappa}.$ Because $\EMT_{\mu \nu}^{(Maxwell)}$ is symmetric, $(g^{-1})^{\kappa \lambda} \EMT_{\kappa \lambda}^{(Maxwell)} = 0,$ and $\overline{K}$ is a conformal Killing field, it thus follows that

\begin{align} \label{E:divJ0intro}
	\nabla_{\mu} J_{(Maxwell)}^{\mu} = 0.
\end{align}
Additionally, $\EMT_{\mu \nu}^{(Maxwell)}$ has the following positivity property: for every pair of future-directed causal vectors $X,Y,$ we have that $\EMT_{\kappa \lambda}^{(Maxwell)} X^{\kappa} Y^{\lambda} \geq 0.$ In particular, choosing 
$X^{\mu} \eqdef \delta_0^{\mu},$ $Y^{\mu} \eqdef \overline{K}^{\mu},$ it can be shown that (see Lemma \ref{L:CanonicalStressErrorTermExpansion})

\begin{align} \label{E:J0intro}
	J_{(Maxwell)}^0 = \frac{1}{2} \Big\lbrace(1 + q^2)|\ualpha|^2 + (1+s^2)|\alpha|^2 + (2 + q^2 + s^2)(\rho^2 + \sigma^2) \Big\rbrace,
\end{align}
where $\ualpha,$ $\alpha,$ $\rho,$ and $\sigma$ are the null components of $\Far.$ If we define the energy $\mathcal{E} \geq 0$ by

\begin{align} \label{E:Energyintro}
	\mathcal{E}^2(t) \eqdef \int_{\mathbb{R}^3} \, J_{(Maxwell)}^0(t,\ux) d^3 \ux,
\end{align}
then it follows from \eqref{E:divJ0intro} and the divergence theorem that $\mathcal{E}(t)$ is constant in time if it is initially finite:

\begin{align} \label{E:Energyconstantintro}
	\frac{d}{dt} \big(\mathcal{E}^2(t) \big) = 0.
\end{align} 

The various weights in \eqref{E:J0intro} are the first hint that different null components of $\Far$ have different 
$L^{\infty}$ decay properties. However, the full proof of decay requires that we commute the Maxwell equations with various conformal Killing fields and apply the \emph{global Sobolev inequality}. Let us explain what we mean by this. Given any solution $\Far$ of \eqref{E:Maxwellintro} and any conformal Killing field $Z,$ it can be shown that $\Lie_Z \Far$ 
is also a solution to the linear Maxwell-Maxwell equations. Here, $\Lie_Z \Far$ is the Lie derivative of $\Far$ with respect to the vectorfield $Z.$ Iterating this process, we conclude that $\Lie_{\mathcal{Z}}^I \Far$ is a solution, where $I$ is a multi-index, and $\Lie_{\mathcal{Z}}^I$ is shorthand notation for iterated Lie derivatives with respect to vectorfields $Z \in \mathcal{Z}.$ In this article, the relevant set of conformal Killing fields $\mathcal{Z}$ consists of: the $4$ translations 
$T_{(\mu)} \eqdef \partial_{\mu},$ $(\mu=0,1,2,3);$ the $3$ rotations $\Omega_{(jk)} \eqdef x_j \partial_k - x_k \partial_j,$
$(1 \leq j < k \leq 3);$ the $3$ Lorentz boosts $\Omega_{(0j)} \eqdef - t \partial_j - x_j \partial_t,$ $(j=1,2,3);$
and the scaling vectorfield $S \eqdef x^{\kappa} \partial_{\kappa}.$ Furthermore, as in \eqref{E:Energyconstantintro}, the weighted $L^2$ norms of the various null components of the $\Lie_{\mathcal{Z}}^I \Far$ are constant in time. Now in order to derive $L^{\infty}$ decay, we need to connect these weighted $L^2$ norms of $\Lie_{\mathcal{Z}}^I \Far$ to weighted $L^{\infty}$ norms of $\Far.$ This is exactly what the global Sobolev inequality provides; see Proposition \ref{P:GlobalSobolev} for the details. 

Let us also discuss the heuristic mechanism for the fact that derivatives tangential to the $C_q^+$ (i.e. $\nabla_L,$ $\nabla_{e_A}$) have better decay properties than the transversal derivative $\nabla_{\uL}.$  As examples, we consider the outgoing vectorfield $L \eqdef \partial_t + \partial_r$ and the ingoing vectorfield $\uL \eqdef \partial_t - \partial_r,$ where $\partial_r$ denotes the radial derivative. Simple algebraic computations lead to the identities

\begin{align} \label{E:heuristic}
	L =  \frac{S - \omega^a \Omega_{(0a)}}{s},
	&& \uL = - \frac{S + \omega^a \Omega_{(0a)}}{q},
\end{align}
where $\omega^i \eqdef x^i/r,$ and $q,s$ are the null coordinates mentioned above. Therefore, if we have achieved good control of the quantities $\nabla_{S} \Far$ and $\nabla_{\Omega_{(0i)}} \Far,$
then the formulas \eqref{E:heuristic} suggest that we can achieve even better control of the outgoing derivative $\nabla_L \Far,$ because of the favorable denominator $s^{-1}.$ On the other hand, the transversal derivative $\nabla_{\uL} \Far$ features a less favorable denominator $q^{-1}.$ More specifically, the $q^{-1}$ term fails to provide decay in the ``wave zone'' $r \approx t,$ while in the entire region $\lbrace t \geq 0 \rbrace$ we have that $s^{-1} \leq \min \lbrace t^{-1}, r^{-1} \rbrace;$ i.e., decay in $s$ implies decay in $r$ and $t.$ 

\subsubsection{Nonlinear analysis} \label{SSS:NonlinearAnalysis}
We now outline the key differences between the proof of decay of solutions to the linear Maxwell-Maxwell equations, and the proof of the global existence of and decay of solutions to the MBI system in the small-data regime. To analyze solutions to the MBI system, the ``working form'' of which is given below in \eqref{E:modifieddFis0summary} - \eqref{E:HmodifieddMis0summary}, we will use the \emph{same Minkowski null decomposition }of the Faraday tensor described above. In particular, in order to derive our desired estimates, \textbf{we do not need to use the characteristic geometry of the MBI system}; in using the ``wrong'' Minkoswkian geometry (which has the advantage of relative simplicity), 
we are deviating from the correct MBI geometry \big(which is governed by the reciprocal Maxwell-Born-Infeld metric $(b^{-1})^{\mu \nu}$ defined in \eqref{E:MBImetric}\big) by small error terms that are controllable. Now like the Maxwell-Maxwell system, the MBI system has a corresponding divergence-free energy-momentum tensor $\EMT_{\mu \nu}^{(MBI)},$ which is given below in \eqref{E:MBItensorupper}; the availability of this tensor is a well-known consequence of the fact that the MBI Lagrangian $\Ldual_{(MBI)}$ (see \eqref{E:LMBI}) depends covariantly on only the metric $g$ and the field variables $\Far.$ This tensor can be used in conjunction with the vectorfield $\overline{K}$ to estimate the weighted $L^2$ norm of the solutions $\Far$ to the MBI system. However, to estimate the weighted $L^2$ norm of $\Lie_{\mathcal{Z}}^I \Far,$ we need a different tensor, which is described by Christodoulou in detail in \cite{dC2000}: the \emph{canonical stress} $\Stress_{\ \nu}^{\mu},$ which is defined below in \eqref{E:Hstress}. Now in the case of the linear Maxwell-Maxwell system, the canonical stress corresponding to the $\Lie_{\mathcal{Z}}^I \Far$ coincides with the energy momentum tensor $\EMT_{\ \nu}^{\mu(Maxwell)}$ constructed out of the $\Lie_{\mathcal{Z}}^I \Far,$ but in a general nonlinear theory, the two tensors differ. The important point is that the $\Lie_{\mathcal{Z}}^I \Far$ are solutions to the linearized equations \eqref{E:EOVBianchi} - \eqref{E:EOVMBI}, which are derivable from a \emph{linearized Lagrangian} $\dot{\mathscr{L}}$ 
(see \eqref{E:LinearizedLagrangian}) depending on the metric $g,$ the linearized variables $\dot{\Far} \eqdef \Lie_{\mathcal{Z}}^I \Far,$ \emph{and also the background} $\Far.$ Although the dependence of $\dot{\mathscr{L}}$ on the background precludes the availability of a divergence-free tensor for the linearized system, we may nevertheless use Christodoulou's framework to construct the tensor $\Stress_{\ \nu}^{\mu}.$  Although $\Stress_{\mu \nu}$ is in general not even symmetric, nor is $\Stress_{\ \nu}^{\mu}$ divergence-free, the role that $\Stress_{\ \nu}^{\mu}$ plays in the analysis of the linearized equations is roughly analogous to the role played by the energy momentum tensor in the original equations: 
$\nabla_{\mu} \Stress_{\ \nu}^{\mu},$ though non-zero, is of lower order (in terms of the number of derivatives), and furthermore, $\Stress_{\ \nu}^{\mu}$ possesses some positivity properties under contractions against certain pairs $(\xi,X)$ of timelike (covectors, vectors).

Once we have $\Stress_{\ \nu}^{\mu},$ we can again use the vectorfield $\overline{K}$ construct energies 
$\mathcal{E}_N[\Far(t)],$ which are a sum over $|I|\leq N$ of the energy of $\Lie_{\mathcal{Z}}^I \Far,$ that are analogous to the energies \eqref{E:Energyintro} defined in the Maxwell-Maxwell case; the precise definition is given in \eqref{E:mathcalENdef} below. However, in the MBI system, $\mathcal{E}_N[\Far(t)]$ is not constant. Additionally, in the nonlinear problem, the $q, s$ weighted factors appearing in the expression $\mathcal{E}_N[\Far(t)]$ are not manifestly uniform, but instead depend on the solution $\Far$ itself. For these reasons, it is convenient to introduce a norm $\Kintnorm \Far(t) \Kintnorm_{\Lie_{\mathcal{Z}};N}$ whose $q, s$ weights don't depend on $\Far;$ see \eqref{E:MorawetzWeightedLieDerivativeIntegralNormN}. In order to compare the two quantities, we establish inequality \eqref{E:EnergyNormEquivalence}, which shows that in the small-solution regime, $\mathcal{E}_N[\Far(t)] \approx \Kintnorm \Far(t) \Kintnorm_{\Lie_{\mathcal{Z}};N}.$ The crux of the global existence proof is the following: even though $\mathcal{E}_N[\Far(t)]$ is not constant, we are nevertheless able to derive an a-priori bound for $\Kintnorm \Far(t) \Kintnorm_{\Lie_{\mathcal{Z}};N}$ which shows that it remains uniformly small on any time interval of existence for the solution. According to the \emph{continuation principle} of Proposition \ref{P:LocalExistence}, such an a-priori bound for $\Kintnorm \Far(t) \Kintnorm_{\Lie_{\mathcal{Z}};N}$ implies global existence.

Now in order to estimate $\Kintnorm \Far(t) \Kintnorm_{\Lie_{\mathcal{Z}};N},$ we need to handle the numerous ``error'' terms
arising in the expression for $\frac{d}{dt}\big( \mathcal{E}_N^2[\Far(t)] \big).$ The first source of error terms was alluded to above, namely that the divergence of $\Stress_{\ \nu}^{\mu}$ is non-zero. The second source comes from the fact that the $\Lie_{\mathcal{Z}}^I \Far$ are solutions to the linearized equations \emph{with inhomogeneous terms}: many 
inhomogeneous ``error'' terms arise from commuting the operator $\Lie_{\mathcal{Z}}^I$ through the equation \eqref{E:HmodifieddMis0summary}; see Proposition \ref{P:Inhomogeneousterms}. This commuting is accomplished through the use of \emph{modified} Lie derivatives $\Liemod_{Z},$ which are equal to ordinary Lie derivatives plus a scalar multiple of the identity; see Definition \ref{D:Liemoddef} and Lemma \ref{L:LiemodZLiemodMaxwellCommutator}. A careful analysis of the 
special structure of the error terms (which are discussed in the next section), in conjunction with the global Sobolev inequality, leads to the a-priori estimate \eqref{E:energydifferentialinequality}, which is valid in the small-solution regime, and which we restate here for convenience:

\begin{align} \label{E:Norminequalityeintro}
	\Kintnorm \Far(t) \Kintnorm_{\Lie_{\mathcal{Z}};N}^2 \leq 
		C \Big\lbrace \Kintnorm \Far(0) \Kintnorm_{\Lie_{\mathcal{Z}};N}^2 
		+ \int_{\tau = 0}^t \frac{1}{1 + \tau^2} \Kintnorm \Far(\tau) \Kintnorm_{\Lie_{\mathcal{Z}};N}^2 \, d \tau \Big\rbrace.
\end{align}
Applying Gronwall's inequality to \eqref{E:Norminequalityeintro}, we thus conclude the desired result: 
$\Kintnorm \Far(t) \Kintnorm_{\Lie_{\mathcal{Z}};N}$ is globally bounded in time, if $\Kintnorm \Far(0) \Kintnorm_{\Lie_{\mathcal{Z}};N}$ is sufficiently small. In addition, we remark that the aforementioned decay properties of the solution $\Far$ are a by-product of the previous analysis. More specifically, the decay properties of $\Far$ follow directly from the global bound on $\Kintnorm \Far(t) \Kintnorm_{\Lie_{\mathcal{Z}};N}$ and the global Sobolev inequality (Proposition \ref{P:GlobalSobolev}).

We have a few final comments to make concerning the smallness of the data. The initial data consist of a pair of
one-forms $(\mathring{\Displacement},\mathring{\Magneticinduction})$ that are tangent to the Cauchy-hypersurface $\Sigma_0,$
and that satisfy the constraints (which are familiar from linear Maxwell-Maxwell theory) $\SigmafirstfundNabla_a \mathring{\Displacement}^a = \SigmafirstfundNabla_a \mathring{\Magneticinduction}^a = 0.$ Here, $\SigmafirstfundNabla$ denotes the Levi-Civita connection corresponding to the first fundamental form $\Sigmafirstfund$ of $\Sigma_0$ (see Section \ref{S:Geometry}). As is described in Section \ref{SS:electromagneticdecomposition}, $\Far(0)$ can be constructed out of $(\mathring{\Displacement},\mathring{\Magneticinduction}).$ However, the quantity $\Kintnorm \Far(0) \Kintnorm_{\Lie_{\mathcal{Z}};N},$ the smallness of which is required to close the global existence argument, involves derivatives of $\Far$ that are normal to $\Sigma_0$ (i.e., time derivatives). However, by repeatedly using an appropriate version of the MBI system, the normal derivatives of $\Far$ along $\Sigma_0$ can be written in terms of the tangential derivatives (i.e. spatial derivatives) of $(\mathring{\Displacement},\mathring{\Magneticinduction}).$ Consequently, as is explained in detail in Section \ref{SS:DataNorms}, it is possible to devise a smallness condition involving only the data $(\mathring{\Displacement},\mathring{\Magneticinduction})$ and their tangential derivatives, from which the smallness of $\Kintnorm \Far(0) \Kintnorm_{\Lie_{\mathcal{Z}};N}$ necessarily follows. This allows for a ``proper'' formulation of the small-data global existence condition of Theorem \ref{T:GlobalExistence} in terms of quantities inherent to the data.

\subsubsection{The error terms} \label{SS:ErrorTerms}
Let us now make a few remarks concerning the many error terms that arise in our study of $\frac{d}{dt}\big(\mathcal{E}_N^2[\Far(t)] \big),$ since the study of these error terms is at the heart of our analysis. In the small-solution regime, the MBI system is a \emph{cubic} quasilinear perturbation of the linear Maxwell-Maxwell system. It is well-known that for linear hyperbolic PDEs whose solutions possess the decay properties of solutions to the Maxwell-Maxwell system, cubic perturbations\footnote{We are assuming that the perturbations involve only $1$ or fewer derivatives, and that the perturbed system is also hyperbolic.} do not destroy the existence of small-data global solutions. In fact, a much shorter proof of small-data global existence could be provided by using the vectorfield $\partial_t$ in place of the vectorfield $\overline{K}$ in our construction of the energies. However, in order to show that small-data MBI solutions have the same decay properties as solutions to the linear Maxwell-Maxwell system, we make full use of the vectorfield $\overline{K},$ together with an algebraic property of the MBI system: its nonlinearities satisfy the \emph{null condition}. The null condition, which was first identified by Klainerman \cite{sK1986} in the context of nonlinear wave equations, is a collection of algebraic properties that are satisfied by special nonlinearities. Roughly speaking, when a nonlinearity satisfies the null condition, the worst kind of terms (from the point of view of decay) are not present. More specifically, in the case of the MBI system, the expression for $\frac{d}{dt}\big( \mathcal{E}_N^2[\Far(t)] \big)$ involves quartic terms in $\Far$ and its iterated Lie derivatives $\Lie_{\mathcal{Z}}^I \Far,$ multiplied by weights in $q$ and $s$ arising from the vectorfield $\overline{K}$ and its covariant derivative $\nabla \overline{K}.$ Because these terms are fourth order, we do not need to perform a fully detailed null decomposition in order to prove our desired estimates. That is, there is room for imprecision; we only prove estimates that are sufficient recover the full decay properties possessed by solutions to the linear Maxwell-Maxwell system. Let us summarize the version of the null condition that we show is satisfied by the MBI system (see Section \ref{S:NullFormEstimates} for the details): for those terms involving a weight of $s$ or $1 + s^2,$ at most two of the four factors correspond to the worst decaying components $\ualpha[\Lie_{\mathcal{Z}}^I \Far].$ It is also true that for those terms involving a weight of $q$ or $1 + q^2,$ at most three of the four factors correspond to the worst decaying components $\ualpha[\Lie_{\mathcal{Z}}^I \Far].$ However, we do not make use of the availability of the one ``good'' factor since our estimates close without it. 

\subsubsection{The large-data well-posedness of the MBI system}

Finally, we would like to make a few remarks about the local existence proof that is briefly sketched in Section \ref{S:IVP}. This result is interesting in itself because it shows the following fact, which is arguably not manifest: the MBI system's initial value problem is well-posed in every field-strength regime in which its Lagrangian is real-valued, i.e., in every regime in which the theory is well-defined. The crucial estimate in this regard is contained in Proposition \ref{P:LocalExistenceCurrent}, which shows that it is always possible to construct an energy current for the linearized equations with positivity properties that are sufficient to deduce local existence in the weighted Sobolev space of relevance for our global existence result. This fact is strongly related to the internal geometry of and the hyperbolicity of the MBI system (i.e., the characteristic subsets), which will be explored in detail in an upcoming article by the author and his collaborators. We remark that the energy current we use for deducing the local existence result is constructed by contracting the canonical stress against a suitable ``multiplier''
vectorfield $\Vmult,$ so that it differs from the current used in our small-data global existence proof; the vectorfield $\overline{K}$ may not be a suitable multiplier for deducing large-data local existence.

For an alternative proof of the large-data well-posedness of the MBI system's initial value problem, one may consult \cite{yB2004} (see also \cite{dS2004}). In this work, Brenier ``augments'' the MBI system by taking as his $10$ unknowns the non-trivial components of the electromagnetic quantities (see Section \ref{SS:electromagneticdecomposition}) $\Magneticinduction,$ $\Displacement,$ $P,$ and $h.$ Along the ``MBI submanifold,'' $P$ coincides with the Poynting vector ($P = \Magneticinduction \times \Displacement),$ and $h$ coincides with the $00$ component of the MBI energy-momentum tensor
$\EMT_{(MBI)}^{\mu \nu}$ (see \eqref{E:MBItensorupper}), but for general augmented MBI solutions, $P$ and $h$ are independent unknowns. To compensate, the additional evolution equations \cite[Eqn. (1.9)]{yB2004} and \cite[Eqn. (1.10)]{yB2004}, which are redundant for solutions belonging to the MBI submanifold, were added to the MBI system (i.e., so that there are $10$ equations for the $10$ unknowns). From the point of view of hyperbolicity, the most important feature of this augmented system is that the function $S(\Magneticinduction,\Displacement,P,h) \eqdef \frac{1 + |\Magneticinduction|^2 + |\Displacement|^2 + |P|^2}{h},$ which coincides with a constant multiple of the quantity $h = \EMT_{(MBI)}^{00}$ for solutions constrained to the MBI submanifold, satisfies the properties of a \emph{smooth, strictly convex entropy function} of the augmented variables $(\Magneticinduction,\Displacement,P,h)$. Thus, using the general framework of hyperbolic conservation laws (see e.g. \cite{cD2010}), it follows that there exists a change of state-space variables in which the augmented MBI system becomes symmetric hyperbolic. For symmetric hyperbolic systems, there exists a well-established
theory of well-posedness based on energy estimates (see e.g. \cite{rCdH1962}, \cite{cD2010}, \cite{kF1954}, \cite{pL2006}, \cite{aM1984}, \cite{dS1999}).

\subsection{Outline of the article}

The remainder of the article is organized as follows:

\begin{itemize}
	\item In Section \ref{S:Notation}, we collect together much of the notation that is
		introduced throughout the article.
	\item In Section \ref{S:Geometry}, we recall some basic facts from differential geometry.
	\item In Section \ref{S:MBI}, we provide a detailed introduction to the MBI system.
	\item In Section \ref{S:ConformalKilling}, we discuss the collections of Minkowski conformal Killing fields that play a role 
		in our analysis. We also introduce modified Lie derivatives, which have favorable commutation properties with the 
		MBI equations. 
	\item In Section \ref{S:Decompositions}, we introduce the null frame and the null decomposition of a tensor. We then decompose
		the MBI system relative to a null frame. We also introduce several electromagnetic decompositions of the Faraday and 
		Maxwell tensors.
	\item In Section \ref{S:Commutation}, we provide some commutation lemmas that will be used throughout the remainder of the 
		article, especially in Section \ref{S:GlobalSobolev}.
	\item In Section \ref{S:CanonicalStress}, we discuss the energy-momentum tensor associated to the MBI system and the 
		canonical stress tensor associated to the equations of variation. 
	\item In Section \ref{S:NormsandEnergies}, we introduce the norms, seminorms, and energies that will be used in the proof of 
		our main theorem.
	\item In Section \ref{S:NullFormEstimates}, we perform a partial null decomposition of the nonlinear error terms that
		appear in the expression for the time derivative of the energy. It is here that the null condition is revealed.
	\item In Section \ref{S:GlobalSobolev}, we discuss the global Sobolev inequality, which connects weighted $L^2$ bounds
		to weighted $L^{\infty}$ bounds.
	\item In Section \ref{S:EnergyEstimates}, we prove the a-priori bound \eqref{E:Norminequalityeintro}, which is the most 
		important inequality in the article.
	\item In Section \ref{S:IVP}, we briefly discuss local existence for the MBI system. We also recall the availability of 
		a continuation principle, which provides criteria for the existence of a global classical solution.
	\item In Section \ref{S:GlobalExistence}, we combine the results of Sections \ref{S:EnergyEstimates} and 
		\ref{S:IVP} in order to establish our main theorem.
	\end{itemize}

\section{Notation} \label{S:Notation}
In this section, we collect together for convenience much of the notation that is introduced throughout the article.
\\

\noindent \hrulefill
\ \\

\subsection{Constants}
We use the symbol $C$ to denote a generic \emph{positive} constant that is free to vary from line to line. In general, $C$ can depend on many quantities, but in the small-solution regime that we consider in this article, $C$ can be chosen uniformly. Sometimes it is illuminating to explicitly indicate one of the quantities $\mathfrak{Q}$ that $C$ depends on; we do by writing $C_{\mathfrak{Q}}$ or $C(\mathfrak{Q}).$ If $A$ and $B$ are two quantities, then we often write 
\begin{align*}
	A \lesssim B
\end{align*}
to mean that ``there exists a $C > 0$ such that $A \leq C B.$'' Furthermore, if $A \lesssim B$ and $B \lesssim A,$ then we 
often write

\begin{align*}
	A \approx B.
\end{align*}

\subsection{Indices} \label{SS:Indices}
\begin{itemize}
	\item Lowercase Latin indices $a,b,j,k,$ etc. take on the values $1,2,3.$
	\item Greek indices $\kappa, \lambda, \mu, \nu,$ etc. take on the values $0,1,2,3.$
	\item Uppercase Latin indices $A,B$ etc. take on the values $1,2$ and are used to enumerate
		the two orthogonal null frame vectors tangent to the spheres $S_{r,t}.$
	\item Indices are lowered and raised with the Minkowski metric $g_{\mu \nu}$ and its inverse $(g^{-1})^{\mu \nu}.$
	\item Repeated indices are summed over.
	\item We sometimes use parentheses to distinguish indices that are labels from coordinate indices; e.g., 
		the ``$0$'' in	$T_{(0)}$ is a labeling index.
\end{itemize}

\subsection{Coordinates}

\begin{itemize}
	\item $\lbrace x^{\mu} \rbrace_{\mu = 0,1,2,3}$ are the \emph{spacetime coordinates};
		\textbf{in our fixed inertial coordinate system only}, we use the notation $t \eqdef x^0,$ $\ux = (x^1,x^2,x^3).$
	\item Relative to our inertial coordinate system, $r \eqdef |\ux| \eqdef \sqrt{(x^1)^2 + (x^2)^2 + (x^3)^2}$
		denotes the radial coordinate.
	\item $q \eqdef r-t,$ $s \eqdef r+t$ are the \emph{null coordinates}.
\end{itemize}

\subsection{Surfaces}
	Relative to the inertial coordinate system:
\begin{itemize}
	\item $C_{s}^- \eqdef \lbrace (\tau,\uy) \mid |\uy| + t = s \rbrace$
		are the \emph{ingoing null cones}.
	\item $C_{q}^+ \eqdef \lbrace (\tau,\uy) \mid |\uy| - t = q \rbrace$ are the \emph{outgoing null cones}. 
	\item $\Sigma_t \eqdef \lbrace (\tau,\uy) \mid \tau = t \rbrace$ are the constant time slices.
	\item $S_{r,t} \eqdef \lbrace (\tau,\uy) \mid \tau = t, |\uy| = r \rbrace$ are the Euclidean spheres.
\end{itemize}

\subsection{Metrics and volume forms}
\begin{itemize}
	\item $g$ denotes the standard Minkowski metric on $\mathbb{R}^{1+3};$ in our fixed inertial coordinate system,
		$g_{\mu \nu} = \mbox{diag}(-1,1,1,1).$
	\item $\emet$ denotes the standard Euclidean metric on $\mathbb{R}^{1+3};$
		in our fixed coordinate system, $\emet_{\mu \nu} = \mbox{diag}(1,1,1,1).$
	\item $\emet^{-1}$ denotes the inverse of the standard Euclidean metric on $\mathbb{R}^{1+3};$
		in our fixed inertial coordinate system, $(\emet^{-1})^{\mu \nu} = \mbox{diag}(1,1,1,1).$
	\item $\Sigmafirstfund$ denotes the first fundamental form of $\Sigma_t;$ in our fixed inertial coordinate system, 
		$\Sigmafirstfund_{\mu \nu} = \mbox{diag}(0,1,1,1).$
	\item $\angg$ denotes the first fundamental form of $S_{r,t};$ relative to an arbitrary coordinate system, \\
		$\angg_{\mu \nu} = g_{\mu \nu} + \frac{1}{2}\big(L_{\mu} \uL_{\nu} + \uL_{\mu} L_{\nu} \big),$
		where $\uL, L$ are defined in Section \ref{SS:NullFrames}.
\end{itemize}

\begin{itemize}
	\item $\epsilon_{\mu \nu \kappa \lambda} = |\mbox{det}g|^{1/2} [\mu \nu \kappa \lambda]$ denotes the
		\emph{volume form} of $g;$ $[0123]= 1 = -[1023],$ etc.
	\item $\epsilon^{\mu \nu \kappa \lambda} = - |\mbox{det}g|^{-1/2} [\mu \nu \kappa \lambda].$ 
	\item $\uepsilon_{\nu \kappa \lambda} = \epsilon_{\mu \nu \kappa \lambda} (T_{(0)})^{\mu}$ denotes the volume form of 
		$\Sigmafirstfund,$ where $T_{(0)}$ is defined in Section \ref{SS:Killingnotation}.
	\item $\angepsilon_{\mu \nu} = \frac{1}{2}\epsilon_{\mu \nu \kappa \lambda}\uL^{\kappa} L^{\lambda}$
		denotes the volume form of $\angg.$
\end{itemize}

\subsection{Hodge duality} \label{SS:Hodge}
	For an arbitrary two-form $\Far_{\mu \nu}:$
\begin{itemize}
	\item $\Fardual_{\mu \nu} = \frac{1}{2} g_{\mu \widetilde{\mu}} g_{\nu \widetilde{\nu}} 
		\epsilon^{\widetilde{\mu} \widetilde{\nu} \kappa \lambda} \Far_{\kappa \lambda}
		= - \frac{1}{2} |\mbox{det} \, g|^{-1/2} g_{\mu \widetilde{\mu}} g_{\nu \widetilde{\nu}}
		[\widetilde{\mu} \widetilde{\nu} \kappa \lambda] \Far_{\kappa \lambda}$ denotes the Hodge dual
		of $\Far_{\mu \nu}$ with respect to the spacetime metric $g.$
\end{itemize}

\subsection{Derivatives}
\begin{itemize}
	\item In an arbitrary coordinate system $\lbrace x^{\mu} \rbrace_{\mu=0,1,2,3},$ 
		$\partial_{\mu} = \frac{\partial}{\partial x^{\mu}},$ $\nabla_{\mu} = \nabla_{\frac{\partial}{\partial x^{\mu}}}.$
	\item $\nabla$ denotes the Levi-Civita connection corresponding to $g.$
	\item $\SigmafirstfundNabla$ denotes the Levi-Civita connection corresponding to $\Sigmafirstfund.$
	\item $\angn$ denotes the Levi-Civita connection corresponding to $\angg.$
	\item In our fixed inertial coordinate system, $\partial_r = \omega^a \partial_a$ denotes the radial 
		derivative, where $\omega^j = x^j/r.$
	\item $\nabla_X$ denotes the differential operator $X^{\kappa} \nabla_{\kappa}.$
	\item If $X$ is tangent to $\Sigma_t,$ then $\SigmafirstfundNabla_X$ denotes the differential operator $X^{\kappa} 
		\SigmafirstfundNabla_{\kappa}.$
	\item If $X$ is tangent to $S_{r,t},$ then $\angn_X$ denotes the differential operator $X^{\kappa} 
		\angn_{\kappa}.$
	%\item $\partial_s \eqdef \frac{1}{2}(\partial_t + \partial_r), \partial_q \eqdef \frac{1}{2}(\partial_t - \partial_r)$
		%denote the null derivatives.
	\item $\nabla_{(n)}U$ denotes the $n^{th}$ covariant derivative tensorfield of the tensorfield $U.$
	\item $\SigmafirstfundNabla_{(n)}U$ denotes the $n^{th}$ $\Sigma_t-$intrinsic covariant derivative tensorfield of a tensorfield $U$	
		tangent to the hypersurfaces $\Sigma_t.$
	\item $\angn_{(n)}U$ denotes the $n^{th}$ $S_{r,t}-$intrinsic covariant derivative tensorfield of a tensorfield $U$
		tangent to the spheres $S_{r,t}.$
	\item $\udiv U = \Sigmafirstfund_{\kappa}^{\ \lambda} \nabla_{\lambda} U^{\kappa}$ denotes the intrinsic divergence of
		a vectorfield $U$ tangent to the hypersurfaces $\Sigma_t.$
	\item $(\ucurcl U)^{\nu} = \uepsilon_{\ \ \ \lambda}^{\nu \kappa} \nabla_{\kappa} U^{\lambda}$ are the 
		components of the intrinsic divergence of a vectorfield $U$ tangent to the hypersurfaces $\Sigma_t.$
	\item $\angdiv U = \angg_{\kappa}^{\ \lambda} \nabla_{\lambda} U^{\kappa}$ denotes the intrinsic divergence of
		a vectorfield $U$ tangent to the spheres $S_{r,t}.$
	\item $\angcurl U =\angepsilon_{\ \lambda}^{\kappa} \nabla_{\kappa} U^{\lambda}$
		denotes the intrinsic curl of a vectorfield $U$ tangent to the spheres $S_{r,t}.$
	\item $\Lie_X$ denotes the Lie derivative with respect to the vectorfield $X.$
	\item $[X,Y]^{\mu} = (\Lie_X Y)^{\mu} = X^{\kappa}\partial_{\kappa}Y^{\mu} - Y^{\kappa}\partial_{\kappa} X^{\mu}$
		denotes the Lie bracket of the vectorfields $X$ and $Y.$
	\item For $Z \in \mathcal{Z},$ $\Liemod_Z = \Lie_Z + 2c_Z$ denotes the \emph{modified Lie derivative},
		where the constant $c_Z$ is defined in Section \ref{SS:Killingnotation}.
	\item $\Lie_{\mathcal{A}}^I U,$ and $\Liemod_{\mathcal{A}}^I U,$ $\nabla_{\mathcal{A}}^I U$ respectively 
		denote an $|I|^{th}$ order iterated Lie, iterated modified Lie, and iterated covariant derivative of the tensorfield $U$ 
		with respect to vectorfields belonging to the set $\mathcal{A};$ $\angn_{\mathcal{O}}^I U$ is an iterated intrinsic (to the 
		spheres $S_{r,t}$) covariant derivative of $U$ with respect to rotation vectorfields.
\end{itemize}

\subsection{Minkowski conformal Killing fields} \label{SS:Killingnotation} \ \\
Relative to the inertial coordinate system $\lbrace x^{\mu} \rbrace_{\mu=0,1,2,3} = (t,\ux):$
\begin{itemize}
	\item $T_{(\mu)} = \partial_{\mu},$ $(\mu=0,1,2,3),$ denotes a translation vectorfield.
	\item $\Omega_{(jk)} = x_j \partial_k - x_k \partial_j,$ $(1 \leq j < k \leq 3),$ denotes a rotation vectorfield.
	\item $\Omega_{(0j)} = -t \partial_j - x_j \partial_t,$ $(j=1,2,3),$ denotes a Lorentz boost vectorfield.
	\item $S = x^{\kappa} \partial_{\kappa}$ denotes the scaling vectorfield.
	\item $K_{(\mu)} = - 2 x_{\mu} S + g_{\kappa \lambda}x^{\kappa} x^{\lambda} \partial_{\mu},$
		$(\mu=0,1,2,3)$ denotes an acceleration vectorfield.
	\item $\overline{K} = K_{(0)} + T_{(0)}$ denotes the Morawetz vectorfield.
	\item $\mathscr{T} = \lbrace T_{(\mu)} \rbrace_{0 \leq \mu \leq 3}.$
	\item $\mathcal{O} = \lbrace \Omega_{(12)}, \Omega_{(13)}, \Omega_{(23)} \rbrace.$
	\item $\mathcal{Z} = \lbrace T_{(\mu)}, \Omega_{(\mu \nu)}, S \rbrace_{1 \leq \mu < \nu \leq 3}.$
	\item $\mathbf{T}, \mathbf{O},$ and $\mathbf{Z}$ are the Lie algebras generated by
		$\mathscr{T}, \mathcal{O},$ and $\mathcal{Z}$ respectively.
	\item For $Z \in \mathcal{Z},$ $^{(Z)}\pi_{\mu \nu} = \nabla_{\mu} Z_{\nu} + \nabla_{\nu} Z_{\mu} = c_Z g_{\mu \nu},$	
		denotes the deformation tensor of $Z,$ where $c_Z$ is a constant. 
	\item Commutation properties with the Maxwell-Maxwell term 
		$\big[(g^{-1})^{\mu \kappa} (g^{-1})^{\nu \lambda} - (g^{-1})^{\mu \lambda} (g^{-1})^{\nu \kappa}\big] \nabla_{\mu} 
		\Far_{\kappa \lambda}:$ \\
		$\Liemod_{\mathcal{Z}}^I \Big\lbrace \big[(g^{-1})^{\mu \kappa} (g^{-1})^{\nu \lambda} - (g^{-1})^{\mu \lambda} 
			(g^{-1})^{\nu \kappa} \big] \nabla_{\mu} \Far_{\kappa \lambda} \Big\rbrace
		= \big[(g^{-1})^{\mu \kappa} (g^{-1})^{\nu \lambda} - (g^{-1})^{\mu \lambda} (g^{-1})^{\nu \kappa}\big] 
			\nabla_{\mu} \Lie_Z^I \Far_{\kappa \lambda}.$
\end{itemize}

\subsection{Null frames} \label{SS:NullFrames}
\begin{itemize}
	\item $\uL \eqdef \partial_t - \partial_r$ denotes the null vectorfield generating the $C_{s}^-$ and transversal to the $C_q^+.$
	\item $L \eqdef \partial_t + \partial_r$ denotes the null vectorfield generating the $C_q^+.$
	\item $e_A, \ A = 1,2$ denotes orthonormal vectorfields spanning the tangent space of the spheres $S_{r,t}.$  
	\item The set $\mathcal{L} \eqdef \lbrace L \rbrace$ contains only $L.$
	\item The set $\mathcal{T} \eqdef \lbrace L, e_1, e_2 \rbrace$ denotes the frame vector fields tangent
		to the $C_q^+.$
	\item The set $\mathcal{U} \eqdef \lbrace \uL,L, e_1, e_2 \rbrace$ denotes the entire null frame.
\end{itemize}

\subsection{Null frame decomposition}

\begin{itemize}
	\item For an arbitrary vectorfield $X$ and frame vector $U \in \mathcal{U},$ we define
		$X_U \eqdef X_{\kappa} U^{\kappa},$ where $X_{\mu} \eqdef g_{\mu \kappa} X^{\kappa}.$
	\item For an arbitrary vectorfield $X = X^{ \kappa}\partial_{\kappa} = X^{L} L + X^{\uL} \uL
		+ X^{A} e_A,$ where \\
		$X^{L} = - \frac{1}{2}X_{\uL},$ $X^{\uL} = - \frac{1}{2}X_{L},$ $X^A = X_A,$ $X_A \eqdef X_{e_A}.$
	\item For an arbitrary pair of vectorfields $X,Y:$ \\
	$g(X,Y) = X^{\kappa}Y_{\kappa} = -\frac{1}{2}X_{L}Y_{\uL} - \frac{1}{2}X_{\uL}Y_{L} + \delta^{AB}X_A Y_B.$
\end{itemize}

If $\Far$ is any two-form, its null components are

\begin{itemize}
	\item $\ualpha_{\mu} = \angg_{\mu}^{\ \nu} \Far_{\nu \lambda} \uL^{\lambda}.$
	\item $\alpha_{\mu} = \angg_{\mu}^{\ \nu} \Far_{\nu \lambda} L^{\lambda}.$
	\item $\rho = \frac{1}{2} \Far_{\lambda \kappa}\uL^{\kappa} L^{\lambda}.$
	\item $\sigma = \frac{1}{2} \angepsilon^{\kappa \lambda} \Far_{\kappa \lambda}.$
\end{itemize}

\subsection{Null Forms}

For arbitrary two-forms $\Far, \Gar:$

\begin{itemize}
	\item $\mathcal{Q}_{(1)}(F,G) = \Far^{\kappa \lambda} \Gar_{\kappa \lambda}
		= - \delta^{AB} \ualpha_A[\Far]\alpha_B[\Gar]
		- \delta^{AB} \ualpha_A[\Gar] \alpha_B[\Far] - 2 \rho[\Far]\rho[\Gar]
		+ 2\sigma[\Far] \sigma[\Gar].$
	\item $\mathcal{Q}_{(2)}(\Far, \Gar) = \Fardual^{\kappa \lambda} \Gar_{\kappa \lambda} 
			= \angepsilon^{AB} \ualpha_A[\Far] \alpha_B[\Gar] + \angepsilon^{AB} \ualpha_A[\Gar] \alpha_B[\Far] 
			- 2 \sigma[\Far]\rho[\Gar] - 2\rho[\Far] \sigma[\Gar].$
\end{itemize}

\subsection{Electromagnetic decompositions}
Given a two-form $\Far$ and its associated MBI Maxwell tensor \\
$\Max_{\mu \nu} = \ell_{(MBI)}^{-1}\big( \Fardual_{\mu \nu} + \Farinvariant_{(2)} \Far_{\mu \nu} \big),$ its electromagnetic 
components relative to an arbitrary coordinate system are
\begin{itemize}
	\item $\Electricfield_{\mu} = \Far_{\mu \kappa} T_{(0)}^{\kappa}.$ 
	\item $\Magneticinduction_{\mu} = - \Fardual_{\mu \kappa} T_{(0)}^{\kappa}.$ 
	\item $\Displacement_{\mu} = - \Maxdual_{\mu \kappa} T_{(0)}^{\kappa}.$
	\item $\Magneticfield_{\mu} = - \Max_{\mu \kappa} T_{(0)}^{\kappa}.$
	\item $\SFar_{\mu} = \Far_{\mu \kappa} S^{\kappa}.$
	\item $\SFardual_{\mu} = \Fardual_{\mu \kappa} S^{\kappa}.$
\end{itemize}

\subsection{Norms and energies}

For an arbitrary tensor $U$ of type $\binom{n}{m},$ and 
$\mathcal{A} \in \lbrace \mathscr{T}, \mathcal{O}, \mathcal{Z} \rbrace:$
\begin{itemize}
	\item $|U|^2 = |(\emet^{-1})^{\widetilde{\lambda}_1 \lambda_1 } \cdots (\emet^{-1})^{\widetilde{\lambda}_m \lambda_m } 
		\emet_{\widetilde{\kappa}_1 \kappa_1 } \cdots \emet_{\widetilde{\kappa}_n \kappa_n} 
			U_{\widetilde{\lambda}_1 \cdots \widetilde{\lambda}_m}^{\ \ \ \ \ \ \ \widetilde{\kappa}_1 \cdots \widetilde{\kappa}_n}
			U_{\lambda_1 \cdots \lambda_m}^{\ \ \ \ \ \ \ \kappa_1 \cdots \kappa_n}|.$
	\item $|U|_{\Lie_{\mathcal{A};N}}^2 = \sum_{|I| \leq N} |\Lie_{\mathcal{A}}^I U|^2.$ 
	\item $|U|_{\nabla_{\mathcal{A}};N}^2 = \sum_{|I| \leq N} |\nabla_{\mathcal{A}}^I U|^2.$
	\item $|U|_{\angn_{\mathcal{O}};N}^2 \eqdef \sum_{|I| \leq N} |\angn_{\mathcal{O}}^I U|^2.$
\end{itemize}

For an arbitrary type $\binom{0}{2}$ tensor $F,$ and $\mathcal{V}, \mathcal{W} \in \lbrace \mathcal{L}, \mathcal{T},\mathcal{U} \rbrace:$

\begin{itemize}
	\item $|F|_{\mathcal{V} \mathcal{W}} = \sum_{V \in \mathcal{V}, W \in \mathcal{W}} |V^{\kappa} W^{\lambda} 
		F_{\kappa \lambda}|.$
	\item $|\nabla F|_{\mathcal{V} \mathcal{W}} = \sum_{U \in \mathcal{U}, V \in \mathcal{V}, W \in \mathcal{W}} 
		|V^{\kappa} W^{\lambda} U^{\gamma} \nabla_{\gamma} F_{\kappa \lambda}|.$
\end{itemize}

For an arbitrary two-form $\dot{\Far}$ with null components $\dot{\ualpha},$ $\dot{\alpha},$ $\dot{\rho},$ $\dot{\sigma};$
and $\mathcal{A} \in \lbrace \mathscr{T}, \mathcal{O}, \mathcal{Z} \rbrace:$ 
\begin{itemize}
	\item $\Knorm \dot{\Far} \Knorm^2 =	(1 + q^2) |\dot{\ualpha}|^2 + (1 + s^2) |\dot{\alpha}|^2 + (2 + q^2 + s^2)(\dot{\rho}^2 + 
		\dot{\sigma}^2).$ 
	\item $\Knorm \dot{\Far} \Knorm_{\Lie_{\mathcal{A}};N}^2 \eqdef \sum_{|I| \leq N}	
		\Knorm \Lie_{\mathcal{A}}^I \dot{\Far} \Knorm^2.$  
	%\item $\Knorm \dot{\Far} \Knorm|_{\nabla_{\mathcal{A}};N}^2 = \sum_{|I| \leq N}	\Knorm \nabla_{\mathcal{A}}^I \dot{\Far} 
		%\Knorm^2.$	
	\item $\Kintnorm \dot{\Far}(t) \Kintnorm_{\Lie_{\mathcal{Z}};N}^2 \eqdef \int_{\mathbb{R}^3} \Knorm \dot{\Far}(t,\ux) 
		\Knorm_{\Lie_{\mathcal{Z}};N}^2 \, d^3 \ux.$
\end{itemize}

For an arbitrary tensorfield $U$ defined on the Euclidean space $\Sigma$ with Euclidean coordinate system $\ux:$

\begin{itemize}
	\item $\| U \|_{H_{\delta}^N}^2 = \sum_{n=0}^N \int_{\Sigma} (1 + |\ux|^2)^{(\delta + n)} 
		|\SigmafirstfundNabla_{(n)} U(\ux)|^2 \, d^3 \ux$ is a weighted Sobolev norm of $U.$ 
	\item $\| U \|_{C_{\delta}^N}^2 \eqdef \sum_{n=0}^N \sup_{\ux \in \Sigma} (1 + |\ux|^2 )^{(\delta + n)} 
		|\SigmafirstfundNabla_{(n)} U(\ux)|^2$ is a weighted pointwise norm of $U.$
\end{itemize}

For arbitrary two-forms $\Far$ and $\dot{\Far}:$

\begin{itemize}
	\item $H^{\mu \nu \kappa \lambda} \nabla_{\mu} \dot{\Far}_{\kappa \lambda},$ where $H^{\mu \nu \kappa \lambda}$
	depends on $\Far,$ is the principal term in the equations of variation \eqref{E:EOVMBI}.
	\item $\Stress_{\ \nu}^{\mu} = H^{\mu \zeta \kappa \lambda} \dot{\Far}_{\kappa \lambda} \dot{\Far}_{\nu \zeta}
			- \frac{1}{4} \delta_{\nu}^{\mu} H^{\zeta \eta \kappa \lambda} \dot{\Far}_{\zeta \eta} \dot{\Far}_{\kappa \lambda}$ is 
		the canonical stress tensor.
	\item $\dot{J}_{\Far}^{\mu}[\dot{\Far}] = - \Stress_{\ \nu}^{\mu} \overline{K}^{\nu}$ is the energy current (used
		during the proof of global existence) constructed from the variation $\dot{\Far},$ the background $\Far,$
		and the Morawetz-type vectorfield $\overline{K} = \frac{1}{2}\big\lbrace(1+s^2) L + (1 + q^2) \uL \big\rbrace.$
	\item $\mathcal{E}_N^2[\dot{\Far}(t)] = \sum_{|I| \leq N} \int_{\mathbb{R}^3} \dot{J}_{\Far}^0[\Lie_{\mathcal{Z}}^I 	
		\dot{\Far}] \, d^3 \ux$ is the square of the order $N$ energy of $\dot{\Far}.$
\end{itemize}

\subsection{Function spaces and the regularity of maps}
	\begin{itemize}
		\item $H_{\delta}^N$ is the set of all distributions $f$ such that $\| f \|_{H_{\delta}^N} < \infty.$
		\item $C_{\delta}^N$ is the set of all functions $f$ such that $\| f \|_{C_{\delta}^N} < \infty.$
		\item $C^k \big([0,T) \times \mathbb{R}^3\big)$ denotes the set of $k-$times continuously differentiable
			functions on $[0,T) \times \mathbb{R}^3.$
		\item If $X$ is a function space, then $C^k\big([0,T),X \big)$ denotes the set of $k-$times 
			continuously differentiable maps from $[0,T)$ to $X.$
	\end{itemize}

\section{Geometry} \label{S:Geometry}
In this section, we recall some basic facts from differential geometry that will be used throughout the article.

\noindent \hrulefill
\ \\

\subsection{Inertial coordinate systems, the spacetime metric, and the Riemannian metric}

In Minkowski space, there exists a family of global coordinate systems, which we refer to as \emph{inertial coordinate systems}, in which the metric $g_{\mu \nu}$ and its inverse $(g^{-1})^{\mu \nu}$ have the following components:

\begin{align}
	g_{\mu \nu} = (g^{-1})^{\mu \nu} = \mbox{diag}(-1,1,1,1).
\end{align}
It will be convenient to carry out calculations and to define various tensors relative to an inertial coordinate system. \textbf{Therefore, we fix a single inertial coordinate system $\lbrace x^{\mu} \rbrace_{\mu = 0,1,2,3}$ on Minkowski space. For the remainder of the article, when we decompose tensors with respect to an inertial frame, it will always be relative to the frame corresponding to this fixed inertial coordinate system.} When working in this coordinate system, we often use the abbreviations 

\begin{subequations}
\begin{align}
	x^0 & \eqdef t, & & \ux \eqdef (x^1,x^2,x^3), \\
	\partial_{\mu} & \eqdef \frac{\partial}{\partial x^{\mu}}, & & \partial_t \eqdef \partial_0 = T_{(0)}. 
\end{align}
\end{subequations}

We recall the following partitions of $T_p M$ and $T_p^* M$ induced by $g.$
\begin{definition}
	Vectors $X \in T_p M$ are classified as \emph{timelike}, \emph{null}, \emph{causal}, \emph{spacelike} as follows,
	where $g(X,X) \eqdef g_{\kappa \lambda}X^{\kappa} X^{\lambda}:$
	
	\begin{subequations}
	\begin{align}
		g(X,X) & < 0 && (\mbox{timelike}), \\
		g(X,X) & = 0 && (\mbox{null}), \\
		g(X,X) & \leq 0 && (\mbox{causal}), \\
		g(X,X) & > 0 && (\mbox{spacelike}).
	\end{align}
	\end{subequations}
	
	Furthermore, relative to our fixed inertial coordinate system, $X$ is classified as future-directed or past-directed
	as follows:
	
	\begin{subequations}
	\begin{align}
		X^0 & > 0 && (\mbox{future-directed}), \\
		X^0 & < 0 && (\mbox{past-directed}).
	\end{align} 
	\end{subequations}
	
	Covectors $\xi_{\mu}$ are defined to have the same classification as their metric dual $X^{\mu} \eqdef (g^{-1})^{\mu \kappa} 
	\xi_{\kappa}.$ We sometimes refer to $\xi$ as the $g-$dual of $X$ in order to emphasize that this notion of 
	duality depends on $g.$
\end{definition}

In order to measure the size of various tensor, it is convenient to introduce a Riemannian metric on $\mathbb{R}^4.$
A natural choice is the Euclidean metric $\emet,$ which has the following components relative to an arbitrary coordinate system:

\begin{align} \label{E:Riemanninanmetric}
	\emet_{\mu \nu} \eqdef g_{\mu \nu} + 2(T_{(0)})_{\mu} (T_{(0)})_\nu.
\end{align}
In the above formula, $T_{(0)}$ is the ``time translation'' vectorfield, which is defined to coincide with
$\partial_t$ in our inertial coordinate system. Therefore, relative to this coordinate system, 
the metric $\emet$ and its inverse $\emet^{-1}$ have the following components:

\begin{subequations}
\begin{align} \label{E:Riemannainmetricandinverseinertialframe}
	\emet_{\mu \nu} & = \mbox{diag}(1,1,1,1), \\
	(\emet^{-1})^{\mu \nu} & = \mbox{diag}(1,1,1,1).
\end{align}
\end{subequations}

We now define the aforementioned tensorial norm.

\begin{definition} \label{D:Riemanniannorm}
	If $U$ is a tensor of type $\binom{n}{m},$ then we define the norm $|\cdot | \geq 0$ of $U$ by
	
	\begin{align} \label{E:Riemanniannorm}
		|U|^2 = |(\emet^{-1})^{\widetilde{\lambda}_1 \lambda_1 } \cdots (\emet^{-1})^{\widetilde{\lambda}_m \lambda_m } 
			\emet_{\widetilde{\kappa}_1 \kappa_1 } \cdots \emet_{\widetilde{\kappa}_n \kappa_n} 
			U_{\widetilde{\lambda}_1 \cdots \widetilde{\lambda}_m}^{\ \ \ \ \ \ \ \widetilde{\kappa}_1 \cdots \widetilde{\kappa}_n}
			U_{\lambda_1 \cdots \lambda_m}^{\ \ \ \ \ \ \ \kappa_1 \cdots \kappa_n}|.
\end{align}
\end{definition}

\subsection{Lie derivatives and covariant derivatives}

Given any pair of vectorfields $X,Y,$ we recall that relative to an arbitrary coordinate system,
their \emph{Lie bracket} $[X,Y]$ can be expressed as

\begin{align} \label{E:bracket}
	[X,Y]^{\mu} & = X^{\kappa} \partial_{\kappa} Y^{\mu} - Y^{\kappa} \partial_{\kappa} X^{\mu}.
\end{align}
Furthermore, we have that 

\begin{align} \label{E:LieXY}
	\Lie_X Y = [X,Y], 
\end{align}	
where $\Lie$ denotes the \emph{Lie derivative operator}. Given a type $\binom{0}{m}$ tensorfield $U,$ and vectorfields $Y_{(1)}, \cdots Y_{(m)},$ the Leibniz rule for $\Lie$ implies that \eqref{E:LieXY} generalizes as follows:

\begin{align} \label{E:Liederivativebracketexpression}
	(\Lie_X U)(Y_{(1)}, \cdots, Y_{(m)}) & = X \lbrace U(Y_{(1)}, \cdots, Y_{(m)}) \rbrace
		- \sum_{i=1}^m U(Y_{(1)}, \cdots, Y_{(i-1)}, [X,Y_{(i)}], Y_{(i+1)}, \cdots, Y_{(m)}).
\end{align}

\begin{remark} \label{R:LieDerivatives}
	The Lie derivative operator does not commute with the raising and lowering of indices via the metric $g.$ Thus, in order to 
	avoid confusion, we use the convention that Lie derivatives are applied to two-forms $\Far_{\mu \nu}$ with both 
	indices down. In particular, the quantity $\Lie_Z \Far$ is understood to be a two-form with the indices down, and
	$\Lie_Z \Far^{\mu \nu} \eqdef (g^{-1})^{\mu \kappa} (g^{-1})^{\nu \lambda} \Lie_Z \Far_{\kappa \lambda}.$ 
\end{remark}

There is a unique affine connection $\nabla,$ which is known as the \emph{Levi-Civita connection}, that is torsion-free and compatible with the metric $g.$ These properties are equivalent to the requirement that the following identities hold for all vectorfields $X,Y,Z:$

\begin{align} 
	\nabla_X Y - \nabla_Y X & = [X,Y], \label{E:Torsionfree} \\
	\nabla_X \lbrace g(Y,Z) \rbrace & = g(\nabla_X Y,Z) + g(X,\nabla_Y Z). \label{E:derivativeofgiszero}
\end{align}
Furthermore, given a type $\binom{0}{m}$ tensorfield $U,$ and vectorfields $Y_{(1)}, \cdots, Y_{(m)},$ the Leibniz rule implies that
\begin{align} \label{E:CovariantLeibniz}
	(\nabla_X U)(Y_{(1)}, \cdots, Y_{(m)}) = X \lbrace U(Y_{(1)}, \cdots, Y_{(m)}) \rbrace
		- \sum_{i=1}^m U(Y_{(1)}, \cdots, Y_{(i-1)}, \nabla_X Y_{(i)}, Y_{(i+1)}, \cdots, Y_{(m)}).
\end{align}
We remark that relative to an arbitrary coordinate system, \eqref{E:derivativeofgiszero} is equivalent to

\begin{align} \label{E:CovariantDerivativeofgisZeroIndices}
	\nabla_{\lambda} g_{\mu \nu} & = 0, & & (\lambda, \mu, \nu = 0,1,2,3).
\end{align}
Furthermore, in our inertial coordinate system on Minkowski space, if $U$ is any type $\binom{n}{m}$ tensorfield, then $\nabla_{\mu} U_{\mu_1 \cdots \mu_m}^{\ \ \ \ \ \ \ \ \nu_1 \cdots \nu_n} = \partial_{\mu} U_{\mu_1 \cdots \mu_m}^{\ \ \ \ \ \ \ \ \nu_1 \cdots \nu_n}.$ In the above formulas and throughout the article, we use the notation

\begin{align} \label{E:NablaXdef}
	\nabla_X \eqdef X^{\kappa} \nabla_{\kappa}.
\end{align}

The \emph{Riemann curvature} tensor $R(\cdot,\cdot)\cdot$ is defined by the requirement that the following identities hold
for all vectorfields $X,Y,Z:$

\begin{align} \label{E:Curvature}
	\nabla_X \nabla_Y Z - \nabla_Y \nabla_X Z & = R(X,Y)Z + \nabla_{[X,Y]} Z.
\end{align}
In Minkowski space, $R(X,Y)Z \equiv 0.$

The following standard lemma gives a convenient formula relating Lie derivatives and covariant derivatives.

\begin{lemma}\cite{rW1984}  \label{L:Liederivativeintermsofnabla}

Let $X$ be a vectorfield, and let $U$ be a tensorfield of type $\binom{n}{m}.$ Then $\Lie_X U$ can be expressed in terms of covariant derivatives of $U$ and $X$ as follows:

\begin{align} \label{E:Liederivativeintermsofnabla}
	(\Lie_X U)_{\mu_1 \cdots \mu_m}^{\ \ \ \ \ \ \ \ \nu_1 \cdots \nu_n} \eqdef
		(\nabla_X U)_{\mu_1 \cdots \mu_m}^{\ \ \ \ \ \ \ \ \nu_1 \cdots \nu_n} 
	& + U_{\kappa \mu_2 \cdots \mu_m}^{\ \ \ \ \ \ \ \ \nu_1 \cdots \nu_n}\nabla_{\mu_1}X^{\kappa} 
		+ \cdots + U_{\mu_1 \cdots \mu_{m-1} \kappa}^{\ \ \ \ \ \ \ \ \ \nu_1 \cdots \nu_n}\nabla_{\mu_m}X^{\kappa} \\
	& - U_{\mu_1 \cdots \mu_m}^{\ \ \ \ \ \ \ \ \kappa \cdots \nu_n} \nabla_{\kappa}X^{\nu_1}
	 	- \cdots - U_{\mu_1 \cdots \mu_m}^{\ \ \ \ \ \ \ \ \nu_1 \cdots \nu_{n-1} \kappa} \nabla_{\kappa}X^{\nu_n}. \notag
\end{align}

\end{lemma}
\hfill $\qed$

It follows that
\begin{align}
	\Lie_X g_{\mu \nu} = ^{(X)}\pi_{\mu \nu},
\end{align}
where 
\begin{align} \label{E:deformationdef}
	^{(X)}\pi_{\mu \nu} \eqdef \nabla_{\mu} X_{\nu} + \nabla_{\nu} X_{\mu}
\end{align}
is the \emph{deformation tensor} of $X.$

\subsection{Volume forms and Hodge dual}

There is a canonical volume form $\epsilon_{\mu \nu \kappa \lambda}$ associated to the metric $g.$ 
Relative to any local coordinate system, we have that

\begin{align}
	\epsilon_{\kappa \lambda \mu \nu} & = |\mbox{det}(g)|^{1/2} [\kappa \lambda \mu \nu], \label{E:volumeform} \\
	\epsilon^{\kappa \lambda \mu \nu} & = - |\mbox{det}(g)|^{-1/2} [\kappa \lambda \mu \nu], \label{E:volumeforminverse}
\end{align}
where $[\kappa \lambda \mu \nu]$ is totally antisymmetric with normalization $[0123]= 1.$ It can be checked that 
the covariant derivative of the volume form vanishes:

\begin{align}	\label{E:volumeformcovariantlyconstant}
	\nabla_{\beta} \epsilon_{\kappa \lambda \mu \nu} & = 0, && (\beta, \kappa, \lambda, \mu, \nu = 0,1,2,3).
\end{align}

The \emph{Hodge dual} operator, which we denote by $\star,$ plays a fundamental role throughout our discussion. 

\begin{definition}
If $\Far$ is any two-form, then its Hodge dual $\Fardual$ is defined as follows: 

\begin{align}	\label{E:Fardualdef}
	\Fardual^{\mu \nu} & \eqdef \frac{1}{2} \epsilon^{\kappa \lambda \mu \nu} \Far_{\kappa \lambda}.
\end{align}

\end{definition}

\subsection{\texorpdfstring{$\Sigma_t, S_{r,t},$}{Constant time slices, Euclidean spheres,} 
and the first and second fundamental forms}

\begin{definition}

The following two classes of spacelike submanifolds of Minkowski space, which we define relative to the inertial coordinate system $\lbrace x^{\mu} \rbrace_{\mu=0,1,2,3}$ will play a role throughout the remainder of the article:
\begin{align}
	\Sigma_t & \eqdef \lbrace (\tau,\uy) \mid \tau = t \rbrace, \\
	S_{r,t} & \eqdef \lbrace (\tau,\uy) \mid \tau = t, |\uy| = r \rbrace,
\end{align}
where $|\uy|\eqdef \sqrt{(y^1)^2 + (y^2)^2 + (y^3)^2}.$ We refer to the $\Sigma_t$ as ``time slices,'' and the $S_{r,t}$ as ``spheres.'' 

\end{definition}

The future-directed normal to the $\Sigma_t$
is the time translation vectorfield $T_{(0)},$ while the $S_{r,t}$ have two linearly independent \emph{null} normals. We denote the one pointing in the ``outward'' direction by $L,$ and the one pointing in the ``inwards'' direction by $\uL.$ The vectorfields
$\uL$ and $L,$ which are defined on $M \slash 0,$ are unique up to multiplication by a scalar function. We choose the normalization so that they have the following components relative to our inertial coordinate system:

\begin{subequations}
\begin{align}
	\uL^{\mu} & = (1,-\omega^1,-\omega^2,-\omega^3), \label{E:uLdef}\\
	L^{\mu} & = (1,\omega^1,\omega^2,\omega^3), \label{E:Ldef} 
\end{align}
\end{subequations}
where $\omega^j = x^j/r.$ With $\partial_r \eqdef \frac{1}{r}x^a \partial_a$ 
denoting the radial vectorfield, $\uL,L$ can be expressed as

\begin{subequations}
\begin{align}
	\uL & = \partial_t - \partial_r, \label{E:underlineLradialdef} \\
	L & = \partial_t + \partial_r. \label{E:Lradialdef} 
\end{align}
\end{subequations}
We remark that beginning in Section \ref{SS:NullFrame}, $\uL$ and $L$ will play a key role in the 
\emph{null decomposition} of the MBI system.

We now recall the definitions of the first fundamental forms of $\Sigma_t$ and of $S_{r,t}.$
\begin{definition} \label{D:FirstFundamental}
		The \textbf{first fundamental forms} of $\Sigma_t, S_{r,t}$ are the Riemannian metrics on $\Sigma_t, S_{r,t}$
		respectively induced by the spacetime metric $g.$ In an arbitrary local coordinate system, $\Sigmafirstfund, \angg$ can be 
		expressed as follows:
		
	\begin{align} 
		\Sigmafirstfund_{\mu \nu} & \eqdef g_{\mu \nu} + (T_{(0)})_{\mu} (T_{(0)})_{\nu}, \\
		\angg_{\mu \nu} & \eqdef g_{\mu \nu} + \frac{1}{2}\big(\uL_{\mu} L_{\nu} + L_{\mu} \uL_{\nu} \big). \label{E:anggdef}
	\end{align}
	
\end{definition}
	We remark that the tensors $\Sigmafirstfund_{\mu}^{\ \nu} \eqdef \delta_{\mu}^{\nu} + (T_{(0)})_{\mu} (T_{(0)})^{\nu}$ and 
	$\angg_{\mu}^{\ \nu} \eqdef \delta_{\mu}^{\nu} + \frac{1}{2}\big(L_{\mu} \uL^{\nu} + 
	\uL_{\mu} L^{\nu} \big)$ orthogonally project onto $\Sigma_t$ and $S_{r,t}$ respectively. 
	Furthermore, the volume forms of $\Sigmafirstfund$ and $\angg,$ 
	which we respectively denote by $\uepsilon_{\nu \kappa \lambda}$ and $\angepsilon,$
	can be expressed as follows relative to an arbitrary coordinate system:
	
	\begin{align} 
		\uepsilon_{\nu \kappa \lambda} & = \epsilon_{\mu \nu \kappa \lambda} T_{(0)}^{\mu}, 
			\label{E:bargvolumeformdef} \\
		\angepsilon_{\mu \nu} & = \frac{1}{2}\epsilon_{\mu \nu \kappa \lambda}\uL^{\kappa} L^{\lambda}. \label{E:angvolumeformdef}
	\end{align}

\begin{definition}
	Let $U$ be a type $\binom{n}{m}$ spacetime tensor. We say that $U$ is tangent to the time slices $\Sigma_t$ if 
	
	\begin{align}
		U_{\mu_1 \cdots \mu_m}^{\ \ \ \ \ \ \ \nu_1 \cdots \nu_n} 
		= \Sigmafirstfund_{\mu_1}^{\ \widetilde{\mu}_1} \cdots \Sigmafirstfund_{\mu_m}^{\ \widetilde{\mu}_m} 
		\Sigmafirstfund_{\widetilde{\nu}_1}^{\ \nu_1}
		\cdots \Sigmafirstfund_{\widetilde{\nu}_n}^{\ \nu_1} U_{\widetilde{\mu}_1 \cdots \widetilde{\mu}_m}^{\ \ \ \ \ \ \ 
		\widetilde{\nu}_1 \cdots \widetilde{\nu}_n}.
	\end{align}
		Equivalently, $U$ is tangent to the $\Sigma_t$ if and only if any contraction of $U$ with $T_{(0)}$ results 
		in $0.$

	Similarly, we say that $U$ is tangent to the spheres $S_{r,t}$ if
	
	\begin{align}
		U_{\mu_1 \cdots \mu_m}^{\ \ \ \ \ \ \ \nu_1 \cdots \nu_n} 
		= \angg_{\mu_1}^{\ \widetilde{\mu}_1} \cdots \angg_{\mu_m}^{\ \widetilde{\mu}_m} \angg_{\widetilde{\nu}_1}^{\ \nu_1}
		\cdots \angg_{\widetilde{\nu}_n}^{\ \nu_1} U_{\widetilde{\mu}_1 \cdots \widetilde{\mu}_m}^{\ \ \ \ \ \ \ \widetilde{\nu}_1 
		\cdots \widetilde{\nu}_n}.
	\end{align}
	Equivalently, $U$ is tangent to the spheres $S_{r,t}$ if and only if any contraction of $U$ with either $\uL$ or $L$ results 
	in $0.$
	
\end{definition}

	We also recall the following relationships between the Levi-Civita connections $\SigmafirstfundNabla, \angn$ 
	corresponding to $\Sigmafirstfund,\angg$ and the Levi-Civita connection $\nabla$ corresponding to $g,$ which are valid for 
	any tensor $U$ of type $\binom{n}{m}$ tangent to the $\Sigma_t, S_{r,t}$ respectively:
	
	\begin{align} 
		\SigmafirstfundNabla_{\lambda} 
			U_{\mu_1 \cdots \mu_m}^{\ \ \ \ \ \ \ \nu_1 \cdots \nu_n} & = \Sigmafirstfund_{\lambda}^{\ \widetilde{\lambda}}
			\Sigmafirstfund_{\mu_1}^{\ \widetilde{\mu}_1} \cdots 
			\Sigmafirstfund_{\mu_m}^{\ \widetilde{\mu}_m} \Sigmafirstfund_{\widetilde{\nu}_1}^{\ \nu_1} \cdots 
			\Sigmafirstfund_{\widetilde{\nu}_n}^{\ \nu_n} 
			\nabla_{\widetilde{\lambda}} U_{\widetilde{\mu}_1 \cdots \widetilde{\mu}_m}^{\ \ \ \ 
			\ \ \ \widetilde{\nu}_1 \cdots \widetilde{\nu}_n}, 
			\label{E:SigmatIntrinsicintermsofExtrinsic} \\
		\angn_{\lambda} 
			U_{\mu_1 \cdots \mu_m}^{\ \ \ \ \ \ \ \nu_1 \cdots \nu_n} & = \angg_{\lambda}^{\ \widetilde{\lambda}}
			\angg_{\mu_1}^{\ \widetilde{\mu}_1} \cdots 
			\angg_{\mu_m}^{\ \widetilde{\mu}_m} \angg_{\widetilde{\nu}_1}^{\ \nu_1} \cdots \angg_{\widetilde{\nu}_n}^{\ \nu_n} 
			\nabla_{\widetilde{\lambda}}U_{\widetilde{\mu}_1 \cdots \widetilde{\mu}_m}^{\ \ \ \ 
			\ \ \ \widetilde{\nu}_1 \cdots \widetilde{\nu}_n}. \label{E:SphereIntrinsicintermsofExtrinsic}
	\end{align}

As in \eqref{E:NablaXdef}, throughout the article, we use the notation

\begin{align}
	\SigmafirstfundNabla_X & \eqdef X^{\kappa} \SigmafirstfundNabla_{\kappa}, && \mbox{(if $X$ is tangent to $\Sigma_t$)}, \\
	\angn_{X} & \eqdef X^{\kappa} \angn_{\kappa}, && \mbox{(if $X$ is tangent to $S_{r,t}$)}.
\end{align}	

We recall the definitions of the \emph{second fundamental form} of the $\Sigma_t,$ and
\emph{null second fundamental forms} of the $S_{r,t}.$

\begin{definition} \label{D:SecondFundamental}
	The \textbf{second fundamental form} of the hypersurface $\Sigma_t$ is defined to be the tensorfield
	
	\begin{align}
		\nabla_{\mu} (T_{(0)})_{\nu}.
	\end{align}

	The \textbf{null second fundamental forms} of the $S_{r,t}$ are defined to be the following pair of tensorfields:
	\begin{align}
		\nabla_{\mu} \uL_{\nu}, && \nabla_{\mu} L_{\nu}.
	\end{align}
\end{definition}

In the next lemma, we illustrate one of key properties of the second fundamental forms.

\begin{lemma} \label{L:SecondFundamentalFormsSymmetric}
	The second fundamental form $\nabla_{\mu} (T_{(0)})_{\nu}$ is a symmetric type $\binom{0}{2}$ tensorfield that is
	tangent to the time slices $\Sigma_t.$ Similarly, the null second fundamental forms $\nabla_{\mu} L_{\nu}, \nabla_{\mu} 
	\uL_{\nu}$ are symmetric type $\binom{0}{2}$ tensorfields that are tangent to the spheres $S_{r,t}.$
\end{lemma}

\begin{proof}
	The fact that $\nabla_{\mu} L_{\nu}, \nabla_{\mu} \uL_{\nu}$ are tangent to the $S_{r,t}$ follows from contracting
	them with the vectors $\uL,L,$ which form a basis for the orthogonal complement (in $M$) of the tangent space of $S_{r,t},$
	and using \eqref{E:LanduLaregeodesic} - \eqref{E:nablaLuLis0}. For the symmetry property, let $X,Y$ be vectorfields tangent 
	that are tangent to $S_{r,t}.$ Then $[X,Y]$ is also tangent to $S_{r,t}.$ Therefore, using the fact that $\nabla_{\mu} 
	L_{\nu}$ is tangent to the $S_{r,t},$ \eqref{E:Torsionfree}, and \eqref{E:derivativeofgiszero}, we deduce that
	
	\begin{align}
		X^{\mu} Y^{\nu}\nabla_{\mu} L_{\nu} = g(\nabla_X L,Y) & = \nabla_X \overbrace{g(L,Y)}^0 - g(L,\nabla_X Y) \\
		& = - g(L,\nabla_Y X) - \overbrace{g(L,[X,Y])}^0 \notag \\
		& = - \nabla_Y \overbrace{g(L,X)}^0 + g(\nabla_Y L, X) \notag \\
		& = g(\nabla_Y L, X) = Y^{\mu} X^{\nu} \nabla_{\mu} L_{\nu}. \notag
	\end{align}
	The proofs for $\nabla_{\mu} \uL_{\nu}$ and $\nabla_{\mu} (T_{(0)})_{\nu}$ are similar.
\end{proof}

\begin{remark}
	Lemma \ref{L:NullSecondFundamantalFormSimpleExpression} provides very simple expressions for the null
	second fundamental forms.
\end{remark}

\begin{remark}

By Lemma \ref{L:SecondFundamentalFormsSymmetric}, we have that $\angg_{\mu}^{\ \widetilde{\mu}} \nabla_{\widetilde{\mu}} L^{\nu}
= \nabla_{\mu} L^{\nu},$ and similarly for $\uL.$ Therefore, we sometimes use the abbreviations
${\angn_{\mu} L^{\nu}}$ $\eqdef \angg_{\mu}^{\ \widetilde{\mu}} \nabla_{\widetilde{\mu}} L^{\nu}$ and 
$\angn_{\mu} \uL^{\nu} \eqdef \angg_{\mu}^{\ \widetilde{\mu}} \nabla_{\widetilde{\mu}} \uL^{\nu},$ which should cause no confusion.
\end{remark}

To conclude this section, we recall the following basic facts concerning the metrics $\Sigmafirstfund$ and $\angg.$

\begin{lemma} \label{L:BargandAnggareCovariantlyConstant}
	Let $\Sigmafirstfund$ and $\angg$ be the first fundamental forms of $g$ defined in Definition \ref{D:FirstFundamental}.
	Let $\SigmafirstfundNabla,$ $\angn$ be their corresponding Levi-Civita connections, as defined in 
	\eqref{E:SigmatIntrinsicintermsofExtrinsic}, \eqref{E:SphereIntrinsicintermsofExtrinsic} respectively.
	Then
	
	\begin{align}
		\SigmafirstfundNabla_{\lambda} \Sigmafirstfund_{\mu \nu} & = 0, && (\lambda, \mu, \nu = 0,1,2,3), \\
		\angn_{\lambda} \angg_{\mu \nu} & = 0, && (\lambda, \mu, \nu = 0,1,2,3).
		\label{E:AnggIntrinsicDerivativeis0}
	\end{align}
	
\end{lemma}

\begin{proof}
	Lemma \ref{L:BargandAnggareCovariantlyConstant} follows from the expressions \eqref{E:SigmatIntrinsicintermsofExtrinsic} - 
	\eqref{E:SphereIntrinsicintermsofExtrinsic} and Lemma \ref{L:SecondFundamentalFormsSymmetric}.
\end{proof}

\section{The Maxwell-Born-Infeld System} \label{S:MBI}
\ \\

In this section, we first discuss the equations of motion for a generic covariant theory of classical electromagnetism that is derivable from a Lagrangian. We then introduce the Maxwell-Born-Infeld Lagrangian and derive several versions of its equations of motion. The final version, namely equations \eqref{E:modifieddFis0summary} - \eqref{E:HmodifieddMis0summary}, will be the one we use throughout the remainder of the article.

\noindent \hrulefill
\ \\

\subsection{The Lagrangian formulation of nonlinear electromagnetism}

In this section, we recall some facts from classical nonlinear electromagnetic field theory in a Lorentzian
manifold $(M,g)$ of signature $(-,+,+,+).$ We restrict our attention to theories of nonlinear electromagnetism derivable from a Lagrangian $\mathscr{L}.$ The fundamental quantity in such a theory is the \emph{Faraday tensor} $\Far_{\mu \nu},$ a two-form (i.e., an anti-symmetric tensorfield) 
that is postulated to be closed:

\begin{align} \label{E:Farisclosed}
	d \Far = 0,
\end{align}
where $d$ denotes the exterior derivative operator. This equation, which is the first of two equations that will define a particular nonlinear theory, is known as the \emph{Faraday-Maxwell law}. In local coordinates, it can be expressed in the following two ways
\begin{subequations}
\begin{align}
	\partial_{[\lambda} \Far_{\mu \nu]} & = 0, && (\lambda, \mu, \nu = 0,1,2,3) \\
	\nabla_{[\lambda} \Far_{\mu \nu]} & = 0, && (\lambda, \mu, \nu = 0,1,2,3),
\end{align}
\end{subequations}
where $[\cdots]$ denotes antisymmetrization.

In any covariant theory of classical electromagnetism, $\Ldual$ is a scalar-valued function of the two invariants of $\Far,$ which we denote by $\Farinvariant_{(1)}$ and $\Farinvariant_{(2)};$ i.e., $\Ldual = \Ldual \big(\Farinvariant_{(1)}[\Far], \Farinvariant_{(2)}[\Far] \big).$
They can be expressed in the following ways:

\begin{subequations}
\begin{align}
	\Farinvariant_{(1)} & = \Farinvariant_{(1)}[\Far] \eqdef \frac{1}{2} (g^{-1})^{\kappa \mu} (g^{-1})^{\lambda \nu} 
		\Far_{\kappa \lambda} \Far_{\mu \nu} = - {^{\star \hspace{-.025in}}( \Far \wedge \Fardual)} =|\Magneticinduction|^2 - 
		|\Electricfield|^2, \label{E:firstinvariant} \\
	\Farinvariant_{(2)} & = \Farinvariant_{(2)}[\Far] 
		\eqdef \frac{1}{4} (g^{-1})^{\kappa \mu} (g^{-1})^{\lambda \nu} \Far_{\kappa \lambda} \Fardual_{\mu \nu}
		= \frac{1}{8} \epsilon^{\kappa \lambda \mu \nu} \Far_{\kappa \lambda} \Far_{\mu \nu}
		= \frac{1}{2} {^{\star \hspace{-.025in}} (\Far \wedge \Far)} = \Electricfield_{\kappa} \Magneticinduction^{\kappa}, \label{E:secondinvariant}
\end{align}
\end{subequations}
where $\wedge$ denotes the wedge product, and $\Electricfield, \Magneticinduction$ are the
electromagnetic one-forms defined in Section \ref{SS:electromagneticdecomposition}. As we will discuss in Section \ref{SS:NullForms}, the invariants $\Farinvariant_{(1)}$ and $\Farinvariant_{(2)},$ viewed as quadratic forms in $\Far,$ have a special algebraic structure. More specifically, we will see that from the point of view of the decay estimates of Proposition \ref{P:GlobalSobolev}, the worst possible quadratic terms are absent from $\Farinvariant_{(1)}$ and $\Farinvariant_{(2)}.$ This is the one of the fundamental reasons that small-data solutions to the MBI system have the same decay properties as solutions to the linear Maxwell-Maxwell equations.

We now introduce the \emph{Maxwell tensor} $\Max$, a two tensor whose Hodge dual $\Maxdual$ is defined by

\begin{align} \label{E:Maxdualdef}
	\Maxdual^{\mu \nu} & \eqdef \frac{\partial \Ldual}{\partial \Far_{\mu \nu}}.
\end{align}	
Furthermore, we postulate that $\Max$ is closed: 
\begin{align} \label{E:Maxwelltensorclosed}
	d \Max = 0.
\end{align}
Taken together, \eqref{E:Farisclosed} and \eqref{E:Maxwelltensorclosed} are the equations of motion for the theory
arising from the Lagrangian $\mathscr{L}.$ We remark that \eqref{E:Farisclosed} and \eqref{E:Maxwelltensorclosed} are 
respectively equivalent to

\begin{subequations}
\begin{align} 
	\nabla_{\mu} \Fardual^{\mu \nu} & = 0, && (\nu = 0,1,2,3), \label{E:BianchiEM} \\
	\nabla_{\mu} \Maxdual^{\mu \nu} & = 0, && (\nu = 0,1,2,3). \label{E:Euler-Lagrange}
\end{align}
\end{subequations}
Equations \eqref{E:BianchiEM} are sometimes referred to as the \emph{Bianchi identities}. We furthermore remark
that the solutions to \eqref{E:Farisclosed}, \eqref{E:Maxwelltensorclosed} are exactly the stationary points (under
closed variations $d \dot{\Far} = 0$ with support contained in compact subsets $\mathfrak{C}$) of the action functional
\begin{align}
	\mathcal{A}_{\mathfrak{C}}[\Far] \eqdef \int_{\mathfrak{C} \Subset M} 
	\Ldual\big(\Farinvariant_{(1)}[\Far], \Farinvariant_{(2)}[\Far] \big) \, d \mu_g,
\end{align}
and that \eqref{E:Euler-Lagrange} are the Euler-Lagrange equations of $\Ldual.$ In the above formula, 
$d \mu_g \eqdef |\mbox{det} g|^{1/2} d^4 x$ is the measure associated to the spacetime volume form \eqref{E:volumeform}.

The Euler-Lagrange equations \eqref{E:Euler-Lagrange} can be written in the following form:

\begin{align} \label{E:Euler-LagrangeRewritten}
	h^{\mu \nu \kappa \lambda} \nabla_{\mu} \Far_{\kappa \lambda} = 0, && (\nu = 0,1,2,3),
\end{align}
where
\begin{align} \label{E:littlehdef}
	h^{\mu \nu \kappa \lambda} = - \frac{1}{2} \frac{\partial^2 \Ldual}{\partial \Far_{\mu \nu}
		\partial \Far_{\kappa \lambda}}.
\end{align}
The tensorfield $h^{\mu \nu \kappa \lambda},$ which has the properties

\begin{subequations}
\begin{align}
	h^{\nu \mu \kappa \lambda} & = - h^{\mu \nu \kappa \lambda}, \label{E:hminussignproperty1} \\
	h^{\mu \nu \lambda \kappa } & = - h^{\mu \nu \kappa \lambda}, \label{E:hminussignproperty2} \\
	h^{\kappa \lambda \mu \nu } & = h^{\mu \nu \kappa \lambda}, \label{E:hsymmetryproperty}
\end{align}
\end{subequations}
is of fundamental importance throughout this article. As is explained in Section \ref{S:CanonicalStress}, its algebraic and geometric properties are intimately related to the hyperbolic nature of the MBI system. In particular, the symmetry 
properties \eqref{E:hminussignproperty1} - \eqref{E:hsymmetryproperty} are needed to ensure that the \emph{canonical stress} tensor, which is defined in Section \ref{SS:Stress}, has lower-order divergence.

We state as a lemma the following identities, which will be used for various computations. We leave the proof
as an exercise for the reader.

\begin{lemma} \label{L:electromagneticidentities}
	The following identities hold:
	
	\begin{subequations}
	\begin{align}
		\frac{\partial |\mbox{det}g|}{\partial g_{\mu \nu}} & = |\mbox{det} \ g| (g^{-1})^{\mu \nu}, \\
		\frac{\partial (g^{-1})^{\kappa \lambda}}{g_{\mu \nu}} & = -(g^{-1})^{\kappa \mu} (g^{-1})^{\lambda \nu}, \\
		\Farinvariant_{(2)}^2 & = |\mbox{det} \ \Far| |\mbox{det} \ g|^{-1}, \\
		(g^{-1})^{\kappa \lambda} \Far_{\mu \kappa} \Far_{\nu \lambda} 
			- (g^{-1})^{\kappa \lambda} \Fardual_{\mu \kappa} \Fardual_{\nu \lambda} & = \Farinvariant_{(1)} g_{\mu \nu}, \\
		(g^{-1})^{\kappa \lambda }\Far_{\mu \kappa} \Fardual_{\nu \lambda} & = \Farinvariant_{(2)} g_{\mu \nu}, \\
			%\frac{\partial \Farinvariant_{(1)}}{\partial (\partial_{\mu} A_{\nu})} & = 2 \Far^{\mu \nu}, \\
		%\frac{\partial \Farinvariant_{(2)}}{\partial (\partial_{\mu} A_{\nu}) } & = \Fardual^{\mu \nu}, \\
		\frac{\partial \Farinvariant_{(1)}}{\partial g_{\mu \nu}} & = - g_{\kappa \lambda} \Far^{\mu \kappa}\Far^{\nu \lambda}, \\
		\frac{\partial \Farinvariant_{(2)}}{\partial g_{\mu \nu}} & = - \frac{1}{2} \Farinvariant_{(2)} (g^{-1})^{\mu \nu}, \\
		\Maxdual^{\mu \nu} & = 2\frac{\partial \Ldual}{\partial \Farinvariant_{(1)}} \Far^{\mu \nu}
			+ \frac{\partial \Ldual}{\partial \Farinvariant_{(2)}}\Fardual^{\mu \nu}, \\
		\frac{\partial \Farinvariant_{(1)}}{\partial \Far_{\mu \nu}} & = 2 \Far^{\mu \nu}, 
			\label{E:partialfirstinvariantpartialFarmunu} \\
		\frac{\partial \Farinvariant_{(2)}}{\partial \Far_{\mu \nu}} & = \Fardual^{\mu \nu}, \\
		\frac{\partial \Far^{\mu \nu}}{\partial \Far_{\kappa \lambda}} & 
			= (g^{-1})^{\mu \kappa} (g^{-1})^{\nu \lambda} - (g^{-1})^{\mu \lambda}(g^{-1})^{\nu \kappa}, \\
		\frac{\partial \Fardual^{\mu \nu}}{\partial \Far_{\kappa \lambda}} & = 
			\epsilon^{\mu \nu \kappa \lambda}. \label{E:partialFardualmunupartialFarkappalambda}
	\end{align}
	\end{subequations}

\end{lemma}

\subsection{Derivation of the MBI equations}

The Lagrangian for the MBI model is 
\begin{align} \label{E:LMBI}
	\Ldual_{(MBI)}  \eqdef \frac{1}{\upbeta^4} - \frac{1}{\upbeta^4} \big(1 + \upbeta^4 \Farinvariant_{(1)}[\Far] - \upbeta^8 
		\Farinvariant_{(2)}^2[\Far] \big)^{1/2} = \frac{1}{\upbeta^4} - \frac{1}{\upbeta^4} \big(\mbox{det}_g(g + \Far) \big)^{1/2},
\end{align}
where $\upbeta > 0$ denotes \emph{Born's ``aether'' constant.} \emph{For the remainder of the article, we set $\upbeta = 1$
for simplicity}; however, the analysis in the case $\upbeta \neq 1$ easily reduces to the case $\upbeta = 1$ by making change of variable $\widetilde{\Far} = \upbeta^2 \Far$ in the equations. For future use, we introduce the abbreviation

\begin{align} \label{E:ldef}
	\ell_{(MBI)} \eqdef \big(1 + \Farinvariant_{(1)} - \Farinvariant_{(2)}^2 \big)^{1/2}, 
\end{align}
which implies that 
\begin{align} \label{E:LMBIbetaequals1}
	\Ldual_{(MBI)} = 1 - \ell_{(MBI)}.
\end{align}
Using definition \eqref{E:Maxdualdef} and Lemma \ref{L:electromagneticidentities}, we compute that in the MBI model, $\Maxdual^{\mu \nu}$ can be expressed as follows:

\begin{align} \label{E:MaxdualMBI}
	\Maxdual^{\mu \nu} & = - \ell_{(MBI)}^{-1}\big( \Far^{\mu \nu} - \frac{1}{4} \Far_{\kappa \lambda} \Fardual^{\kappa \lambda} 
	\Fardual^{\mu \nu} \big) = - \ell_{(MBI)}^{-1}\big( \Far^{\mu \nu} - \Farinvariant_{(2)} \Fardual^{\mu \nu} \big).
\end{align}
Taking the Hodge dual of \eqref{E:MaxdualMBI}, we have that

\begin{align} \label{E:MaxMBI}
	\Max^{\mu \nu} & = \ell_{(MBI)}^{-1}\big( \Fardual^{\mu \nu} + \Farinvariant_{(2)} \Far^{\mu \nu} \big).
\end{align}
From \eqref{E:MaxdualMBI}, it follows that the Euler-Lagrange equations \eqref{E:Euler-Lagrange} for the MBI model are 

\begin{align} \label{E:MBIEuler-Lagrange}
	\nabla_{\mu} \Far^{\mu \nu} & - \frac{1}{4} \Fardual^{\mu \nu} \nabla_{\mu}(\Far_{\kappa \lambda}
		\Fardual^{\kappa \lambda}) - \frac{1}{2}\ell_{(MBI)}^{-2}\Big(\Far^{\mu \nu} - \frac{1}{4}\Far_{\kappa \lambda}
	\Fardual^{\kappa \lambda} \Fardual^{\mu \nu} \Big) \nabla_{\mu}\Big( \frac{1}{2} \Far_{\kappa \lambda} 
		\Far^{\kappa \lambda} - \frac{1}{16} (\Far_{\kappa \lambda} \Fardual^{\kappa \lambda})^2 \Big) = 0. 
\end{align}
Furthermore, it follows from \eqref{E:MBIEuler-Lagrange} that the tensorfield $h^{\mu \nu \kappa \lambda}$ from \eqref{E:littlehdef} can be expressed as

\begin{align} \label{E:MBIlittleh}
	h^{\mu \nu \kappa \lambda} & = \frac{1}{2} \bigg\lbrace \ell_{(MBI)}^{-1}\big[(g^{-1})^{\mu \kappa} (g^{-1})^{\nu \lambda} - 	
		(g^{-1})^{\mu \lambda} (g^{-1})^{\nu \kappa}\big] - \ell_{(MBI)}^{-3} \Far^{\mu \nu} \Far^{\kappa \lambda} 
		+ \Farinvariant_{(2)} \ell_{(MBI)}^{-3} \Big(\Far^{\mu \nu} \Fardual^{\kappa \lambda} + \Fardual^{\mu \nu} \Far^{\kappa \lambda} 
		\Big) \\
	& \hspace{2in} - \Big(\ell_{(MBI)}^{-1} + \Farinvariant_{(2)}^2 \ell_{(MBI)}^{-3} \Big)\Fardual^{\mu \nu} \Fardual^{\kappa \lambda} 
		- \ell_{(MBI)}^{-1}\Farinvariant_{(2)} \epsilon^{\mu \nu \kappa \lambda}  \bigg\rbrace \notag.
\end{align}

\subsection{$H^{\mu \nu \kappa \lambda}$ and the working version of the MBI equations} \label{SS:MBISummary}

To ease the calculations, it is convenient to perform two simple modifications of the tensorfield 
$h^{\mu \nu \kappa \lambda}$ from \eqref{E:MBIlittleh} obtaining a new tensorfield; the modifications
will not alter the set of solutions to the MBI system. First, we drop the $\ell_{(MBI)}^{-1}\Farinvariant_{(2)} \epsilon^{\mu \nu \kappa \lambda}$ term from \eqref{E:MBIlittleh}. This is permissible because its contribution to the Euler-Lagrange equations is
$\ell_{(MBI)}^{-1}\Farinvariant_{(2)} \epsilon^{\mu \nu \kappa \lambda} \nabla_{\mu} \Far_{\kappa \lambda} = 
2 \ell_{(MBI)}^{-1}\Farinvariant_{(2)} \nabla_{\mu} \Fardual^{\mu \nu} = 0,$ on account of equation \eqref{E:BianchiEM}.
Second, we multiply the remaining terms in \eqref{E:MBIlittleh} by $\ell_{(MBI)}.$ We denote the resulting tensorfield by $H^{\mu \nu \kappa \lambda}.$ Furthermore, it is convenient to split $H^{\mu \nu \kappa \lambda}$ into a main term, which coincides with the tensorfield in the case of the Maxwell-Maxwell equations, and a quadratic error term, which we denote by $H_{\triangle}^{\mu \nu \kappa \lambda}.$ The end result is 

\begin{subequations}
\begin{align}
	H^{\mu \nu \kappa \lambda} & \eqdef \ell_{(MBI)} \big(h^{\mu \nu \kappa \lambda} 
		+ \frac{1}{2}\ell_{(MBI)}^{-1}\Farinvariant_{(2)} \epsilon^{\mu \nu \kappa \lambda}\big) 
		=\frac{1}{2} \big[(g^{-1})^{\mu \kappa } (g^{-1})^{\nu \lambda} - (g^{-1})^{\mu \lambda} (g^{-1})^{\nu \kappa}\big] 
		+ H_{\triangle}^{\mu \nu \kappa \lambda}, \label{E:Hdef} \\
	H_{\triangle}^{\mu \nu \kappa \lambda} & \eqdef \frac{1}{2}\bigg\lbrace 
		- \ell_{(MBI)}^{-2} \Far^{\mu \nu} \Far^{\kappa \lambda} 
		+ \Farinvariant_{(2)} \ell_{(MBI)}^{-2} \Big(\Far^{\mu \nu} \Fardual^{\kappa \lambda} + \Fardual^{\mu \nu} \Far^{\kappa \lambda} 		\Big) - \Big(1 + \Farinvariant_{(2)}^2 \ell_{(MBI)}^{-2} \Big)\Fardual^{\mu \nu} \Fardual^{\kappa \lambda} \bigg\rbrace.
	 	\label{E:Htriangledef}
\end{align}
\end{subequations}
It follows that the system \eqref{E:BianchiEM}, \eqref{E:Euler-Lagrange}, \eqref{E:MaxdualMBI} is equivalent to the following version of the MBI system:

\begin{subequations}
\begin{align} 
	\nabla_{\lambda} \Far_{\mu \nu} + \nabla_{\mu} \Far_{\nu \lambda} + \nabla_{\nu} \Far_{\lambda \mu} & = 0,
		&& (\lambda,\mu,\nu = 0,1,2,3), \label{E:modifieddFis0summary} \\
	H^{\mu \nu \kappa \lambda} \nabla_{\mu} \Far_{\kappa \lambda} & = 0,
	&& (\nu = 0,1,2,3). \label{E:HmodifieddMis0summary}
\end{align}
\end{subequations}

\section{Conformal Killing Fields and Modified Lie Derivatives} \label{S:ConformalKilling}

In this section, we recall the definition of conformal Killing fields. This collection of vectorfields, which has the structure of a Lie algebra under the Lie bracket operator $X,Y \rightarrow [X,Y],$ comprises the generators of the conformal symmetries of the spacetime $(M,g).$ We focus on the case of Minkowski space, which has the maximum possible number of generators (15). In particular, we introduce several subsets of the Minkowski conformal Killing fields, each of which will play a role throughout the remainder of the article. More specifically, they appear in the definitions of the norms and energies (see Section \ref{S:NormsandEnergies}) that are used during our global existence argument. Finally, for a special collection of Minkowski conformal Killing fields $\mathcal{Z}$ we define \emph{modified} Lie derivatives $\Liemod_{\mathcal{Z}},$ which are equal to ordinary Lie derivatives plus a scalar multiple of the identity. This definition is justified by Lemma \ref{L:LiemodZLiemodMaxwellCommutator}, which shows that the operator $\Liemod_{\mathcal{Z}}$ has favorable commutation properties with the linear Maxwell-Maxwell equation  $\big[(g^{-1})^{\mu \kappa} (g^{-1})^{\nu \lambda} - (g^{-1})^{\mu \lambda} (g^{-1})^{\nu \kappa} \big] \nabla_{\mu} \Far_{\kappa \lambda} = 0.$ \\

\noindent \hrulefill
\ \\

\begin{definition}

A \emph{Killing field} of the metric $g_{\mu \nu}$ is a vectorfield $X$ such that
\begin{align}
	^{(X)}\pi_{\mu \nu} = 0,
\end{align}
while a \emph{conformal Killing field} $X$ satisfies
\begin{align} \label{E:conformalKilling}
	^{(X)}\pi_{\mu \nu} = \phi_X g_{\mu \nu}
\end{align}
for some scalar-valued function $\phi_X(t,\ux).$ In the above formulas, the deformation tensor
$^{(X)}\pi_{\mu \nu}$ is defined in \eqref{E:deformationdef}.

\end{definition}

The conformal Killing fields of the Minkowski metric are generated by the following $15$ vectorfields
(see e.g. \cite{dC2008}): 
\begin{enumerate}
	\item the four \emph{translations} $T_{(\mu)} , \qquad (\mu =0,1,2,3),$ 
	\item the three \emph{rotations} $\Omega_{(jk)}, \qquad (1 \leq j < k \leq 3),$
	\item the three \emph{Lorentz boosts} $\Omega_{(0j)}, \qquad (j = 1,2,3),$
	\item the \emph{scaling} vectorfield $S,$
	\item the four \emph{acceleration} vectorfields $K_{(\mu)}, \qquad (\mu = 0,1,2,3).$ 
\end{enumerate}
Relative to the inertial coordinate system $\lbrace x^{\mu} \rbrace_{\mu=0,1,2,3},$ the above vectorfields can be expressed as

\begin{subequations}
\begin{align}
	T_{(\mu)} & = \partial_{\mu}, \label{E:Translationsdef} \\
	\Omega_{(\mu \nu)} & = x_{\mu} \partial_{\nu} - x_{\nu} \partial_{\mu}, \\
	S & = x^{\kappa} \partial_{\kappa}, \label{E:Sdef} \\
	K_{(\mu)} & = - 2 x_{\mu} S + g_{\kappa \lambda}x^{\kappa} x^{\lambda} \partial_{\mu}.
\end{align}	
\end{subequations}
When working in \emph{our fixed inertial coordinate system}, we use the notation $T_{(0)} = \partial_t$ interchangeably.
In this article, we will primarily make use of the vectorfields in $(1) - (4),$ together with $\overline{K} \eqdef K_{(0)} + T_{(0)},$ which has the following components relative to the inertial coordinate system:

\begin{subequations}
\begin{align} \label{E:Kfirstdef}
	\overline{K}^{0} & = 1 + t^2 + (x^1)^2 + (x^2)^2 + (x^3)^2, \\
	\overline{K}^j & = 2t x^j, && (j = 1,2,3).
\end{align}
\end{subequations}
We remark that the translations, rotations, and boosts are Killing fields, while relative to the inertial coordinate system,
we have that

\begin{subequations}
\begin{align}
	^{(S)}\pi_{\mu \nu} & = 2 g_{\mu \nu}, \\
	^{(K_{(\lambda)})}\pi_{\mu \nu} & = -4x_{\lambda} g_{\mu \nu}, \\
	^{(\overline{K})}\pi_{\mu \nu} & = 4t g_{\mu \nu}. \label{E:overlineKdeformationtensor}
\end{align}
\end{subequations}
  
The subset $\mathscr{T},$ which consists of all the translations\footnote{This is not to be confused with the subset $\mathcal{T}$ of frame field vectors, which is defined in \eqref{E:Framefieldsubsets}.},
the subset $\mathcal{O},$ which consists of the rotations, and the subset $\mathcal{Z},$ 
which consists of all generators except for the accelerations, and which have
cardinalities $4,$ $3,$ and $11$ respectively, will play a distinguished role throughout this article:

\begin{subequations}
\begin{align}
	\mathscr{T} & \eqdef \lbrace T_{(\mu)} \rbrace_{0 \leq \mu \leq 3}, \label{E:Translationsetdef} \\
	\mathcal{O} & \eqdef \lbrace \Omega_{(12)}, \Omega_{(23)}, \Omega_{(13)} \rbrace, \label{E:Rotationsetdef} \\
	\mathcal{Z} & \eqdef \lbrace T_{(\mu)}, \Omega_{(\mu \nu)}, S \rbrace_{0 \leq \mu < \nu \leq 3}.  \label{E:Zsetdef} 
\end{align}
\end{subequations}
We denote the Lie algebras generated by $\mathscr{T}, \mathcal{O},$ and $\mathcal{Z}$ by 
$\mathbf{T}, \mathbf{O},$ and $\mathbf{Z}$ respectively. Note that for each vectorfield $Z \in \mathbf{Z},$ there is a constant $c_Z$ such that

\begin{subequations}
\begin{align}
	\Lie_Z g_{\mu \nu} = c_Z g_{\mu \nu}, \label{E:LieZonmlower} \\
	\Lie_Z (g^{-1})^{\mu \nu} = - c_Z (g^{-1})^{\mu \nu}. \label{E:LieZonmupper}
\end{align}
\end{subequations}
Also note that on the left-hand side of \eqref{E:LieZonmupper}, the indices of $g$ are raised before differentiation occurs.

It will be convenient for us to work with \emph{modified Lie derivatives}\footnote{Note that these are not the same
modified Lie derivatives that appear in \cite{lBnZ2009}, \cite{dCsK1993}, and \cite{nZ2000}.} $\Liemod_Z.$

\begin{definition} \label{D:Liemoddef}
	For each vectorfield $Z \in \mathcal{Z},$ we define the modified Lie derivative $\Liemod_Z$ by
	\begin{align} \label{E:Liemoddef}
		\Liemod_Z \eqdef \Lie_Z + 2c_Z,
	\end{align}
	where $c_Z$ denotes the constant from \eqref{E:LieZonmlower}.
\end{definition}
The crucial feature of the above definition is captured by Lemma \ref{L:LiemodZLiemodMaxwellCommutator} below, which shows that 
for each $Z \in \mathcal{Z},$ the operator $\Liemod_Z$ can be commuted through the linear Maxwell-Maxwell equation \\ ${\big[(g^{-1})^{\mu \kappa} (g^{-1})^{\nu \lambda} - (g^{-1})^{\mu \lambda} (g^{-1})^{\nu \kappa}\big] \nabla_{\mu} \Far_{\kappa \lambda} = 0},$ resulting in the identity
${\big[(g^{-1})^{\mu \kappa} (g^{-1})^{\nu \lambda} - (g^{-1})^{\mu \lambda} (g^{-1})^{\nu \kappa}\big] \nabla_{\mu} \Lie_Z \Far_{\kappa \lambda} = 0}.$
As is shown in Lemmas \ref{L:HtriangleCommutator} and the proof of Lemma \ref{L:divergenceofdotJLinfinitybound}, a similar result (involving nonlinear error terms) also holds for the MBI equation \eqref{E:HmodifieddMis0summary}.

We now introduce some notation that will allow us to compactly express iterated derivatives. If $\mathcal{A}$
is one of the sets from \eqref{E:Translationsetdef} - \eqref{E:Zsetdef}, then we label the vectorfields in $\mathcal{A}$ as $Z^{\iota_{1}}, \cdots, Z^{\iota_{m}},$ where $m$ is the cardinality of $\mathcal{A}.$ Then for any multi-index 
$I = (\iota_1, \cdots, \iota_k)$ of length $k,$ where each $\iota_i \in \lbrace 1,2,\cdots, m \rbrace,$ we define

\begin{definition} \label{D:iterated}
\begin{subequations}
\begin{align} 
	\Lie_{\mathcal{A}}^I & \eqdef \Lie_{Z^{\iota_1}} \circ \cdots \circ \Lie_{Z^{\iota_k}}, \label{E:iteratedLie} \\
	\Liemod_{\mathcal{A}}^I & \eqdef \Liemod_{Z^{\iota_1}} \circ \cdots \circ \Liemod_{Z^{\iota_k}}, \\
	\nabla_{\mathcal{Z}}^I & \eqdef \nabla_{Z^{\iota_1}} \circ \cdots \circ \nabla_{Z^{\iota_k}},  \label{E:iteratedCovariant} \\
	\angn_{\mathcal{O}}^I & \eqdef \angn_{Z^{\iota_1}} \circ \cdots \circ \nabla_{Z^{\iota_k}},
\end{align}
\end{subequations}
etc.
\end{definition}
Under this convention, the Leibniz rule can be written as

\begin{align}
	\Lie_{\mathcal{Z}}^I (UV) = \sum_{I_1 + I_2 = I} (\Lie_{\mathcal{Z}}^I U)(\Lie_{\mathcal{Z}}^I V),
\end{align}
etc., where by a sum over $I_1 + I_2 = I,$ we mean a sum over all order preserving partitions of the index $I$ into two 
multi-indices; i.e., if $I = (\iota_1, \cdots, \iota_k),$ then $I_1 = (\iota_{i_1}, \cdots, \iota_{i_a}), 
I_2 = (\iota_{i_{a+1}}, \cdots, \iota_{i_k}),$ where $i_1, \cdots, i_k$ is any re-ordering of the integers
$1,\cdots,k$ such that $i_1 < \cdots < i_a,$ and $i_{a+1} < \cdots < i_k.$

We end this section with the following lemmas, which provide expressions for the commutators 
$[\Lie_X, \star]$ and $[\nabla_X, \star]$ acting on two-forms.

\begin{lemma} \label{L:nablaXhodgedualcommutation}
	If $X$ is a vectorfield, and $\Far$ is a two-form, then
	
	\begin{align} \label{E:LieXhodgedualcommutation}
		{^{\star \hspace{-.03in}} (\nabla_X \Far)}_{\mu \nu} = \nabla_X \Fardual_{\mu \nu}.
	\end{align}
	
\end{lemma}

\begin{proof}
	Lemma \ref{L:nablaXhodgedualcommutation} follows easily from definition \ref{E:Fardualdef} and
	the fact $\nabla_{X} \epsilon_{\mu \nu \kappa \lambda} = 0,$ which is a simple consequence of 
	\eqref{E:volumeformcovariantlyconstant}.
\end{proof}

\begin{lemma} \label{L:LieXhodgedualcommutation} \cite[Eqn. 3.25]{dCsK1990}
	If $X$ is a vectorfield, and $\Far$ is a two-form, then
	
	\begin{align} \label{E:LieXhodgedualcommutation}
		{^{\star \hspace{-.03in}} (\Lie_X \Far)}_{\mu \nu} = \Lie_X \Fardual_{\mu \nu} 
			- \Fardual_{\mu}^{\ \beta} {^{(X)}\pi_{\nu \beta}}
			+ {^{(X)}\pi_{\mu \beta}} \Fardual_{\nu}^{\ \beta} 
			+ \frac{1}{2} {^{(X)} \pi_{\ \beta}^{\beta}} \Fardual_{\mu \nu}.
	\end{align}
\end{lemma}

\begin{proof}
	The relation \eqref{E:LieXhodgedualcommutation} follows from the using the expressions
	\eqref{E:Fardualdef} and \eqref{E:Liederivativeintermsofnabla} for $\Fardual$ and $\Lie_X,$ together
	with the well-known identities
	
	\begin{align}
		\Lie_X \epsilon_{\kappa \lambda \mu \nu} & = \frac{1}{2} {^{(X)}\pi_{\ \beta}^{\beta}}
			\epsilon_{\kappa \lambda \mu \nu}, \label{E:LieX4volume} \\
		\Lie_X (g^{-1})^{\mu \nu} & = - {^{(X)}\pi^{\mu \nu}}, \label{E:LieZginverse} \\
		\epsilon^{\lambda \beta \mu \nu} \epsilon_{\alpha \beta \rho \sigma} & =
			\delta_{\alpha}^{\lambda} \delta_{\rho}^{\mu} \delta_{\sigma}^{\nu}
			- \delta_{\alpha}^{\lambda} \delta_{\rho}^{\nu} \delta_{\sigma}^{\mu}
			+ \delta_{\alpha}^{\mu} \delta_{\rho}^{\nu} \delta_{\sigma}^{\lambda}
			- \delta_{\alpha}^{\mu} \delta_{\rho}^{\lambda} \delta_{\sigma}^{\nu}
			+	\delta_{\alpha}^{\nu} \delta_{\rho}^{\lambda} \delta_{\sigma}^{\mu}
			- \delta_{\alpha}^{\nu} \delta_{\rho}^{\mu} \delta_{\sigma}^{\lambda},
	\end{align}
	to express $\Lie_X \Fardual_{\mu \nu}$ in terms of the components of
	$\Fardual.$ We leave the details to the reader.
\end{proof}

We are particularly interested in the case that $X$ is a conformal Killing field. Under this assumption, we have the following simple corollary of the lemma.

\begin{corollary} \label{C:ConformalKillingLieXhodgedualcommutation} \cite[Corollary of Proposition 3.3]{dCsK1990}
	If $X$ is a conformal Killing field and $\Far$ is a two-form, then
	
	\begin{align} \label{E:ConformalKillingLieXhodgedualcommutation}
	{^{\star \hspace{-.03in}} (\Lie_X \Far)}	= \Lie_X \Fardual.
	\end{align}
\end{corollary}

\begin{proof}
	Corollary \ref{C:ConformalKillingLieXhodgedualcommutation} follows from \eqref{E:conformalKilling} and Lemma 
	\ref{L:LieXhodgedualcommutation}.
\end{proof}

\section{Tensorial Decompositions} \label{S:Decompositions}

In this section, we have three main goals. First, we will introduce the \emph{null frame} decomposition of
tensorfields. Related to this decomposition is the notion of a \emph{null form} $\mathcal{Q}(\cdot,\cdot),$ which is a 
quadratic form that acts on a pair of type $\binom{0}{2}$ tensors, and that has a special algebraic property: the ``worst'' possible quadratic combinations, from the point of view of the null decomposition, are absent. Since the two invariants $\Farinvariant_{(i)}$ of the Faraday tensor are each multiples of a corresponding null form $\mathcal{Q}_{(i)}(\Far, \Far),$ 
the net effect is that every nonlinear term in the Euler-Lagrange equation \eqref{E:MBIEuler-Lagrange} of the MBI system can be expressed as functions of null forms in $\Far, \nabla \Far.$ Later in the article, with the help of the additional null form structure present in the the expression \eqref{E:divJdot} for the divergence of our energy currents, together with Proposition \ref{P:GlobalSobolev}, we will be able to prove sharp decay estimates for the null components of solutions $\Far$ to the MBI system. Next in this section, we decompose the MBI system relative to the null frame. From the point of view of proving the sharpest possible decay estimates, the most important equation is \eqref{E:MBIalphanulldecomp}, which shows that the ``worst'' derivative of the $\alpha$ null component of $\Far$ can be expressed in terms of ``good'' derivatives of other null components of $\Far,$ plus cubic error terms. Finally, we introduce electromagnetic decompositions of $\Far$ and $\Max,$
where $\Max$ is the Maxwell tensor from \eqref{E:MaxMBI}. These decompositions will be useful for proving various identities and inequalities concerning $\Far,$ and for expressing the smallness condition in our global existence theorem directly in terms of the data $(\mathring{\Magneticinduction},\mathring{\Displacement}),$ which are one-forms inherent to the Cauchy hypersurface $\Sigma_0.$
\\

\noindent \hrulefill
\ \\

\subsection{The Null frame} \label{SS:NullFrame}

Before proceeding, we introduce the subsets $C_{q}^+, C_{s}^-$ of Minkowski space,
which will play a role in the sequel.
\begin{definition}
	In our fixed inertial coordinate system $(t,\uy),$ we define the \emph{outgoing Minkowski 
	null cones} $C_{q}^+,$ and \emph{ingoing Minkowski null cones} $C_{s}^-,$ as follows:
	
	\begin{subequations}
	\begin{align}
		C_{q}^+ & \eqdef \lbrace (\tau,\uy) \mid |\uy| - \tau = q \rbrace, \\
		C_{s}^- & \eqdef \lbrace (\tau,\uy) \mid |\uy| + \tau = s \rbrace.
	\end{align}
	\end{subequations}
	In the above formulas, $\uy \eqdef (y^1,y^2,y^3),$ and 
	$|\uy| \eqdef \big((y^1)^2 + (y^2)^2 + (y^3)^2 \big)^{1/2}.$
\end{definition}

A \emph{Minkowski null frame} is a locally defined collection of four vectorfields on $M \slash 0$
that we will denote by $\uL, L, e_1, e_2,$ and that have the following properties:
\begin{itemize}
	\item At each point $p,$ the set $\lbrace \uL, L, e_1, e_2 \rbrace$ spans $T_p M \simeq M.$
	\item For each $s \in \mathbb{R},$ and for all nonzero $p \in C_{s}^-,$ $\uL|_p$ is a 
		future-directed, \emph{ingoing null geodesic vectorfield} tangent to $C_{s}^-.$
	\item For each $q \in \mathbb{R},$ and for all nonzero $p \in C_{q}^+,$ $L|_p$ is a future-directed, \emph{outgoing null 
		geodesic vectorfield} tangent to $C_{q}^+.$
	\item For all $t \in \mathbb{R},$ for all $r > 0,$ $r \eqdef |\ux|,$ and all points $p \in S_{r,t},$
		$\uL|_p$ and $L_p$ are normal to $S_{r,t}.$ 
	\item $g(\uL,L) = -2.$ 
	\item At each point $p \in S_{r,t},$ $e_1|_p$ and $e_2|_p$ are tangent to $S_{r,t}.$
	\item $g(e_A,e_B) = \delta_{AB}.$
\end{itemize} 
In formulas \eqref{E:Ldef} - \eqref{E:uLdef}, we provided the components of $\uL, L$ relative to our inertial coordinate
system. Relative to an arbitrary coordinate system, the aforementioned properties of $\uL,$ $L,$ $e_1,$ and $e_2$ can be expressed
as follows:

\begin{subequations}
\begin{align}
	\nabla_L L & = \nabla_{\uL} \uL = 0, \label{E:LanduLaregeodesic} \\
	\nabla_L \uL & = \nabla_{\uL} L = 0, \label{E:nablaLuLis0} \\
	\uL_{\kappa} L^{\kappa} & = - 2, \label{E:LunderlineLcontracted} \\
	e_A^{\kappa} \uL_{\kappa} & = e_A^{\kappa} L_{\kappa} = 0,  \label{E:LunderlineLnormaltospheres} \\
	g_{\kappa \lambda} e_A^{\kappa} e_B^{\lambda} & = \delta_{AB}.
\end{align}
\end{subequations}
Additionally, it follows from \eqref{E:Torsionfree} and \eqref{E:nablaLuLis0} that $\uL$ and $L$ commute as vectorfields: 
\begin{align}
	[\uL, L] & = 0. \label{E:LuLBracketis0}
\end{align}

In the analysis that will follow, we will see that the decay rates of the null components (see Section \ref{SS:NullComponents}) 
of $\Far$ will be distinguished according to the kinds of contractions of $\Far$ taken against $\uL,$ $L,$ $e_1,$ and $e_2.$ With these considerations in mind, we introduce the following sets of vectorfields:

\begin{subequations}
\begin{align} \label{E:Framefieldsubsets}
	\mathcal{L} \eqdef \lbrace L \rbrace, & & \mathcal{T} \eqdef \lbrace L, e_1, e_2 \rbrace,
	& & \mathcal{U} \eqdef \lbrace \uL, L, e_1, e_2 \rbrace.
\end{align}
\end{subequations}

We will often need to measure the size of the contractions of various tensors and their covariant derivatives against vectors
belonging to the sets $\mathcal{L}, \mathcal{T}, \mathcal{U}.$ This motivates the next definition. 

\begin{definition} \label{D:contractionnomrs}

If $\mathcal{V}, \mathcal{W}$ denote any two of the above sets, and $F$ is an arbitrary type $\binom{0}{2}$
tensor, then we define the following pointwise seminorms:

\begin{subequations}
\begin{align}
	|F|_{\mathcal{V} \mathcal{W}} & \eqdef \sum_{V \in \mathcal{V}, W \in \mathcal{W}} |V^{\kappa} W^{\lambda} F_{\kappa 
		\lambda}|, \label{E:contractionnorm} \\
	|\nabla F|_{\mathcal{V} \mathcal{W}} & \eqdef \sum_{U \in \mathcal{U}, V \in \mathcal{V}, W \in \mathcal{W}} 
		|V^{\kappa} W^{\lambda} U^{\gamma} \nabla_{\gamma} F_{\kappa \lambda}|.
\end{align}
\end{subequations}

\end{definition}
Observe that if $F$ is a two tensor, then $|F| \approx |F|_{\mathcal{U} \mathcal{U}},$ where
$|F|$ is defined in \eqref{E:Riemanniannorm}.

\subsection{Null frame decomposition of a tensorfield}

For an arbitrary vectorfield $X$ and frame vector $U \in \mathcal{U},$ we define 

\begin{align} \label{E:XlowerUdef}
	X_U & \eqdef X_{\kappa} U^{\kappa}, \ \mbox{where} \ X_{\mu} \eqdef g_{\mu \kappa} X^{\kappa}. 
\end{align}
The components $X_U$ are known as the \emph{null components} of $X.$ In the sequel, we will abbreviate

\begin{align}
	X_A \eqdef X_{e_A}, && \nabla_A \eqdef \nabla_{e_A}, \ \mbox{etc}.
\end{align}
It follows from \eqref{E:XlowerUdef} that

\begin{subequations}
\begin{align}
	X & = X^{\kappa} \partial_{\kappa} = X^{L} L + X^{\uL} \uL
		+ X^{A} e_A,  \label{X:AbstractLuLAdecomp} \\
	X^{L} & = - \frac{1}{2}X_{\uL}, && X^{\uL} = - \frac{1}{2}X_L, && X^A = X_A.
		\label{E:XupperUdef}
\end{align}
\end{subequations}
Furthermore, it is easy to check that
\begin{align} \label{E:XYnullframeinnerproduct}
	g(X,Y) = X^{\kappa}Y_{\kappa} = -\frac{1}{2}X_{L}Y_{\uL} - \frac{1}{2}X_{\uL}Y_{L} + \delta^{AB} X_A Y_B. 
\end{align}

The above null decomposition of a vectorfield generalizes in the obvious way to higher order tensorfields. In the next section,
we provide a detailed version of the null decomposition of two-forms $\Far,$ since they are the fundamental unknowns in 
any classical theory of electromagnetism.

\subsection{The detailed null decomposition of a two-form} \label{SS:NullComponents}

\begin{definition} \label{D:null}
Given any two-form $\Far,$ we define its \emph{null components} $\ualpha, \alpha, \rho, \sigma$ as follows:

\begin{subequations}
\begin{align}
	\ualpha_{\mu} & \eqdef \angg_{\mu}^{\ \nu} \Far_{\nu \lambda} \uL^{\lambda}, 
		\label{E:ualphadef} \\
	\alpha_{\mu} & \eqdef \angg_{\mu}^{\ \nu} \Far_{\nu \lambda} L^{\lambda}, \\
	\rho & \eqdef \frac{1}{2} \Far_{\kappa \lambda}\uL^{\kappa} L^{\lambda}, \\
	\sigma & \eqdef \frac{1}{2} \angepsilon^{\kappa \lambda} \Far_{\kappa \lambda}. \label{E:sigmadef}
\end{align}
\end{subequations}

\end{definition}

It is a simple exercise to check that $\ualpha, \alpha$ are tangent to the spheres $S_{r,t}:$

\begin{subequations}
\begin{align}
	\ualpha_{\kappa}\uL^{\kappa} & = 0, & \ualpha_{\kappa}L^{\kappa} & = 0, \\
	\alpha_{\kappa}\uL^{\kappa} & = 0, & \alpha_{\kappa}L^{\kappa} & = 0.
\end{align}
\end{subequations}
Furthermore, relative to the null frame $\lbrace \uL, L, e_1, e_2 \rbrace,$ we have that

\begin{subequations}
\begin{align}
	\underline{\alpha}_A & = \Far_{A \uL}, \\
	\alpha_A & = \Far_{AL}, \\
	\rho & = \frac{1}{2} \Far_{\uL L}, \\
	\sigma & = \Far_{12}.
\end{align}
\end{subequations}

In terms of the norms introduced in Definition \ref{D:contractionnomrs}, it follows that

\begin{subequations}
\begin{align}
	|\Far| & \approx |\Far|_{\mathcal{U}\mathcal{U}} \approx |\ualpha| + |\alpha| + |\rho| + |\sigma|, \\
	|\Far|_{\mathcal{L}\mathcal{U}} & \approx |\alpha| + |\rho|, \\
	|\Far|_{\mathcal{T} \mathcal{T}} & \approx |\alpha| + |\sigma|. 
\end{align}
\end{subequations}

The null components of $\Fardual$ can be expressed in terms of the above null components of $\Far.$
Denoting the null components\footnote{We use the symbol $\odot$ in order to avoid confusion with the Hodge dual; i.e., it is
not true that ${^{\star \hspace{-.03in}}(\ualpha[\Far])} = \ualpha[\Fardual].$} of $\Fardual$ by $\ualphadot, \alphadot, \rhodot, \sigmadot,$ we leave it as a simple exercise to the reader to check that  

\begin{subequations}
\begin{align}
	\ualphadot_A & = - \underline{\alpha}^B \angepsilon_{BA}, \label{E:Fardualalpha} \\
	\alphadot_A & = \alpha^B \angepsilon_{BA}, \\
	\rhodot & = \sigma, \\  
	\sigmadot & = - \rho. \label{E:Fardualsigma}
\end{align}
\end{subequations}

\subsection{Null forms} \label{SS:NullForms}

\begin{definition}

Let $F,G$ be any pair of type $\binom{0}{2}$ tensors. We define the null forms 
$\mathcal{Q}_{(1)}(\cdot, \cdot), \mathcal{Q}_{(2)}(\cdot, \cdot)$ as follows:

\begin{subequations}
\begin{align}
	\mathcal{Q}_{(1)}(F,G) & \eqdef F^{\kappa \lambda} G_{\kappa \lambda}, \label{E:nullform1} \\
	\mathcal{Q}_{(2)}(F,G) & \eqdef {^{\star \hspace{-.03in}}F}^{\kappa \lambda} G_{\kappa \lambda} \label{E:nullform2}.
\end{align}	
\end{subequations}
\end{definition}
It is easy to check that $\mathcal{Q}_{(i)}(F,G) = \mathcal{Q}_{(i)}(G,F)$ for $i =1,2.$

\begin{remark}
	Observe that the two invariants \eqref{E:firstinvariant} - \eqref{E:secondinvariant} of the Faraday tensor $\Far$	are
	\begin{align}
		\Farinvariant_{(1)} = \frac{1}{2}\mathcal{Q}_{(1)}(\Far, \Far), 
		&& \Farinvariant_{(2)} = \frac{1}{4}\mathcal{Q}_{(2)}(\Far, \Far). 
	\end{align}
\end{remark}

Based on the above remark, it is clear that we will be primarily interested in the case in which $F,G$ are two-forms. The next lemma describes the fundamental algebraic properties of null forms; i.e., the absence of the ``worst possible'' quadratic terms.

\begin{lemma} \label{L:NullForms}
	Let $\Far, \Gar$ be a pair of two-forms, and let $\mathcal{Q}_{(i)}(\cdot,\cdot)$ be one the null forms defined 
	in \eqref{E:nullform1} - \eqref{E:nullform2}, and let $|\cdot|_{\mathcal{V}\mathcal{W}}$ 
	be the norms defined in $\eqref{E:contractionnorm}$ Then for $i = 1,2,$ the following pointwise estimate holds:
	\begin{align} \label{E:NullFormBound}
		\mathcal{Q}_{(i)}(\Far,\Gar) \lesssim 
			|\Far| |\Gar|_{\mathcal{L}\mathcal{U}} + |\Far|_{\mathcal{L}\mathcal{U}} |\Gar|
			+ |\Far|_{\mathcal{T}\mathcal{T}} |\Gar|_{\mathcal{T}\mathcal{T}}.
	\end{align}
	
\end{lemma}

\begin{proof}
	Let $\ualpha[\Far],\alpha[\Far], \sigma[\Far], \rho[\Far]$ and $\ualpha[\Gar],\alpha[\Gar], \sigma[\Gar], \rho[\Gar]$
	denote the null components of $\Far$ and $\Gar$ respectively. Then using the relations \eqref{E:Fardualalpha} - 
	\eqref{E:Fardualsigma} (for the case $i=2$), we compute that
	\begin{align}
		\Far^{\kappa \lambda} \Gar_{\kappa \lambda} & = - \delta^{AB} \ualpha_A[\Far]\alpha_B[\Gar]
		- \delta^{AB} \ualpha_A[\Gar] \alpha_B[\Far] - 2 \rho[\Far]\rho[\Gar]
		+ 2\sigma[\Far] \sigma[\Gar], \\
		\Fardual^{\kappa \lambda} \Gar_{\kappa \lambda} 
			& = \angepsilon^{AB} \ualpha_A[\Far] \alpha_B[\Gar] + \angepsilon^{AB} \ualpha_A[\Gar] \alpha_B[\Far] 
			- 2 \sigma[\Far]\rho[\Gar] - 2\rho[\Far] \sigma[\Gar].
	\end{align}
	from which \eqref{E:NullFormBound} follows. 
	
\end{proof}

\subsection{Intrinsic divergence and curl, and the cross product}

In this section, we recall the definitions of the intrinsic divergence and curl of 
vectorfields $U$ that are tangent to the submanifolds $\Sigma_t$ and $S_{r,t}.$

\begin{definition}
	
	If $U$ is a vectorfield tangent to $\Sigma_t,$ then its intrinsic
	divergence and curls are defined relative to an arbitrary spacetime coordinate system as follows:
	
	\begin{align}
		\udiv U & \eqdef \SigmafirstfundNabla_{\kappa} U^{\kappa} = \Sigmafirstfund_{\kappa}^{\ \lambda} \nabla_{\lambda} 
		U^{\kappa}, && \\
		(\ucurcl U)^{\nu} & \eqdef \uepsilon_{\ \ \ \lambda}^{\nu \kappa} \nabla_{\kappa} U^{\lambda}, && (\nu = 0,1,2,3),		
	\end{align}
	where the volume form $\uepsilon_{\nu \kappa \lambda}$ of $\Sigma_t$ is defined in \eqref{E:bargvolumeformdef}.
	Relative to the Euclidean coordinate system $\bar{x} = (x^1, x^2, x^3)$ on $\Sigma_t,$ we have that
	\begin{align}
		\udiv U & = \SigmafirstfundNabla_a U^a, && \\
		(\ucurcl U)^j & \eqdef	\uepsilon^{ja}_{\ \ b} \SigmafirstfundNabla_a U^b, && (j=1,2,3).	
	\end{align}
	
	If $U$ is a vectorfield tangent to the spheres $S_{r,t},$ then its intrinsic
	divergence and curls are defined to be the following scalar quantities:
	
	\begin{align}
		\angdiv U & \eqdef \angn_{\kappa} U^{\kappa} = \angg_{\kappa}^{\ \lambda} \nabla_{\lambda} U^{\kappa}, 
			\\
		\angcurl U & \eqdef	\angepsilon_{\ \lambda}^{\kappa} \nabla_{\kappa} U^{\lambda},
	\end{align}
	where $\angepsilon_{\mu \nu},$ the volume form $S_{r,t},$ is defined in \eqref{E:angvolumeformdef}. Relative to a null frame, we
	have that
	
	\begin{align}
		\angdiv U & \eqdef \delta^{AB} \angn_A U_B, \label{E:angdivframe} \\
		\angcurl U & \eqdef	\angepsilon_{AB} \angn_A U_B. \label{E:angcurlframe}
	\end{align}
	
	Note that in the above definitions, contractions against the frame vectors $e_1, e_2$ is taken \emph{after}
	covariant differentiation; e.g., $\angn_A U_B \eqdef e_A^{\kappa} e_B^{\lambda} \angn_{\kappa} U_{\lambda}.$ 
\end{definition}

\begin{definition} \label{D:Crossproduct}
	If $U$ and $V$ are vectors tangent to $\Sigma_t,$ then we define their cross product, which is also tangent to $\Sigma_t,$ 
	as follows, relative to the Euclidean coordinate system $\ux = (x^1, x^2, x^3)$ on $\Sigma_t:$
	
	\begin{align} \label{E:Crossproduct}
		(U \times V)^j & \eqdef \uepsilon^{j}_{\ ab} U^a V^b.
	\end{align}
\end{definition}

\subsection{The null components of \texorpdfstring{$\overline{K}$ and $\nabla \overline{K}$}{the Morawetz vectorfield and its 
covariant derivatives}}

In this section, we provide the null components of the vectorfield $\overline{K},$ which is defined in \eqref{E:Kfirstdef},
and its covariant derivative $\nabla \overline{K}.$ $\overline{K}$ is central to our global existence argument,
because it is a fundamental ingredient in the energies we construct; see \eqref{E:Jdotdef} and \eqref{E:mathcalENdef}.
We don't provide any proofs in this section, but instead leave the simple computations as an exercise for the reader. 

We first note that using definitions \eqref{E:Lradialdef} and \eqref{E:underlineLradialdef}, $\overline{K}$ can be expressed as follows:
\begin{align} \label{E:Kdef}
	\overline{K} = \frac{1}{2}\big\lbrace(1+s^2) L + (1 + q^2) \uL \big\rbrace.
\end{align}
From \eqref{E:LunderlineLcontracted}, \eqref{E:LunderlineLnormaltospheres}, and \eqref{E:Kdef}, it easily follows that
\begin{subequations}
\begin{align}   
	\overline{K}_{L} & = -(1 + q^2), \label{E:MorawetznulldecompL} \\
	\overline{K}_{\uL} & = -(1 + s^2), \\
	\overline{K}_{A} & = 0. \label{E:MorawetznulldecompA}
\end{align}
\end{subequations}
Finally, we compute that 

\begin{subequations}
\begin{align}
	\nabla_L \overline{K}_{\uL} & \eqdef L^{\kappa} \uL^{\lambda} \nabla_{\kappa} \overline{K}_{\lambda} = -4s, 
		\label{E:nablaoverlineKLunderlineL} \\
	\nabla_{\uL} \overline{K}_L & \eqdef \uL^{\kappa} L^{\lambda} \nabla_{\kappa} \overline{K}_{\lambda} = 4q, \\
	\nabla_A \overline{K}_B & \eqdef e_A^{\kappa} e_{B}^{\lambda} \nabla_{\kappa} \overline{K}_{\lambda} = 2t \delta_{AB}, \\
	\nabla_L \overline{K}_L & \eqdef L^{\kappa} L^{\lambda} \nabla_{\kappa} \overline{K}_{\lambda} = 0, \\
	\nabla_L \overline{K}_A & \eqdef L^{\kappa} e_A^{\lambda} \nabla_{\kappa} \overline{K}_{\lambda} = 0, \\
	\nabla_{\uL} \overline{K}_{\uL} & \eqdef \uL^{\kappa} \uL^{\lambda} \nabla_{\kappa} \overline{K}_{\lambda} = 0, \\
	\nabla_{\uL} \overline{K}_A & \eqdef \uL^{\kappa} e_A^{\lambda} \nabla_{\kappa} \overline{K}_{\lambda} = 0, \\
	\nabla_A \overline{K}_L & \eqdef e_A^{\kappa} L^{\lambda} \nabla_{\kappa} \overline{K}_{\lambda} = 0, \\
	\nabla_A \overline{K}_L & \eqdef e_A^{\kappa} \uL^{\lambda} \nabla_{\kappa} \overline{K}_{\lambda} = 0.
	 \label{E:nablaoverlineKAL}
\end{align}
\end{subequations}

\subsection{The null decomposition of the MBI system}

In this section, we decompose the MBI system into equations for the null components of $\Far.$ 
We begin by noting that after simple computations, the MBI system \eqref{E:dFis0intro} - \eqref{E:Maxsdefintro}
can be expressed in the following equivalent form:

\begin{subequations}
\begin{align}
	& \nabla_{\lambda} \Far_{\mu \nu} + \nabla_{\mu} \Far_{\nu \lambda} + \nabla_{\nu} \Far_{\lambda \mu} = 0,
		\label{E:dFis0nullsection} \\
	& \nabla_{\lambda} \Fardual_{\mu \nu} + \nabla_{\mu} \Fardual_{\nu \lambda} + \nabla_{\nu} \Fardual_{\lambda \mu} 
		 \label{E:dMis0nullsection} \\
	& \ \ \ - \frac{1}{2}\ell_{(MBI)}^{-2} \Big\lbrace (\nabla_{\lambda}\Farinvariant_{(1)} - 2 \Farinvariant_{(2)} \nabla_{\lambda} \Farinvariant_{(2)})\Fardual_{\mu \nu}
			+ (\nabla_{\mu}\Farinvariant_{(1)} - 2 \Farinvariant_{(2)} \nabla_{\mu} \Farinvariant_{(2)})\Fardual_{\nu \lambda}
			+ (\nabla_{\nu}\Farinvariant_{(1)} - 2 \Farinvariant_{(2)} \nabla_{\nu} \Farinvariant_{(2)})\Fardual_{\lambda \mu} \Big\rbrace \notag  \\
	& \ \ \ - \frac{1}{2}\ell_{(MBI)}^{-2} \Farinvariant_{(2)} \Big\lbrace (\nabla_{\lambda}\Farinvariant_{(1)} - 2 \Farinvariant_{(2)} \nabla_{\lambda} \Farinvariant_{(2)})\Far_{\mu 
		\nu} + (\nabla_{\mu}\Farinvariant_{(1)} - 2 \Farinvariant_{(2)} \nabla_{\mu} \Farinvariant_{(2)})\Far_{\nu \lambda}
		+ (\nabla_{\nu}\Farinvariant_{(1)} - 2 \Farinvariant_{(2)} \nabla_{\nu} \Farinvariant_{(2)})\Far_{\lambda \mu} \Big\rbrace \notag \\
	& \ \ \ + (\nabla_{\lambda} \Farinvariant_{(2)}) \Far_{\mu \nu} + (\nabla_{\mu}\Farinvariant_{(2)}) \Far_{\nu \lambda} + (\nabla_{\nu}\Farinvariant_{(2)}) 
			\Far_{\lambda \mu} = 0. \notag
\end{align}
\end{subequations}

In our calculations below, we will make use of the identities

\begin{align} \label{E:nablaALnablaAuL}
	\nabla_A \uL = -r^{-1} e_A, && \nabla_A L = r^{-1} e_A, 
\end{align}
which can be directly calculated in our inertial coordinate system using \eqref{E:uLdef} - \eqref{E:Ldef}. We will also make use
of the identity 

\begin{align} \label{E:nablaAeBcovriantintrintermsofinsicnablaAeBcovraint}
	\angn_A e_B & = \nabla_A e_B + \frac{1}{2} g(\nabla_A e_B, \uL) L + \frac{1}{2} g(\nabla_A e_B, L) \uL \\
	& =  \nabla_A e_B - \frac{1}{2} g(e_B, \nabla_A \uL) L - \frac{1}{2} g(e_B, \nabla_A L)\uL \notag \\
	& = \nabla_A e_B + \frac{1}{2}r^{-1} \delta_{AB} (L - \uL), \notag
\end{align}
which follows from \eqref{E:SphereIntrinsicintermsofExtrinsic} and \eqref{E:nablaALnablaAuL}.

Furthermore, if $U$ is a type $\binom{0}{m}$ tensorfield, and $X_{(i)}, 1 \leq i \leq m$ and $Y$ are vectorfields, then 
by \eqref{E:CovariantLeibniz}, we have that
\begin{align} \label{E:nablaLeibnizrule}
	\nabla_Y \big(U(X_{(1)}, X_{(2)}, \cdots, X_{(m)})\big) 
	& = (\nabla_Y U)(X_{(1)}, X_{(2)}, \cdots, X_{(m)}) + U(\nabla_Y X_{(1)}, X_{(2)}, \cdots , X_{(m)}) \\
	& \ \ + \cdots + U(X_{(1)}, X_{(2)}, \cdots , \nabla_Y X_{(m)}). \notag
\end{align}
Similarly, if $U$ is tangent to the spheres $S_{r,t},$ and $e_{B_{(1)}}, \cdots, e_{B_{(m)}} \in \lbrace e_1, e_2 \rbrace,$
then

\begin{align} \label{E:instrinsicnablaLiebnizrule}
	\angn_A (U_{B_{(1)} \cdots B_{(m)}}) & \eqdef \angn_{e_A} \big(U(e_{B_{(1)}}, e_{B_{(2)}}, \cdots , e_{B_{(m)}} \big) \\
	& = (\angn_A U)(e_{B_{(1)}}, e_{B_{(2)}}, \cdots ,e_{B_{(m)}}) + 
		U(\angn_A e_{B_{(1)}}, e_{B_{(2)}}, \cdots , e_{B_{(m)}}) + \cdots + U(e_{B_{(1)}}, e_{B_{(2)}}, \cdots , \angn_A 
		e_{B_{(m)}}). \notag
\end{align}
Applying \eqref{E:nablaLeibnizrule} and \eqref{E:instrinsicnablaLiebnizrule} to $\Far,$ and using \eqref{E:nablaALnablaAuL} plus \eqref{E:nablaAeBcovriantintrintermsofinsicnablaAeBcovraint}, we compute that the identities contained in the following lemma hold.

\begin{lemma} \cite[pg. 161]{dCsK1990}
Let $\Far$ be a two-form, and let $\ualpha,$ $\alpha,$ $\rho,$ $\sigma$ be its null components as defined in Definition \ref{D:null}. 
Then the following identities hold:

\begin{subequations}
\begin{align}
	\nabla_A \Far_{B \uL} & = \angn_A \ualpha_B - r^{-1}(\rho \delta_{AB} + \sigma \angepsilon_{AB}), 
		\label{E:nablaAFarBuLintermsofintrinsic} \\
	\nabla_A \Far_{BL} & = \angn_A \alpha_B - r^{-1}(\rho \delta_{AB} - \sigma \angepsilon_{AB}), \\
	\nabla_A \Fardual_{B \uL} & = - \angepsilon_{CB}\angn_A \ualpha_C - r^{-1}(\sigma \delta_{AB} - \rho \angepsilon_{AB}), \\
	\nabla_A \Fardual_{BL} & = \angepsilon_{CB}\angn_A \alpha_C - r^{-1}(\sigma \delta_{AB} + \rho \angepsilon_{AB}), \\
	\frac{1}{2} \nabla_A \Far_{\uL L} & = \angn_A \rho + \frac{1}{2} r^{-1}(\ualpha_A + \alpha_A), \\
	\frac{1}{2} \nabla_A \Fardual_{\uL L} & = \angn_A \sigma + \frac{1}{2} r^{-1}(-\angepsilon_{BA}\ualpha_B + \angepsilon_{BA} \alpha_B), \\
	\nabla_A \Far_{BC} & = \angepsilon_{BC} \Big\lbrace \angn_A \sigma + \frac{1}{2}r^{-1}
		(-\angepsilon_{DA}\ualpha_D + \angepsilon_{DA} \alpha_D) \Big\rbrace. \label{E:nablaAFarBCintermsofintrinsic}
\end{align}
\end{subequations}
Note that in all of the above expressions, contractions are taken after differentiating; e.g., 
$\nabla_A \Far_{BL} \eqdef e_A^{\mu} e_B^{\kappa} \uL^{\lambda} \nabla_{\mu} \Far_{\kappa \lambda}.$

\end{lemma}

\hfill $\qed$

\begin{lemma}
The MBI system can be decomposed into principal terms and ``cubic error terms'' as follows, relative to a 
Minkowski null frame:

\begin{subequations}
\begin{align}
	\nabla_L \ualpha_A & + r^{-1}\ualpha_A + \angn_A \rho - \angepsilon_{AB} \angn_B \sigma  \label{E:MBIualphanulldecomp} \\
	& - \frac{1}{4}\ell_{(MBI)}^{-2} \Big\lbrace (\nabla_{\uL}\Farinvariant_{(1)} - 2 \Farinvariant_{(2)} \nabla_{\uL} 
		\Farinvariant_{(2)})\alpha_A
		- 2\angepsilon_{AB}(\angn_{B}\Farinvariant_{(1)} - 2 \Farinvariant_{(2)} \angn_{B} \Farinvariant_{(2)})\sigma
		+ (\nabla_{L}\Farinvariant_{(1)} - 2 \Farinvariant_{(2)} \nabla_{L} \Farinvariant_{(2)}) \ualpha_A \Big\rbrace  
		\notag  \\
	& - \frac{1}{4}\ell_{(MBI)}^{-2} \Farinvariant_{(2)} \Big\lbrace (\nabla_{\uL}\Farinvariant_{(1)} - 2 \Farinvariant_{(2)} 
		\nabla_{\uL} \Farinvariant_{(2)}) \angepsilon_{AB} \alpha_B
		- 2\angepsilon_{AB}(\angn_{B}\Farinvariant_{(1)} - 2 \Farinvariant_{(2)} \angn_{B} \Farinvariant_{(2)}) \rho
		- (\nabla_{L}\Farinvariant_{(1)} - 2 \Farinvariant_{(2)} \nabla_{L} \Farinvariant_{(2)})\angepsilon_{AB} \ualpha_B 
		\Big\rbrace  \notag \\
	& + \frac{1}{2}\Big\lbrace (\nabla_{\uL} \Farinvariant_{(2)}) \angepsilon_{AB} \alpha_B 
		- 2\angepsilon_{AB}(\angn_B \Farinvariant_{(2)}) \rho  
		- (\nabla_{L}\Farinvariant_{(2)}) \angepsilon_{AB} \ualpha_B \Big\rbrace = 0,	 \notag \\
	\nabla_{\uL} \alpha_A & - r^{-1} \alpha_A - \angn_A \rho - \angepsilon_{AB} \angn_B \sigma 
		\label{E:MBIalphanulldecomp} \\
	& - \frac{1}{4}\ell_{(MBI)}^{-2} \Big\lbrace (\nabla_{\uL}\Farinvariant_{(1)} - 2 \Farinvariant_{(2)} \nabla_{\uL} 
		\Farinvariant_{(2)})\alpha_A
		- 2\angepsilon_{AB}(\angn_{B}\Farinvariant_{(1)} - 2 \Farinvariant_{(2)} \angn_{B} \Farinvariant_{(2)})\sigma
		+ (\nabla_{L}\Farinvariant_{(1)} - 2 \Farinvariant_{(2)} \nabla_{L} \Farinvariant_{(2)}) \ualpha_A \Big\rbrace 
		\notag \\
	& - \frac{1}{4}\ell_{(MBI)}^{-2} \Farinvariant_{(2)} \Big\lbrace (\nabla_{\uL}\Farinvariant_{(1)} - 2 \Farinvariant_{(2)} 
		\nabla_{\uL} \Farinvariant_{(2)}) \angepsilon_{AB} \alpha_B
		- 2\angepsilon_{AB}(\angn_{B}\Farinvariant_{(1)} - 2 \Farinvariant_{(2)} \angn_{B} \Farinvariant_{(2)}) \rho
		- (\nabla_{L}\Farinvariant_{(1)} - 2 \Farinvariant_{(2)} \nabla_{L} \Farinvariant_{(2)})\angepsilon_{AB} \ualpha_B \Big\rbrace  		\notag \\
	& + \frac{1}{2}\Big\lbrace (\nabla_{\uL} \Farinvariant_{(2)}) \angepsilon_{AB} \alpha_B 
		- 2\angepsilon_{AB}(\angn_B \Farinvariant_{(2)}) \rho  
		- (\nabla_{L}\Farinvariant_{(2)}) \angepsilon_{AB} \ualpha_B \Big\rbrace = 0, \notag \\
	- \angdiv \ualpha & - \nabla_{\uL} \rho + 2r^{-1} \rho    
		- \frac{1}{2}\ell_{(MBI)}^{-2} \Big\lbrace - (\nabla_{\uL}\Farinvariant_{(1)} - 2 \Farinvariant_{(2)} \nabla_{\uL} 	
		\Farinvariant_{(2)}) \rho
		- \delta^{AB}(\angn_A\Farinvariant_{(1)} - 2 \Farinvariant_{(2)} \angn_A \Farinvariant_{(2)})\ualpha_B \Big\rbrace 
		\label{E:intrinsicdivualphanulldecomp} \\
		& - \frac{1}{2}\ell_{(MBI)}^{-2} \Farinvariant_{(2)} \Big\lbrace (\nabla_{\uL}\Farinvariant_{(1)} - 2 \Farinvariant_{(2)} 
			\nabla_{\uL} \Farinvariant_{(2)}) \sigma
			+ \angepsilon_{AB}(\angn_A\Farinvariant_{(1)} - 2 \Farinvariant_{(2)} \angn_A \Farinvariant_{(2)})\ualpha_B \Big\rbrace 
			+ \Big\lbrace (\nabla_{\uL} \Farinvariant_{(2)}) \sigma + \angepsilon_{AB} (\angn_A \Farinvariant_{(2)}) \ualpha_B 	
			\Big\rbrace = 0, \notag \\
	\angcurl \ualpha & + \nabla_{\uL} \sigma -2 r^{-1} \sigma  = 0, \label{E:intrinsiccurlualphanulldecomp} \\
	\angdiv \alpha & - \nabla_L \rho - 2 r^{-1} \rho 
		- \frac{1}{2}\ell_{(MBI)}^{-2} \Big\lbrace - (\nabla_L\Farinvariant_{(1)} - 2 \Farinvariant_{(2)} \nabla_L 	
		\Farinvariant_{(2)}) \rho
		+ \delta^{AB}(\angn_A\Farinvariant_{(1)} - 2 \Farinvariant_{(2)} \angn_A \Farinvariant_{(2)})\alpha_B \Big\rbrace 
		\label{E:intrinsicdivalphanulldecomp} 	\\
	& - \frac{1}{2}\ell_{(MBI)}^{-2} \Farinvariant_{(2)} \Big\lbrace (\nabla_{L}\Farinvariant_{(1)} - 2 \Farinvariant_{(2)} 	
		\nabla_{L} \Farinvariant_{(2)}) \sigma
		+ \angepsilon_{AB}(\angn_A\Farinvariant_{(1)} - 2 \Farinvariant_{(2)} \angn_A \Farinvariant_{(2)})\alpha_B \Big\rbrace 
		+ \Big\lbrace (\nabla_L \Farinvariant_{(2)}) \sigma + \angepsilon_{AB} (\angn_A \Farinvariant_{(2)}) \alpha_B \Big\rbrace = 0, 
		\notag \\
	\angcurl \alpha & + \nabla_L \sigma + 2 r^{-1} \sigma  = 0, \label{E:intrinsiccurlalphanulldecomp}
\end{align}
\end{subequations}
where the differential operators $\angdiv$ and $\angcurl$ are defined in \eqref{E:angdivframe} - \eqref{E:angcurlframe}.

\end{lemma}

\begin{remark}
	If we discard the nonlinear terms, then the resulting system is the null decomposition of the Maxwell-Maxwell system.
\end{remark}

\begin{proof}
	Contract the vectors $\uL^{\lambda} L^{\mu} e_A^{\nu}$ against equation \eqref{E:dFis0nullsection}, and then against
	equation \eqref{E:dMis0nullsection}. Adding the resulting equations and using \eqref{E:nablaAFarBuLintermsofintrinsic} - 
	\eqref{E:nablaAFarBCintermsofintrinsic} gives \eqref{E:MBIualphanulldecomp}, while subtracting the resulting equations and 
	using \eqref{E:nablaAFarBuLintermsofintrinsic} - \eqref{E:nablaAFarBCintermsofintrinsic} gives \eqref{E:MBIalphanulldecomp}.
	Equations \eqref{E:intrinsicdivualphanulldecomp} - \eqref{E:intrinsiccurlualphanulldecomp} follow from contracting
	$\uL^{\lambda} e_A^{\mu} e_B^{\nu}$ against \eqref{E:dMis0nullsection} and \eqref{E:dFis0nullsection} respectively, and using
	\eqref{E:nablaAFarBuLintermsofintrinsic} - \eqref{E:nablaAFarBCintermsofintrinsic}. Similarly, 	
	\eqref{E:intrinsicdivalphanulldecomp} - \eqref{E:intrinsiccurlalphanulldecomp} follow from contracting
	$L^{\lambda} e_A^{\mu} e_B^{\nu}$ against \eqref{E:dMis0nullsection} and \eqref{E:dFis0nullsection} respectively, and using
	\eqref{E:nablaAFarBuLintermsofintrinsic} - \eqref{E:nablaAFarBCintermsofintrinsic}.
	
\end{proof}

\subsection{The electromagnetic decompositions of \texorpdfstring{$\Far_{\mu \nu}$}{the Faraday tensor} and
 \texorpdfstring{$\Max_{\mu \nu}$}{the Maxwell tensor}} \label{SS:electromagneticdecomposition}

In this section, we provide the familiar decomposition of the Faraday tensor $\Far$ into the \emph{electric field} $\Electricfield_{\mu},$ and the \emph{magnetic induction field} $\Magneticinduction_{\mu}.$ We also introduce a related decomposition of the Maxwell tensor $\Max$ into the \emph{electric displacement} $\Displacement_{\mu},$ and the \emph{magnetic field} $\Magneticfield_{\mu}.$ For computational purposes, it will also be convenient to introduce two additional one-forms,  $\SFar_{\mu}$ and $\SFardual_{\mu},$ which are related to $E_{\mu}$ and $B_{\mu}.$

The decompositions require the introduction of $i,$ the interior product operator, the action of which on $\Far$
is defined by the requirement that the following relation should hold for all vectors $X,Y:$
\begin{align}
	i_X \Far(Y) = \Far(Y,X) = \Far_{\kappa \lambda}Y^{\kappa} X^{\lambda}. 
\end{align}
In an arbitrary coordinate system, $i_X \Far$ can be expressed as
\begin{align}
	(i_X \Far)_{\mu} \eqdef \Far_{\mu \kappa} X^{\kappa}.
\end{align}

\begin{definition}
In terms of the Faraday tensor $\Far$ and the Maxwell tensor $\Max$ (an expression for $\Max$ in the MBI
system is given in \eqref{E:MaxMBI}), the one-forms $E,B,D,H,\SFar,$ and $\SFardual$ are 

\begin{subequations}
\begin{align} \label{E:EBDHdef}
	\Electricfield & \eqdef i_{T_{(0)}} \Far, & \Magneticinduction & \eqdef - \iota_{T_{(0)}} \Fardual, & \Displacement \eqdef - 
		i_{T_{(0)}} \Maxdual, &  & \Magneticfield & \eqdef - \iota_{T_{(0)}} \Max,
\end{align}

\begin{align} \label{E:PQdef}
	\SFar & \eqdef i_{S} \Far, & \SFardual & \eqdef i_{S} \Fardual,
\end{align}
\end{subequations}
where $T_{(0)}$ and $S$ are the time translation and scaling vectorfields defined in \eqref{E:Translationsdef} and \eqref{E:Sdef}.

\end{definition}

In components relative to the inertial coordinate system $\lbrace x^{\mu} \rbrace_{\mu = 0,1,2,3},$ we have that

\begin{subequations}
\begin{align} \label{E:EBDHinertialcomponents}
	\Electricfield_{\mu} & = \Far_{\mu 0}, & \Magneticinduction_{\mu} & = - \Fardual_{\mu 0}, & \Displacement_{\mu} & = - \Maxdual_{\mu 0},  	& \Magneticfield_{\mu} & = - \Max_{\mu 0},
\end{align}

\begin{align} \label{E:PQinertialcomponents}
	\SFar_{\mu} & = x^{\kappa}\Far_{\mu \kappa}, & \SFardual_{\mu} & = x^{\kappa} \Fardual_{\mu \kappa}.
\end{align}
\end{subequations}
The anti-symmetry of $\Far, \Max$ implies that $\Electricfield, \Magneticinduction, \Displacement,$ and $\Magneticfield$ are tangent to the hyperplane $\Sigma_t.$ In the inertial coordinate system, this is equivalent to $\Electricfield_0 = \Magneticinduction_0 = \Displacement_0 = \Magneticfield_0 = 0.$ We may therefore view these four quantities as one-forms that are intrinsic to $\Sigma_t.$

\begin{remark}
	Our definition of $\Magneticinduction$ coincides with the one commonly found in the physics literature, but it has the 
	opposite sign convention of the definition given in \cite{dCsK1990}.
\end{remark}

The identity 
\begin{align}
	g_{\kappa \lambda} X^{\kappa} X^{\lambda} \Far_{\mu \nu} = (i_X \Far)_{\mu} X_{\nu}  - (i_X \Far)_{\nu} X_{\mu} 
		+ X^{\kappa}(i_X \Fardual)^{\lambda} \epsilon_{\kappa \lambda \mu \nu}
\end{align}
shows that $\Far$ is completely determined by $i_X \Far$ and $i_X \Fardual$ at any spacetime point 
where $g_{\kappa \lambda} X^{\kappa} X^{\lambda} \neq 0.$ In particular, $\Electricfield$ and $\Magneticinduction$ completely determine $\Far,$ and it can be checked that in the inertial coordinate system,

\begin{subequations}
\begin{align} \label{E:FarEBrelations}
	\Far_{j0} = \Electricfield_j, && \Far_{jk} = \uepsilon_{ijk}\Magneticinduction^i, 
		&& \Magneticinduction_j = \frac{1}{2} \uepsilon_j^{\ ab} \Far_{ab}.
\end{align}

\begin{align} \label{E:MaxDHrelations}
	\Max_{j0} = -\Magneticfield_j, && \Max_{jk} = \uepsilon_{ijk}\Displacement^i, 
		&& \Displacement_j = \frac{1}{2} \uepsilon_j^{\ ab} \Max_{ab}.
\end{align}
\end{subequations}

Now in linear Maxwell-Maxwell theory, the relations $\Electricfield = \Displacement,$ $\Magneticinduction = \Magneticfield$ hold. In contrast, in the case of the MBI system, tedious computations lead to the following relations 
(see e.g. \cite{iBB1983}, \cite{mK2004a}):

\begin{subequations}
\begin{align}
	\Displacement & = \frac{\Electricfield + (\Electricfield_a \Magneticinduction^a)\Magneticinduction}{\big(1 + |\Magneticinduction|^2 - 
		|\Electricfield|^2 - (\Electricfield_a \Magneticinduction^a)^2 \big)^{1/2}}, 
		\label{E:DisplacementintermsofElectricfieldMagneticinduction} \\
	\Magneticfield & = \frac{\Magneticinduction - (\Electricfield_a \Magneticinduction^a)\Electricfield}{\big(1 + |\Magneticinduction|^2 
		- |\Electricfield|^2 - (\Electricfield_a \Magneticinduction^a)^2 \big)^{1/2}}.
\end{align}
\end{subequations}

Furthermore, we have that

\begin{subequations}
\begin{align} 
	\Electricfield & = \frac{\Displacement + \Magneticinduction \times (\Displacement \times \Magneticinduction)}
		{(1 + |\Magneticinduction|^2 + |\Displacement|^2 + |\Displacement \times \Magneticinduction|^2)^{1/2}}, 
		\label{E:EintermsofDB} \\
	\Magneticfield & = \frac{\Magneticinduction - \Displacement \times (\Displacement \times \Magneticinduction)}
		{(1 + |\Magneticinduction|^2 + |\Displacement|^2 + |\Displacement \times \Magneticinduction|^2)^{1/2}}
		\label{E:HintermsofDB}.
\end{align}
\end{subequations}
In the above formulas, $\times$ denotes the usual intrinsic cross product on $\Sigma_t;$ see \eqref{E:Crossproduct}.

Working in our inertial coordinate system, we set $\nu = 0$ in equations \eqref{E:BianchiEM} - \eqref{E:Euler-Lagrange} and use \eqref{E:EBDHinertialcomponents} to deduce the \emph{constraint} equations for $\Magneticinduction$ and $\Displacement:$
\begin{subequations}
\begin{align}
	\udiv \Displacement & = 0, \label{E:Dconstraint} \\ 
	\udiv \Magneticinduction & = 0. \label{E:Bconstraint}
\end{align}
\end{subequations}
Setting $\lambda = 0, \mu = a, \nu = b$ in \eqref{E:dFis0nullsection}, then contracting against $\uepsilon_j^{\ a b}$
and using \eqref{E:FarEBrelations}, we deduce that

\begin{subequations}
\begin{align} \label{E:partialtBisminuscurlE}
	\partial_t \Magneticinduction = - \ucurcl \Electricfield.
\end{align}
Similarly, we use an equivalent (modulo equation \eqref{E:dFis0nullsection}) version  of \eqref{E:dMis0nullsection},
namely $\nabla_{\lambda} \Max_{\mu \nu} + \nabla_{\mu} \Max_{\nu \lambda} + \nabla_{\nu} \Max_{\lambda \mu} = 0,$
set $\lambda = 0, \mu = a, \nu = b,$ contract against $\uepsilon_j^{\ a b},$ and use \eqref{E:MaxDHrelations} to deduce that

\begin{align} \label{E:partialtDiscurlH}
	\partial_t \Displacement = \ucurcl \Magneticfield.
\end{align}
\end{subequations}

Now from \eqref{E:EintermsofDB}, \eqref{E:HintermsofDB}, \eqref{E:partialtBisminuscurlE}, and \eqref{E:partialtDiscurlH},
it follows that $\Magneticinduction$ and $\Displacement$ satisfy the following evolution equations:

\begin{subequations}
\begin{align}
	\partial_t \Magneticinduction & = - \ucurcl \bigg\lbrace \frac{\Displacement + \Magneticinduction \times (\Displacement \times \Magneticinduction)}
		{(1 + |\Magneticinduction|^2 + |\Displacement|^2 + |\Displacement \times \Magneticinduction|^2)^{1/2}} \bigg\rbrace,  
		\label{E:partialtBintermsofBandD} \\
	\partial_t \Displacement & = \ucurcl \bigg\lbrace \frac{\Magneticinduction - \Displacement \times (\Displacement \times 
		\Magneticinduction)}{(1 + |\Magneticinduction|^2 + |\Displacement|^2 + |\Displacement \times \Magneticinduction|^2)^{1/2}}
		\bigg\rbrace. \label{E:partialtDintermsofBandD}
\end{align}
\end{subequations}
Equations \eqref{E:Dconstraint} - \eqref{E:Bconstraint} plus \eqref{E:partialtBintermsofBandD} - \eqref{E:partialtDintermsofBandD} are
equivalent to the MBI equations \eqref{E:modifieddFis0summary} - \eqref{E:HmodifieddMis0summary}.

In regards to the terminology ``constraint equations'' used above, we make the following remark: it follows that if
\eqref{E:Dconstraint} - \eqref{E:Bconstraint} are satisfied along $\Sigma_0,$ and if $\Displacement, \Magneticinduction$ are classical solutions to
the evolution equations \eqref{E:partialtBintermsofBandD} - \eqref{E:partialtDintermsofBandD} existing on the slab $[T_-,T_+] \times \mathbb{R}^3$, then \eqref{E:Dconstraint} - \eqref{E:Bconstraint} are satisfied in the same slab; i.e., the well-known
identity $\udiv \circ \ucurcl = 0$ implies that $\nabla_t \udiv \Magneticinduction = \nabla_t \udiv \Displacement = 0.$

\begin{remark} \label{R:BDHyperbolicity}
	Using definition \eqref{E:ldef} and the relation \eqref{E:EintermsofDB}, 
	and performing some tedious calculations, it follows that
	
	\begin{align} \label{E:lMBIintermsofBD}
		\ell_{(MBI)}^2 = \frac{(1 + |\Magneticinduction|^2)^2}{1 + 
		|\Magneticinduction|^2 + |\Displacement|^2 + |\Magneticinduction \times \Displacement|^2}.
	\end{align}
	Consequently, as shown by Proposition \ref{P:LocalExistenceCurrent} and Proposition \ref{P:LocalExistence},
	the regime of hyperbolicity for the MBI system has a 
	particularly nice interpretation in terms of the state-space variables $(B,D):$ the MBI equations are
	well-defined and hyperbolic for all finite values of $(\Magneticinduction,\Displacement).$ 
\end{remark}

\section{Commutation Lemmas} \label{S:Commutation}
\setcounter{equation}{0}

In this section, we prove some basic commutation lemmas that will be used in the following sections.
\\

\noindent \hrulefill
\ \\

\begin{lemma} \label{L:Minkowskiisflat}
	Let $X$ and $Y$ be vectorfields, let $U$ be any tensorfield, and let $\nabla$ denote the Levi-Civita connection corresponding 
	to the Minkowski spacetime metric. Then 
	\begin{align} \label{E:Minkowskiisflat}
		[\nabla_X, \nabla_Y] U & = \nabla_{[X,Y]} U.
	\end{align}
\end{lemma}
\begin{proof}
	Simply use \eqref{E:Curvature} and the fact that Minkowski space is flat.
\end{proof}

\begin{lemma} \label{L:Vectorfieldderivativenorms}
	Let $\mathcal{O}$ and $\mathcal{Z}$ be the subsets of Minkowski conformal Killing fields defined in 
	\eqref{E:Rotationsetdef} - \eqref{E:Zsetdef}. Then for any $Z \in \mathcal{Z},$ we have that
	\begin{align} \label{E:Zcovariantderivativebounds}
		|Z| & \lesssim s, & |\nabla Z| & \lesssim 1, & |\nabla_{(2)} Z| & = 0.
	\end{align}
	
	Furthermore, for any $O \in \mathcal{O},$ we have that
	\begin{align} \label{E:Rotationcovariantderivativebounds}
		|O| & \lesssim r, & |\nabla O| & = const, & |\nabla_{(2)} O| & = 0.
	\end{align}
	
\end{lemma}

\begin{proof}
	The simple computations are easily performed in the inertial coordinate system.
\end{proof}

\begin{lemma} \label{L:Liecommuteswithcoordinatederivatives}
Let $\nabla$ denote the Levi-Civita connection corresponding to the Minkowski metric, and let $I$ denote a multi-index for the set $\mathcal{Z}$ of Minkowski conformal Killing fields. Let $\Liemod_{\mathcal{Z}}^I$ be the iterated modified Lie derivative from Definitions \ref{D:Liemoddef} and \ref{D:iterated}. Then

\begin{align} \label{E:Liecommuteswithcoordinatederivatives}
	[\nabla, \Lie_{\mathcal{Z}}^I] & = 0, && [\nabla, \Liemod_{\mathcal{Z}}^I] = 0.
\end{align}

In an arbitrary coordinate system, equations \eqref{E:Liecommuteswithcoordinatederivatives} are equivalent to the following relations, which hold for all type $\binom{n}{m}$ tensorfields $U:$

\begin{align}
	\nabla_{\mu}\big\lbrace (\Lie_{\mathcal{Z}}^I U)_{\mu_1 \cdots \mu_m}^{\ \ \ \ \ \ \ \ \nu_1 \cdots \nu_n} \big\rbrace 
	& = \Lie_{\mathcal{Z}}^I \big\lbrace \nabla_{\mu}U_{\mu_1 \cdots \mu_m}^{\ \ \ \ \ \ \ \ \nu_1 \cdots \nu_n} \big\rbrace, \\
	\nabla_{\mu}\big\lbrace (\Liemod_{\mathcal{Z}}^I U)_{\mu_1 \cdots \mu_m}^{\ \ \ \ \ \ \ \ \nu_1 \cdots \nu_n} \big\rbrace 
	& = \Liemod_{\mathcal{Z}}^I \big\lbrace \nabla_{\mu} U_{\mu_1 \cdots \mu_m}^{\ \ \ \ \ \ \ \ \nu_1 \cdots \nu_n} \big\rbrace. \notag
\end{align}

\end{lemma}

\begin{proof}
	The relation \eqref{E:Liecommuteswithcoordinatederivatives} can be shown via induction in $|I|$ using
	\eqref{E:Liederivativeintermsofnabla} and the fact that $\nabla_{(2)} Z = 0$ (i.e. 
	\eqref{E:Zcovariantderivativebounds}).
\end{proof}

\begin{lemma} \label{L:LiemodZLiemodMaxwellCommutator}
	Let $I$ denote a multi-index for the set $\mathcal{Z}$ of Minkowski conformal Killing fields, and let $\Far$ be a
	two-form. Let $\Liemod_{\mathcal{Z}}^I$ be the iterated modified Lie derivative from Definitions 
	\ref{D:Liemoddef} and \ref{D:iterated}. Then
	\begin{align} \label{E:LiemodZLiemodMaxwellCommutator}
		\Liemod_{\mathcal{Z}}^I \Big\lbrace \big[(g^{-1})^{\mu \kappa} (g^{-1})^{\nu \lambda} - (g^{-1})^{\mu \lambda} 
			(g^{-1})^{\nu \kappa} \big] \nabla_{\mu} \Far_{\kappa \lambda} \Big\rbrace
		= \big[(g^{-1})^{\mu \kappa} (g^{-1})^{\nu \lambda} - (g^{-1})^{\mu \lambda} (g^{-1})^{\nu \kappa} \big] \nabla_{\mu} 
		\Lie_{\mathcal{Z}}^I \Far_{\kappa \lambda}.
	\end{align}
\end{lemma}

\begin{proof}
	Let $Z \in \mathcal{Z}.$ By the Leibniz rule, \eqref{E:LieZonmupper}, and Lemma \ref{L:Liecommuteswithcoordinatederivatives}, 
	we have that 
	
	\begin{align}
		\Lie_{Z}\Big\lbrace \big[(g^{-1})^{\mu \kappa} (g^{-1})^{\nu \lambda} - (g^{-1})^{\mu \lambda} (g^{-1})^{\nu \kappa} \big] 
			\nabla_{\mu} \Far_{\kappa \lambda} \Big\rbrace
		& = -2c_Z \big[(g^{-1})^{\mu \kappa} (g^{-1})^{\nu \lambda} - (g^{-1})^{\mu \lambda} (g^{-1})^{\nu \kappa} \big] 
			\nabla_{\mu} \Far_{\kappa \lambda} \\
		& \ \ + \big[(g^{-1})^{\mu \kappa} (g^{-1})^{\nu \lambda} - (g^{-1})^{\mu \lambda} (g^{-1})^{\nu \kappa}\big] \nabla_{\mu} 
			\Lie_Z \Far_{\kappa \lambda}. \notag
	\end{align}
	It thus follows from Definition \ref{D:Liemoddef} that 
	\begin{align}
		\Liemod_{Z}\Big\lbrace\big[(g^{-1})^{\mu \kappa} (g^{-1})^{\nu \lambda} - (g^{-1})^{\mu \lambda} (g^{-1})^{\nu \kappa}\big] 
			\nabla_{\mu} \Far_{\kappa \lambda} \Big\rbrace
		= \big[(g^{-1})^{\mu \kappa} (g^{-1})^{\nu \lambda} - (g^{-1})^{\mu \lambda} (g^{-1})^{\nu \kappa}\big] \nabla_{\mu} \Lie_Z 
			\Far_{\kappa \lambda}.
	\end{align}
	This implies \eqref{E:LiemodZLiemodMaxwellCommutator} in the case $|I| = 1.$ The general case now follows inductively.
\end{proof}

\begin{lemma} \label{L:LLunderlinecommutewithLieO}
	If $O \in \mathcal{O},$ then the vectorfields, $\uL, L, O$ mutually commute:
\begin{align}
	[\uL, L] & = 0, && [\uL,O] = 0, && [L, O] = 0. \label{E:RotationLuLBracketis0} 
\end{align}

Furthermore, with $q \eqdef r - t,$ $s \eqdef r + t$ denoting the null coordinates, and $S \eqdef x^{\kappa} \partial_{\kappa}$ denoting the scaling vectorfield, the following differential operator commutation relations hold when applied to arbitrary tensorfields:

\begin{subequations}
\begin{align}
	[\nabla_L, \nabla_{\uL}] & = 0, & [q \nabla_{\uL}, s \nabla_L] & = 0, \label{E:nablaLnablaunderlineLcommute} \\
	[\nabla_{\uL}, \nabla_O] & = 0, & [\nabla_L, \nabla_O] & = 0, \label{E:nablaLuLnablaOcommute} \\
	[\nabla_{\uL}, \Lie_{O}] & = 0, & [\nabla_L, \Lie_{O}] & = 0, \label{E:nablaLuLLieOcommute} \\
	[q \nabla_{\uL}, \Lie_S] & = 0, & [s \nabla_{L}, \Lie_S] & = 0. \label{E:snablauLLieScommute}
\end{align}
\end{subequations}

\end{lemma}

\begin{proof}
	\eqref{E:RotationLuLBracketis0} follows from simple computations. \eqref{E:nablaLnablaunderlineLcommute} then follows from 
	\eqref{E:Minkowskiisflat}, \eqref{E:RotationLuLBracketis0}, and the fact that $\nabla_{\uL} s = \nabla_L q = 0.$
	\eqref{E:nablaLuLnablaOcommute} also follows from \eqref{E:Minkowskiisflat} and \eqref{E:RotationLuLBracketis0}. 
	\eqref{E:nablaLuLLieOcommute} follows from \eqref{E:Liederivativeintermsofnabla}, 
	\eqref{E:Rotationcovariantderivativebounds}, and \eqref{E:nablaLuLnablaOcommute}. 
	\eqref{E:snablauLLieScommute} follows from \eqref{E:nablaLnablaunderlineLcommute}, identity $2 S = s L - q \uL,$
	and the fact that by \eqref{E:Liederivativeintermsofnabla}, $\Lie_S U = \nabla_S U + (m-n) U$ for a type
	$\binom{n}{m}$ tensorfield $U.$
\end{proof}

\begin{lemma} \label{L:derivativesofanggandepsilon}
		Let $O \in \mathcal{O},$ and let 
		$\epsilon_{\kappa \lambda \mu \nu},$ $\angg_{\mu \nu},$ and $\angepsilon_{\mu \nu}$
		be as defined in \eqref{E:volumeform}, \eqref{E:anggdef}, and \eqref{E:angvolumeformdef} respectively.
		Then	
	
	\begin{subequations}
	\begin{align}
		\nabla_{\uL} \angg_{\mu \nu}  & = \nabla_L \angg_{\mu \nu} = 0, \label{E:nablaLanggis0} \\
		\nabla_{\uL} \angepsilon_{\mu \nu}  & = \nabla_L \angepsilon_{\mu \nu} = 0, \label{E:nablaL2epsilonis0} \\
		\angn_{\lambda} \angepsilon_{\mu \nu} & = 0, \label{E:angnangepsilonis0} \\
		\Lie_O \epsilon_{\kappa \lambda \mu \nu} & = 0, \label{E:LieO4epsilonis0} \\
		\Lie_O \angg_{\mu \nu} & = 0, \label{E:LieOanggis0} \\
		\Lie_O \angepsilon_{\mu \nu} & = 0. \label{E:LieO2epsilonis0}
	\end{align}
	\end{subequations}

\end{lemma}

\begin{proof}
	The relation \eqref{E:nablaLanggis0} follows from definition \eqref{E:CovariantDerivativeofgisZeroIndices},
	\eqref{E:anggdef}, and \eqref{E:LanduLaregeodesic} - \eqref{E:nablaLuLis0}. Using also 
	\eqref{E:volumeformcovariantlyconstant}, \eqref{E:nablaL2epsilonis0} follows similarly. 
	\eqref{E:angnangepsilonis0} follows from \eqref{E:volumeformcovariantlyconstant}, the formula 
	\eqref{E:SphereIntrinsicintermsofExtrinsic}, and the fact that the null second fundamental forms of the spheres $S_{r,t}$ are 
	tangent to $S_{r,t}$ (i.e., Lemma \ref{L:SecondFundamentalFormsSymmetric}).
	\eqref{E:LieO4epsilonis0} follows from \eqref{E:LieX4volume} and the fact that ${^{(O)}\pi_{\ \nu}^{\mu}} = 0$ (i.e., $O$ is 
	a Killing field). \eqref{E:LieOanggis0} and \eqref{E:LieO2epsilonis0} follow from definitions \eqref{E:volumeform} and 
	\eqref{E:anggdef}, \eqref{E:RotationLuLBracketis0}, and \eqref{E:LieO4epsilonis0}. 
\end{proof}

\begin{corollary} \label{C:LieRotationcommuteswithnulldecomp}
	Let $\Far$ be a two-form and let $\ualpha,$ $\alpha,$ $\rho,$ $\sigma$ be its null components.
	Let $O \in \mathcal{O}.$ Then $\Lie_O \ualpha[\Far] = \ualpha[\Lie_O \Far],$
	$\Lie_O \alpha[\Far] = \alpha[\Lie_O \Far],$ $\Lie_O \rho[\Far] = \rho[\Lie_O \Far],$ 
	and $\Lie_O \sigma[\Far] = \sigma[\Lie_O \Far].$ An analogous result
	holds the operators $\nabla_{\uL}$ and $\nabla_L;$ i.e.,
	$\Lie_O, \nabla_{\uL},$ and $\nabla_L$ commute with the null decomposition of $\Far.$
\end{corollary}

\begin{proof}
	Corollary \ref{C:LieRotationcommuteswithnulldecomp} follows from Definition \ref{D:null},
	\eqref{E:LanduLaregeodesic} - \eqref{E:nablaLuLis0}, \eqref{E:RotationLuLBracketis0},
	\eqref{E:nablaLanggis0} - \eqref{E:nablaL2epsilonis0},
	and \eqref{E:LieOanggis0} - \eqref{E:LieO2epsilonis0}.
\end{proof}

\begin{lemma} \label{L:Liederivativeandcovariantderivativeinteriorproductcommutator}
	Let $X$ and $Y$ be vectorfields, and let $\Far$ be a two-form. Let $i_X \Far$ be the interior product defined by the 
	requirement that $i_X \Far(Y) = \Far(Y,X)$ holds for all pairs of vectors $X,Y.$ Then
	
	\begin{subequations}
	\begin{align} 
		i_X \Lie_Y \Far - \Lie_Y i_X \Far & = i_{[X,Y]} \Far, \label{E:Liederivativeinteriorproductcommutator} \\
		\nabla_X (i_Y \Far) - i_Y(\nabla_X \Far) & = i_{\nabla_X Y} \Far. \label{E:Covariantderivativeinteriorproductcommutator}
	\end{align}
	\end{subequations}
\end{lemma}

\begin{proof}
	The relation \eqref{E:Liederivativeinteriorproductcommutator} follows from the Leibniz rule and the fact that
	$- \Lie_Y X = [X,Y].$ The relation \eqref{E:Covariantderivativeinteriorproductcommutator} follows from the Leibniz rule.
\end{proof}

We leave the computations required to prove the next lemma up to the reader.
\begin{lemma} \label{L:ConformalKillingFieldCommuatators} \cite[pg. 139]{dCsK1990}
	The $15$ generators $\lbrace T_{(\mu)}, \Omega_{(\mu \nu)}, S, K_{(\mu)} 
	\rbrace_{0 \leq \mu < \nu \leq 3}$ of the Minkowski conformal Killing fields, which are defined in Section 	
	\ref{S:ConformalKilling}, satisfy the following commutation relations:
	
	\begin{subequations}
	\begin{align}
		[T_{(\mu)}, T_{(\nu)}] & = 0, && (\mu, \nu = 0,1,2,3), \label{E:translationscommute} \\
		[T_{(\kappa)}, \Omega_{(\mu \nu)}] & = g_{\kappa \mu}T_{(\nu)} - g_{\kappa \nu}T_{(\mu)}, 
			&& (\kappa, \mu, \nu = 0,1,2,3), \\
		[T_{(\mu)}, S] & = T_{(\mu)}, & & (\kappa, \mu, \nu = 0,1,2,3), \\
		[K_{(\mu)}, T_{(\nu)}] & = 2 g_{\mu \nu} S + 2 \Omega_{(\mu \nu)}, && (\mu, \nu = 0,1,2,3), \\
		[\Omega_{(\kappa \lambda)}, \Omega_{(\mu \nu)}] & = g_{\kappa \mu} \Omega_{(\nu \lambda)} 
			- g_{\kappa \nu} \Omega_{(\mu \lambda)} + g_{\lambda \mu} \Omega_{(\kappa \nu)}
			- g_{\lambda \nu} \Omega_{(\kappa \mu)}, && (\kappa, \lambda, \mu, \nu = 0,1,2,3), \\
		[\Omega_{(\mu \nu)}, S] & = 0, && (\mu, \nu = 0,1,2,3), \\
		[K_{(\kappa)}, \Omega_{(\mu \nu)}] & = g_{\kappa \mu}K_{(\nu)} - g_{\mu \nu} K_{(\mu)}, 
			&& (\mu, \nu = 0,1,2,3), \\
		[K_{(\mu)}, S] & = K_{(\mu)}, && (\mu = 0,1,2,3), \\
		[K_{(\mu)}, K_{(\nu)}] & = 0, && (\mu, \nu = 0,1,2,3).
	\end{align}	
	\end{subequations}
\end{lemma}

\hfill $\qed$

The following simple corollary follows directly from Lemma \ref{L:ConformalKillingFieldCommuatators}.
\begin{corollary} \label{C:CommutatorofZTandS}
	Let ${\mathbf{T}}$ and $\mathbf{Z}$ denote the Lie algebras of vectorfields
	generated by the sets $\mathscr{T}$ and $\mathcal{Z}$ respectively. Then for $\mu = 0,1,2,3,$
	we have
	
	\begin{subequations}
	\begin{align}
		[T_{(\mu)},\mathbf{Z}] & \subset \mathbf{T}, \label{E:commutatorTZ}\\
		[S,\mathbf{Z}] & \subset \mathbf{T}. \label{E:commutatorSZ} 
 \end{align}
	\end{subequations}
	
\end{corollary}

\hfill $\qed$

\begin{lemma} \label{L:NullSecondFundamantalFormSimpleExpression}
	The null second fundamental forms $\nabla_{\mu} \uL^{\nu}$ and $\nabla_{\mu} L^{\nu}$
	(see Definition \ref{D:SecondFundamental}) can be expressed as
	
	\begin{subequations}
	\begin{align}
		\nabla_{\mu} \uL^{\nu} & = - r^{-1} \angg_{\mu}^{\ \nu}, 
			\label{E:NablauL} \\
		\nabla_{\mu} L^{\nu} & = r^{-1} \angg_{\mu}^{\ \nu}. \label{E:NablaL}
	\end{align}
	\end{subequations}
	
	Furthermore, the intrinsic covariant derivatives of $\nabla \uL$ and $\nabla L$
	vanish:
	
	\begin{subequations}
	\begin{align} \label{E:HigherIntrinsicCovariantDerivativesofuLandLVanish}
		\angn_{(M)} \uL \eqdef \angn_{(M-1)} \nabla \uL & = 0, && (M \geq 2), \\
		\angn_{(M)} L \eqdef \angn_{(M-1)} \nabla L & = 0, && (M \geq 2).
	\end{align}
	\end{subequations}

\end{lemma}

\begin{proof}
	\eqref{E:NablauL} - \eqref{E:NablaL} follow from simple computations.
	\eqref{E:HigherIntrinsicCovariantDerivativesofuLandLVanish} follows then 
	follows from the Leibniz rule, \eqref{E:AnggIntrinsicDerivativeis0}, and the fact that $\angn r$ = 0.
\end{proof}

\begin{lemma} \label{L:LunderlineLderivativesofSrttangenttensorfield}
	Let $U$ be a type $\binom{n}{m}$ tensorfield tangent to the spheres $S_{r,t}.$ Then 
	for all integers $k,l \geq 0,$ $\nabla_{\uL}^k \nabla_{L}^l U$ is a tensorfield tangent to the 
	spheres $S_{r,t},$ and
	
	\begin{align} \label{E:LunderlineLderivativesofSrttangenttensorfield}
		\angn_{\lambda}(\nabla_{L}^l \nabla_{\uL}^k U_{\mu_1 \cdots \mu_m}^{\ \ \ \ \ \ \ \nu_1 \cdots \nu_n})
		& = \nabla_{\uL}^k  \nabla_{L}^l \angn_{\lambda}U_{\mu_1 \cdots \mu_m}^{\ \ \ \ \ \ \ \nu_1 \cdots \nu_n} \\
		& \ \ - k r^{-1} \nabla_L^l \nabla_{\uL}^{k-1} \angn_{\lambda}
			U_{\mu_1 \cdots \mu_m}^{\ \ \ \ \ \ \ \nu_1 \cdots \nu_n} \notag \\
		& \ \ + l r^{-1} \nabla_L^{l-1} \nabla_{\uL}^k \angn_{\lambda}
			U_{\mu_1 \cdots \mu_m}^{\ \ \ \ \ \ \ \nu_1 \cdots \nu_n}. \notag 
	\end{align}
\end{lemma}

\begin{proof}
	The fact that $\nabla_{L}^l \nabla_{\uL}^k U$ is tangent to the spheres follows from contracting
	any of its indices with either $\uL$ or $L,$ and using \eqref{E:LanduLaregeodesic} - \eqref{E:nablaLuLis0} to
	commute the contractions through the derivatives; the result is $0.$
	The relation \eqref{E:LunderlineLderivativesofSrttangenttensorfield} follows from the Leibniz rule,
	\eqref{E:SphereIntrinsicintermsofExtrinsic}, and Lemma \ref{L:NullSecondFundamantalFormSimpleExpression}.	
\end{proof}

\begin{corollary} \label{C:CommuteIntrinsicCovraiantandLderivatives}
	Let $k,l,m \geq 0$ be integers, and let $U$ be a tensorfield tangent to the spheres $S_{r,t}.$
	Then
	
	\begin{align} \label{E:CommuteIntrinsicCovraiantandLderivatives}
		|\nabla_{\uL}^k \nabla_L^l \angn_{(m)} U - \angn_{(m)} \nabla_{\uL}^k \nabla_L^l U | 
			\lesssim	\mathop{\sum_{k' + l' < k + l}}_{l' \leq l, k' \leq k, m' < m} r^{-\lbrace(k - k') + (l - l')\rbrace}| 
			 \nabla_{\uL}^{k'} \nabla_L^{l'} \angn_{(m')} U|. 
	\end{align}

\end{corollary}

\begin{proof}
	Inequality \eqref{E:CommuteIntrinsicCovraiantandLderivatives} follows from repeated use of the Leibniz rule,
	Lemma \ref{L:LunderlineLderivativesofSrttangenttensorfield}, and the fact that
	$\angn r = 0.$
\end{proof}

	The following simple lemma gives pointwise bounds for the tensorfields $L,\uL,\angg,$
	and their full \emph{spacetime} covariant derivatives.
	
	\begin{lemma} \label{L:SphereProjectionCovariantDerivativesNorms}
		Let $M \geq 0$ be any integer. Let $\uL, L$ be the null vectorfields defined in \eqref{E:uLdef} - \eqref{E:Ldef}, 
		and let $\angg$ be the first fundamental form of the $S_{r,t}$ defined in \eqref{E:anggdef}. Then
		
		\begin{subequations}
		\begin{align}
				|\nabla_{(M)} \uL| & \lesssim r^{-M}, \label{E:underlineLcovariantderivativetensorbound} \\
			|\nabla_{(M)} L| & \lesssim r^{-M}, \label{E:Lcovariantderivativetensorbound} 
		\end{align}
		\end{subequations}
		and
		
		\begin{align}
			|\nabla_{(M)} \angg| & \lesssim r^{-M}. \label{E:Srtmetriccovariantderivativetensorbound} 
		\end{align}
		
	\end{lemma}

\begin{proof}
	Since in the inertial coordinate system, $\uL^{\mu} = (-1,\omega^1,\omega^2,\omega^3),$ $L^{\mu} = 
	(1,\omega^1,\omega^2,\omega^3),$ and $\omega^j \eqdef x^j/r,$ it is easy to check directly that
	\eqref{E:underlineLcovariantderivativetensorbound} - \eqref{E:Lcovariantderivativetensorbound} hold.
	Inequality \eqref{E:Srtmetriccovariantderivativetensorbound} then follows from \eqref{E:CovariantDerivativeofgisZeroIndices},
	definition \eqref{E:anggdef},\eqref{E:underlineLcovariantderivativetensorbound}, \eqref{E:Lcovariantderivativetensorbound}, 
	and the Leibniz rule.
\end{proof}

\begin{lemma} \label{L:rotationvectorfieldsintrinsiccovariantderivative}
	Let $O \in \mathcal{O}.$ Then
	
	\begin{align} \label{E:rotationvectorfieldsintrinsiccovariantderivative}
		|\angn_{(M)} O| \lesssim r^{1-M}.
	\end{align}
\end{lemma}

\begin{proof}
	follows from repeated use of \eqref{E:SphereIntrinsicintermsofExtrinsic}, \eqref{E:Rotationcovariantderivativebounds},
	and Lemma \ref{L:SphereProjectionCovariantDerivativesNorms}.
\end{proof}

\begin{lemma}  \label{L:rotationsareabasis}
		Let $U_{\mu_1 \cdots \mu_m}$ be a type $\binom{0}{m}$ tensorfield tangent to the spheres 
		$S_{r,t}.$ Let $O_{\iota}, \iota =1,2,3$ be an enumeration of the three rotational vectorfields belonging
		to the set $\mathcal{O}.$ For any rotational multi-index $I= (\iota_1, \cdots, \iota_m)$ of length $m,$ where
		$\iota_i \in \lbrace 1,2,3 \rbrace,$ we define 
		$|\mathcal{O}^I U| \eqdef |O_{\iota_1}^{\kappa_1} \cdots O_{\iota_m}^{\kappa_m} U_{\kappa_1 \cdots \kappa_m}|.$
		Then the following pointwise estimate holds:
	
	\begin{align} \label{E:rotationsareabasis}
		\sum_{|I| = m} |\mathcal{O}^I U| \approx r^m |U|.
	\end{align}
\end{lemma}

\begin{proof}
	It is easy to check that at any point in $p \in S_{r,t},$ a pair of the rotations, say $O_{\iota_1}$ and $O_{\iota_2},$ has 
	the following two properties: \textbf{i)} $|O_{\iota_1}|^2, |O_{\iota_2}|^2 \geq r^2/3;$ \textbf{ii)} 
	$|\angg(O_{\iota_1}, O_{\iota_2})|^2 \leq \frac{1}{4} |O_{\iota_1}|^2 |O_{\iota_2}|^2$ (where
	$\angg(X,Y) \eqdef \angg_{\kappa \lambda} X^{\kappa}Y^{\lambda}$). This latter property implies that 
	the (smallest) angle between $O_{\iota_1}$ and $O_{\iota_2},$ viewed as vectors in the two-dimensional plane $T_p S_{r,t},$
	is at least $60$ degrees. Therefore, any covector $\xi \in T_p^* S_{r,t},$ has a corresponding $\angg-$dual vector $X \in T_p 
	S_{r,t}$ that makes an angle $\leq 60$ degrees with one of $^+ _- O_{\iota_1}, \ ^+_- O_{\iota_2},$ and it follows that 
	either $|\angg(O_{\iota_1}, X)|^2 \geq \frac{r^2}{12} |X|^2 = \frac{r^2}{12} |\xi|^2$ or 
	$|\angg(O_{\iota_2},X)|^2 \geq \frac{r^2}{12} |X|^2 = \frac{r^2}{12} |\xi|^2.$ Thus, $\sum_{|I| = 1} |\mathcal{O}^I \xi|^2 
	\gtrsim r^{2} |\xi|^2.$
	
	On the other hand, the reverse inequality $\sum_{|I| = 1} |\mathcal{O}^I \xi|^2 \lesssim r^{2} |\xi|^2$ trivially follows 
	from \eqref{E:Rotationcovariantderivativebounds}. We have thus shown \eqref{E:rotationsareabasis} in the case $m=1.$ The 
	general case \eqref{E:rotationsareabasis} follows from applying similar reasoning to the elements 
	$\lbrace e_{i_1}^* \otimes \cdots \otimes e_{i_m}^* \rbrace$ of a  $\angg-$orthonormal basis for 
	$\underbrace{T_p^* S_{r,t} \otimes \cdots \otimes T_p^* S_{r,t}}_{\mbox{$m$ copies}}.$

\end{proof}

\begin{lemma} \label{L:LieDerivativeIntrinsicCovariantDerivativeComparison}
	Let $N \geq 0$ be an integer. Then for any tensorfield $U,$ we have that
	\begin{align} \label{E:rotatiaonalLiederivativesequivalenttocovariantrotationalderivatives}
		|U|_{\Lie_{\mathcal{O}};N} \approx | U |_{\nabla_{\mathcal{O}};N}.
	\end{align}
	
	Furthermore, if $U$ is tangent to the spheres $S_{r,t},$ we have the following estimates:
	
	\begin{subequations}
	\begin{align}
		|U|_{\angn_{\mathcal{O}};N} & \approx \sum_{n=0}^N r^{n} |\angn_{(n)} U|, 
			\label{E:rotationalintrinsiccovarintderivativesequivalentrweightedintrinsiccovariantderitves} \\
		|U|_{\angn_{\mathcal{O}};N} & \approx |U|_{\Lie_{\mathcal{O}};N},
			\label{E:rotatiaonalLiederivativesequivalenttointrinsiccovariantrotationalderivatives} \\
		|U|_{\Lie_{\mathcal{O}};N} & \approx \sum_{n=0}^N r^{n} |\angn_{(n)} U|. 
			\label{E:LieDerivativeIntrinsicCovariantDerivativeComparison}
	\end{align}
	\end{subequations}
	The norms in the above inequalities are defined in Definition \ref{D:UnweightedPointwiseNorms}.
\end{lemma}

\begin{proof}
	We use the notation defined in Lemma \ref{L:rotationsareabasis}. Inequality 
	\eqref{E:rotatiaonalLiederivativesequivalenttocovariantrotationalderivatives} can be proved inductively
	using the relation \eqref{E:Liederivativeintermsofnabla}, together with the inequalities 
	\eqref{E:Rotationcovariantderivativebounds}.
	
	To prove \eqref{E:rotationalintrinsiccovarintderivativesequivalentrweightedintrinsiccovariantderitves}, we first note that 
	Lemma \ref{L:rotationsareabasis} implies that for each integer $N \geq 0,$ the inequalities in 	
	\eqref{E:rotationalintrinsiccovarintderivativesequivalentrweightedintrinsiccovariantderitves}
	are equivalent the following inequalities:
	
	\begin{align} \label{E:covariantrotationalderivativesequivalentrotationalcontractioncovariantderivative}
		|U|_{\angn_{\mathcal{O}};N} \eqdef \sum_{|I| \leq N} |\angn_{\mathcal{O}}^I U| 
		\approx \sum_{|I| \leq N} |\mathcal{O}^I \angn_{(|I|)}U|.
	\end{align}
	To prove \eqref{E:covariantrotationalderivativesequivalentrotationalcontractioncovariantderivative}, we argue by induction,
	the base case $N=0$ being trivial. For the induction, we assume that the inequality is true in the case $N.$ Let
	$I=(\iota_1,\cdots, \iota_k)$ be a rotational multi-index with $|I| = k \leq N+1.$
	Then by the Leibniz rule and \eqref{E:rotationvectorfieldsintrinsiccovariantderivative}, we have that
	
	\begin{align}
		|\underbrace{\angn_{O_{\iota_1}} \cdots \angn_{O_{\iota_{k}}} U}_{\angn_{\mathcal{O}}^I U} 
			- \underbrace{O_{\iota_1}^{\alpha_1} \cdots O_{\iota_{k}}^{\alpha_{k}} 
			\angn_{\alpha_1} \cdots \angn_{\alpha_{k}} U}_{\mathcal{O}^I \angn_{|I|}U}|
			& \lesssim \sum_{p=1}^{k} \mathop{\sum_{a_1 + \cdots + a_{k} = p}}_{0 \leq a_i \leq i - 1} 
				\Big(\prod_{i=1}^{k} |\angn_{(a_i)} O_{\iota_i}|\Big) |\angn_{({k}-p)} U| \\ 
				& \lesssim \sum_{p'=0}^{k-1} r^{p'}|\angn_{(p')} U| \lesssim |U|_{\angn_{\mathcal{O}};k-1}, \notag 
	\end{align}
	where in the last step, we have used 
	\eqref{E:rotationalintrinsiccovarintderivativesequivalentrweightedintrinsiccovariantderitves}
	under the induction hypothesis. Summing over all $|I| = k \leq N + 1,$ we conclude that 

	\begin{align} \label{E:IntrinsiccovariantrotationminusintrinsiccovariantrotationcontractionError}
		\Big| |U|_{\angn_{\mathcal{O}};N+1} - \sum_{|I| \leq N + 1} |\mathcal{O}^I \angn_{(|I|)}U| \Big| 
		& \lesssim |U|_{\angn_{\mathcal{O}};N} 
		\lesssim \sum_{|I| \leq N} |\mathcal{O}^I \angn_{(|I|)}U|,
		\notag
	\end{align}
	where in the last step, we have used 
	\eqref{E:covariantrotationalderivativesequivalentrotationalcontractioncovariantderivative} under the
	induction hypothesis. From \eqref{E:IntrinsiccovariantrotationminusintrinsiccovariantrotationcontractionError}, 
	the induction step for \eqref{E:covariantrotationalderivativesequivalentrotationalcontractioncovariantderivative}
	easily follows.

	To prove \eqref{E:rotatiaonalLiederivativesequivalenttointrinsiccovariantrotationalderivatives}, we
	also argue by induction using the identity \eqref{E:Liederivativeintermsofnabla} (which is valid if $\nabla$ is replaced with 
	$\angn$), the base case being trivial. For the induction, we assume that inequalities hold in the case $N.$
  Let $I$ be any rotational multi-index with $|I| = k \leq N + 1.$ Repeatedly applying the Leibniz rule and identity 
  \eqref{E:Liederivativeintermsofnabla},  
	and using \eqref{E:rotationvectorfieldsintrinsiccovariantderivative}
	plus \eqref{E:rotationalintrinsiccovarintderivativesequivalentrweightedintrinsiccovariantderitves}, 
	we deduce the following inequalities:
	
	\begin{align}
		\Big|\Lie_{\mathcal{O}}^I U - \angn_{\mathcal{O}}^I U \Big|
		& \lesssim \sum_{p=1}^{k} \mathop{\sum_{a_1 + \cdots + a_{k} = p}}_{0 \leq a_i} 
			\Big(\prod_{i=1}^{k} |\angn_{(a_i)} O_{\iota_i}|\Big) |\angn_{({k}-p)} U| \\ 
			& \lesssim \sum_{p'=0}^{k-1} r^{p'}|\angn_{(p')} U|
				\lesssim |U|_{\angn_{\mathcal{O}};k-1}.  \notag 
	\end{align}
	Summing over all $|I| = k \leq N + 1,$ we conclude that 
	
	\begin{align} \label{E:RotaionalLiederivativesminusintrinsicrotationalcovariantderivativeserror}
		\Big||U|_{\Lie_{\mathcal{O}};N+1} - |U|_{\angn_{\mathcal{O}};N+1} \Big|
		& \lesssim |U|_{\angn_{\mathcal{O}};N} \lesssim |U|_{\Lie_{\mathcal{O}};N},
	\end{align}
	where in the last step, we have used the induction hypothesis. The induction step easily follows from 
	\eqref{E:RotaionalLiederivativesminusintrinsicrotationalcovariantderivativeserror}, which completes
	the proof of \eqref{E:rotatiaonalLiederivativesequivalenttointrinsiccovariantrotationalderivatives}.
	
	The estimate \eqref{E:LieDerivativeIntrinsicCovariantDerivativeComparison} then follows trivially from
	\eqref{E:rotationalintrinsiccovarintderivativesequivalentrweightedintrinsiccovariantderitves} and 	
	\eqref{E:rotatiaonalLiederivativesequivalenttointrinsiccovariantrotationalderivatives}.
\end{proof}

\begin{corollary} \label{C:nablaLnablaunderlineLLieDerivativeIntrinsicCovariantDerivativeComparison}
	If $k,l,M \geq 0$ are integers,
	and $U$ is any tensorfield tangent to the spheres $S_{r,t},$ then 
	we have the following estimate: 
	\begin{align} \label{E:nablaLnablaunderlineLLieDerivativeIntrinsicCovariantDerivativeComparison}
		|\nabla_{\uL}^k \nabla_L^l U|_{\Lie_{\mathcal{O}};M} & \approx \sum_{m=0}^M r^m 
			|\nabla_{\uL}^k \nabla_L^l \angn_{(m)} U|. 
	\end{align}
	The norm on the left-hand side of 
	\eqref{E:nablaLnablaunderlineLLieDerivativeIntrinsicCovariantDerivativeComparison} is defined in Definition 
	\ref{D:UnweightedPointwiseNorms}.
\end{corollary}

\begin{proof}
	Corollary \ref{C:nablaLnablaunderlineLLieDerivativeIntrinsicCovariantDerivativeComparison} follows from
	Corollary \ref{C:CommuteIntrinsicCovraiantandLderivatives} and Lemma
	\ref{L:LieDerivativeIntrinsicCovariantDerivativeComparison}.
\end{proof}

\section{The Energy-Momentum Tensor and the Canonical Stress} \label{S:CanonicalStress}
In this section, we discuss the building blocks of our energies. We will begin by defining the MBI system's energy-momentum tensor, which we denote by $\EMT_{(MBI)}^{\mu \nu},$ and recalling its key properties. We remark that $\EMT_{(MBI)}^{\mu \nu}$ is the usual tensor associated with energy estimates. However, in order to derive energy estimates for the derivatives
of a solution, we will need a different tensor, namely the canonical stress $\Stress_{\ \nu}^{\mu}.$ The bulk of this section is therefore devoted to an analysis of the properties of $\Stress_{\ \nu}^{\mu}.$
\\

\noindent \hrulefill
\ \\

\subsection{The energy-momentum tensor \texorpdfstring{$\EMT^{\mu \nu}$}{}}

The energy-momentum tensor corresponding to an electromagnetic Lagrangian $\Ldual$ is defined as follows:
\begin{align} \label{E:electromagnetictensorupper}
	\EMT^{\mu \nu} & \eqdef 2 \frac{\partial \Ldual}{ \partial g_{\mu \nu}} + (g^{-1})^{\mu \nu} \Ldual, & & (\mu, \nu = 0,1,2,3).
\end{align}	
It follows trivially from the definition that $\EMT^{\mu \nu}$ is symmetric:
%= \frac{1}{2} \Big( \Fardual^{\mu \kappa} \Max_{\ \kappa}^{\nu} -  \Maxdual^{\mu \kappa} \Far_{\ \kappa}^{\nu} \Big)

\begin{align} \label{E:Tissymmetric}
	\EMT^{\mu \nu} & = \EMT^{\nu \mu}.
\end{align}
We recall the fundamental divergence-free property: for energy-momentum tensors $\EMT^{\mu \nu}$ 
constructed out of a solution of the equations of motion \eqref{E:BianchiEM} - \eqref{E:Euler-Lagrange} corresponding to
$\Ldual,$ we have that 

\begin{align} \label{E:Tdivfree}
	\nabla_{\mu} \EMT^{\mu \nu} & = 0, && (\nu = 0,1,2,3).
\end{align}

In the particular case of the MBI model, \eqref{E:LMBIbetaequals1} and Lemma \ref{L:electromagneticidentities} imply that

\begin{align} \label{E:MBItensorupper}
	\EMT_{(MBI)}^{\mu \nu} & = \ell_{(MBI)}^{-1}\big(g_{\kappa \lambda} \Far^{\mu \kappa} \Far^{\nu \lambda} 
		- \Farinvariant_{(2)}^2[\Far] (g^{-1})^{\mu \nu}\big) + (g^{-1})^{\mu \nu} (1 - \ell_{(MBI)}). 
\end{align}
The next lemma is not needed for any of the main results presented in this article, but it is of interest in itself. 
It shows that $\EMT_{\mu \nu}^{(MBI)}$ satisfies the \emph{dominant energy condition}.

\begin{lemma} \label{L:dominantenergycondition}
	Let $\Far$ be any two-form for which the quantity $\ell_{(MBI)},$ which is defined in \eqref{E:ldef} is real; 
	i.e., any $\Far$ for which $1 + \Farinvariant_{(1)}[\Far] - \Farinvariant_{(2)}^2[\Far] \geq 0.$ Then the MBI energy-momentum 
	tensor of $\Far$ satisfies the dominant energy condition; that is, for any pair of future-directed causal vectors $X,Y$ we 
	have that $\EMT_{(MBI)}(X,Y) \geq 0.$ 
\end{lemma}

\begin{proof}
	Let $p \in M,$ and let $X,Y$ be any future-directed causal 
	vectors belonging to $T_p M.$ Then in the plane spanned by $X,Y,$ 
	there exist a pair of null vectors $L', \uL'$ such that $g(\uL',L') = -2,$ and such that ${X= a L' + b \uL',}$ 
	${Y= cL' + d \uL',}$ with $a,b,c,d \geq 0.$ Let us complement $\uL', L'$ with a pair of orthonormal vectors $e_1', e_2'$ 
	belonging to the 
	$g-$orthogonal complement of
	$\mbox{span} \lbrace \uL', L' \rbrace.$ The set $\lbrace \uL', L', e_1', e_2'\rbrace$ is therefore a null frame at $p.$
	Let $\alpha',\ualpha', \rho', \sigma'$ denote the null decomposition of $\Far$ with respect to this frame. By bilinearity, it 
	suffices to check that $\EMT_{(MBI)}(\uL', \uL') \geq 0, \ \EMT_{(MBI)}(L',L') \geq 0,$ and $\EMT_{(MBI)}(\uL',L') \geq 0.$ 
	We leave the following two simple calculations to the reader:
	
	\begin{align}
		\EMT_{(MBI)}(\uL', \uL') & = \ell_{(MBI)}^{-1} |\ualpha'|^2, 
		&& \EMT_{(MBI)}(L',L') = \ell_{(MBI)}^{-1} |\alpha'|^2. 
	\end{align}
	
	For the remaining term, we first calculate that
	\begin{align}
		\ell_{(MBI)} \EMT_{(MBI)}(\uL',L') & = \ualpha'_A \alpha'^A + 2 \rho'^2 + 2 \Farinvariant_{(2)}^2 + 
		2\ell_{(MBI)}(\ell_{(MBI)} - 1).
	\end{align}
	We now express $2\ell_{(MBI)}(\ell_{(MBI)} - 1) = f(\Farinvariant_{(1)} - \Farinvariant_{(2)}^2),$ 
	where $f(v) = 2(1 + v) - 2(1 + v)^{1/2},$ and $f(0) 
	= 0,$ $f'(0) = 1,$ $f''(0) = 1/2.$ Setting $v = \Farinvariant_{(1)} - \Farinvariant_{(2)}^2,$ using the convexity of $f,$ and 
	using the identity $\Farinvariant_{(1)} = - \ualpha'_A \alpha'^A - \rho'^2 + \sigma'^2,$ it thus follows that 
	
	\begin{align}
		\ell_{(MBI)} \EMT_{(MBI)}(\uL',L') \geq \ualpha'_A \alpha'^A + 2 \rho'^2 + 2 \Farinvariant_{(2)}^2 
			+ (\Farinvariant_{(1)} - \Farinvariant_{(2)}^2) \geq \rho'^2 + \sigma'^2 + \Farinvariant_{(2)}^2.
	\end{align}
	This concludes our proof of Lemma \ref{L:dominantenergycondition}.

\end{proof}

\subsection{Equations of variation}

The definition \eqref{E:mathcalENdef} of our energy involves integrals of weighted squares of the components of 
$\Lie_{\mathcal{Z}}^I \Far$ over the spacelike hypersurfaces $\Sigma_t.$ In order to prove our global existence theorem, we need to understand the evolution the energy, which in turn requires that we investigate the evolution of the $\Lie_{\mathcal{Z}}^I \Far.$ These quantities are solutions to the \emph{equations of variation}, which are the linearization of the MBI system \eqref{E:modifieddFis0summary} - \eqref{E:HmodifieddMis0summary} around the background $\Far.$ More specifically, the equations of variation in the unknowns $\dot{\Far}_{\mu \nu}$ are defined to be

\begin{subequations} 
\begin{align}
	\nabla_{\lambda} \dot{\Far}_{\mu \nu} + \nabla_{\mu} \dot{\Far}_{\nu \lambda} + \nabla_{\nu} \dot{\Far}_{\lambda \mu}
		& = \mathfrak{J}_{\lambda \mu \nu}, & & (\lambda, \mu, \nu = 0,1,2,3), \label{E:EOVBianchi} \\
	H^{\mu \nu \kappa \lambda} \nabla_{\mu} \dot{\Far}_{\kappa \lambda} & = \mathfrak{I}^{\nu}, & & (\nu = 0,1,2,3),
		\label{E:EOVMBI} 
\end{align}
\end{subequations}
where the tensorfield $H^{\mu \nu \kappa \lambda}$ is defined in \eqref{E:Hdef}, and
$\mathfrak{J}_{\lambda \mu \nu}$ and $\mathfrak{I}^{\nu}$ are inhomogeneous terms that need to be specified. In our applications below, the $\Lie_{\mathcal{Z}}^I \Far$ will play the role of $\dot{\Far}.$ In order to understand the evolution of the 
$\Lie_{\mathcal{Z}}^I \Far,$ we of course need to understand the nature of the inhomogeneous terms appearing in equations of variation that they satisfy. The next proposition provides detailed expressions for these inhomogeneities.

\begin{proposition} \label{P:Inhomogeneousterms}
	If $\Far_{\mu \nu}$ is a solution to the MBI system \eqref{E:modifieddFis0summary} - \eqref{E:HmodifieddMis0summary}, then
	$\dot{\Far}_{\mu \nu} \eqdef \Lie_{\mathcal{Z}}^I \Far_{\mu \nu}$ is a solution to the equations of variation 
	\eqref{E:EOVBianchi} - \eqref{E:EOVMBI} with inhomogeneous terms given by
	
	\begin{subequations}
	\begin{align}
		\mathfrak{J}_{\lambda \mu \nu}^{(I)} & = 0, && (\lambda, \mu, \nu = 0,1,2,3), 
			\label{E:MBIInhomogeneoustermsJvanish} \\
		\mathfrak{I}_{(I)}^{\nu} & = H_{\triangle}^{\mu \nu \kappa \lambda} \nabla_{\mu} \Big(\Lie_{\mathcal{Z}}^I \Far_{\kappa 	
			\lambda}\Big)
			-	\Liemod_{\mathcal{Z}}^I \Big(H_{\triangle}^{\mu \nu \kappa \lambda} \nabla_{\mu} \Far_{\kappa \lambda} \Big), && 
			(\nu = 0,1,2,3). \label{E:MBIInhomogenoustermsI}
	\end{align}
	\end{subequations}
	In the above formula, the tensorfield $H_{\triangle}^{\mu \nu \kappa \lambda},$ which depends quadratically on $\Far,$ is 
	defined in \eqref{E:Htriangledef}, while the iterated modified Lie derivatives $\Liemod_{\mathcal{Z}}^I$ are defined in 
	Section \ref{S:ConformalKilling}.

\end{proposition}

\begin{proof}
	Recall equation \eqref{E:modifieddFis0summary}, which states 
	that $\Far$ is a solution to $\nabla_{[\kappa} \Far_{\mu \nu]} = 0,$ 
	where $[\cdots]$ denotes antisymmetrization. Using
	property \eqref{E:Liecommuteswithcoordinatederivatives}, it therefore follows that
	\begin{align}
		0 = \Lie_{\mathcal{Z}}^I \nabla_{[\kappa} \Far_{\mu \nu]} = \nabla_{[\kappa} \Lie_{\mathcal{Z}}^I \Far_{\mu \nu]}.
	\end{align}
	This proves \eqref{E:MBIInhomogeneoustermsJvanish}.
	
	To prove \eqref{E:MBIInhomogenoustermsI}, we first recall equation \eqref{E:HmodifieddMis0summary}, 
	which states that $\Far$ is a solution to
	
	\begin{align} \label{E:MBIequationagain}
				H^{\mu \nu \kappa \lambda} \nabla_{\mu} \Far_{\kappa \lambda} \eqdef
				\bigg\lbrace \frac{1}{2}\big[(g^{-1})^{\mu \kappa} (g^{-1})^{\nu \lambda} 
					- (g^{-1})^{\mu \lambda} (g^{-1})^{\nu \kappa}\big] 
				+ H_{\triangle}^{\mu \nu \kappa \lambda} \bigg\rbrace \nabla_{\mu} \Far_{\kappa \lambda} = 0.
	\end{align}
	Applying $\Liemod_{\mathcal{Z}}^I$ to \eqref{E:MBIequationagain}
	and using \eqref{E:LiemodZLiemodMaxwellCommutator}, we conclude that
	
	\begin{align}
		\frac{1}{2} \big[(g^{-1})^{\mu \kappa} (g^{-1})^{\nu \lambda} - (g^{-1})^{\mu \lambda} (g^{-1})^{\nu \kappa}\big] 
			\nabla_{\mu} \Lie_{\mathcal{Z}}^I + \Liemod_{\mathcal{Z}}^I (H_{\triangle}^{\mu \nu \kappa \lambda} \nabla_{\mu} 
			\Far_{\kappa \lambda}) = 0.
	\end{align}
	It therefore follows that
	
	\begin{align}
			H^{\mu \nu \kappa \lambda} \nabla_{\mu} \Lie_{\mathcal{Z}}^I \Far_{\kappa \lambda}
				= \bigg\lbrace \frac{1}{2}\big[(g^{-1})^{\mu \kappa} (g^{-1})^{\nu \lambda}
					- (g^{-1})^{\mu \lambda} (g^{-1})^{\nu \kappa}\big]
				+ H_{\triangle}^{\mu \nu \kappa \lambda} \bigg\rbrace \nabla_{\mu} \Lie_{\mathcal{Z}}^I \Far_{\kappa \lambda}
				= H_{\triangle}^{\mu \nu \kappa \lambda} \nabla_{\mu} \Big(\Lie_{\mathcal{Z}}^I \Far_{\kappa \lambda}\Big)
				-	\Liemod_{\mathcal{Z}}^I \Big(H_{\triangle}^{\mu \nu \kappa \lambda} \nabla_{\mu} \Far_{\kappa \lambda} \Big),
	\end{align}
	which proves \eqref{E:MBIInhomogenoustermsI}.
\end{proof}

\subsection{The canonical stress $\Stress_{\ \nu}^{\mu}$} \label{SS:Stress}

Although the energy momentum tensor \eqref{E:electromagnetictensorupper} is useful for estimating
$\Far_{\mu \nu},$ it is not quite the right object for estimating its derivatives $\Lie_{\mathcal{Z}}^I \Far.$
As explained in detail in Section \ref{SSS:NonlinearAnalysis}, the reason is that the $\Lie_{\mathcal{Z}}^I \Far$ are solutions 
\emph{not to the MBI system itself}, but rather to the to the equations of variation \eqref{E:EOVBianchi} - \eqref{E:EOVMBI}, whose linearized Lagrangian, which is defined in \eqref{E:LinearizedLagrangian} below, depends on the background $\Far.$ Nonetheless, we will be able to construct the \emph{canonical stress}, which encodes the approximate conservation laws satisfied by solutions to the equations of variation, and which we denote by $\Stress_{\ \nu}^{\mu}.$ We remark that a general framework concerning the properties of the canonical stress was developed by Christodoulou in \cite{dC2000}; here, we only recall its basic features. As we will see, $\Stress_{\ \nu}^{\mu}$ has the following two key properties:

\begin{itemize}
	\item Its divergence is lower-order (in terms of the number of derivatives). 
	\item It has positivity properties related to, but in general distinct from those of 
	the energy-momentum tensor $\EMT_{\ \nu}^{\mu}.$
\end{itemize}
The first property is explained in detail at the end of this section, while the second is discussed in limited fashion
in Section \ref{SS:Positivity}.

In order to understand that origin of the canonical stress, let us first define the \emph{linearized Lagrangian} $\dot{\mathscr{L}},$ which is, despite its name, a quadratic form in the variations $\dot{\Far}$ with coefficients depending on the background $\Far.$

\begin{definition}
	The linearized Lagrangian $\dot{\mathscr{L}}[\dot{\Far};\Far]$ corresponding to the Lagrangian $\mathscr{L}[\cdot]$ 
	and the background $\Far$ is defined as follows:

\begin{align} \label{E:LinearizedLagrangian}
	\dot{\mathscr{L}} = \dot{\mathscr{L}}[\dot{\Far};\Far] 
		\eqdef \frac{1}{2} \frac{\partial^2 \Ldual[\Far]}{\partial \Far_{\zeta \eta} \partial \Far_{\kappa 
		\lambda}} \dot{\Far}_{\zeta \eta} \dot{\Far}_{\kappa \lambda} 
		= - \frac{1}{4}h^{\zeta \eta \kappa \lambda} \dot{\Far}_{\zeta \eta} \dot{\Far}_{\kappa \lambda},
\end{align}
where the $\Far-$dependent tensorfield $h^{\zeta \eta \kappa \lambda}$ is defined in \eqref{E:littlehdef}.
\end{definition}

The significance of $\dot{\mathscr{L}}[\dot{\Far};\Far]$ is that its corresponding equations of motion are the equations of variation. More specifically, if we consider $\dot{\Far}$ to be the unknowns, then the 
\emph{principal part} of the Euler-Lagrange equations (assuming the stationarity of the linearized action
$\mathcal{A}_{\mathfrak{C}}[\dot{\Far}] \eqdef \int_{\mathfrak{C} \Subset M} \dot{\mathscr{L}}[\dot{\Far};\Far] \, d \mu_g,$
under closed variations of $\dot{\Far}$) corresponding to $\dot{\mathscr{L}}[\dot{\Far};\Far]$ are the linearization (around $\Far$) of the Euler-Lagrange equations \eqref{E:Euler-LagrangeRewritten} corresponding to $\mathscr{L}.$ 

\begin{definition}
Given a linearized Lagrangian $\dot{\mathscr{L}} = \dot{\mathscr{L}}[\dot{\Far};\Far],$ we define the 
corresponding canonical stress $\Stress_{\ \nu}^{\mu}$ as follows:

\begin{align}
	\Stress_{\ \nu}^{\mu} \eqdef - 2\frac{\partial \dot{\mathscr{L}}}{\partial \dot{\Far}_{\mu \zeta}}\dot{\Far}_{\nu \zeta}
		+ \delta_{\nu}^{\mu} \dot{\mathscr{L}} 
		=	h^{\mu \zeta \kappa \lambda} \dot{\Far}_{\kappa \lambda} \dot{\Far}_{\nu \zeta}
			- \frac{1}{4} \delta_{\nu}^{\mu} h^{\zeta \eta \kappa \lambda} \dot{\Far}_{\zeta \eta} \dot{\Far}_{\kappa \lambda},
\end{align}
where the tensorfield $h$ is defined in \eqref{E:littlehdef}.
\end{definition}

In Section \ref{SS:MBISummary}, we modified the tensorfield $h$ corresponding to the MBI system, obtaining a new tensorfield $H.$ Therefore, for the purposes in this article, it is convenient to construct the MBI canonical stress using $H:$

\begin{align} \label{E:widetildeTHdef}
	\Stress_{\ \nu}^{\mu} & \eqdef H^{\mu \zeta \kappa \lambda} \dot{\Far}_{\kappa \lambda} \dot{\Far}_{\nu \zeta}
			- \frac{1}{4} \delta_{\nu}^{\mu} H^{\zeta \eta \kappa \lambda} \dot{\Far}_{\zeta \eta} \dot{\Far}_{\kappa \lambda},
\end{align}
where $H$ is defined in \eqref{E:Hdef}. \textbf{For the remainder of the article, the term ``canonical stress'' is used to refer to the tensor defined in \eqref{E:widetildeTHdef}}. Note that $\Stress_{\ \nu}^{\mu}$ depends quadratically on $\dot{\Far},$ and it also depends on the background $\Far.$

In general, $\Stress_{\mu \nu}$ is not symmetric, nor is $\Stress_{\ \nu}^{\mu}$ traceless. However, in the case of the MBI system, it can be checked that $\Stress_{\ \nu}^{\mu}$ is in fact traceless. More specifically, using \eqref{E:Hdef}, we compute that

\begin{align} \label{E:Hstress}
	\Stress_{\mu \nu} & = \overbrace{\dot{\Far}_{\mu}^{\ \zeta} \dot{\Far}_{\nu \zeta} - \frac{1}{4} g_{\mu \nu}
		\dot{\Far}_{\zeta \eta}\dot{\Far}^{\zeta \eta}}^{\mbox{linear Maxwell-Maxwell terms}} 
		+ \frac{1}{2}\ell_{(MBI)}^{-2} \Big\lbrace - \Far_{\mu}^{\ \zeta}\dot{\Far}_{\nu \zeta} \Far^{\kappa 
			\lambda}\dot{\Far}_{\kappa 
			\lambda} + \frac{1}{4} g_{\mu \nu} (\Far^{\kappa \lambda}\dot{\Far}_{\kappa \lambda})^2 \Big\rbrace   \\
		& \ \ \ + \frac{1}{2} (1 + \Farinvariant_{(2)}^2 \ell_{(MBI)}^{-2})
			\Big\lbrace - \Fardual_{\mu}^{\ \zeta}\dot{\Far}_{\nu \zeta} \Fardual^{\kappa \lambda}\dot{\Far}_{\kappa \lambda} 
			+ \frac{1}{4} g_{\mu \nu} (\Fardual^{\kappa \lambda}\dot{\Far}_{\kappa \lambda})^2 \Big\rbrace \notag \\
		& \ \ \ + \frac{1}{2} \Farinvariant_{(2)} \ell_{(MBI)}^{-2} \Big\lbrace \Far_{\mu}^{\ \zeta}\dot{\Far}_{\nu \zeta} 
			\Fardual^{\kappa \lambda}\dot{\Far}_{\kappa \lambda} - \frac{1}{4} g_{\mu \nu} \Far^{\zeta \eta}\dot{\Far}_{\zeta \eta} 
			\Fardual^{\kappa \lambda}\dot{\Far}_{\kappa \lambda} \Big\rbrace \notag \\
		& \ \ \ + \frac{1}{2} \Farinvariant_{(2)} \ell_{(MBI)}^{-2} \Big\lbrace \Fardual_{\mu}^{\ \zeta}\dot{\Far}_{\nu \zeta} 
			\Far^{\kappa \lambda}\dot{\Far}_{\kappa \lambda} - \frac{1}{4} g_{\mu \nu} \Far^{\zeta \eta}\dot{\Far}_{\zeta \eta} 
			\Fardual^{\kappa \lambda}\dot{\Far}_{\kappa \lambda}\Big\rbrace, \notag
\end{align}
from which it also follows that

\begin{align} \label{E:widetildeTtracefree}
	\Stress_{\ \kappa}^{\kappa} = 0.
\end{align}
We also note that the first two terms on the right-hand side of \eqref{E:Hstress} are the components of the energy-momentum tensor for the linear Maxwell-Maxwell equations (in the unknown $\dot{\Far}$).

The failure of symmetry of $\Stress_{\mu \nu}$ is captured by the following expression:

\begin{align}  \label{E:widetildeTantisymmetricpart}
	\Stress_{\mu \nu} - \Stress_{\nu \mu} 
		& = \frac{1}{2} \ell_{(MBI)}^{-2} \Far^{\kappa \lambda}\dot{\Far}_{\kappa \lambda} 
		\Big\lbrace \Far_{\nu}^{\ \zeta}\dot{\Far}_{\mu \zeta} 
			- \Far_{\mu}^{\ \zeta}\dot{\Far}_{\nu \zeta} \Big\rbrace 
			+ \frac{1}{2}(1 + \Farinvariant_{(2)}^2 \ell_{(MBI)}^{-2})
			\Fardual^{\kappa \lambda}\dot{\Far}_{\kappa \lambda} \Big\lbrace 
			\Fardual_{\nu}^{\ \zeta}\dot{\Far}_{\mu \zeta} - \Fardual_{\mu}^{\ \zeta}\dot{\Far}_{\nu \zeta} \Big\rbrace 
			\\
		& \ \ \ + \frac{1}{2}\Farinvariant_{(2)} \ell_{(MBI)}^{-2} \Fardual^{\kappa \lambda}\dot{\Far}_{\kappa \lambda}
			\Big\lbrace \Far_{\mu}^{\ \zeta}\dot{\Far}_{\nu \zeta} 
			- \Far_{\nu}^{\ \zeta}\dot{\Far}_{\mu \zeta} \Big\rbrace 
			+ \frac{1}{2} \Farinvariant_{(2)} \ell_{(MBI)}^{-2} \Far^{\kappa \lambda}\dot{\Far}_{\kappa \lambda}
			\Big\lbrace \Fardual_{\mu}^{\ \zeta}\dot{\Far}_{\nu \zeta}  
				- \Fardual_{\nu}^{\ \zeta}\dot{\Far}_{\mu \zeta} \Big\rbrace. \notag
\end{align}

We conclude this section by proving a lemma that illustrates the first key property of $\Stress_{\ \nu}^{\mu},$ namely that its
divergence is lower order. 

\begin{lemma} \label{L:divergenceofwidetildeT}
Let $\dot{\Far}$ be a solution to the equations of variation \eqref{E:EOVBianchi} - \eqref{E:EOVMBI} 
corresponding to the background $\Far,$ and let $\Stress_{\ \nu}^{\mu}$ be the tensor 
defined in \eqref{E:Hstress}. Let $\mathfrak{J}_{\lambda \mu \nu},$ $\mathfrak{I}^{\nu}$ be the inhomogeneous terms on the right-hand sides of \eqref{E:EOVBianchi} - \eqref{E:EOVMBI}. Then

\begin{align} \label{E:divergenceofwidetildeT}
	\nabla_{\mu} \Stress_{\ \nu}^{\mu} & = - \frac{1}{2} H^{\zeta \eta \kappa \lambda} \dot{\Far}_{\zeta \eta} 
		\mathfrak{J}_{\nu \kappa \lambda}
		+ \dot{\Far}_{\nu \eta} \mathfrak{I}^{\eta}
		+ (\nabla_{\mu}H^{\mu \zeta \kappa \lambda}) \dot{\Far}_{\kappa \lambda} \dot{\Far}_{\nu \zeta}
		- \frac{1}{4} (\nabla_{\nu}H^{\zeta \eta \kappa \lambda}) \dot{\Far}_{\zeta \eta} \dot{\Far}_{\kappa \lambda}.
\end{align}

\end{lemma}

\begin{proof}
	Lemma \ref{L:divergenceofwidetildeT} follows from using the properties 
	\eqref{E:hminussignproperty1} - \eqref{E:hsymmetryproperty}, 
	which are satisfied by $H^{\mu \nu \kappa \lambda},$ and simple computations.
\end{proof}

\subsection{Positivity properties of $\Stress_{\ \nu}^{\mu}$} \label{SS:Positivity}

The canonical stress has positivity properties that are analogous to, but distinct from those of the energy-momentum tensor.
For a complete discussion of these properties, which are related to the geometry of the equations\footnote{In general, the geometry of the electromagnetic equations (i.e., the characteristic cones) is distinct from the geometry of spacetime; however, in the case of Maxwell-Maxwell theory, the two geometries coincide.}, see \cite{dC2000}. Here, we only discuss the positivity properties that are relevant to our small-data global existence proof and our sketch of a large-data local existence proof. 
For purposes of constructing an energy suitable for the global existence proof (see Section \ref{S:NormsandEnergies}), the main quantity of interest will be $\Stress(\xi^{(0)}, \overline{K}) \eqdef \Stress_{\ \lambda}^{\kappa}\xi_{\kappa}^{(0)} \overline{K}^{\lambda},$ where $\xi^{(0)}$ is the $g-$dual of the timelike translation Killing vectorfield $T_{(0)}$ defined in \eqref{E:Translationsetdef}, and $\overline{K} = \frac{1}{2}\big\lbrace(1+s^2) L + (1+q^2) \uL \big\rbrace$ is the vectorfield defined in \eqref{E:Kdef}. As we will see in the next lemma, in the small-field regime, the resulting expression is positive definite in the null components of the variation $\dot{\Far}.$ However, different components carry different weights. Ultimately, the different weights will translate into the fact that the various null components of a solution $\Far$ and its derivatives have distinct rates of decay; see the global Sobolev inequality (Proposition \ref{P:GlobalSobolev}). For large-data
local existence (see Proposition \ref{P:LocalExistence}), we will use a vectorfield $\Vmult$ (see Proposition \ref{P:LocalExistenceCurrent}) in place of $\overline{K}.$ The reason for this modification is that the positivity of
$\Stress(\xi^{(0)}, \overline{K}) \eqdef \Stress_{\ \lambda}^{\kappa}\xi_{\kappa}^{(0)} \overline{K}^{\lambda}$ may break down in the large-data regime. On the other hand, the vectorfield $\Vmult,$ which is constructed with the aid of the 
\emph{reciprocal Born-Infeld metric} $(b^{-1})^{\mu \nu},$ maintains its positivity in all regimes in which the MBI equations are well-defined. Although $\Vmult$ does not provide good $q,s$ weights for the null components of $\dot{\Far},$ 
a bound of the form $\Stress(\xi^{(0)}, \Vmult) \geq C |\dot{\Far}|^2$ for some positive constant $C$ is sufficient to prove local existence.

We begin with a lemma that addresses the positivity properties of $\Stress(\xi^{(0)},\overline{K}).$

\begin{lemma} \label{L:CanonicalStressErrorTermExpansion}
Let $\Far, \dot{\Far}$ be arbitrary two-forms, and let $\dot{\ualpha}, \dot{\alpha}, \dot{\rho}, \dot{\sigma},$
be the null components of $\dot{\Far}.$ Let $\Stress$ be the canonical stress \eqref{E:Hstress}
associated to $\Far, \dot{\Far},$ and let $T_{(0)},$ $\overline{K} = \frac{1}{2}\big\lbrace(1+s^2) L + (1+q^2) \uL \big\rbrace$ be the conformal Killing fields defined in \eqref{E:Translationsdef} and \eqref{E:Kdef} respectively. Let 
$\xi^{(0)}$ be the $g-$dual of $T_{(0)}.$ There exists an $\epsilon > 0$ such that if $|\Far| < \epsilon,$ then

\begin{align} \label{E:CanonicalStressErrorTermExpansion}
	\Stress(\xi^{(0)},\overline{K}) & = \Stress_{\mu \nu}T_{(0)}^{\mu}\overline{K}^{\nu}
		= \frac{1}{2}\Big\lbrace(1 + q^2)|\dot{\ualpha}|^2 + 
		(1+s^2)|\dot{\alpha}|^2 + (2 + q^2 + s^2)(\dot{\rho}^2 + \dot{\sigma}^2) \Big\rbrace \\
	& \ \ + |\dot{\ualpha}|^2 O\Big((1 + s)^2|\Far|_{\mathcal{L}\mathcal{U}}^2 \Big)
		+ (1 + s)^2(|\dot{\alpha}|^2 + \dot{\rho}^2 + \dot{\sigma}^2) O(|\Far|^2) \notag \\
	& \ \ + (1 + |q|)^2(|\dot{\ualpha}|^2 + |\dot{\alpha}|^2 + \dot{\rho}^2 + \dot{\sigma}^2) O(|\Far|^2) \notag. 
\end{align}

In the above expression, $O(U)$ denotes a quantity which is $\leq C U$ in magnitude for some positive constant $C.$ 

\end{lemma}

\begin{remark}
	As we shall see, in the small-solution regime, the dominant term on the right-hand side of 
	\eqref{E:CanonicalStressErrorTermExpansion} is $\frac{1}{2}\Big\lbrace(1 + q^2)|\dot{\ualpha}|^2 + 
	(1+s^2)|\dot{\alpha}|^2 + (2 + q^2 + s^2)(\dot{\rho}^2 + \dot{\sigma}^2) \Big\rbrace.$ This motivates definition 
	\eqref{E:weightedpointwisenorm} below. Furthermore, at first sight, one might worry that the $(1 + s^2)$ factor could cause 
	the ``error terms'' $O(\cdots)$ to become large. This worry is alleviated by Corollary \ref{C:GlobalSobolev}.
\end{remark}

\begin{proof}
	Using \eqref{E:Hstress}, we decompose $\Stress$ into its linear ``Maxwell'' part 
	$\Stress^{(Maxwell)},$ and the remaining ``error terms:''
	\begin{align}
		\Stress_{\mu \nu} & = \Stress_{\mu \nu}^{(Maxwell)} 
			+ \frac{1}{2}\ell_{(MBI)}^{-2} \Big\lbrace - \Far_{\mu}^{\ \zeta}\dot{\Far}_{\nu \zeta} \Far^{\kappa 
			\lambda}\dot{\Far}_{\kappa \lambda} + \frac{1}{4} g_{\mu \nu} (\Far^{\kappa \lambda}\dot{\Far}_{\kappa \lambda})^2 
			\Big\rbrace  \label{E:StressDecomposition} \\
		& \ \ \ + \frac{1}{2}(1 + \Farinvariant_{(2)}^2 \ell_{(MBI)}^{-2})
			\Big\lbrace - \Fardual_{\mu}^{\ \zeta}\dot{\Far}_{\nu \zeta} \Fardual^{\kappa \lambda}\dot{\Far}_{\kappa \lambda} 
			+ \frac{1}{4} g_{\mu \nu} (\Fardual^{\zeta \eta}\dot{\Far}_{\zeta \eta})^2 \Big\rbrace \notag \\
		& \ \ \ + \frac{1}{2}\Farinvariant_{(2)} \ell_{(MBI)}^{-2} \Big\lbrace \Far_{\mu}^{\ \zeta}\dot{\Far}_{\nu \zeta} \Fardual^{\kappa 
			\lambda}\dot{\Far}_{\kappa \lambda} - \frac{1}{4} g_{\mu \nu} \Far^{\zeta \eta}\dot{\Far}_{\zeta \eta} 
			\Fardual^{\kappa \lambda}\dot{\Far}_{\kappa \lambda} \Big\rbrace \notag \\
		& \ \ \ + \frac{1}{2}\Farinvariant_{(2)} \ell_{(MBI)}^{-2} \Big\lbrace \Fardual_{\mu}^{\ \zeta}\dot{\Far}_{\nu \zeta} \Far^{\kappa 
			\lambda}\dot{\Far}_{\kappa \lambda} - \frac{1}{4} g_{\mu \nu} \Far^{\zeta \eta}\dot{\Far}_{\zeta \eta} 
			\Fardual^{\kappa \lambda}\dot{\Far}_{\kappa \lambda} \Big\rbrace, \notag \\
		\Stress_{\mu \nu}^{(Maxwell)} & \eqdef \dot{\Far}_{\mu}^{\ \zeta} \dot{\Far}_{\nu \zeta} - \frac{1}{4} g_{\mu \nu}
			\dot{\Far}_{\zeta \eta}\dot{\Far}^{\zeta \eta}. \label{E:Tmaxwelldef}
	\end{align}
	
	As in the proof of Lemma \ref{L:dominantenergycondition}, leave it as a simple exercise for the reader to check that

\begin{align}
	\Stress^{(Maxwell)}(\uL, \uL) & = |\dot{\ualpha}|^2, 
	&& \Stress^{(Maxwell)}(L,L) = 	|\dot{\alpha}|^2, 
	&& \Stress^{(Maxwell)}(\uL,L) = (\dot{\rho}^2 + \dot{\sigma}^2). \label{E:widetildeTMaxwelluLLcalculations}
\end{align}
Since $\Stress^{(Maxwell)}(T_{(0)},\overline{K}) = \Stress^{(Maxwell)}\big(\frac{1}{2}(L + \uL),(1 + s^2)L + (1 + q^2) \uL \big),$ it easily follows from \eqref{E:widetildeTMaxwelluLLcalculations} that

\begin{align} 
	\Stress^{(Maxwell)}(T_{(0)},\overline{K}) = \frac{1}{2}\Big\lbrace(1 + q^2)|\dot{\ualpha}|^2 + 
		(1+s^2)|\dot{\alpha}|^2 + (2 + q^2 + s^2)(\dot{\rho}^2 + \dot{\sigma}^2) \Big\rbrace,
\end{align}
which gives the principal term (i.e., the term in braces) on the right-hand side of \eqref{E:CanonicalStressErrorTermExpansion}.
Note that in the above formulas, we have abused notation by identifying covectors with their $g-$dual vectors.

For the error terms, we first note that $\ell_{(MBI)}^{-2}$ is an order $1$ factor when $\Far$ is sufficiently small, so that
we may ignore it. With this fact in mind, we then evenly divide the $8$ error terms on the right-hand side of \eqref{E:StressDecomposition} into two classes: those that contain the $g_{\mu \nu}$ factor, and those that do not. For the first class, we note that $|g_{\mu \nu} T_{(0)}^{\mu} \overline{K}^{\nu}| =\frac{1}{2}|2 + q^2 + s^2| \lesssim (1 + s^2).$ Therefore, using Lemma \ref{L:NullForms}, it follows that
\begin{align}
	|g_{\mu \nu} T_{(0)}^{\mu} \overline{K}^{\nu} (\Far^{\kappa \lambda}\dot{\Far}_{\kappa \lambda})^2| 
	& \lesssim (1 + s^2) \Big\lbrace |\Far|_{\mathcal{L}\mathcal{U}}^2 |\dot{\ualpha}|^2
		+ |\Far|^2 (|\dot{\alpha}|^2 + \dot{\rho}^2 + \dot{\sigma}^2) \Big\rbrace,
\end{align}	
and similarly for the other 3 terms of this type.

For the first term belonging to the second class, we use the null decomposition \eqref{E:MorawetznulldecompL} - \eqref{E:MorawetznulldecompA}
of $\overline{K}$ to conclude that
\begin{align}
	|T_{(0)}^{\mu} \overline{K}^{\nu} \Far_{\mu}^{\ \zeta}\dot{\Far}_{\nu \zeta}| \lesssim 
	(1 + q^2) |\Far| |\dot{\Far}| + (1 + s^2)|\Far| |\dot{\Far}|_{\mathcal{L}\mathcal{U}}.
\end{align}
Therefore, by a second application of Lemma \ref{L:NullForms} and the inequality $|ab| \lesssim a^2 + b^2,$ it follows that
\begin{align}
	|T_{(0)}^{\mu}\overline{K}^{\nu}\Far_{\mu}^{\ \zeta}\dot{\Far}_{\nu \zeta} \Far^{\kappa \lambda}\dot{\Far}_{\kappa \lambda}|
	& \lesssim (1 + q^2)|\Far|^2 |\dot{\Far}|^2 + (1 + s^2)|\Far|_{\mathcal{L}\mathcal{U}}^2 |\dot{\Far}|^2
		+ (1 + s^2)|\Far|^2|\dot{\Far}|_{\mathcal{L}\mathcal{U}}^2 + (1 + s^2)|\Far|^2|\dot{\Far}|_{\mathcal{T}\mathcal{T}}^2 \\
	& \lesssim (1 + q^2)|\Far|^2 (|\dot{\ualpha}|^2 + |\dot{\alpha}|^2 + \dot{\rho}^2 + \dot{\sigma}^2) 
		+ (1 + s^2)|\Far|_{\mathcal{L}\mathcal{U}}^2 |\dot{\ualpha}|^2
		+ (1 + s^2)|\Far|^2(|\dot{\alpha}|^2 + \dot{\rho}^2 + \dot{\sigma}^2). \notag
\end{align}
The remaining $3$ terms of this type can be estimated similarly. Inequality \eqref{E:CanonicalStressErrorTermExpansion} thus follows.

\end{proof}

The next proposition provides the fundamental estimate that is needed to deduce the ``large-data'' local existence result 
of Proposition \ref{P:LocalExistence}.

\begin{proposition} \label{P:LocalExistenceCurrent}
	Let $\mathscr{H} \eqdef \lbrace \Far \mid \ell_{(MBI)}[\Far] \eqdef 
	\big(1 + \Farinvariant_{(1)}[\Far] - \Farinvariant_{(2)}^2[\Far] \big)^{1/2} > 0 \rbrace$ 
	denote the interior of the subset of state-space for which the Maxwell-Born-Infeld Lagrangian is well-defined,
	and let $\mathfrak{K}$ be a compact subset of $\mathscr{H}.$ Let $\Far \in \mathfrak{K},$ let $\dot{\Far}$ be any two-form,
	and let $\Stress_{\ \nu}^{\mu}$ be the canonical stress \eqref{E:Hstress} for the MBI equations of variation 
	corresponding to the background $\Far$ and the variation $\dot{\Far}.$ Let $\Vmult$ be the
	vectorfield defined below in \eqref{E:LocalExistenceMultiplier}. Then there exists a constant $C > 0,$ depending only 
	on $\mathfrak{K},$ such that the following inequality holds:
	
	\begin{align} \label{E:LocalExistenceCurrent}
		\Stress(\xi^{(0)},\Vmult) \eqdef \Stress_{\ \nu}^{\mu}\xi_{\mu}^{(0)}\Vmult^{\nu}
		\geq C_{\mathfrak{K}} |\dot{\Far}|^2.
	\end{align}
	In the above expression, the covector $\xi^{(0)}$ is the $g-$dual of the
	the time translation Killing field $T_{(0)},$ which is defined in \eqref{E:Translationsdef}.
\end{proposition}

\begin{remark}
	Let us view $\Stress_{\ \nu}^{\mu} \xi_{\mu} \Vmult^{\nu}$ 
	as a function of the spacetime point $p$ and of the covector $\xi \in T_p^* M,$
	i.e., as a function on $T^* M.$ Then since inequality \eqref{E:LocalExistenceCurrent}
	is strict, there is an open neighborhood $\mathcal{N} \subset T^* M$ of $(p,\xi^{(0)})$ and a uniform constant 
	$C_{\mathcal{N}} > 0$ such that $\Stress_{\ \nu}^{\mu} \xi_{\mu} \Vmult^{\nu}
	\geq C_{\mathcal{N}} |\dot{\Far}|^2$ holds for all $(p,\xi) \in \mathcal{N}.$ Using the methods described in 
	\cite[Chapter 5]{dC2000}, such an inequality implies the hyperbolicity of the MBI equations of variation and finite speed of 
	propagation for their solutions.
\end{remark}

\begin{remark}
	Since \eqref{E:LocalExistenceCurrent} is a pointwise inequality, an analogous inequality can be shown to hold
	on a neighborhood of any point $p$ belonging to any sufficiently smooth Lorentzian manifold (simply use a local coordinate 
	system in which $g_{\mu \nu}|_p = \mbox{diag}(-1,1,1,1),$ define $\xi^{(0)}$ to be the $g-$dual of 
	$\frac{\partial}{\partial x^0},$ and define $\Vmult^{\nu}$ as in \eqref{E:LocalExistenceMultiplier} below). Thus, inequality 
	\eqref{E:LocalExistenceCurrent} implies the hyperbolicity of the MBI equations of variation for any 
	sufficiently smooth Lorentzian manifold.
\end{remark}

\begin{proof}
	
	Let $(\Electricfield, \Magneticinduction)$ and $(\dot{\Electricfield}, \dot{\Magneticinduction})$
	be the electromagnetic decompositions of $\Far$ and $\dot{\Far}$ described in Section \ref{SS:electromagneticdecomposition}.
	Then simple calculations imply that the following identities hold in the inertial coordinate system:
	
	\begin{align} 
		\Far_{0}^{\ \kappa} \dot{\Far}_{0 \kappa} & = \Electricfield_a \dot{\Electricfield}^a, 
			\label{E:Far0dotFar0contraction} \\ 
		\Fardual_{0}^{\ \kappa} \dot{\Far}_{0 \kappa} & = - \Magneticinduction_a \dot{\Electricfield}^a, \\
		\Far^{\kappa \lambda} \dot{\Far}_{\kappa \lambda} & = 2 \Magneticinduction_a \dot{\Magneticinduction}^a
			- 2 \Electricfield_a \dot{\Electricfield}^a,  \\
		\Fardual^{\kappa \lambda} \dot{\Far}_{\kappa \lambda} & = 2 \Magneticinduction_a \dot{\Electricfield}^a
			+ 2 \Electricfield_a \dot{\Magneticinduction}^a. \label{E:FardualdotFarcontraction}
	\end{align}
	
	We now introduce the \emph{reciprocal Maxwell-Born-Infeld metric} $(b^{-1})^{\mu \nu},$ which has the following components
	relative to an arbitrary coordinate system:
	
	\begin{align} \label{E:MBImetric}
		(b^{-1})^{\mu \nu} \eqdef (g^{-1})^{\mu \nu} -(1 + \Farinvariant_{(1)}[\Far])^{-1} \Far^{\mu \kappa} \Far_{\ \kappa}^{\nu}.
	\end{align}
	In a future article, we will discuss the significance of $(b^{-1})^{\mu \nu}$ in detail. For now, we simply make the
	following two claims.
	
	\begin{enumerate}
		\item If $(g^{-1})^{\kappa \lambda} \xi_{\kappa \lambda} < 0,$ then 
			$(b^{-1})^{\kappa \lambda} \xi_{\kappa \lambda} < 0.$
		\item If the covector $\xi_{\mu}$ belongs to the characteristic subset\footnote{The characteristic subset 
		of $T_p^* M$ is defined in Section \ref{S:Introduction}.} of $T_p^* M,$ then the vector
		$X^{\mu} \eqdef (b^{-1})^{\mu \kappa}\xi_{\kappa}$ belongs to the characteristic subset of $T_p M.$
	\end{enumerate}
	The claim $(1),$ which is easy to check in the inertial coordinate system, shows that any covector $\xi$
	that is timelike relative to $(g^{-1})^{\mu \nu}$ is also timelike relative to $(b^{-1})^{\mu \nu}.$ This is related 
	to the fact that the energy-momentum tensor satisfies the dominant energy condition. As is illustrated
	below in \eqref{E:LocalExistenceMultiplier} the claim $(2)$ allows us to construct a suitable 
	(for deducing local existence) multiplier 
	vectorfield $X^{\mu}$ from a $g-$timelike covector $\xi.$ For, as is discussed in \cite{dC2000}, in order
	for $\int_{\mathbb{R}^3} \Stress_{\ \nu}^{\mu}\xi_{\mu}X^{\nu} \, d^3 \ux$ to be bounded from below
	by a multiple of $\int_{\mathbb{R}^3} |\dot{\Far}|^2 \, d^3 \ux,$ 
	it is sufficient\footnote{As discussed in \cite{dC2000}, in general, we must choose $\xi$ to be an element of the \emph{inner 
	core} of the  characteristic subset of $T_p^* M,$ and $X$ to be an of the \emph{inner core} of the characteristic subset of 
	$T_p M.$ However, in the case of the MBI system, the inner core of the characteristic subset of $T_p^* M$ coincides with
	the interior of the characteristic subset of $T_p^* M,$ and similarly for the inner core of the characteristic subset of 	
	$T_p M.$ This is because the characteristic subset of $T_p^* M$ in the case of the MBI system is particularly simple:
	it is the cone $\lbrace \xi | \ (b^{-1})^{\kappa \lambda} \xi_{\kappa} \xi_{\lambda} = 0 \rbrace.$} to choose $\xi$ to be a
	covector lying \emph{strictly inside} of the characteristic subset of $T_p^* M,$ and $X$ to be a vector lying 
	\emph{strictly inside} of the characteristic subset of $T_p M$ satisfying\footnote{Many of our
	definitions differ from those in \cite{dC2000} by a minus sign.} $\xi(X) \eqdef \xi_{\kappa}X^{\kappa} < 0.$ By claims $(1)$ 
	and $(2),$ we can choose $\xi$ to be any $g-$timelike covector, and define $X^{\mu} \eqdef (b^{-1})^{\mu \kappa} 
	\xi_{\kappa}.$ 
	
	Rather than relying on the above abstract framework, we will instead directly show the \emph{positive definiteness} of the 
	quadratic form (in $\dot{\Far}$) $\Stress_{\ \nu}^{\mu}\xi_{\mu}X^{\nu}$ for a particular choice of $\xi$ and $X;$ this 
	pointwise positivity is stronger than the integrated positivity that follows from the general framework. We choose $\xi$ to 
	be equal to $\xi^{(0)},$ the covector that is $g-$dual to the timelike Killing field $T_{(0)}.$ Furthermore, motivated by the 
	above discussion, we define the multiplier vectorfield $\Vmult$ by
	
	\begin{align} 
		\Vmult^{\mu} \eqdef 2 \ell_{(MBI)}^2 (1 + \Farinvariant_{(1)}[\Far]) (b^{-1})^{\mu \nu} \xi_{\nu}^{(0)} 
			& =	2 \ell_{(MBI)}^2 (1 + \Farinvariant_{(1)}[\Far]) T_{(0)}^{\mu} 
			- 2 \ell_{(MBI)}^2 T_{(0)}^{\nu} \Far^{\mu \kappa} \Far_{\nu \kappa}
			\label{E:LocalExistenceMultiplier} \\
			& = 2 \ell_{(MBI)}^2 (1 + \Farinvariant_{(1)}[\Far]) \delta_{0}^{\mu} - 2 \ell_{(MBI)}^2 
			\Far^{\mu \kappa} \Far_{0 \kappa}, \notag
	\end{align}
	where the last equality is valid in the inertial coordinate system. We remark that the 
	$2 \ell_{(MBI)}^2 (1 + \Farinvariant_{(1)}[\Far])$ factor on the left-hand side of \eqref{E:LocalExistenceMultiplier}
	is a normalization factor that was chosen out of computational convenience. Using the identities 
	\eqref{E:Far0dotFar0contraction} - \eqref{E:FardualdotFarcontraction}, it can be checked that
	
	\begin{align} \label{E:LocalExistenceEnergyDensity}
		\Stress_{\ \nu}^{\mu} \xi_{\mu}^{(0)} \Vmult^{\nu} =
		2(1 + |\Magneticinduction|^2) \ell_{(MBI)}^2 \Stress_{00} - 2\ell_{(MBI)}^2 \Stress_{0 a} \Far^{a \rho}\Far_{0 \rho}.
	\end{align}
	
	Our goal is to show that the right-hand side of \eqref{E:LocalExistenceEnergyDensity} is a uniformly positive definite
	quadratic form in $\dot{\Far}$ whenever $\Far \in \mathfrak{K}.$ Using \eqref{E:Hstress} and \eqref{E:Far0dotFar0contraction} 
	- \eqref{E:FardualdotFarcontraction}, and defining (for notational convenience)
	
	\begin{align}
		\langle U,V \rangle \eqdef U_a V^a
	\end{align}
	for any pair of vectors $U,V \in \Sigma_t,$ we compute that
	
	\begin{align} \label{E:widetildeT00expresssion}
		2 \ell_{(MBI)}^2 \Stress_{00} 
		& = \ell_{(MBI)}^2 \Big\lbrace |\dot{\Electricfield}|^2 + |\dot{\Magneticinduction}|^2 \Big\rbrace
			+ \Big\lbrace \langle \Electricfield, \dot{\Electricfield} \rangle^2
			+ (1 + |\Magneticinduction|^2 - |\Electricfield|^2)\langle \Magneticinduction, \dot{\Electricfield} \rangle^2 
			+ 2 \langle \Electricfield, \Magneticinduction \rangle 
			\langle \Electricfield, \dot{\Electricfield} \rangle 
			\langle \Magneticinduction, \dot{\Electricfield} \rangle \Big\rbrace \\
		& \ \ \ - \Big\lbrace \langle \Magneticinduction, \dot{\Magneticinduction} \rangle^2
			+ (1 + |\Magneticinduction|^2 - |\Electricfield|^2) \langle \Electricfield, \dot{\Magneticinduction} \rangle^2 
			- 2\langle \Electricfield, \Magneticinduction \rangle 
			\langle \Electricfield, \dot{\Magneticinduction}\rangle 
			\langle \Magneticinduction, \dot{\Magneticinduction} \rangle \Big\rbrace, \notag 
	\end{align}
	and
	
	\begin{align}
		\ell_{(MBI)}^2 \Stress_{0 a} \Far^{a \rho}\Far_{0 \rho} 
		& = - \langle \Magneticinduction, \dot{\Magneticinduction} \rangle^2 \Big\lbrace
			|\Electricfield|^2 + \langle \Electricfield, \Magneticinduction \rangle^2 \Big\rbrace 
		+  \langle \Magneticinduction, \dot{\Magneticinduction} \rangle
			\langle \Magneticinduction, \dot{\Electricfield} \rangle 
			(1 + |\Magneticinduction|^2)\langle \Electricfield, \Magneticinduction \rangle  \\
		& \ \ + 2 \langle \Magneticinduction, \dot{\Magneticinduction} \rangle
				\langle \Electricfield, \dot{\Magneticinduction} \rangle
				(1 + |\Magneticinduction|^2) \langle \Electricfield, \Magneticinduction \rangle 
		+ \langle \Magneticinduction, \dot{\Magneticinduction} \rangle 
			\langle \Electricfield, \dot{\Electricfield} \rangle
			(1 + |\Magneticinduction|^2)  \notag \\
		& \ \ - \langle \Electricfield, \dot{\Magneticinduction} \rangle^2
			\Big\lbrace (1 + |\Magneticinduction|^2 - |\Electricfield|^2)|\Magneticinduction|^2
			+ \langle \Electricfield, \Magneticinduction \rangle^2 \Big\rbrace
		- \langle \Electricfield, \dot{\Magneticinduction} \rangle
			\langle \Magneticinduction, \dot{\Electricfield} \rangle 
			(1 + |\Magneticinduction|^2 - |\Electricfield|^2)(1 + |\Magneticinduction|^2) \notag \\
		& \ \ - \langle \Electricfield, \dot{\Magneticinduction} \rangle
			\langle \Electricfield, \dot{\Electricfield} \rangle
			\langle \Electricfield, \Magneticinduction \rangle (1 + |\Magneticinduction|^2). \notag
	\end{align}

	To simplify the subsequent analysis, we choose a frame of $\Sigmafirstfund-$orthonormal vectors 
	$\lbrace e_{\parallel}, e_{\perp}, e_{\times} \rbrace \subset \Sigma_t$ 
	such that $\Electricfield \in \mbox{span}\lbrace e_{\parallel} 
	\rbrace,$ and such that $\Magneticinduction \in \mbox{span}\lbrace 
	e_{\parallel}, e_{\perp} \rbrace.$ We can therefore express (in a slight abuse of notation) the electric field as 
	$\Electricfield e_{\parallel}$ and the magnetic induction as $\Magneticinduction_{\parallel} e_{\parallel} + 
	\Magneticinduction_{\perp} e_{\perp}.$ Furthermore, we can decompose $\dot{\Magneticinduction}$ and
	$\dot{\Electricfield}$ relative to this frame as
	
	\begin{align}
		\dot{\Magneticinduction} & = \dot{\Magneticinduction}_{\parallel} e_{\parallel} 
			+ \dot{\Magneticinduction}_{\perp} e_{\perp}
			+ \dot{\Magneticinduction}_{\times} e_{\times}, \\
		\dot{\Electricfield} & = \dot{\Electricfield}_{\parallel} e_{\parallel} 
			+ \dot{\Electricfield}_{\perp} e_{\perp}
			+ \dot{\Electricfield}_{\times} e_{\times},
	\end{align} 
	where
	
	\begin{align}
		\langle \dot{\Magneticinduction}, e_{\parallel} \rangle & \eqdef \dot{\Magneticinduction}_{\parallel}, \\
		\langle \dot{\Magneticinduction}, e_{\perp} \rangle & \eqdef \dot{\Magneticinduction}_{\perp}, \\
		\langle \dot{\Electricfield}, e_{\parallel} \rangle & \eqdef \dot{\Electricfield}_{\parallel}, \\
		\langle \dot{\Electricfield}, e_{\perp} \rangle & \eqdef \dot{\Electricfield}_{\perp},
	\end{align}
	and
	
	\begin{align}
		|\dot{\Magneticinduction}|^2 & = \dot{\Magneticinduction}_{\parallel}^2
			+ \dot{\Magneticinduction}_{\perp}^2 + \dot{\Magneticinduction}_{\times}^2, \\
		|\dot{\Electricfield}|^2 & = \dot{\Electricfield}_{\parallel}^2
			+ \dot{\Electricfield}_{\perp}^2 + \dot{\Electricfield}_{\times}^2.
	\end{align}
	
	Let us first address the positivity of the quadratic form $\Stress_{\ \nu}^{\mu} \xi_{\mu}^{(0)} \Vmult^{\nu}$
	in $(\dot{\Magneticinduction}_{\times}, \dot{\Electricfield}_{\times}),$
	since these components are very simple to analyze, as they are completely decoupled from
	the remaining components. The only term that contributes to these components is the 
	$\ell_{(MBI)}^2 \Big\lbrace |\dot{\Electricfield}|^2 + |\dot{\Magneticinduction}|^2) \Big\rbrace$ term 
	on the right-hand side of \eqref{E:widetildeT00expresssion}, which leads to the trivial inequality
	
	\begin{align} \label{E:OrthogonalComponentsExpression}
		\Stress_{\ \nu}^{\mu} \xi_{\mu}^{(0)} \Vmult^{\nu} \geq \ell_{(MBI)}^2 (1 + |\Magneticinduction|^2) 
		(\dot{\Electricfield}_{\times}^2 + \dot{\Magneticinduction}_{\times}^2).
	\end{align}
	We note that the uniform positivity of the right-hand side of \eqref{E:OrthogonalComponentsExpression} 
	for $\Far \in \mathfrak{K}$ is manifest, since $\ell_{(MBI)}^2[\Far]$ is uniformly positive whenever $\Far \in 
	\mathfrak{K}.$ We therefore conclude from \eqref{E:OrthogonalComponentsExpression} that
	
	\begin{align} \label{E:OrthogonalComponentsBoundedFromBelow}
		\Stress_{\ \nu}^{\mu} \xi_{\mu}^{(0)} \Vmult^{\nu} \geq C_{\mathfrak{K}} (\dot{\Electricfield}_{\times}^2 + 
		\dot{\Magneticinduction}_{\times}^2).
	\end{align}
	
	We now address the components $(\dot{\Magneticinduction}_{\parallel}, \dot{\Magneticinduction}_{\perp}, 
	\dot{\Electricfield}_{\parallel}, \dot{\Electricfield}_{\parallel}).$ After many tedious but simple calculations, 
	and ignoring the components $(\dot{\Magneticinduction}_{\times}, \dot{\Electricfield}_{\times})$ analyzed above,
	it follows that the right-hand side of \eqref{E:LocalExistenceEnergyDensity} 
	can be viewed as a quadratic form in the components $(\dot{\Magneticinduction}_{\parallel}, \dot{\Magneticinduction}_{\perp}, 
	\dot{\Electricfield}_{\parallel}, \dot{\Electricfield}_{\parallel}),$ whose corresponding symmetric matrix $\lbrace A_{ij} 
	\rbrace_{1 \leq i,j \leq 4}$ has the following entries:
	
	\begin{align}
		A_{11} & =  (\ell_{(MBI)}^2 + \Electricfield^2 \Magneticinduction_{\perp}^2 )
			(1 + \Magneticinduction_{\perp}^2 - \Electricfield^2), \\
		A_{12} & = - (\ell_{(MBI)}^2 + \Electricfield^2 \Magneticinduction_{\perp}^2)
			\Magneticinduction_{\parallel} \Magneticinduction_{\perp}, \\
		A_{13} & = (1 + |\Magneticinduction|^2)\Electricfield \Magneticinduction_{\parallel} \Magneticinduction_{\perp}^2, \\
		A_{14} & = (1 + |\Magneticinduction|^2)
			(1 + \Magneticinduction_{\perp}^2 - \Electricfield^2)\Electricfield\Magneticinduction_{\perp}, \\
		A_{22} & = (\ell_{(MBI)}^2 + \Electricfield^2 \Magneticinduction_{\perp}^2)(1 + \Magneticinduction_{\parallel}^2), \\
		A_{23} & = -(1 + |\Magneticinduction|^2)(1 + \Magneticinduction_{\parallel}^2)\Electricfield\Magneticinduction_{\perp}, \\
		A_{24} & = -(1 + |\Magneticinduction|^2) \Electricfield \Magneticinduction_{\parallel} \Magneticinduction_{\perp}^2, \\
		A_{33} & = (1 + |\Magneticinduction|^2)^2 (1 + \Magneticinduction_{\parallel}^2), \\
		A_{34} & = (1 + |\Magneticinduction|^2)^2 \Magneticinduction_{\parallel} \Magneticinduction_{\perp}, \\
		A_{44} & = (1 + |\Magneticinduction|^2)^2 (1 + \Magneticinduction_{\perp}^2 - \Electricfield^2). 
	\end{align}
	
	By Sylvester's criterion, the positive definiteness of $\lbrace A_{ij} \rbrace_{1 \leq i,j \leq 4}$
	is equivalent to the positivity of the following four quantities, which are the determinants of an increasing sequence of
	sub-blocks along the diagonal:
	
	\begin{align}
		A_{11} & =  (\ell_{(MBI)}^2 + \Electricfield^2 \Magneticinduction_{\perp}^2 )
			(1 + \Magneticinduction_{\perp}^2 - \Electricfield^2), \label{E:FirstBlockDet} \\
		\mbox{det} \Big( \lbrace A_{ij} \rbrace_{1 \leq i,j \leq 2} \Big) & = 
			(\ell_{(MBI)}^2 + \Electricfield^2 \Magneticinduction_{\perp}^2)^2 \ell_{(MBI)}^2, 
			\label{E:SecondBlockDet} \\
		\mbox{det} \Big( \lbrace A_{ij} \rbrace_{1 \leq i,j \leq 3} \Big) & 
			= (1 + \Magneticinduction_{\parallel}^2)(1 + |\Magneticinduction|^2)^2 
			(\ell_{(MBI)}^2 + \Electricfield^2 \Magneticinduction_{\perp}^2) \ell_{(MBI)}^4, \label{E:ThirdBlockDet} \\
		\mbox{det} \Big( \lbrace A_{ij} \rbrace_{1 \leq i,j \leq 4} \Big) & = 
			(1 + |\Magneticinduction|^2)^4 \ell_{(MBI)}^8. \label{E:ADet}
	\end{align}
	We remark that we calculated the right-hand sides of \eqref{E:FirstBlockDet} and \eqref{E:SecondBlockDet} by hand, while we
	used version $11$ of Maple to compute the right-hand sides of \eqref{E:ThirdBlockDet} and \eqref{E:ADet}.
	
	As before, the uniform positivity of the right-hand sides of \eqref{E:SecondBlockDet} - \eqref{E:ADet}
	for $\Far \in \mathfrak{K}$ is manifest. The uniform positivity of the right-hand side of \eqref{E:FirstBlockDet} 
	follows from the fact that
	
	\begin{align}
		\ell_{(MBI)}^2[\Far] = \delta \iff
		E^2 = 1 + \frac{\Magneticinduction_{\perp}^2}{1+ \Magneticinduction_{\parallel}^2} 
			- \frac{\delta}{1+ \Magneticinduction_{\parallel}^2}.
	\end{align}
	We have thus shown that $\Stress_{\ \nu}^{\mu} \xi_{\mu}^{(0)} \Vmult^{\nu}$ 
	$\geq C_{\mathfrak{K}}(\dot{\Magneticinduction}_{\parallel}^2 + \dot{\Magneticinduction}_{\perp}^2 
	+ \dot{\Electricfield}_{\parallel}^2 + \dot{\Electricfield}_{\perp}^2).$ Combining this bound with 
	\eqref{E:OrthogonalComponentsBoundedFromBelow}, and using the relation \eqref{E:FaradaynormsquaredintermsofEB},
	we have that $\Stress_{\ \nu}^{\mu} \xi_{\mu}^{(0)} \Vmult^{\nu}$
	$\geq C_{\mathfrak{K}}(\dot{\Magneticinduction}_{\parallel}^2 + \dot{\Magneticinduction}_{\perp}^2
	+ \dot{\Magneticinduction}_{\times}^2 + \dot{\Electricfield}_{\parallel}^2 + \dot{\Electricfield}_{\perp}^2 + 
	\dot{\Electricfield}_{\times}^2)$ $= C_{\mathfrak{K}}(|\dot{\Electricfield}|^2 + |\dot{\Magneticinduction}|^2)$
	$\geq C_{\mathfrak{K}}|\dot{\Far}|^2.$ This completes our proof of \eqref{E:LocalExistenceCurrent}.
\end{proof}

\section{Norms, Seminorms, Energies, and Comparison Lemmas} \label{S:NormsandEnergies}

In this section, we introduce a collection of norms and seminorms that will be used in the remaining sections to estimate 
solutions to the MBI system. We also introduce the energy $\mathcal{E}_N = \mathcal{E}_N[\Far(t)],$ a related positive 
integral quantity that is constructed via the canonical stress and the ``multiplier'' vectorfield $\overline{K}.$ In Section \ref{S:EnergyEstimates}, we will study the time derivative of $\mathcal{E}_N,$ and in order to close the estimates, we need to prove inequalities that relate the norms to the energy; we provide these inequalities in Proposition \ref{P:Equivalences}. We also introduce another norm $\| \cdot \|_{H_1^N}$ on the electromagnetic initial data $(\mathring{\Displacement}, \mathring{\Magneticinduction}),$ and for small-data, we prove that this norm is equivalent to $\mathcal{E}_N[\Far(0)].$ This allows us to express the global existence smallness condition of 
Theorem \ref{T:GlobalExistence} in terms of $(\mathring{\Displacement}, \mathring{\Magneticinduction}),$ which are inherent to the Cauchy hypersurface $\Sigma_0.$ In particular, the norm on $(\mathring{\Displacement}, \mathring{\Magneticinduction})$ does not involve normal (i.e. time) derivatives. Finally, in \eqref{E:divJdot}, we provide a preliminary expression that will be needed in our estimates of $\frac{d}{dt} \big( \mathcal{E}_N^2[\Far(t)] \big).$ This computation motivates Section \ref{S:NullFormEstimates}, in which we provide algebraic estimates of the terms appearing in \eqref{E:divJdot}.
\\

\noindent \hrulefill
\ \\
	
\setcounter{equation}{0} 	
	
	\begin{definition} \label{D:UnweightedPointwiseNorms}
	Let $N \geq 0$ be an integer. If $\mathcal{A}$ is one of the sets $\mathscr{T},$ $\mathcal{O},$ or $\mathcal{Z}$ defined in
	\eqref{E:Translationsetdef} - \eqref{E:Zsetdef}, and $U$ is any tensorfield, 
	then we define the following pointwise norms of $U:$ 
	
	\begin{subequations}
	\begin{align} 
		|U|_{\Lie_{\mathcal{A};N}}^2 \eqdef \sum_{|I| \leq N} |\Lie_{\mathcal{A}}^I U|^2, 
			\label{E:UnweightedPointwiseNormLieA} \\
		|U|_{\nabla_{\mathcal{A}};N}^2 \eqdef \sum_{|I| \leq N} |\nabla_{\mathcal{A}}^I U|^2.
			\label{E:UnweightedPointwiseNormNablaA}
	\end{align}
	\end{subequations}
	
	Furthermore, if $U$ is tangent to the spheres $S_{r,t},$ and $\angn$ denotes the Levi-Civita connection corresponding to
	$\angg,$ then we define
	
	\begin{align}
		|U|_{\angn_{\mathcal{O}};N}^2 \eqdef \sum_{|I| \leq N} |\angn_{\mathcal{O}}^I U|^2.
	\end{align}
	In the above formulas, the pointwise norm $|\cdot|$ is defined in \eqref{E:Riemanniannorm},
	while the iterated derivatives $\nabla_{\mathcal{A}}^I,$ etc., are defined in Definition \ref{D:iterated}.
	
\end{definition}

\begin{definition} \label{D:weightedpointwisenorm}
	Let $\Stress^{(Maxwell)}$ be the Maxwellian canonical stress corresponding to the two-form $\dot{\Far}$ 
	defined in \eqref{E:Tmaxwelldef}, and let $\dot{\ualpha}, \dot{\alpha}, \dot{\rho}, \dot{\sigma}$ be the null components of 
	$\dot{\Far}.$ Let $\mathcal{A}$ be one of the sets $\mathscr{T}, \mathcal{O},$ or $\mathcal{Z}$ defined in   
	\eqref{E:Translationsetdef} - \eqref{E:Zsetdef}, and let $\overline{K}$ be the conformal Killing field defined in
	\eqref{E:Kfirstdef}. Then for each integer $N \geq 0,$ we define the following pointwise norms of 
	$\dot{\Far}:$ 
	
	\begin{align} 
		\Knorm \dot{\Far} \Knorm^2 & \eqdef \Stress^{(Maxwell)}(\xi^{(0)},\overline{K}) 
			= (1 + q)^2 |\dot{\ualpha}|^2 + (1 + s^2) |\dot{\alpha}|^2 + (2 + q^2 + s^2)(\dot{\rho}^2 + \dot{\sigma}^2), 
			\label{E:weightedpointwisenorm} \\
		\Knorm \dot{\Far} \Knorm_{\Lie_{\mathcal{A}};N}^2 & \eqdef \sum_{|I| \leq N}	\Knorm \Lie_{\mathcal{A}}^I \dot{\Far} 	
		\Knorm^2,
	\end{align}
	where $\xi^{(0)}$ is the $g-$dual of the time translation vectorfield $T_{(0)}.$

\end{definition}
Notice that the different components of $\dot{\Far}$ in \eqref{E:weightedpointwisenorm} carry different weights. Additionally,
we have that ${(1 + |q|)^2|\dot{\Far}|^2 \leq \Knorm \dot{\Far} \Knorm^2}.$

\begin{definition} \label{D:EnergyCurrent}
	Let $\Stress$ be the MBI canonical stress \eqref{E:Hstress} corresponding to the ``background'' $\Far$ and the variation 
	$\dot{\Far},$ and let $\overline{K}$ be the conformal Killing field defined in \eqref{E:Kfirstdef}. 
	We define the \emph{energy current} $\dot{J}_{\Far}^{\mu}[\dot{\Far}]$ corresponding to $\dot{\Far}$ to be the following 
	vectorfield:
	
	\begin{align} \label{E:Jdotdef}
		\dot{J}_{\Far}^{\mu}[\dot{\Far}] \eqdef - \Stress_{\ \nu}^{\mu} \overline{K}^{\nu}.
	\end{align}
	We note that $\dot{J}_{\Far}^{\mu}[\dot{\Far}]$ depends on the background $\Far$ through $\Stress_{\ \nu}^{\mu},$
	and that $\dot{J}_{\Far}^0[\dot{\Far}] = \Stress(\xi^{(0)}, \overline{K}).$
\end{definition}

\begin{definition} \label{D:MorawetzWeightedLieDerivativeIntegralNormN}
	Let $N \geq 0$ be an integer, and let $\Far$ and $\dot{\Far}$ be a pair of two-forms. 
	We define the weighted integral norm $\Kintnorm \dot{\Far}(t) \Kintnorm_{\Lie_{\mathcal{Z}};N}$ of $\dot{\Far}$ as follows: 
	
	\begin{align} \label{E:MorawetzWeightedLieDerivativeIntegralNormN}
		\Kintnorm \dot{\Far}(t) \Kintnorm_{\Lie_{\mathcal{Z}};N} & \eqdef 
		\left( \int_{\mathbb{R}^3} \Knorm \dot{\Far}(t,\ux) 
			\Knorm_{\Lie_{\mathcal{Z}};N}^2 \, d^3 \ux \right)^{1/2}.
	\end{align}	

	Furthermore, we define the energy $\mathcal{E}_N[\dot{\Far}(t)]$ of $\dot{\Far}$ as follows:
	
	\begin{align} \label{E:mathcalENdef}
		\mathcal{E}_N[\dot{\Far}(t)] & \eqdef \left(\sum_{|I| \leq N}
			\int_{\mathbb{R}^3} \dot{J}_{\Far}^0[\Lie_{\mathcal{Z}}^I \dot{\Far}(t,\ux)] \, d^3 \ux \right)^{1/2},
	\end{align}
	where the component $\dot{J}_{\Far}^0[\dot{\Far}(t,\ux)]$ is defined in \eqref{E:Jdotdef}. 

\end{definition}

\begin{remark}
	The $q,s$ weights under the integral in definition \eqref{E:MorawetzWeightedLieDerivativeIntegralNormN} are exactly the ones 
	needed in the global Sobolev inequality (Proposition \ref{P:GlobalSobolev}).
\end{remark}

In the next proposition, we collect together a large number of comparison estimates that will be used throughout the remainder of the article.

\begin{proposition} \label{P:Equivalences}
	Let $N \geq 0$ be an integer, and let $U$ be any tensorfield. Then the following pointwise estimates hold:
	
	\begin{subequations}
	\begin{align} 
		|U|_{\Lie_{\mathcal{Z}};N} & \approx |U|_{\nabla_{\mathcal{Z}};N}, \label{E:LieZequivalenttonablaZpointwise} \\
		|\nabla_{(N)} U| & \approx \sum_{|I| = N} |\nabla_{\mathscr{T}}^I U|, 
			\label{E:spacetimecovariantderivatvenormequivalenttotranslationalcovariantderivativenorm} \\
		\sum_{n=0}^N (1 + |q|)^{n} |\nabla_{(n)} U| & \lesssim  |U|_{\nabla_{\mathcal{Z}};N}, 
			\label{E:spacetimenablaboundedbyqpowertimenablaZ} \\
		|U|_{\nabla_{\mathcal{Z}};N}| & \lesssim \sum_{n=0}^N (1 + s)^{n} |\nabla_{(n)} U|,
			\label{E:Zcovariantderivativeslessthansweightedcovarianttensor} \\
		|U|_{\nabla_{\mathcal{Z}};N}^2|_{\Sigma_0} & \approx \sum_{n=0}^N (1 + r^2)^n |\nabla_{(n)}U|^2|_{\Sigma_0}.
			\label{E:ZcovariantderivativeslessthanrweightedcovarianttensorAlongSimga0}
	\end{align}
	\end{subequations}
	Although $|_{\Sigma_0}$ denotes restriction to $\Sigma_0$ in the above formulas, we emphasize that $\nabla_{(n)}$
	denotes the \textbf{full spacetime} covariant derivative operator of order $n.$

	Let $\Far$ be an arbitrary two-form, let $\ualpha,$ $\alpha,$ $\rho,$ $\sigma$ be its null components
	as defined in Section \ref{SS:NullComponents},
	and let $\Electricfield, \Magneticinduction, \SFar, \SFardual$ 
	be its electromagnetic decompositions as defined in Section \ref{SS:electromagneticdecomposition}.
	Let $\mathcal{A}$ be one of the three sets $\mathcal{T}, \mathcal{O}, \mathcal{Z}$ defined in 
	\eqref{E:Translationsetdef} - \eqref{E:Zsetdef}. Let $q \eqdef r - t, s \eqdef r + t$ denote the null
	coordinates. Then the following pointwise estimates hold:
	
	\begin{subequations}
	\begin{align}
		|\Far|^2 & = 2(|\Electricfield|^2 + |\Magneticinduction|^2),   
			\label{E:FaradaynormsquaredintermsofEB} \\
		|\Far|^2 & = |\ualpha|^2 + |\alpha|^2 + 2(\rho^2 + \sigma^2), \label{E:Faradaynormsquaredintermsofnullcomponents} \\
		\Knorm \Far \Knorm^2 & = |\Electricfield|^2 + |\Magneticinduction|^2 + |\SFar|^2 + |\SFardual|^2
			\label{E:FaradayKNormsquaredintermsofEBPQ} \\
			& = (1 + q^2)|\ualpha|^2 + (1+s^2)|\alpha|^2 + (2 + q^2 + s^2)(\rho^2 + \sigma^2), \notag \\
		\Knorm \Far \Knorm_{\Lie_{\mathcal{A}};N}^2 
			& = \sum_{|I| \leq N} \Big( |i_{T_{(0)}} \Lie_{\mathcal{A}}^I \Far|^2 
			+ |i_{T_{(0)}} \Lie_{\mathcal{A}}^I \Fardual|^2 
			+ |i_{S} \Lie_{\mathcal{A}}^I \Far|^2
			+ |i_{S} \Lie_{\mathcal{A}}^I \Fardual|^2 \Big), \label{E:FaradayLieAKnormsquaredintermsofinteriorproducts} \\
		\Knorm \Fardual \Knorm_{\Lie_{\mathcal{A}};N}^2 
			& = \sum_{|I| \leq N} \Big( |i_{T_{(0)}} \Lie_{\mathcal{A}}^I \Far|^2 
			+ |i_{T_{(0)}} \Lie_{\mathcal{A}}^I \Fardual|^2 
			+ |i_{S} \Lie_{\mathcal{A}}^I \Far|^2
			+ |i_{S} \Lie_{\mathcal{A}}^I \Fardual|^2 \Big), \label{E:FaradaydualLieAKnormsquaredintermsofinteriorproducts} \\
		|\nabla_{(n)} \Far|^2 & \approx |\nabla_{(n)} \Electricfield|^2 + |\nabla_{(n)} \Magneticinduction|^2, 
			\label{E:CovariantderivativesFaradaynormsquaredapproxCovariantderivatiesEB} \\
		\Knorm \Far \Knorm_{\Lie_{\mathcal{Z}};N}^2 & \approx 
			| \Electricfield |_{\Lie_{\mathcal{Z}};N}^2 
			+ | \Magneticinduction |_{\Lie_{\mathcal{Z}};N}^2
			+ | \SFar |_{\Lie_{\mathcal{Z}};N}^2
			+ | \SFardual |_{\Lie_{\mathcal{Z}};N}^2, \label{E:PointwiseweightedFarLieZnormapproxLieZnormEBPQ} \\
		|\Electricfield |_{\nabla_{\mathcal{Z}};N}^2 
				+ | \Magneticinduction |_{\nabla_{\mathcal{Z}};N}^2
				+ | \SFar |_{\nabla_{\mathcal{Z}};N}^2
				+ | \SFardual |_{\nabla_{\mathcal{Z}};N}^2
			& \lesssim \sum_{n=0}^{N} (1 + s)^{2(n+1)} (|\nabla_{(n)} \Electricfield|^2 + |\nabla_{(n)} \Magneticinduction|^2),
				\label{E:PointwiseEBPQnablaZnormlessthanHigherspowerweightednormcovariantderivativetensorEB} \\
		\sum_{|I| \leq N} \Big\lbrace (1 + |q|)^{|I|} \Knorm \nabla_{\mathscr{T}}^I \Far \Knorm \Big\rbrace
			& \lesssim \Knorm \Far \Knorm_{\Lie_{\mathcal{Z}};N},  \label{E:PointwiseqweightedqnormtranslationderivatieslessthanKnormLieZ} \\
		\sum_{n=0}^N \Big\lbrace (1 + |q|)^{n+1} | \nabla_{(n)} \Far | \Big\rbrace
			& \lesssim \Knorm \Far \Knorm_{\Lie_{\mathcal{Z}};N}, 
			\label{E:PointwiseHigherpowerqweightedcovariantderivativesofFarlessthanKnormLieZderivatives} \\
		\Knorm \Far \Knorm_{\Lie_{\mathcal{Z}};N}^2|_{\Sigma_0} 
			&\approx \sum_{n = 0}^N (1 + r^2)^{n + 1} \Big(|\nabla_{(n)} \Electricfield|^2|_{\Sigma_0} 
		+ |\nabla_{(n)} \Magneticinduction|^2|_{\Sigma_0} \Big). 
			\label{E:FarLieZintegralnormintermsofweightedrintegralnormofEBalongSigma0}
	\end{align}
	\end{subequations}
	Although $|_{\Sigma_0}$ denotes restriction to $\Sigma_0$ in the above formulas, we emphasize that $\nabla_{(n)}$
	denotes the \textbf{full spacetime} covariant derivative operator of order $n.$ Furthermore,
	in \eqref{E:PointwiseqweightedqnormtranslationderivatieslessthanKnormLieZ}, $I$ is a translational mutli-index.
	
	Let $\Far$ be any two-form. Then for any $r > 0,$ $r \eqdef |\ux|,$ we have that
	\begin{align} \label{E:nablauLnablaLFarLierotaionsnormLieZnormcomparison}
		\sum_{k + l = 0}^N (1 + |q|)^k s^l \Knorm \nabla_{\uL}^k \nabla_L^l \Far \Knorm_{\Lie_{\mathcal{O}};N-k-l} 
		\lesssim \Knorm \Far \Knorm_{\Lie_{\mathcal{Z}};N}.
	\end{align}
	
	Let $\dot{\Far}$ be another arbitrary two-form, and let $\dot{J}_{\Far}^{\mu}[\dot{\Far}]$ 
	be the energy current vectorfield defined in \eqref{E:Jdotdef}. There exists an $\epsilon > 0$ such that if 
	$\Knorm \Far \Knorm < \epsilon,$ then	the following estimates hold:
	
	\begin{subequations}
	\begin{align} 
		\dot{J}_{\Far}^{0}[\dot{\Far}] & \approx \Knorm \dot{\Far} \Knorm^2, \label{E:dotJ0Knormequivalence} \\
		\mathcal{E}_N[\dot{\Far}] & \approx \Kintnorm \dot{\Far} \Kintnorm_{\Lie_{\mathcal{Z}};N}. \label{E:EnergyNormEquivalence}
	\end{align}	
	\end{subequations}
	
\end{proposition}

\begin{remark}
	We remark that the estimates \eqref{E:LieZequivalenttonablaZpointwise} - 	
	\eqref{E:nablauLnablaLFarLierotaionsnormLieZnormcomparison} either were proven directly by Christodoulou and 
	Klainerman in \cite{dCsK1990}, or follow easily from the estimates of \cite{dCsK1990}; we nevertheless
	provide proofs here for convenience.
\end{remark}

\noindent \textbf{Proof of Proposition \ref{P:Equivalences}} 
\\

\noindent \emph{Proof of \eqref{E:LieZequivalenttonablaZpointwise}}:
The estimate \eqref{E:LieZequivalenttonablaZpointwise} can be proven inductively using 
\eqref{E:Liederivativeintermsofnabla} and Lemma \ref{L:Vectorfieldderivativenorms}.  
\\

\noindent \emph{Proof of \eqref{E:spacetimecovariantderivatvenormequivalenttotranslationalcovariantderivativenorm}
- \eqref{E:ZcovariantderivativeslessthanrweightedcovarianttensorAlongSimga0}}:

The estimate \eqref{E:spacetimecovariantderivatvenormequivalenttotranslationalcovariantderivativenorm} is easily
deduced in the inertial coordinate system $\lbrace x^{\mu} \rbrace_{\mu = 0,1,2,3}.$

To prove \eqref{E:spacetimenablaboundedbyqpowertimenablaZ}, we first note that in 
the inertial coordinate system, the following identity holds for $\mu = 0,1,2,3:$

\begin{align} \label{E:Translationidentityfavorabledenominator}
	T_{(\mu)} = \frac{x^{\kappa} \Omega_{(\kappa \mu)} + x_{\mu}S}{qs}.
\end{align}
Therefore, we have that

\begin{align}
	|\nabla_{T_{(\mu)}} U| = |(qs)^{-1}||x^{\kappa} \nabla_{\Omega_{(\kappa \mu)}} U + x_{\mu} \nabla_S U|
	\lesssim |q|^{-1} |U|_{\nabla_{\mathcal{Z}};1}.
\end{align}
Since we also have the trivial identity $|\nabla_{T_{(\mu)}} U| \leq |U|_{\nabla_{\mathcal{Z}};1},$
it follows that

\begin{align} \label{E:basecasespacetimenablaboundedbyqpowertimenablaZ}
	|\nabla_{T_{(\mu)}} U| & \lesssim (1 + |q|)^{-1} |U|_{\nabla_{\mathcal{Z}};1}.
\end{align}

Similarly, we use Lemma \ref{L:Minkowskiisflat}, Corollary \ref{C:CommutatorofZTandS}, and \eqref{E:Translationidentityfavorabledenominator} to inductively derive the following inequality:

\begin{align} \label{E:TranslationaldervativeslessthanqweightedLieZKnorm}
	|\nabla_{T_{(\mu_1)}} \cdots \nabla_{T_{(\mu_n)}} U| \lesssim (1 + |q|)^{-n} | U |_{\nabla_{\mathcal{Z}};n}.
\end{align}
Combining \eqref{E:TranslationaldervativeslessthanqweightedLieZKnorm} and \eqref{E:spacetimecovariantderivatvenormequivalenttotranslationalcovariantderivativenorm},
we deduce \eqref{E:spacetimenablaboundedbyqpowertimenablaZ}.

Inequality \eqref{E:Zcovariantderivativeslessthansweightedcovarianttensor} follows 
trivially from \eqref{E:Zcovariantderivativebounds} and the Leibniz rule. Inequality \eqref{E:ZcovariantderivativeslessthanrweightedcovarianttensorAlongSimga0} then follows from
\eqref{E:Zcovariantderivativeslessthansweightedcovarianttensor} and the fact that $s = r$ along $\Sigma_0.$ \\

\emph{Proof of \eqref{E:dotJ0Knormequivalence} and \eqref{E:EnergyNormEquivalence}:}
Since $\dot{J}_{\Far}^0 = \Stress(\xi^{(0)},\overline{K}),$ \eqref{E:dotJ0Knormequivalence} is a simple consequence of Lemma 
\ref{L:CanonicalStressErrorTermExpansion}. Inequality \eqref{E:EnergyNormEquivalence} then follows 
from integrating inequality \eqref{E:dotJ0Knormequivalence} over $\Sigma_t.$ \\

\noindent \emph{Proof of \eqref{E:FaradaynormsquaredintermsofEB} - \eqref{E:CovariantderivativesFaradaynormsquaredapproxCovariantderivatiesEB}}:

If $X$ is any vectorfield in the plane spanned by $\uL$ and $L,$ then it can be checked that
	
	\begin{align}
		|i_X \Far|^2 & = \frac{1}{4} (X_{\uL})^2 |\ualpha|^2 + \frac{1}{4} (X_{L})^2 |\alpha|^2 
			+ \frac{1}{2} X_L X_{\uL}\angg(\ualpha,\alpha) 
			+ \frac{1}{2} \big\lbrace(X_{L})^2 + (X_{\uL})^2 \big\rbrace \rho^2,  \\
		|i_X \Fardual|^2 & = \frac{1}{4} (X_{\uL})^2 |\ualpha|^2 + \frac{1}{4} (X_{L})^2 |\alpha|^2  - \frac{1}{2} X_L 
			X_{\uL}\angg(\ualpha,\alpha) + \frac{1}{2} \big\lbrace(X_L)^2 + (X_{\uL})^2\big\rbrace\sigma^2,
	\end{align}
	where $\angg(\ualpha,\alpha) = \angg^{\kappa \lambda}\ualpha_{\kappa} \alpha_{\lambda} = (g^{-1})^{\kappa \lambda}\ualpha_{\kappa} 
	\alpha_{\lambda}.$ Therefore, 
	
	\begin{align}
			|i_X \Far|^2 + |i_X \Fardual|^2 = \frac{1}{2} \Big\lbrace (X_L)^2 |\ualpha|^2 + (X_{\uL})^2 |\alpha|^2 
				+ \big[(X_L)^2 + (X_{\uL})^2\big](\rho^2 + \sigma^2) \Big\rbrace. \label{E:FarcontractXnormidentity}
	\end{align}
	
	Now taking first $X=S,$ and then $X= T_{(0)},$ using $S_L = q,$ $S_{\uL} = -s,$ $T_{(0)L} = T_{(0)\uL} = -1,$ and 
	recalling that $\SFar = i_S \Far, \SFardual = i_S \Fardual,$ $\Electricfield = i_{T_{(0)}} \Far, \Magneticinduction = 
	- i_{T_{(0)}} \Fardual,$ we find that
	
	\begin{align} 
		|\SFar|^2 + |\SFardual|^2 & = \frac{1}{2} \Big\lbrace q^2 |\ualpha|^2 + s^2 |\alpha|^2 
				+ (q^2 + s^2)(\rho^2 + \sigma^2) \Big\rbrace, \label{E:PQsumofsquares} \\
		|\Electricfield|^2 + |\Magneticinduction|^2 & = \frac{1}{2} \Big\lbrace |\ualpha|^2 + |\alpha|^2 
			+ 2(\rho^2 + \sigma^2) \Big\rbrace. \label{E:EBsumofsquares}
	\end{align}
	Adding \eqref{E:PQsumofsquares} and \eqref{E:EBsumofsquares}, and comparing with \eqref{E:weightedpointwisenorm}, 
	we conclude the following:
	\begin{align}
		\Knorm \Far \Knorm^2 = |\Electricfield|^2 + |\Magneticinduction|^2 + |\SFar|^2 + |\SFardual|^2.
	\end{align}
	We have thus proved \eqref{E:FaradaynormsquaredintermsofEB} - \eqref{E:FaradayLieAKnormsquaredintermsofinteriorproducts}. 
	With the help of Corollary \ref{C:ConformalKillingLieXhodgedualcommutation}, 
	\eqref{E:FaradaydualLieAKnormsquaredintermsofinteriorproducts} 
	follows similarly.

	Inequality \eqref{E:CovariantderivativesFaradaynormsquaredapproxCovariantderivatiesEB} follows from
	\eqref{E:spacetimecovariantderivatvenormequivalenttotranslationalcovariantderivativenorm},
	\eqref{E:FaradaynormsquaredintermsofEB}, and the fact that the electric field and magnetic induction one-forms
	associated to $\nabla_{\mathscr{T}}^I \Far$ are respectively
	$\nabla_{\mathscr{T}}^I \Electricfield$ and $\nabla_{\mathscr{T}}^I \Magneticinduction.$ \\

\noindent \emph{Proof of \eqref{E:PointwiseweightedFarLieZnormapproxLieZnormEBPQ}}:

Using the commutation identities \eqref{E:Liederivativeinteriorproductcommutator}, \eqref{E:translationscommute}, 
\eqref{E:commutatorTZ}, and \eqref{E:commutatorSZ}, it follows that for any $\mathcal{Z}-$multi-index $I,$ we have

\begin{align}
	|i_{T_{(0)}} \Lie_{\mathcal{Z}}^I \Far - \Lie_{\mathcal{Z}}^I i_{(T_{(0)})} \Far|
		& \lesssim  \sum_{\mu = 0}^3 \sum_{|J| \leq |I| - 1} |i_{T_{(\mu)}} \Lie_{\mathcal{Z}}^J \Far|, 
		\label{E:T0contractionLieZIFarcommutationidentity} \\
	|i_{T_{(0)}} \Lie_{\mathcal{Z}}^I \Fardual - \Lie_{\mathcal{Z}}^I i_{(T_{(0)})} \Fardual|
		& \lesssim  \sum_{\mu = 0}^3 \sum_{|J| \leq |I| - 1} |i_{T_{(\mu)}} \Lie_{\mathcal{Z}}^J \Fardual|, \\	
	|i_S \Lie_{\mathcal{Z}}^I \Far - \Lie_{\mathcal{Z}}^I i_S \Far|
		& \lesssim  \sum_{\mu = 0}^3 \sum_{|J| \leq |I| - 1} |i_{T_{(\mu)}} \Lie_{\mathcal{Z}}^J \Far|, \\
	|i_S \Lie_{\mathcal{Z}}^I \Fardual - \Lie_{\mathcal{Z}}^I i_S \Fardual|
		& \lesssim  \sum_{\mu = 0}^3 \sum_{|J| \leq |I| - 1} |i_{T_{(\mu)}} \Lie_{\mathcal{Z}}^J \Fardual|.
		\label{E:T0contractionLieZIFardualcommutationidentity}
\end{align}
Furthermore, it follows from \eqref{E:Faradaynormsquaredintermsofnullcomponents} and \eqref{E:FarcontractXnormidentity} 
that for any two-form $\Far$ and each translational Killing field $T_{(\mu)},$ $(\mu = 0,1,2,3),$ we have that

\begin{align} \label{E:TranslationcontractwithFarpointwiseNormbound}
	|i_{T_{(\mu)}} \Far|^2 + |i_{T_{(\mu)}} \Fardual|^2 
	\lesssim |\Far|^2 \leq \Knorm \Far \Knorm^2.
\end{align}
Combining \eqref{E:T0contractionLieZIFarcommutationidentity} - \eqref{E:T0contractionLieZIFardualcommutationidentity}
with \eqref{E:TranslationcontractwithFarpointwiseNormbound}, we deduce that

\begin{align}
	|i_{T_{(0)}} \Lie_{\mathcal{Z}}^I \Far - \Lie_{\mathcal{Z}}^I i_{T_{(0)}} \Far| & \lesssim \Knorm \Far 
		\Knorm_{\Lie_{\mathcal{Z}};|I|-1}, \label{E:T0contractionLieZIFarcommutationbound} \\
	|i_S \Lie_{\mathcal{Z}}^I \Far - \Lie_{\mathcal{Z}}^I i_S \Far| & \lesssim \Knorm \Far \Knorm_{\Lie_{\mathcal{Z}};|I|-1}.
		\label{E:ScontractionLieZIFarcommutationbound}
\end{align}

We will now prove \eqref{E:PointwiseweightedFarLieZnormapproxLieZnormEBPQ} by induction. The base case was
established in \eqref{E:FaradayKNormsquaredintermsofEBPQ}. We thus inductively assume
that \eqref{E:PointwiseweightedFarLieZnormapproxLieZnormEBPQ} holds in the case $N-1.$
Using \eqref{E:T0contractionLieZIFarcommutationbound} - \eqref{E:ScontractionLieZIFarcommutationbound}, it follows that

\begin{align}
	\Knorm \Far \Knorm_{\Lie_{\mathcal{Z}};N}^2 & = \Knorm \Far \Knorm_{\Lie_{\mathcal{Z}};N - 1}^2
		+ \sum_{|I| = N} \Big( |i_{T_{(0)}} \Lie_{\mathcal{Z}}^I \Far|^2 
		+ |i_{T_{(0)}} \Lie_{\mathcal{Z}}^I \Fardual|^2 
		+ |i_{S} \Lie_{\mathcal{Z}}^I \Far|^2
		+ |i_{S} \Lie_{\mathcal{Z}}^I \Fardual|^2 \Big) \\
	& \lesssim \Knorm \Far \Knorm_{\Lie_{\mathcal{Z}};N - 1}^2
		+ \sum_{|I| = N} \Big( |\Lie_{\mathcal{Z}}^I i_{T_{(0)}} \Far|^2 
			+ |\Lie_{\mathcal{Z}}^I i_{T_{(0)}} \Fardual|^2 
			+ |\Lie_{\mathcal{Z}}^I i_{S} \Far|^2
			+ |\Lie_{\mathcal{Z}}^I i_{S} \Fardual|^2  \Big) \notag \\
	& = \Knorm \Far \Knorm_{\Lie_{\mathcal{Z}};N - 1}^2 
		+ \sum_{|I| = N} \Big( |\Lie_{\mathcal{Z}}^I \Electricfield |^2 
		+ |\Lie_{\mathcal{Z}}^I \Magneticinduction|^2 
		+ |\Lie_{\mathcal{Z}}^I \SFar|^2
		+ |\Lie_{\mathcal{Z}}^I \SFardual|^2 \Big) \notag \\
	& \lesssim \sum_{|I| \leq N} \Big( |\Lie_{\mathcal{Z}}^I \Electricfield |^2 
		+ |\Lie_{\mathcal{Z}}^I \Magneticinduction|^2 + |\Lie_{\mathcal{Z}}^I \SFar|^2 + |\Lie_{\mathcal{Z}}^I\SFardual|^2 \Big) 
			\notag \\
	& = | \Electricfield |_{\Lie_{\mathcal{Z}};N}^2 
			+ | \Magneticinduction |_{\Lie_{\mathcal{Z}};N}^2
			+ | \SFar |_{\Lie_{\mathcal{Z}};N}^2
			+ | \SFardual |_{\Lie_{\mathcal{Z}};N}^2, \notag
\end{align}
where the induction hypothesis was used to deduce the next to last line.

For the opposite inequality, we again use \eqref{E:T0contractionLieZIFarcommutationbound} - \eqref{E:ScontractionLieZIFarcommutationbound} and the induction hypothesis to conclude 
that
\begin{align}
	& |\Electricfield |_{\Lie_{\mathcal{Z}};N}^2 
			+ | \Magneticinduction |_{\Lie_{\mathcal{Z}};N}^2
			+ | \SFar |_{\Lie_{\mathcal{Z}};N}^2
			+ | \SFardual |_{\Lie_{\mathcal{Z}};N}^2 \\
	& = \sum_{|I| = N} \Big(| \Lie_{\mathcal{Z}}^I i_{T_{(0)}} \Far|^2 + | \Lie_{\mathcal{Z}}^I i_{T_{(0)}} \Fardual|^2
		+ | \Lie_{\mathcal{Z}}^I i_{S} \Far|^2 + | \Lie_{\mathcal{Z}}^I i_{S} \Fardual|^2 \Big)
			+ | \Electricfield |_{\Lie_{\mathcal{Z}};N-1}^2 
			+ | \Magneticinduction |_{\Lie_{\mathcal{Z}};N-1}^2
			+ | \SFar |_{\Lie_{\mathcal{Z}};N-1}^2
			+ | \SFardual |_{\Lie_{\mathcal{Z}};N-1}^2 \notag \\
	& \lesssim \sum_{|I| = N} \Big(|i_{T_{(0)}} \Lie_{\mathcal{Z}}^I \Far|^2 + |i_{T_{(0)}} \Lie_{\mathcal{Z}}^I \Fardual|^2
		+ |i_{S} \Lie_{\mathcal{Z}}^I \Far|^2 + |i_{S} \Lie_{\mathcal{Z}}^I \Fardual|^2 \Big)
		+ \Knorm \Far \Knorm_{\Lie_{\mathcal{Z}};N-1}^2 \notag \\
	& = \Knorm \Far \Knorm_{\Lie_{\mathcal{Z}};N}^2. \notag 
\end{align}
\\

\noindent \emph{Proof of \eqref{E:PointwiseqweightedqnormtranslationderivatieslessthanKnormLieZ}}:

Using the facts that $\nabla_{T_{(\mu)}} S = T_{(\mu)}$ and $\nabla_{T_{(\mu)}} T_{(0)} = 0$ for $\mu = 0,1,2,3,$
together with the commutation identity \eqref{E:Covariantderivativeinteriorproductcommutator}, it follows that

\begin{align}
	|\nabla_{\mathscr{T}}^I i_{T_{(0)}} \Far - i_{T_{(0)}}  \nabla_{\mathscr{T}}^I \Far| & = 0, 
		\label{E:nablaTcontractionT0Farcommutatoris0} \\
	|\nabla_{\mathscr{T}}^I i_{T_{(0)}} \Fardual - i_{T_{(0)}}  \nabla_{\mathscr{T}}^I \Fardual| & = 0, 
		\label{E:nablaTcontractionT0Fardualcommutatoris0} \\
	|\nabla_{\mathscr{T}}^I i_S \Far - i_S \nabla_{\mathscr{T}}^I \Far| & \lesssim 
	 \sum_{\mu = 0}^3 \sum_{|J| \leq |I| - 1} |i_{T_{(\mu)}} \nabla_{\mathscr{T}}^J \Far|, 
	 \label{E:nablaT0contractionSFarcommutatorestimate} \\
	|\nabla_{\mathscr{T}}^I i_S \Fardual - i_S \nabla_{\mathscr{T}}^I \Fardual| & \lesssim 
	 \sum_{\mu = 0}^3 \sum_{|J| \leq |I| - 1} |i_{T_{(\mu)}} \nabla_{\mathscr{T}}^J \Fardual|.
	 \label{E:nablaT0contractionSFardualcommutatorestimate}
\end{align}
Furthermore, from definition \eqref{E:weightedpointwisenorm},
it follows that any two-form $\dot{\Far}$ satisfies the following inequality:

\begin{align} \label{E:T0contractionKnormweightedqcomparison}
	|i_{T_{(0)}} \dot{\Far}|^2 + |i_{T_{(0)}} {^{\star \hspace{-.03in}}\dot{\Far}}|^2 & \lesssim (1 + |q|)^{-2} \Knorm \dot{\Far} 
	\Knorm^2.
\end{align}
It thus follows from \eqref{E:nablaTcontractionT0Farcommutatoris0} - \eqref{E:nablaT0contractionSFardualcommutatorestimate}
and \eqref{E:T0contractionKnormweightedqcomparison} that

\begin{align} 
	|\nabla_{\mathscr{T}}^I i_S \Far - i_S \nabla_{\mathscr{T}}^I \Far|^2 
	 & \lesssim (1 + |q|)^{-2} \Knorm \Far \Knorm_{\nabla_{\mathscr{T}};|I|-1}^2, 
	 \label{E:nablaTIcontractionSFarcommutatorestimate}\\
	|\nabla_{\mathscr{T}}^I i_S \Fardual - i_S \nabla_{\mathscr{T}}^I \Fardual|^2 
	 & \lesssim (1 + |q|)^{-2} \Knorm \Far \Knorm_{\nabla_{\mathscr{T}};|I|-1}^2. 
	 	\label{E:nablaTIcontractionSFardualcommutatorestimate}
\end{align}

We now prove inequality \eqref{E:PointwiseqweightedqnormtranslationderivatieslessthanKnormLieZ} by induction, the base case being the trivial inequality $\Knorm \Far \Knorm^2 \lesssim \Knorm \Far \Knorm^2.$ We assume that the inequality holds in the 
case $N-1,$ and we let $I$ be any  multi-index $I$ with $|I| = N.$
Using \eqref{E:nablaTcontractionT0Farcommutatoris0} - \eqref{E:nablaTcontractionT0Fardualcommutatoris0} and
\eqref{E:nablaTIcontractionSFarcommutatorestimate} - \eqref{E:nablaTIcontractionSFardualcommutatorestimate},
it follows that 

\begin{align} \label{E:nablaTIFarKnromcomparedtocontractionT0SnablaTIFarandFardual}
	\Knorm \nabla_{\mathscr{T}}^I \Far \Knorm^2 
	& =|i_{T_{(0)}} \nabla_{\mathscr{T}}^I \Far |^2 + |i_{T_{(0)}} \nabla_{\mathscr{T}}^I \Fardual |^2 
		+ |i_{S} \nabla_{\mathscr{T}}^I \Far |^2 + |i_{S} \nabla_{\mathscr{T}}^I \Fardual |^2 \\
	& \lesssim |\nabla_{\mathscr{T}}^I i_{T_{(0)}} \Far |^2 + |\nabla_{\mathscr{T}}^I i_{T_{(0)}} \Fardual |^2 
		+ |\nabla_{\mathscr{T}}^I i_{S} \Far |^2 + |\nabla_{\mathscr{T}}^I i_{S} \Fardual |^2 
		+ (1 + |q|)^{-2} \Knorm \Far \Knorm_{\nabla_{\mathscr{T}};N-1}^2 \notag \\
	& \lesssim |\nabla_{\mathscr{T}}^I \Electricfield |^2 + |\nabla_{\mathscr{T}}^I \Magneticinduction |^2
		+ |\nabla_{\mathscr{T}}^I \SFar|^2 + |\nabla_{\mathscr{T}}^I \SFardual |^2
		+ (1 + |q|)^{-2N} \Knorm \Far \Knorm_{\Lie_{\mathcal{Z}};N-1}^2. \notag 
\end{align}
We remark that we have used the induction hypothesis to arrive at the last inequality. 

To estimate the $|\nabla_{\mathscr{T}}^I \Electricfield |^2 + |\nabla_{\mathscr{T}}^I \Magneticinduction |^2 + |\nabla_{\mathscr{T}}^I \SFar|^2 + |\nabla_{\mathscr{T}}^I \SFardual |^2$ term on the right-hand side of \eqref{E:nablaTIFarKnromcomparedtocontractionT0SnablaTIFarandFardual}, we use
\eqref{E:LieZequivalenttonablaZpointwise}, \eqref{E:spacetimecovariantderivatvenormequivalenttotranslationalcovariantderivativenorm}, 
\eqref{E:spacetimenablaboundedbyqpowertimenablaZ}, and \eqref{E:PointwiseweightedFarLieZnormapproxLieZnormEBPQ}
to deduce that

\begin{align} \label{E:mathcalTIEBPQcomparedtoqweightedKnormFar}
	|\nabla_{\mathscr{T}}^I \Electricfield|^2 + |\nabla_{\mathscr{T}}^I \Magneticinduction|^2 
		+ |\nabla_{\mathscr{T}}^I \SFar|^2 + |\nabla_{\mathscr{T}}^I \SFardual|^2
	& \lesssim (1 + |q|)^{-2N} \Big(|\Electricfield|_{\Lie_{\mathcal{Z}};N}^2
		+ |\Magneticinduction|_{\Lie_{\mathcal{Z}};N}^2 + |\SFar|_{\Lie_{\mathcal{Z}};N}^2
		+ |\SFardual|_{\Lie_{\mathcal{Z}};N}^2 \Big) \\
	& \approx  (1 + |q|)^{-2N} \Knorm \Far \Knorm_{\Lie_{\mathcal{Z}};N}^2. \notag
\end{align}
Combining \eqref{E:nablaTIFarKnromcomparedtocontractionT0SnablaTIFarandFardual} and \eqref{E:mathcalTIEBPQcomparedtoqweightedKnormFar}, we deduce the induction step. This completes the proof of 
\eqref{E:PointwiseqweightedqnormtranslationderivatieslessthanKnormLieZ}. 
\\

\noindent \emph{Proof of \eqref{E:PointwiseEBPQnablaZnormlessthanHigherspowerweightednormcovariantderivativetensorEB}
and \eqref{E:PointwiseHigherpowerqweightedcovariantderivativesofFarlessthanKnormLieZderivatives}
- \eqref{E:FarLieZintegralnormintermsofweightedrintegralnormofEBalongSigma0}}:

In the inertial coordinate system $\lbrace x^{\mu} \rbrace_{\mu=0,1,2,3},$ inequality \eqref{E:PointwiseEBPQnablaZnormlessthanHigherspowerweightednormcovariantderivativetensorEB} follows from the definition \eqref{E:PQinertialcomponents} of $\SFar$ and $\SFardual,$
\eqref{E:Zcovariantderivativeslessthansweightedcovarianttensor}, \eqref{E:CovariantderivativesFaradaynormsquaredapproxCovariantderivatiesEB}, 
the Leibniz rule, and the fact that the coordinate functions $x^{\mu}$ satisfy $|x^{\mu}| \leq s,$ $|\nabla x^{\mu}| \leq 1,$ and $|\nabla_{(2)} x^{\mu}| = 0.$

Inequality \eqref{E:PointwiseHigherpowerqweightedcovariantderivativesofFarlessthanKnormLieZderivatives} 
follows from \eqref{E:spacetimecovariantderivatvenormequivalenttotranslationalcovariantderivativenorm},
\eqref{E:PointwiseqweightedqnormtranslationderivatieslessthanKnormLieZ},
and the fact that  $(1 + |q|) |\dot{\Far}| \lesssim \Knorm \dot{\Far} \Knorm$ holds for any two-form $\dot{\Far}.$

Inequality \eqref{E:FarLieZintegralnormintermsofweightedrintegralnormofEBalongSigma0} then follows
from \eqref{E:LieZequivalenttonablaZpointwise}, 
\eqref{E:CovariantderivativesFaradaynormsquaredapproxCovariantderivatiesEB},
\eqref{E:PointwiseweightedFarLieZnormapproxLieZnormEBPQ},
\eqref{E:PointwiseEBPQnablaZnormlessthanHigherspowerweightednormcovariantderivativetensorEB},
\eqref{E:PointwiseHigherpowerqweightedcovariantderivativesofFarlessthanKnormLieZderivatives},
and that fact that $q= s = r$ along $\Sigma_0.$ 
\\

\noindent \emph{Proof of \eqref{E:nablauLnablaLFarLierotaionsnormLieZnormcomparison}}

Let $k \geq 0$ be an integer. Then since $\uL = T_{(0)} - \omega^a T_{(a)},$ 
$\omega^j \eqdef x^j/r,$ we use \eqref{E:PointwiseqweightedqnormtranslationderivatieslessthanKnormLieZ} and the fact that 
$T_{(0)} \omega^j = \omega^a T_{(a)} \omega^j = 0,$ $(j=1,2,3),$ to conclude that the following inequality holds for $r > 0:$

\begin{align} \label{E:nablauLFarKnormboundedbyqweightedLieZFarKnorm}
	\Knorm \nabla_{\uL}^k \Far \Knorm \lesssim (1 + |q|)^{-k} \Knorm \Far \Knorm_{\mathcal{Z};k}.
\end{align}

Next, we observe the following identity, which holds for any two-form $\Far:$
\begin{align} \label{E:snablaLFarintermsofSqnablauLFar}
	s \nabla_L \Far = 2 \Lie_S \Far + q \nabla_{\uL} \Far - 4 \Far.
\end{align}
Now using the commutation properties of Lemma \ref{L:LLunderlinecommutewithLieO}, we may iterate \eqref{E:snablaLFarintermsofSqnablauLFar} to conclude that there exist constants $C_{abl}$ such that

\begin{align} \label{E:snablaLFarintermsofqnablauLLieSFariterated}
	(s \nabla_L)^l \Far = \sum_{b=0}^l \sum_{a=0}^b C_{abl} (q \nabla_{\uL})^a \Lie_S^{b-a} \Far.
\end{align}
Using the facts that $- \nabla_{\uL} q = \nabla_{L} s = 2,$ it follows from \eqref{E:snablaLFarintermsofqnablauLLieSFariterated} that there exist constants $\widetilde{C}_{abl}$ such that

\begin{align} \label{E:snablaLFarintermsofqnablauLLieSFaralternate}
	s^l \nabla_L^l \Far = \sum_{b=0}^l \sum_{a=0}^b \widetilde{C}_{abl} q^a\nabla_{\uL}^a \Lie_S^{b-a} \Far.
\end{align}
Furthermore, if $I$ is any rotational multi-index, then using Lemma \ref{L:LLunderlinecommutewithLieO}, the facts
that $\nabla_{\uL} s = \nabla_L q = \Lie_{\mathcal{O}}^I s = \Lie_{\mathcal{O}}^I q = 0,$ that $\nabla_{\uL} q = - 2,$ and \eqref{E:nablauLFarKnormboundedbyqweightedLieZFarKnorm}, it follows that

\begin{align} \label{E:notyetIsummednablauLnablaLFarLierotaionsnormLieZnormcomparison}
	s^l \Knorm \Lie_{\mathcal{O}}^I \nabla_L^l \nabla_{\uL}^k \Far \Knorm 
	& \lesssim \sum_{b=0}^l \sum_{a=0}^b |q|^a \Knorm \nabla_{\uL}^{a + k} \Lie_{\mathcal{O}}^I \Lie_S^{b-a} \Far \Knorm \\
	& \lesssim \sum_{b=0}^l \sum_{a=0}^b |q|^a (1 + |q|)^{-(a + k)} \Knorm \Far \Knorm_{\Lie_{\mathcal{Z}};k + b + |I|} 
		\notag \\
	& \lesssim (1 + |q|)^{-k} \Knorm \Far \Knorm_{\Lie_{\mathcal{Z}};k + l + |I|}. \notag
\end{align}
Multiplying each side of \eqref{E:notyetIsummednablauLnablaLFarLierotaionsnormLieZnormcomparison}
by $(1 + |q|)^k$ and summing over all rotational multi-indices $|I| \leq N - k - l,$ we arrive at \eqref{E:nablauLnablaLFarLierotaionsnormLieZnormcomparison}. \hfill $\qed$

\subsection{Norms for \texorpdfstring{$\mathring{\Displacement}$ and $\mathring{\Magneticinduction}$}{the data}} \label{SS:DataNorms}
In this section, we introduce a weighted Sobolev norm that will be used to express our global existence smallness condition in terms of the restriction of $(\Displacement, \Magneticinduction)$ to $\Sigma_0,$ which we denote by $(\mathring{\Displacement}, \mathring{\Magneticinduction}).$ This norm is computed using only quantities that are inherent to $\Sigma_0.$
However, during the course of our global existence argument, we analyze the norm 
$\Kintnorm \Far(t) \Kintnorm_{\Lie_{\mathcal{Z}};N},$ and in particular that we need the smallness of $\Kintnorm \Far(0) \Kintnorm_{\Lie_{\mathcal{Z}};N}$ to close our estimates. Since this latter quantity involves derivatives that are normal to $\Sigma_0,$ we need to use the MBI equations to relate the size of $\Kintnorm \Far(0) \Kintnorm_{\Lie_{\mathcal{Z}};N}$ to the size of $(\mathring{\Displacement}, \mathring{\Magneticinduction})$ as measured by the $\Sigma_0-$inherent norm. This is accomplished in Lemma \ref{E:LiederivateFarnormEBnorminitialequivalence}.
 
We now define the aforementioned weighted Sobolev norm.

\begin{definition} \label{D:HNdeltanorm}
	Let $U$ be a tensorfield tangent to the Cauchy hypersurface $\Sigma_0.$ 
	Let $d(\ux)$ denote the Riemannian distance from the origin in $\Sigma_0$ to the point $\ux \in \Sigma_0;$
	i.e., in a Euclidean coordinate system $\lbrace x^j \rbrace_{j=1,2,3}$ on $\Sigma,$ $d^2(\ux) = |\ux|^2 = (x^1)^2 + 
	(x^2)^2 + (x^3)^2.$ Then for any integer $N \geq 0,$ and any real number $\delta,$ we define the $H_{\delta}^N$ norm of $U$ 
	by 
	\begin{align} \label{E:HNdeltanorm}
		\| U \|_{H_{\delta}^N}^2 \eqdef \sum_{n=0}^N \int_{\Sigma_0} \big(1 + d^2(\ux)\big)^{(\delta + n)} 
		|\SigmafirstfundNabla_{(n)} U(\ux)|^2 \, d^3 \ux.
	\end{align}
\end{definition}

As discussed in detail in \cite{yCBdC1981}, the following norm, which is used in our proof of
Lemma \ref{L:EBFarintegralnormequivalence}, can be bounded by a suitable choice of one of the above
weighted Sobolev norms; see Lemma \ref{L:SobolevEmbeddingHNdeltaCNprimedeltamprime}.

\begin{definition} \label{D:CNdeltanorm}
	Let $U$ be a tensorfield tangent to the Cauchy hypersurface $\Sigma_0.$ 
	Let $d(\ux)$ denote the Riemannian distance from the origin in $\Sigma_0$ to the point $\ux \in \Sigma_0;$
	i.e., in a Euclidean coordinate system on $\Sigma_0,$ $d^2(\ux) = |\ux|^2 = (x^1)^2 + (x^2)^2 + (x^3)^2.$ Then for 
	any integer $N \geq 0,$ and any real number $\delta,$ we define the $C_{\delta}^N$ norm of $U$ by
	
	\begin{align} \label{E:CNdeltanorm}
		\| U \|_{C_{\delta}^N}^2 \eqdef \sum_{n=0}^N \sup_{\ux \in \Sigma_0} \big(1 + d^2(\ux)\big)^{(\delta + n)} 
		|\SigmafirstfundNabla_{(n)} U(\ux)|^2.
	\end{align}
\end{definition}

The next lemma is used in the proof of the subsequent lemma, which is the main result of this section.

\begin{lemma} 
		Let $N \geq 0$ be an integer. Let $\Far(t,\ux)$ be a two-form, and let 
		$\big(\Electricfield(t,\ux),\Magneticinduction(t,\ux)\big)$ 
		be its electromagnetic decomposition as defined in Section \ref{SS:electromagneticdecomposition}. Then
	 
	 \begin{align} \label{E:LiederivateFarnormEBnorminitialequivalence}
		\Kintnorm \Far(0) \Kintnorm_{\Lie_{\mathcal{Z}};N}^2 
		\approx \sum_{n=0}^N \int_{\mathbb{R}^3} (1 + |\uy|^2)^{n+1} (|\nabla_{(n)} \Electricfield (0,\uy)|^2 
			+ |\nabla_{(n)} \Magneticinduction (0,\uy)|^2) \, d^3 \uy.
	\end{align}
\end{lemma}

\begin{remark}
	Note that on the right-hand side of \eqref{E:LiederivateFarnormEBnorminitialequivalence},$\nabla_{(n)}$
	denotes the \textbf{full spacetime} covariant derivative operator of order $n.$ In particular, 
	$\nabla_{(n)} \Electricfield$ and $\nabla_{(n)} \Magneticinduction$ involve both tangential and normal derivatives
	of $\Electricfield$ and $\Magneticinduction.$
\end{remark}

\begin{proof}
	Inequality \eqref{E:LiederivateFarnormEBnorminitialequivalence} follows directly from integrating 
	\eqref{E:FarLieZintegralnormintermsofweightedrintegralnormofEBalongSigma0} over $\Sigma_0.$
\end{proof}

We now state and prove the main result of this section.

\begin{lemma} \label{L:EBFarintegralnormequivalence}
	Let $N \geq 2$ be an integer. Let $\Far$ be a solution to the MBI system, and let $(\mathring{\Displacement}, 
	\mathring{\Magneticinduction})$ be its electromagnetic decomposition along the Cauchy hypersurface $\Sigma_0$ as defined in 
	Section \ref{SS:electromagneticdecomposition}. Let $\Kintnorm \cdot \Kintnorm_{\Lie_{\mathcal{Z}};N}$ and 
	$\| \cdot \|_{H_1^N}$ be the weighted integral norms defined in \eqref{E:MorawetzWeightedLieDerivativeIntegralNormN} and 
	\eqref{E:HNdeltanorm} respectively. There exists a constant $\epsilon > 0$ such that if 
	$\| (\mathring{\Displacement}, \mathring{\Magneticinduction}) \|_{H_1^N} < \epsilon,$
	then
	
	\begin{align}	\label{E:EBFarintegralnormequivalence}
		\Kintnorm \Far(0) \Kintnorm_{\Lie_{\mathcal{Z}};N} 
			\approx \| (\mathring{\Displacement}, \mathring{\Magneticinduction}) \|_{H_1^N}.
	\end{align}
\end{lemma}

\begin{proof}
	We first recall that by \eqref{E:EintermsofDB}, we have that
	
	\begin{align} \label{E:mathringEintermsofmathringDmathringB}
		\mathring{\Electricfield} & = \frac{\mathring{\Displacement} + \mathring{\Magneticinduction} \times 
		(\mathring{\Magneticinduction} 
			\times \mathring{\Displacement})}{(1 + |\mathring{\Magneticinduction}|^2 + |\mathring{\Displacement}|^2 
			+ |\mathring{\Displacement} \times \mathring{\Magneticinduction}|^2)^{1/2}}.
	\end{align}
	Using Corollary \ref{C:CompositionProductHNdelta}, it follows from \eqref{E:mathringEintermsofmathringDmathringB} 
	that if $\epsilon$ is sufficiently small, then
	
	\begin{align} \label{E:DBEHsmalldataHNdeltanormequivalence}
		\| (\mathring{\Displacement}, \mathring{\Magneticinduction}) \|_{H_1^N} 
		\approx \| (\mathring{\Electricfield}, \mathring{\Magneticinduction}) \|_{H_1^N}.
	\end{align}
	
	We now introduce the abbreviation $V \eqdef (\Magneticinduction,\Displacement)$
	and re-write equations \eqref{E:partialtBintermsofBandD} and \eqref{E:partialtDintermsofBandD} as
	
	\begin{align} \label{E:partialtV}
		\nabla_{T_{(0)}} V = \mathfrak{L}(\SigmafirstfundNabla V) + \mathfrak{F}(V, \SigmafirstfundNabla V),
	\end{align}
	where $\mathfrak{L}(\cdot)$ is a constant linear map (i.e., a constant-coefficient matrix),
	and $\mathfrak{F}(V,\SigmafirstfundNabla V),$ which consists of ``cubic-order'' error terms,
	is a smooth function $V, \SigmafirstfundNabla V.$ We will make use of the following
	consequence of Lemma \ref{L:SobolevEmbeddingHNdeltaCNprimedeltamprime}:
	
	\begin{align} \label{E:SobolevEmbeddingHNdeltaCNprimedeltamprimeagain}
		\| V \|_{C_1^{N - 2}} \lesssim \| V \|_{H_1^N}.
	\end{align}
	
	We now repeatedly differentiate \eqref{E:partialtV} with $\nabla_{T_{(0)}}$ and 
	$\SigmafirstfundNabla,$ using the resulting equations to replace $\nabla_{T_{(0)}}$
	derivatives with $\Sigma_0-$tangential derivatives, and using
	\eqref{E:SobolevEmbeddingHNdeltaCNprimedeltamprimeagain} to inductively 
	conclude that if $\epsilon$ is sufficiently small and $0 \leq n \leq N,$ then
	
	\begin{align} \label{E:partialtaVfirstinequality}
		|\nabla_{(n)} V||_{t=0} \lesssim |\SigmafirstfundNabla_{(n)} \mathring{V}|
			+ \sum \prod_{i=1}^{n} |\SigmafirstfundNabla_{(m_i)} \mathring{V}|,
	\end{align}
	where the sum is taken over all non-negative integers $m_1,m_2,\cdots,m_n$ such that $\sum_{i=1}^{n} m_i = n.$ We remark 
	that the implicit constant in \eqref{E:partialtaVfirstinequality} depends on $\mathfrak{F}$ and its first $n$ derivatives 
	with respect to $V, \SigmafirstfundNabla V,$ and that we have used \eqref{E:SobolevEmbeddingHNdeltaCNprimedeltamprimeagain} to bound 
	non-differentiated factors of $\mathring{V}$ in $L^{\infty}$ from above by $C \epsilon.$ 
	Multiplying each side of \eqref{E:partialtaVfirstinequality} by 
	$(1 + |\ux|^2)^{(n+1)/2} = (1 + |\ux|^2)^{(1 + \sum_{i=1}^n m_i)/2},$
	squaring, integrating, and using Corollary \ref{C:WeightedSobolevSpaceMultiplicationProperties},
	it follows that
	
	\begin{align} \label{E:partialtaVthirdinequality}
		\int_{\mathbb{R}^3} (1 + |\uy|^2)^{n+1} |\nabla_{(n)} V(0,\uy)|^2 \, d^3 \uy 
		& \lesssim \sum_{a=0}^N \int_{\mathbb{R}^3} (1 + |\uy|^2)^{a+1} |\SigmafirstfundNabla_{(a)} \mathring{V}(\uy)|^2 \, d^3 \uy \\
		& \lesssim \| (\mathring{\Displacement}, \mathring{\Magneticinduction}) \|_{H_{1}^N}^2. \notag
	\end{align}
	We have thus shown that for sufficiently small $\epsilon,$ we have
	
	\begin{align} \label{E:DBinitialfullspacetimeweightednormcomparedtointrinsicweightedSobolevnorm}
		\sum_{n=0}^N \int_{\mathbb{R}^3} (1 + |\uy|^2)^{n+1} (|\nabla_{(n)} \Displacement (0,\uy)|^2 
			+ |\nabla_{(n)} \Magneticinduction (0,\uy)|^2)\, d^3 \uy 
		\approx \| (\mathring{\Displacement}, \mathring{\Magneticinduction}) \|_{H_{1}^N}^2,
	\end{align}
	the $\gtrsim$ direction inequality being trivial. We remark that the left-hand side of 
	\eqref{E:DBinitialfullspacetimeweightednormcomparedtointrinsicweightedSobolevnorm} involves normal derivatives, while
	the right-hand side involves only $\Sigma_0-$tangential derivatives.
	
	 Applying similar reasoning to equation \eqref{E:EintermsofDB}, we derive the following inequality, valid for all 
	 sufficiently small $\epsilon:$
	 
	 \begin{align} \label{E:EBDBinitialfullspacetimeweightedSobolevnormrelation}
			\int_{\mathbb{R}^3} (1 + |\uy|^2)^{n+1} (|\nabla_{(n)} \Electricfield (0,\uy)|^2 + |\nabla_{(n)} \Magneticinduction 
			(0,\uy)|^2) \, d^3 \uy \lesssim
			\int_{\mathbb{R}^3} (1 + |\uy|^2)^{n+1} (|\nabla_{(n)} \Displacement (0,\uy)|^2 + |\nabla_{(n)} \Magneticinduction 
			(0,\uy)|^2) \, d^3 \uy.
		\end{align}
	From \eqref{E:DBEHsmalldataHNdeltanormequivalence}, 
	\eqref{E:DBinitialfullspacetimeweightednormcomparedtointrinsicweightedSobolevnorm}
	and \eqref{E:EBDBinitialfullspacetimeweightedSobolevnormrelation}, it follows that
	
	\begin{align} \label{E:EBallderivativesH1Nnormintermsofdata}
		\sum_{n=0}^N \int_{\mathbb{R}^3} (1 + |\uy|^2)^{n+1} (|\nabla_{(n)} \Electricfield (0,\uy)|^2 
			+ |\nabla_{(n)} \Magneticinduction (0,\uy)|^2) \, d^3 \uy
		\approx \| (\mathring{\Displacement}, \mathring{\Magneticinduction}) \|_{H_{1}^N}^2.
	\end{align}
	The desired inequalities \eqref{E:EBFarintegralnormequivalence} now follow from 
	\eqref{E:LiederivateFarnormEBnorminitialequivalence} and \eqref{E:EBallderivativesH1Nnormintermsofdata}.
	
\end{proof}

\subsection{An expression for \texorpdfstring{$\nabla_{\mu}(\dot{J}_{\Far}^{\mu}[\dot{\Far}])$}{the divergence of
the energy current}}

In Section \ref{S:EnergyEstimates}, we will bound (from above) the time derivative of the energy $\mathcal{E}_N[\Far]$ defined in \eqref{E:mathcalENdef}. By the divergence theorem, the analysis amounts to estimating the $L^1$ norms of the quantities 
$\nabla_{\mu}(\dot{J}_{\Far}^{\mu}[\dot{\Far}]),$ where $\dot{\Far} \eqdef \Lie_{\mathcal{Z}}^I \Far.$ In this section, we take a preliminary step in this direction by providing an expression for $\nabla_{\mu}(\dot{J}_{\Far}^{\mu}[\dot{\Far}]),$ where $\dot{\Far}$ is a solution to the equations of variation \eqref{E:EOVBianchi} - \eqref{E:EOVMBI}. 

To begin, we compute that for any vectorfield $X,$

\begin{align} \label{E:divergenceofTX}
	\nabla_{\mu} (\Stress_{\ \nu}^{\mu} X^{\nu}) = (\nabla_{\mu} \Stress_{\ \nu}^{\mu})X^{\nu}
		+ \frac{1}{2} \Stress^{\mu \nu} {^{(X)}\pi_{\mu \nu}}
		+ \frac{1}{2} (\Stress^{\mu \nu} - \Stress^{\nu \mu})\nabla_{\mu} X_{\nu},
\end{align}
where ${^{(X)}\pi_{\mu \nu}}$ is defined in \eqref{E:deformationdef}. Now according to \eqref{E:overlineKdeformationtensor}, 
we have that ${^{(\overline{K})}\pi_{\mu \nu}} = 4tg_{\mu \nu}.$ Therefore, by \eqref{E:widetildeTtracefree}, 
the $\Stress^{\mu \nu} {^{(\overline{K})}\pi_{\mu \nu}}$ term on the right-hand side of 
\eqref{E:divergenceofTX} vanishes. If we also make use of equations \eqref{E:widetildeTantisymmetricpart} and
\eqref{E:divergenceofwidetildeT}, then we easily arrive at the following lemma.

\begin{lemma} \label{L:divJdot}
Let the two-form $\dot{\Far}$ be a solution to the equations of variation \eqref{E:EOVBianchi} - \eqref{E:EOVMBI} associated to the background two-form $\Far.$ Let $\dot{J}_{\Far}^{\mu}[\dot{\Far}] = 
- \Stress_{\ \nu}^{\mu} \overline{K}^{\nu}$ be the energy current vectorfield defined in \eqref{E:Jdotdef}. Then 

\begin{align} \label{E:divJdot}
	\nabla_{\mu}(\dot{J}_{\Far}^{\mu}[\dot{\Far}]) 
	& = \frac{1}{2} H^{\zeta \eta \kappa \lambda} \dot{\Far}_{\zeta \eta} \mathfrak{J}_{\nu \kappa \lambda} \overline{K}^{\nu} 
		- \overline{K}^{\nu} \dot{\Far}_{\nu \eta} \mathfrak{I}^{\eta}
		- (\nabla_{\mu}H^{\mu \zeta \kappa \lambda}) \dot{\Far}_{\kappa \lambda} \dot{\Far}_{\nu \zeta} \overline{K}^{\nu}
		+ \frac{1}{4} (\overline{K}^{\nu} \nabla_{\nu}H^{\zeta \eta \kappa \lambda}) \dot{\Far}_{\zeta \eta} 
		\dot{\Far}_{\kappa \lambda} \\
	& \ \ \ - \frac{1}{4} \bigg\lbrace
		\ell_{(MBI)}^{-2} \Far^{\kappa \lambda}\dot{\Far}_{\kappa \lambda} 
			\big( \Far^{\nu \zeta}\dot{\Far}_{\ \zeta}^{\mu} - \Far^{\mu \zeta} \dot{\Far}_{\ \zeta}^{\nu} \big) 
		+ \big(1 + \Farinvariant_{(2)}^2 \ell_{(MBI)}^{-2}\big)
			\Fardual^{\kappa \lambda} \dot{\Far}_{\kappa \lambda} \big( 
			\Fardual^{\nu \zeta}\dot{\Far}_{\ \zeta}^{\mu} - \Fardual^{\mu \zeta}\dot{\Far}_{\ \zeta}^{\nu} \big) 
			\notag \\
		& \ \ \ \ \ \ + \Farinvariant_{(2)} \ell_{(MBI)}^{-2} \Fardual^{\kappa \lambda}\dot{\Far}_{\kappa \lambda}
				\big( \Far^{\mu \zeta}\dot{\Far}_{\ \zeta}^{\nu} - \Far^{\nu \zeta}\dot{\Far}_{\ \zeta}^{\mu} \big) 
			+ \Farinvariant_{(2)} \ell_{(MBI)}^{-2} \Far^{\kappa \lambda}\dot{\Far}_{\kappa \lambda}
			\big( \Fardual^{\mu \zeta}\dot{\Far}_{\ \zeta}^{\nu}  
			- \Fardual^{\nu \zeta}\dot{\Far}_{\ \zeta}^{\mu} \big) \bigg\rbrace 
			\nabla_{\mu} \overline{K}_{\nu}. \notag
\end{align}

\end{lemma}

\begin{remark}
	Because of \eqref{E:MBIInhomogeneoustermsJvanish}, the first term on the right-hand side of \eqref{E:divJdot} 
	is $0$ for all of the variations of interest in this article.
\end{remark}

\section{The Null Condition and Geometric/Algebraic Estimates of the Nonlinearities} \label{S:NullFormEstimates}
\setcounter{equation}{0}

In this section, we provide a partial\footnote{We can prove our desired estimates without the use of a fully detailed null decomposition.} null decomposition of the terms appearing on the right-hand side of 
equation \eqref{E:divJdot}, where $\dot{\Far}$ is equal to one of the iterated Lie derivatives 
$\Lie_{\mathcal{Z}}^I \Far$ of a solution $\Far$ to the MBI system. These geometric/algebraic 
estimates form the backbone of the proof of  Proposition \ref{P:Energydifferentialinequality}, which is our main energy estimate. It is in this section that the null condition is revealed; as we will see, and as is described at the end of Section \ref{SS:Commentsonanalysis}, the worst possible combinations of terms are absent from the right-hand side of \eqref{E:divJdot}.
\\

\noindent \hrulefill
\ \\

We begin with the following simple lemma, which show that covariant and Lie derivatives of null form expressions can also be expressed in terms of null forms.

\begin{lemma} \label{L:Liederivativeofnullformsaresumsofnullforms}
	Let $m \geq 0$ be an integer, and let $\mathcal{Q}_{(i)}(\cdot,\cdot),$ $i = 1,2,$ denote the null forms defined in 
	\eqref{E:nullform1} - \eqref {E:nullform2}. Then for any vectorfield $X$ and any pair of two-forms $\Far, \Gar,$ and $i=1,2,$ 
	we have that
	
	\begin{align} \label{E:Liederivativeofnullformsaresumsofnullforms}
		\nabla_X^m \Big(\mathcal{Q}_{(i)}(\Far,\Gar)\Big) = \sum_{a + b = m} \binom{m}{a}
		\mathcal{Q}_{(i)}(\nabla_X^a \Far, \nabla_X^b \Gar).
	\end{align}

	Furthermore, under the convention of Remark \ref{R:LieDerivatives},
	for any $\mathcal{Z}-$multi-index $I,$ there exist constants 
	$C_{I_1,I_2}$ such that
	
	\begin{align} \label{E:Liederivativeofnullformsaresumsofnullforms}
		\Lie_{\mathcal{Z}}^I \Big(\mathcal{Q}_{(i)}(\Far,\Gar)\Big) = \sum_{I_1 + I_2 \leq I} 
		C_{I_1,I_2} \mathcal{Q}_{(i)}(\Lie_{\mathcal{Z}}^{I_1}\Far,\Lie_{\mathcal{Z}}^{I_2}\Gar).
	\end{align}

\end{lemma}

\begin{proof}
	To prove \eqref{E:Liederivativeofnullformsaresumsofnullforms} in the case $i=1,$ we use the Leibniz rule and the fact 
	that $\nabla_X g = 0.$ In the case $i=2,$ we use similar reasoning, plus Lemma \ref{L:nablaXhodgedualcommutation}.
	
	To prove \eqref{E:Liederivativeofnullformsaresumsofnullforms}
	in the case $i=1,$ we use the Leibniz rule
	and the fact that \eqref{E:LieZonmupper} holds for any $Z \in \mathcal{Z}.$ In 
	the case $i=2,$ we use similar reasoning, plus Corollary \ref{C:ConformalKillingLieXhodgedualcommutation}.
	
\end{proof}

We now state a lemma concerning the null structure of some of the factors appearing in the terms in braces on
the right-hand side of \eqref{E:divJdot}. This lemma is in the spirit of Lemma \ref{L:NullForms}.

\begin{lemma} \label{L:twoformsMorawetzderivativecontraction}
	If $\Far$ and $\Gar$ are two-forms, then
	\begin{align} \label{E:twoformsMorawetzderivativecontraction}
		|\Far_{\nu}^{\ \zeta} \Gar_{\mu \zeta} \nabla^{\mu} \overline{K}^{\nu}| 
		\lesssim s \big(|\Far|_{\mathcal{L}\mathcal{U}} |\Gar| + |\Far| |\Gar|_{\mathcal{L}\mathcal{U}}
		+ |\Far|_{\mathcal{T} \mathcal{T}} |\Gar|_{\mathcal{T} \mathcal{T}} \big).
	\end{align}
\end{lemma}

\begin{proof}
	Lemma \ref{L:twoformsMorawetzderivativecontraction} follows from \eqref{E:XYnullframeinnerproduct} together with
	the null decomposition of $\nabla \overline{K}$ given in \eqref{E:nablaoverlineKLunderlineL} - \eqref{E:nablaoverlineKAL}.
\end{proof}

The next lemma is a technical precursor to the subsequent one. 

\begin{lemma} \label{L:Htrianglecommutatorsingleterm}
	Let $N \geq 0$ be an integer. Let $\Far$ be a two-form, and let $H_{\triangle}^{\mu \zeta \kappa \lambda}$ be the 
	($\Far-$dependent) tensorfield defined in \eqref{E:Htriangledef}. Suppose that $J,J'$ are $\mathcal{Z}-$multi-indices, and 
	that $|J| \leq N.$ There exists an $\epsilon > 0$ such that if $|\Far|_{\Lie_{\mathcal{Z}}; \lfloor N/2 \rfloor} 
	\leq \epsilon,$ then the following pointwise estimate holds for $\zeta = 0,1,2,3:$
	
	\begin{align} \label{E:Htrianglecommutatorsingleterm}
		\Big| \big(\Lie_{\mathcal{Z}}^{J} (H_{\triangle}^{\mu \zeta \kappa \lambda}) \big) 
		\nabla_{\mu} \Lie_{\mathcal{Z}}^{J'} \Far_{\kappa \lambda} \Big| 
		\lesssim	\sum_{|J_1| + |J_2| \leq |J|}
			|\Lie_{\mathcal{Z}}^{J_1} \Far|
			\Big\lbrace |\Lie_{\mathcal{Z}}^{J_2} \Far||\nabla(\Lie_{\mathcal{Z}}^{J'} \Far)|_{\mathcal{L}\mathcal{U}}  
			+ |\Lie_{\mathcal{Z}}^{J_2} \Far|_{\mathcal{L}\mathcal{U}} |\nabla(\Lie_{\mathcal{Z}}^{J'} \Far)|
			+ |\Lie_{\mathcal{Z}}^{J_2} \Far|_{\mathcal{T}\mathcal{T}}|\nabla(\Lie_{\mathcal{Z}}^{J'} 
			\Far)|_{\mathcal{T}\mathcal{T}}  \Big\rbrace.
	\end{align}
\end{lemma}

\begin{remark}
	In the above inequality, $ \lfloor N/2 \rfloor$ denotes the largest integer less than or equal to $N/2,$
	and the seminorms $|\cdot|_{\mathcal{L}\mathcal{U}}, |\nabla \cdot|_{\mathcal{L}\mathcal{U}},$ etc. are 
	defined in Definition \ref{D:contractionnomrs}.
\end{remark}

\begin{proof}
	
	We begin by recalling that
	\begin{align} \label{E:Htrianglerecall}
		H_{\triangle}^{\mu \zeta \kappa \lambda} & = \frac{1}{2} (g^{-1})^{\mu \widetilde{\mu}} (g^{-1})^{\zeta \widetilde{\zeta}}
			(g^{-1})^{\kappa \widetilde{\kappa}} (g^{-1})^{\lambda \widetilde{\lambda}} \bigg\lbrace 
		- \ell_{(MBI)}^{-2} \Far_{\widetilde{\mu} \widetilde{\zeta}} \Far_{\widetilde{\kappa} \widetilde{\lambda}} 
		+ \Farinvariant_{(2)} \ell_{(MBI)}^{-2} \Big(\Far_{\widetilde{\mu} \widetilde{\zeta}} 
			\Fardual_{\widetilde{\kappa} \widetilde{\lambda}} + \Fardual_{\widetilde{\mu} \widetilde{\zeta}} 
			\Far_{\widetilde{\kappa} \widetilde{\lambda}} \Big) - \Big(1 + \Farinvariant_{(2)}^2 
			\ell_{(MBI)}^{-2} \Big)\Fardual_{\widetilde{\mu} \widetilde{\zeta}} \Fardual_{\widetilde{\kappa} \widetilde{\lambda}}
		\bigg\rbrace.
	\end{align}
	Note that to avoid the possible confusion described in Remark \ref{R:LieDerivatives}, we have lowered all of the indices on 
	$\Far$ in preparation for Lie differentiation. We now claim that
	
	\begin{align} \label{E:Htrianglecommutatorsingletermfirstbound}
		\Big| \big(\Lie_{\mathcal{Z}}^{J} (H_{\triangle}^{\mu \zeta \kappa \lambda}) \big) 
		\nabla_{\mu} \Lie_{\mathcal{Z}}^{J'} \Far_{\kappa \lambda} \Big| 
		\lesssim \sum_{i=1}^2 \sum_{|J_1| + |J_2| \leq |J|}
			 |\Lie_{\mathcal{Z}}^{J_1} \Far|
			\big |\mathcal{Q}_{(i)}\big(\Lie_{\mathcal{Z}}^{J_2} \Far, \nabla(\Lie_{\mathcal{Z}}^{J'} \Far) \big) \big|.
	\end{align}
	Inequality \eqref{E:Htrianglecommutatorsingleterm} then follows from 
	\eqref{E:Htrianglecommutatorsingletermfirstbound} and Lemma \ref{L:NullForms}. 
	
	To obtain \eqref{E:Htrianglecommutatorsingletermfirstbound}, we first split $H_{\triangle}^{\mu \zeta \kappa \lambda}$
	into the following four pieces:
	
	\begin{align}
		(i)^{\mu \zeta \kappa \lambda} & = - \frac{1}{2} (g^{-1})^{\mu \widetilde{\mu}} (g^{-1})^{\zeta \widetilde{\zeta}}
			(g^{-1})^{\kappa \widetilde{\kappa}} (g^{-1})^{\lambda \widetilde{\lambda}} 
			\ell_{(MBI)}^{-2} \Far_{\widetilde{\mu} \widetilde{\zeta}} \Far_{\widetilde{\kappa} \widetilde{\lambda}}, 
			\label{E:Hdecompositioni} \\
		(ii)^{\mu \zeta \kappa \lambda} & = \frac{1}{2} (g^{-1})^{\mu \widetilde{\mu}} (g^{-1})^{\zeta \widetilde{\zeta}}
			(g^{-1})^{\kappa \widetilde{\kappa}} (g^{-1})^{\lambda \widetilde{\lambda}} \Farinvariant_{(2)} \ell_{(MBI)}^{-2} \Far_{\widetilde{\mu} 
			\widetilde{\zeta}} \Fardual_{\widetilde{\kappa} \widetilde{\lambda}}, \\
		(iii)^{\mu \zeta \kappa \lambda} & =\frac{1}{2} (g^{-1})^{\mu \widetilde{\mu}} (g^{-1})^{\zeta \widetilde{\zeta}}
			(g^{-1})^{\kappa \widetilde{\kappa}} (g^{-1})^{\lambda \widetilde{\lambda}}
			\Farinvariant_{(2)} \ell_{(MBI)}^{-2} \Fardual_{\widetilde{\mu} \widetilde{\zeta}} 
			\Far_{\widetilde{\kappa} \widetilde{\lambda}}, \\
		(iv)^{\mu \zeta \kappa \lambda} & = - \frac{1}{2} (g^{-1})^{\mu \widetilde{\mu}} (g^{-1})^{\zeta \widetilde{\zeta}}
			(g^{-1})^{\kappa \widetilde{\kappa}} (g^{-1})^{\lambda \widetilde{\lambda}} \Big(1 + \Farinvariant_{(2)}^2 
			\ell_{(MBI)}^{-2} \Big)\Fardual_{\widetilde{\mu} \widetilde{\zeta}} \Fardual_{\widetilde{\kappa} \widetilde{\lambda}}.
			\label{E:Hdecompositioniv}
	\end{align}
	Since the analysis is the roughly the same for each piece, we will focus only on the term $(iv).$
	
	Differentiating the term $(iv)$ with $\Lie_{\mathcal{Z}}^{J},$ contracting with
	$\nabla_{\mu} \Lie_{\mathcal{Z}}^{J'} \Far_{\kappa \lambda},$ and 
	using \eqref{E:LieZonmupper} plus Corollary \ref{C:ConformalKillingLieXhodgedualcommutation},
	we see that a typical term that arises after expanding via the Leibniz rule is of the form 
	
	\begin{align} \label{E:typicalterm}
		(g^{-1})^{\mu \widetilde{\mu}} (g^{-1})^{\zeta \widetilde{\zeta}}
			(g^{-1})^{\kappa \widetilde{\kappa}} (g^{-1})^{\lambda \widetilde{\lambda}} 
			\Big\lbrace \Lie_{\mathcal{Z}}^{J_1'} \Big[ \big(1 + \Farinvariant_{(2)}^2 \ell_{(MBI)}^{-2} \big) 
			\Fardual_{\widetilde{\mu} \widetilde{\zeta}} \Big] \Big\rbrace
			{^{\star \hspace{-.03in}}(\Lie_{\mathcal{Z}}^{J_2} \Far_{\widetilde{\kappa} \widetilde{\lambda}})}
			\nabla_{\mu} \Lie_{\mathcal{Z}}^{J'} \Far_{\kappa \lambda},
	\end{align}
	where $|J_1'| + |J_2| = |J|.$ Note that the factor
	$(g^{-1})^{\kappa \widetilde{\kappa}} (g^{-1})^{\lambda \widetilde{\lambda}} {^{\star \hspace{-.03in}}(\Lie_{\mathcal{Z}}^{J_2} 
	\Far_{\widetilde{\kappa} \widetilde{\lambda}})}
	\nabla_{\mu} \Lie_{\mathcal{Z}}^{J'} \Far_{\kappa \lambda}$ is equal to 
	$\mathcal{Q}_{(2)}(\Lie_{\mathcal{Z}}^{J_2}\Far,\nabla_{\mu} \Lie_{\mathcal{Z}}^{J'} \Far).$ Now among the remaining 	
	factors $f \eqdef (g^{-1})^{\mu \widetilde{\mu}} (g^{-1})^{\zeta \widetilde{\zeta}} \Big\lbrace \Lie_{\mathcal{Z}}^{J_1'} \Big[ \big(1 + 
	\Farinvariant_{(2)}^2 \ell_{(MBI)}^{-2} \big) (\ell_{(MBI)}^{-2} \Fardual_{\widetilde{\mu} \widetilde{\zeta}}) \Big] 
	\Big\rbrace,$ after fully expanding via the Leibniz and chain rules, there is at most one factor of $\Far$ with more than 
	$N/2$ derivatives falling on it. Thus, by Corollary \ref{C:ConformalKillingLieXhodgedualcommutation} and the assumptions,
	we have that $|f| \lesssim \sum_{|J_1| \leq |J_1'|} |\Lie_{\mathcal{Z}}^{J_1} \Far|.$ 
	Thus, in total, we have shown that such a 
	typical term \eqref{E:typicalterm} is bounded in magnitude from above by the right-hand side of 
	\eqref{E:Htrianglecommutatorsingletermfirstbound}. This completes the proof for the term $(iv).$
	Terms $(i) - (iii)$ can be handled similarly.
	
\end{proof}

The next lemma (more precisely, its corollary) will be used to control the $\overline{K}^{\nu} \dot{\Far}_{\nu \eta} \mathfrak{I}^{\eta}$ term on the right-hand side of \eqref{E:divJdot}, where 
$\dot{\Far} \eqdef \Lie_{\mathcal{Z}} \Far,$ and $\mathfrak{I}^{\nu}$ is the inhomogeneous term from Proposition \ref{P:Inhomogeneousterms}.

\begin{lemma} \label{L:HtriangleCommutator}
	Let $N \geq 0$ be an integer, let $\Far$ be a two-form, and let $H_{\triangle}^{\mu \zeta \kappa \lambda}$ be the
	($\Far-$dependent) tensorfield defined in \eqref{E:Htriangledef}. 
	There exists an $\epsilon > 0$ such that if $|\Far|_{\Lie_{\mathcal{Z}}; \lfloor N/2 \rfloor} \leq \epsilon,$ 
	and $I$ is any $\mathcal{Z}-$multi-index satisfying $|I| \leq N,$ then  
	the following pointwise estimate holds:
	
	\begin{align} \label{E:HtriangleCommutator}
		\Big| \Big\lbrace H_{\triangle}^{\mu \zeta \kappa \lambda} & \nabla_{\mu} \Big(\Lie_{\mathcal{Z}}^I \Far_{\kappa 	
			\lambda}\Big) -	\Liemod_{\mathcal{Z}}^I \Big(H_{\triangle}^{\mu \zeta \kappa \lambda} \nabla_{\mu} \Far_{\kappa 
			\lambda} \Big) \Big\rbrace (\Lie_{\mathcal{Z}}^I \Far_{\nu \zeta}) \overline{K}^{\nu} \Big| \\
		& \lesssim \Big\lbrace (1 + s)^2  |\Lie_{\mathcal{Z}}^I \Far|_{\mathcal{L}\mathcal{U}} + (1 + |q|)^2  
				|\Lie_{\mathcal{Z}}^I \Far| \Big\rbrace \notag \\
		& \ \ \ \times \mathop{\sum_{|J_3| \leq |I| - 1}}_{|J_1| + |J_2| + |J_3| \leq |I|}  
			|\Lie_{\mathcal{Z}}^{J_1} \Far|
			\Big\lbrace |\Lie_{\mathcal{Z}}^{J_2} \Far||\nabla(\Lie_{\mathcal{Z}}^{J_3} \Far)|_{\mathcal{L}\mathcal{U}}  
			+ |\Lie_{\mathcal{Z}}^{J_2} \Far|_{\mathcal{L}\mathcal{U}} |\nabla(\Lie_{\mathcal{Z}}^{J_3} \Far)|
			+ |\Lie_{\mathcal{Z}}^{J_2} \Far|_{\mathcal{T}\mathcal{T}}|\nabla(\Lie_{\mathcal{Z}}^{J_3} 
			\Far)|_{\mathcal{T}\mathcal{T}}  \Big\rbrace. \notag 
	\end{align}

\end{lemma}	

\begin{proof}
	Using \eqref{E:Htrianglerecall}, the definition \eqref{E:Liemoddef} of $\Liemod_{\mathcal{Z}}^I,$
	the null decomposition \eqref{E:MorawetznulldecompL} - \eqref{E:MorawetznulldecompA} for $\overline{K},$ 
	and Lemma \ref{L:Liecommuteswithcoordinatederivatives}, it follows that
	
	\begin{align} \label{E:HtriangleCommutatorMorawetz}
		\Big| \Big\lbrace H_{\triangle}^{\mu \zeta \kappa \lambda} & \nabla_{\mu} \Big(\Lie_{\mathcal{Z}}^I \Far_{\kappa 	
			\lambda}\Big) -	\Liemod_{\mathcal{Z}}^I \Big(H_{\triangle}^{\mu \zeta \kappa \lambda} \nabla_{\mu} \Far_{\kappa 
			\lambda} \Big) \Big\rbrace (\Lie_{\mathcal{Z}}^I \Far_{\zeta \nu}) \overline{K}^{\nu} \Big| \\
		& \lesssim (1 + s)^2 |\Lie_{\mathcal{Z}}^I \Far|_{\mathcal{L} \mathcal{U}}  
				\mathop{\sum_{|J'| \leq |I| - 1}}_{|J| + |J'| \leq |I|} 
			\sum_{\zeta =0}^3 \Big| \Big\lbrace \big( \Lie_{\mathcal{Z}}^{J} (H_{\triangle}^{\mu \zeta \kappa \lambda}) \big) 
			\nabla_{\mu} \Lie_{\mathcal{Z}}^{J'} \Far_{\kappa \lambda} \Big\rbrace \Big| \notag \\
		& \ \ \ + (1 + |q|)^2 |\Lie_{\mathcal{Z}}^I \Far| 
			\mathop{\sum_{|J'| \leq |I| - 1}}_{|J| + |J'| \leq |I|} \sum_{\zeta =0}^3
			\Big|	\Big\lbrace \big( \Lie_{\mathcal{Z}}^{J} (H_{\triangle}^{\mu \zeta \kappa \lambda}) \big) \nabla_{\mu} 
			\Lie_{\mathcal{Z}}^{J'} \Far_{\kappa \lambda} \Big\rbrace \Big|. 
			\notag
	\end{align}
	Inequality \eqref{E:HtriangleCommutator} now follows from applying Lemma \ref{L:Htrianglecommutatorsingleterm} to
	\eqref{E:HtriangleCommutatorMorawetz}. 
	
\end{proof}

\begin{corollary} \label{C:HtriangleCommutator}
	Under the assumptions of Lemma \ref{L:HtriangleCommutator}, we have that
	
	\begin{align} \label{E:algebraicCauchySchwarzHtriangleCommutator}
		\sum_{|I| \leq N}
			\Big| \Big\lbrace H_{\triangle}^{\mu \zeta \kappa \lambda} & \nabla_{\mu} \Big(\Lie_{\mathcal{Z}}^I \Far_{\kappa 	
			\lambda}\Big) -	\Liemod_{\mathcal{Z}}^I \Big(H_{\triangle}^{\mu \zeta \kappa \lambda} \nabla_{\mu} \Far_{\kappa 
			\lambda} \Big) \Big\rbrace (\Lie_{\mathcal{Z}}^I \Far_{\nu \zeta}) \overline{K}^{\nu} \Big| \\
		& \lesssim \bigg\lbrace \sum_{|J| \leq \lfloor N/2 \rfloor} (1 + s)^2|\Lie_{\mathcal{Z}}^J 
			\Far|_{\mathcal{L}\mathcal{U}}^2 + (1 + s)^2 |\Lie_{\mathcal{Z}}^J \Far|_{\mathcal{T}\mathcal{T}}^2 
			+ |\Lie_{\mathcal{Z}}^J \Far|^2 \bigg\rbrace \notag \\
		& \ \ \times \bigg\lbrace \sum_{|I| \leq N} (1 + s)^2|\Lie_{\mathcal{Z}}^I \Far|_{\mathcal{L}\mathcal{U}}^2 
			+ (1 + s)^2 |\Lie_{\mathcal{Z}}^I \Far|_{\mathcal{T}\mathcal{T}}^2 
			+ (1 + |q|)^2|\Lie_{\mathcal{Z}}^I \Far|^2 \bigg\rbrace. \notag 
		\end{align}
\end{corollary}

\begin{proof}
	Inequality \eqref{E:algebraicCauchySchwarzHtriangleCommutator} follows from \eqref{E:HtriangleCommutator},
	from the facts that
	
	\begin{align}
		|\nabla \Far|_{\mathcal{L} \mathcal{U}} + |\nabla \Far|_{\mathcal{T}\mathcal{T}}
		\lesssim \sum_{Z \in \mathcal{Z}} |\Lie_Z \Far|_{\mathcal{L} \mathcal{U}} + |\Lie_Z \Far|_{\mathcal{T}\mathcal{T}}
	\end{align}
	holds for any two-for $\Far,$ and from simple algebraic estimates of the form $ab \lesssim a^2 + b^2.$
\end{proof}

The following lemma will be used to control the terms
$(\nabla_{\mu}H^{\mu \zeta \kappa \lambda}) \dot{\Far}_{\kappa \lambda} \dot{\Far}_{\nu \zeta} \overline{K}^{\nu}
- \frac{1}{4} (\overline{K}^{\nu} \nabla_{\nu}H^{\zeta \eta \kappa \lambda}) \dot{\Far}_{\zeta \eta} 
\dot{\Far}_{\kappa \lambda}$ appearing on the right-hand side of \eqref{E:divJdot}.

\begin{lemma} \label{L:DerivativeofHtriangleerrorterms}
	Let $\Far, \dot{\Far}$ be a pair of two-forms, and let $H_{\triangle}^{\mu \zeta \kappa \lambda}$ be the ($\Far-$dependent)
	tensorfield defined in \eqref{E:Htriangledef}. There exists an $\epsilon > 0$ 
	such that if $|\Far|, |\nabla \Far| \leq \epsilon,$ then the following pointwise estimates hold:
	
	\begin{align} \label{E:DerivativeofHtriangleerrorterm1}
	|(\nabla_{\mu}H_{\triangle}^{\mu \zeta \kappa \lambda}) \dot{\Far}_{\kappa \lambda} \dot{\Far}_{\ \zeta}^{\nu} 
		\overline{K}_{\nu}| & \lesssim (1 + s)^2 \sum_{|I| \leq 1}	
			\Big\lbrace |\Lie_{\mathcal{Z}}^I \Far|_{\mathcal{L}\mathcal{U}}^2 |\dot{\Far}|^2 
			+ |\Lie_{\mathcal{Z}}^I\Far|^2 |\dot{\Far}|_{\mathcal{L}\mathcal{U}}^2 
			+ |\Lie_{\mathcal{Z}}^I \Far|_{\mathcal{T}\mathcal{T}}^2 |\dot{\Far}|_{\mathcal{T}\mathcal{T}}^2 \Big\rbrace \\
	& \ \ \ + (1 + |q|)^2 \sum_{|I| \leq 1} |\Lie_{\mathcal{Z}}^I \Far|^2 |\dot{\Far}|^2, 
		\notag \\
	\Big|(\overline{K}^{\nu}\nabla_{\nu}H_{\triangle}^{\zeta \eta \kappa \lambda}) \dot{\Far}_{\zeta \eta} 
		\dot{\Far}_{\kappa \lambda} \Big| & \lesssim
		(1 + s)^2	\sum_{|I| \leq 1} \Big\lbrace |\Lie_{\mathcal{Z}}^I \Far|_{\mathcal{L}\mathcal{U}}^2 |\dot{\Far}|^2 
			+ |\Lie_{\mathcal{Z}}^I \Far|^2 |\dot{\Far}|_{\mathcal{L}\mathcal{U}}^2 
			+ |\Lie_{\mathcal{Z}}^I \Far|_{\mathcal{T}\mathcal{T}}^2 |\dot{\Far}|_{\mathcal{T}\mathcal{T}}^2 \Big\rbrace.
			\label{E:DerivativeofHtriangleerrorterm2}  
	\end{align}
\end{lemma}

\begin{proof}
	
	Consider the decomposition 
	$\nabla_{\mu}H_{\triangle}^{\mu \zeta \kappa \lambda}$ $= \nabla_{\mu}(i)^{\mu \zeta \kappa \lambda}$
	$+ \nabla_{\mu}(ii)^{\mu \zeta \kappa \lambda}$ $+ \nabla_{\mu}(iii)^{\mu \zeta \kappa \lambda}$ 
	$+ \nabla_{\mu}(iv)^{\mu \zeta \kappa \lambda}$ implied by \eqref{E:Hdecompositioni} - \eqref{E:Hdecompositioniv}. We will 
	focus only on the case of term $(i)$; terms
	$(ii) - (iv)$ can be handled similarly. We now further decompose $\nabla_{\mu}(i)^{\mu \zeta \kappa \lambda}$ into 
	three pieces, which up to constant factors can be expressed as follows:
	
	\begin{align}
		\nabla_{\mu}(i')^{\mu \zeta \kappa \lambda} & \eqdef  \ell_{(MBI)}^{-2} (\nabla_{\mu}\Far^{\mu \zeta}) \Far^{\kappa 
			\lambda}, \\
		\nabla_{\mu}(i'')^{\mu \zeta \kappa \lambda} & \eqdef  \ell_{(MBI)}^{-2} \Far^{\mu \zeta} 
			\nabla_{\mu}\Far^{\kappa \lambda}, \\
		\nabla_{\mu}(i''')^{\mu \zeta \kappa \lambda} & \eqdef (\nabla_{\mu}\ell_{(MBI)}^{-2}) \Far^{\mu \zeta} 
			\Far^{\kappa \lambda}.
	\end{align}
	
	Using \eqref{E:XYnullframeinnerproduct} and	the null decomposition \eqref{E:MorawetznulldecompL} - 
	\eqref{E:MorawetznulldecompA} for $\overline{K},$ it follows that if $\epsilon$ is sufficiently small, then
	
	\begin{align}
		\Big|[\nabla_{\mu}(i')^{\mu \zeta \kappa \lambda}] \dot{\Far}_{\kappa \lambda} \dot{\Far}_{\ \zeta}^{\nu} 
			\overline{K}_{\nu}\Big| & \lesssim \Big\lbrace (1 + s)^2|\dot{\Far}|_{\mathcal{L}\mathcal{U}}|\nabla \Far|
				+ (1 + |q|)^2 |\dot{\Far}| |\nabla \Far|
			\Big\rbrace \big|\mathcal{Q}_{(1)}(\Far,\dot{\Far})\big|, \label{E:nablaiprimenulldecomposition}\\
		\Big|[\nabla_{\mu}(ii')^{\mu \zeta \kappa \lambda}] \dot{\Far}_{\kappa \lambda} \dot{\Far}_{\ \zeta}^{\nu} 
			\overline{K}_{\nu}\Big| & \lesssim \Big\lbrace (1 + s)^2|\dot{\Far}|_{\mathcal{L}\mathcal{U}}|\Far|
				+ (1 + |q|)^2 |\dot{\Far}| |\Far|
			\Big\rbrace \sum_{\mu = 0}^3 \big|\mathcal{Q}_{(1)}(\nabla_{\mu}\Far,\dot{\Far})\big|, 
			\label{E:nablaitwoprimenulldecomposition}\\
		\Big|[\nabla_{\mu}(i''')^{\mu \zeta \kappa \lambda}] \dot{\Far}_{\kappa \lambda} \dot{\Far}_{\ \zeta}^{\nu} 
			\overline{K}_{\nu}\Big| & \lesssim \Big\lbrace (1 + s)^2|\dot{\Far}|_{\mathcal{L}\mathcal{U}}|\Far|
			+ (1 + |q|)^2 |\dot{\Far}| |\Far|
			\Big\rbrace \big|\mathcal{Q}_{(1)}(\Far,\dot{\Far})\big|, \label{E:nablaithreeprimenulldecomposition}
	\end{align}
	where the $\mathcal{Q}_{(1)}(\cdot,\cdot)$ terms arise from the $\kappa, \lambda$ indices.
	Also applying the null decomposition of Lemma \ref{L:Liederivativeofnullformsaresumsofnullforms} to 
	the $\mathcal{Q}_{(1)}(\cdot,\cdot)$ terms, the fact that $|\nabla \Far|_{\mathcal{V}\mathcal{W}} 
	\lesssim \sum_{|I| = 1} |\Lie_{\mathcal{Z}}^I \Far|_{\mathcal{V}\mathcal{W}},$
	and using simple algebraic estimates of the form $|ab| \lesssim a^2 + b^2,$
	we conclude that each of the right-hand sides of 
	\eqref{E:nablaiprimenulldecomposition}
	- \eqref{E:nablaithreeprimenulldecomposition} are $\lesssim$ the right-hand side of 
	\eqref{E:DerivativeofHtriangleerrorterm1}. Consequently,
	the same is true of $\Big|[\nabla_{\mu}(i)^{\mu \zeta \kappa \lambda}]\dot{\Far}_{\kappa \lambda} \dot{\Far}_{\ 
	\zeta}^{\nu} \overline{K}_{\nu}\Big|.$ Note in particular that our estimates for the terms in 
	\eqref{E:nablaiprimenulldecomposition} - \eqref{E:nablaithreeprimenulldecomposition} involving the $(1 + |q|)^2$ 
	factor are not optimal; we have simply bounded them by 
	$(1 + |q|)^2 \sum_{|I| \leq 1} |\Lie_{\mathcal{Z}}^I \Far|^2|\dot{\Far}|^2.$ The cases $(ii) - (iv)$ can be handled 
	similarly; this completes our proof of \eqref{E:DerivativeofHtriangleerrorterm1}.
	
To prove \eqref{E:DerivativeofHtriangleerrorterm2}, we first consider the decomposition
$(\overline{K}^{\nu} \nabla_{\nu}H_{\triangle}^{\zeta \eta \kappa \lambda}) \dot{\Far}_{\zeta \eta} \dot{\Far}_{\kappa \lambda}$
$= [\overline{K}^{\nu}\nabla_{\nu} (i)^{\zeta \eta \kappa \lambda}]\dot{\Far}_{\zeta \eta} \dot{\Far}_{\kappa \lambda}$ \\
$+ [\overline{K}^{\nu} \nabla_{\nu} (ii)^{\zeta \eta \kappa \lambda}]\dot{\Far}_{\zeta \eta} \dot{\Far}_{\kappa \lambda}$
$+ [\overline{K}^{\nu} \nabla_{\nu} (iii)^{\zeta \eta \kappa \lambda}]\dot{\Far}_{\zeta \eta} \dot{\Far}_{\kappa \lambda}$
$+ [\overline{K}^{\nu} \nabla_{\nu} (iv)^{\zeta \eta \kappa \lambda}]\dot{\Far}_{\zeta \eta} \dot{\Far}_{\kappa \lambda}$
implied by \eqref{E:Hdecompositioni} - \eqref{E:Hdecompositioniv}. We will 
focus only on the case of term $(i)$; terms $(ii) - (iv)$ can be handled similarly. We now further decompose
$[\overline{K}^{\nu} \nabla_{\nu} (i)^{\zeta \eta \kappa \lambda}]$ 
$= (I')^{\zeta \eta \kappa \lambda}$
$+ (I'')^{\zeta \eta \kappa \lambda}$
$+ (I''')^{\zeta \eta \kappa \lambda},$ where
	 	
\begin{align}
		(I')^{\zeta \eta \kappa \lambda} & \eqdef  
			\ell_{(MBI)}^{-2} (\overline{K}^{\nu}\nabla_{\nu}\Far^{\zeta \eta}) \Far^{\kappa \lambda}, \\
		(I'')^{\zeta \eta \kappa \lambda} & \eqdef  \ell_{(MBI)}^{-2} \Far^{\zeta \eta} 
			\overline{K}^{\nu}\nabla_{\nu} \Far^{\kappa \lambda}, \\
		(I''')^{\zeta \eta \kappa \lambda} & \eqdef (\overline{K}^{\nu}\nabla_{\nu} \ell_{(MBI)}^{-2}) 
			\Far^{\zeta \eta} \Far^{\kappa \lambda}.
	\end{align}	 	
	 	Using the null decomposition $K^{\nu} \nabla_{\nu} = \frac{1}{2}\big\lbrace(1+s^2) \nabla_{L} + (1 + q^2) \nabla_{\uL} 
	 	\big\rbrace,$ it follows that if $\epsilon$ is sufficiently small, then
	 	
	 \begin{align}
		\big|(I')^{\zeta \eta \kappa \lambda} \dot{\Far}_{\zeta \eta} \dot{\Far}_{\kappa \lambda} \big| 
			& \lesssim \Big\lbrace (1 + s)^2 \big| \mathcal{Q}_{(1)}(\nabla_L \Far, \dot{\Far})\big| 
				+ (1 + |q|)^2 \big|\mathcal{Q}_{(1)}(\nabla_{\uL} \Far, \dot{\Far})\big| \Big\rbrace 
				\big|\mathcal{Q}_{(1)}(\Far,\dot{\Far})\big|, \label{E:Kderivativeiprimenulldecomposition}\\
		\big|(I'')^{\zeta \eta \kappa \lambda} \dot{\Far}_{\zeta \eta} 
			\dot{\Far}_{\kappa \lambda} \big| 
			& \lesssim \Big\lbrace (1 + s)^2 \big|\mathcal{Q}_{(1)}(\nabla_L \Far, \dot{\Far})\big| 
				+ (1 + |q|)^2 \big|\mathcal{Q}_{(1)}(\nabla_{\uL} \Far, \dot{\Far})\big| \Big\rbrace 
				\big|\mathcal{Q}_{(1)}(\Far,\dot{\Far})\big|, \label{E:Kderivativeitwoprimenulldecomposition} \\
		\big|(I''')^{\zeta \eta \kappa \lambda} \dot{\Far}_{\zeta \eta} \dot{\Far}_{\kappa \lambda} \big| 
			& \lesssim \sum_{i=1}^2 \Big\lbrace (1 + s)^2 \big|\mathcal{Q}_{(i)}(\nabla_L \Far, \Far)\big| 
				+ (1 + |q|)^2 \big|\mathcal{Q}_{(i)}(\nabla_{\uL} \Far, \Far)\big| \Big\rbrace 
				\big|\mathcal{Q}_{(1)}(\Far,\dot{\Far})\big|^2, 
				\label{E:Kderivativeithreeprimenulldecomposition}
	\end{align}	 
	where the $\mathcal{Q}_{(1)}(\nabla_L \Far, \dot{\Far}), \mathcal{Q}_{(1)}(\nabla_{\uL} \Far, \dot{\Far})$
	terms in \eqref{E:Kderivativeiprimenulldecomposition} arise from the $\eta, \zeta$ indices, the 
	$\mathcal{Q}_{(1)}(\Far,\dot{\Far})$ terms in \eqref{E:Kderivativeiprimenulldecomposition} 
	arise from the $\kappa, \lambda$  indices, the estimate \eqref{E:Kderivativeitwoprimenulldecomposition}
	follows from \eqref{E:Kderivativeiprimenulldecomposition} by interchanging the roles of $\eta, \zeta$ and $\kappa, \lambda,$
	the $\big|\mathcal{Q}_{(1)}(\Far,\dot{\Far})\big|^2$ 
	term in \eqref{E:Kderivativeithreeprimenulldecomposition} arises from both the $\zeta, \eta$ and the $\kappa, \lambda$ 
	indices, and the $\mathcal{Q}_{(i)}(\nabla_L \Far, \dot{\Far}), \mathcal{Q}_{(i)}(\nabla_{\uL} \Far, \dot{\Far})$ terms in
	\eqref{E:Kderivativeithreeprimenulldecomposition} arise from the fact that 
	$\ell_{(MBI)}$ can be expressed as a function of null forms; see \eqref{E:ldef}. Also applying the null decomposition of 
	Lemma \ref{L:Liederivativeofnullformsaresumsofnullforms} to the $\mathcal{Q}_{(i)}(\cdot,\cdot)$ terms, 
	using the fact that $|\nabla \Far|_{\mathcal{V}\mathcal{W}} \lesssim \sum_{|I| = 1} |\Lie_{\mathcal{Z}}^I 
	\Far|_{\mathcal{V}\mathcal{W}},$ and using simple 
	algebraic estimates of the form $|ab| \lesssim a^2 + b^2,$ we conclude that each of the right-hand sides of 
	\eqref{E:Kderivativeiprimenulldecomposition} - \eqref{E:Kderivativeithreeprimenulldecomposition} are $\lesssim$ the 
	right-hand side of \eqref{E:DerivativeofHtriangleerrorterm2}. Therefore, the same is true of
	$|\overline{K}^{\nu} \nabla_{\nu} (i)^{\zeta \eta \kappa \lambda}|.$ We remark that we do not need the full
	structure of the right-hand sides of \eqref{E:Kderivativeiprimenulldecomposition} - 
	\eqref{E:Kderivativeithreeprimenulldecomposition} to conclude the desired estimates; rather, we only need the fact that the 
	right-hand sides of \eqref{E:Kderivativeiprimenulldecomposition} - \eqref{E:Kderivativeithreeprimenulldecomposition} are  
	at least quadratic in the $\mathcal{Q}_{(i)}(\cdot,\cdot).$ The cases $(ii) - (iv)$ can be handled similarly.
	 
\end{proof}

Finally, the last lemma in this section will be used to control the terms inside braces on the right-hand side of \eqref{E:divJdot}.

\begin{lemma} \label{L:CanonicalStressMorawetzDerivativeTerm}
	There exists an $\epsilon > 0$ that if $|\Far|, |\nabla \Far| \leq \epsilon,$ then the following estimates hold:
	\begin{align} \label{E:CanonicalStressMorawetzDerivativeTerm}
		& \bigg| \bigg\lbrace
		\ell_{(MBI)}^{-2} \Far^{\kappa \lambda}\dot{\Far}_{\kappa \lambda} 
		\big( \Far^{\nu \zeta}\dot{\Far}_{\ \zeta}^{\mu} 
			- \Far^{\mu \zeta}\dot{\Far}_{\ \zeta}^{\nu} \big) + \big(1 + \Farinvariant_{(2)}^2 \ell_{(MBI)}^{-2} \big)
			\Fardual^{\kappa \lambda}\dot{\Far}_{\kappa \lambda} \big( 
			\Fardual^{\nu \zeta}\dot{\Far}_{\ \zeta}^{\mu} - \Fardual^{\mu \zeta}\dot{\Far}_{\ \zeta}^{\nu} \big) 
			\\
		& \ \ \ \ \ \ + \Farinvariant_{(2)} \ell_{(MBI)}^{-2} \Fardual^{\kappa \lambda}\dot{\Far}_{\kappa \lambda}
			\big( \Far^{\mu \zeta}\dot{\Far}_{\ \zeta}^{\nu} 
			- \Far^{\nu \zeta}\dot{\Far}_{\ \zeta}^{\mu} \big) + 
			\Farinvariant_{(2)} \ell_{(MBI)}^{-2} \Far^{\kappa \lambda}\dot{\Far}_{\kappa \lambda}
			\big( \Fardual^{\mu \zeta}\dot{\Far}_{\ \zeta}^{\nu} 
			- \Fardual^{\nu \zeta} \dot{\Far}_{\ \zeta}^{\mu} \big) \bigg\rbrace 
			\nabla_{\mu} \overline{K}_{\nu} \bigg| \notag \\
		& \lesssim s \big(|\Far|_{\mathcal{L}\mathcal{U}}^2 |\dot{\Far}|^2 + |\Far|^2 |\dot{\Far}|_{\mathcal{L}\mathcal{U}}^2
			+ |\Far|_{\mathcal{T} \mathcal{T}}^2 |\dot{\Far}|_{\mathcal{T} \mathcal{T}}^2 \big). \notag
	\end{align}
\end{lemma}

\begin{proof}
	Inequality \eqref{E:CanonicalStressMorawetzDerivativeTerm} follows from \eqref{E:Fardualalpha} - \eqref{E:Fardualsigma},
	Lemma \ref{L:NullForms}, and Lemma \ref{L:twoformsMorawetzderivativecontraction}.
\end{proof}

\section{Global Sobolev Inequality} \label{S:GlobalSobolev}
\setcounter{equation}{0}

In this section, we recall a version of the global Sobolev inequality that was proved in \cite{dCsK1990}. This fundamental inequality allows us to deduce weighted $L^{\infty}$ bounds for a two-form $\Far$ from weighted $L^2$ bounds of the quantities
$\Lie_{\mathcal{Z}}^I \Far.$ It provides the mechanism for deducing the $\frac{1}{1 + \tau^2}$ factor in our estimate
\eqref{E:LieZintegralnormintegralinequality}. Although many of the estimates in this section were proved in \cite{dCsK1990}, we reprove some of them for convenience. However, in order to derive the improved decay estimates \eqref{E:GlobalSobolevalpha} - \eqref{E:GlobalSobolevalphaupgraded} for the $\alpha$ component of a solution $\Far$ to the MBI system, we will have to make use of the null decomposition equation \eqref{E:MBIualphanulldecomp}. The fact that inequalities \eqref{E:GlobalSobolevalpha} - \eqref{E:GlobalSobolevalphaupgraded} hold is another manifestation that the nonlinearities 
in the MBI system satisfy a version of the null condition.
\\

\noindent \hrulefill
\ \\

\begin{proposition} \label{P:GlobalSobolev}
	Let $\Far$ be a two-form, and let $\ualpha,$ $\alpha,$ $\rho,$ $\sigma$ be its null decomposition
	as defined in \eqref{E:ualphadef} - \eqref{E:sigmadef}. Let $M \geq 2$ be an integer.
	Then in the interior region $\lbrace (t,\ux) \mid |\ux| \leq 1 + t/2 \rbrace,$ we have that
	
	\begin{align} \label{E:Interiorestimates}
		|\nabla_{(m)} \Far(t,\ux)| & \lesssim (1 + s)^{-5/2 - m} \Kintnorm \Far \Kintnorm_{\Lie_{\mathcal{Z}};M},
			&& m = 0,\cdots, M-2.
	\end{align}
	
	In the exterior region $\lbrace (t,\ux) \mid |\ux| \geq 1 + t/2 \rbrace,$ we have the following estimates:
	
	\begin{subequations}
	\begin{align}
		|\nabla_{\uL}^k \nabla_{L}^l \angn_{(m)} \ualpha(t,\ux)| & \lesssim (1 + s)^{-1 - l - m}(1 + |q|)^{-3/2-k} 
			\Kintnorm \Far \Kintnorm_{\Lie_{\mathcal{Z}};M},
			&& 0 \leq k+ l + m  \leq M-2, \label{E:GlobalSobolevualpha} \\
		|\nabla_{\uL}^k \nabla_{L}^l \angn_{(m)} \big(\alpha(t,\ux),\rho(t,\ux),\sigma(t,\ux)\big)| 
		& \lesssim (1 + s)^{-2-l-m}(1 + |q|)^{-1/2 - k} \Kintnorm \Far \Kintnorm_{\Lie_{\mathcal{Z}};M}, 
			&& 0 \leq k + l + m  \leq M-2. \label{E:GlobalSobolevalpharhosigma} 
	\end{align}
	\end{subequations}
	
	Furthermore, if $\Far$ is a solution to the MBI system \eqref{E:modifieddFis0summary} - \eqref{E:HmodifieddMis0summary}, then
	we have the following improved estimates for $\alpha:$
	
	\begin{subequations}
	\begin{align}
		|\nabla_{L}^l \angn_{(m)} \alpha(t,\ux)| & \lesssim (1 + s)^{-5/2-l-m} \Kintnorm \Far \Kintnorm_{\Lie_{\mathcal{Z}};M},
			&& 0 \leq  l + m \leq M-2, \label{E:GlobalSobolevalpha} \\
		|\nabla_{\uL}^{k+1} \nabla_{L}^l \angn_{(m)} \alpha(t,\ux)| & \lesssim (1 + s)^{-3 - l - m}(1 + |q|)^{-1/2 - k} 
			\Kintnorm \Far \Kintnorm_{\Lie_{\mathcal{Z}};M},&& 0 \leq k + l + m  \leq M-3 && (\mbox{if} \ M \geq 3). 
			\label{E:GlobalSobolevalphaupgraded}
	\end{align}
	\end{subequations}
\end{proposition}

The following corollary follows easily by using $\Lie_{\mathcal{Z}}^I \Far$ in place of $\Far$ in
Proposition \ref{P:GlobalSobolev}. It plays a fundamental role in our 
proof of Proposition \ref{P:Energydifferentialinequality}.

\begin{corollary} \label{C:GlobalSobolev}
	Let $\Far$ be any two-form, and let $I$ be any $\mathcal{Z}-$multi-index
	such that $|I| \leq M-2$. Let the pointwise seminorms $|\cdot|_{\mathcal{V}\mathcal{W}}$ be as in Definition 
	\ref{D:contractionnomrs}. Then with $r \eqdef |\ux|, \ q \eqdef r - 
	t, \ s \eqdef r + t,$ we have that
	
	\begin{subequations}
	\begin{align}
		|\Lie_{\mathcal{Z}}^I \Far| & \lesssim (1 + s)^{-1} (1 + |q|)^{-3/2} \Kintnorm \Far \Kintnorm_{\Lie_{\mathcal{Z}};M}, \\
		|\Lie_{\mathcal{Z}}^I \Far|_{\mathcal{L}\mathcal{U}} & \lesssim (1 + s)^{-2}(1 + |q|)^{-1/2} 
			\Kintnorm \Far \Kintnorm_{\Lie_{\mathcal{Z}};M}, \\
		|\Lie_{\mathcal{Z}}^I \Far|_{\mathcal{T}\mathcal{T}} & \lesssim (1 + s)^{-2}(1 + |q|)^{-1/2} 
			\Kintnorm \Far \Kintnorm_{\Lie_{\mathcal{Z}};M}.
	\end{align}
	\end{subequations}
\end{corollary}

\hfill $\qed$

The proof of Proposition \ref{P:GlobalSobolev} is heavily based on the next lemma, which was proved in \cite{dCsK1990}.

\begin{lemma} \label{L:Sobolev} \cite[Lemmas 2.2 and 2.3]{dCsK1990}
 	Let $U(\ux)$ be a tensorfield defined on Euclidean space $\mathbb{R}^3.$
 	Then for any real number $t \geq 1,$ we have that
 	
	\begin{align}
		\sup_{|\ux| \leq 1 + t/2} |U(\ux)| & \lesssim (1 + t)^{-3/2} \bigg(\sum_{m=0}^2 t^{2m} 
			\int_{|\uy| \leq 1 + 3t/4} |\nabla_{(m)}U(\uy)|^2 \, d^3\uy \bigg)^{1/2}.
	\end{align}
	
	For all $\ux \in \mathbb{R}^3$ with $|\ux| \eqdef |r| \geq 1,$ we have that
	\begin{align} \label{E:alphaimprovementExteriorRegionSobolev}
		|U(\ux)| & \lesssim r^{-3/2} \bigg( \int_{|\uy| \geq r} |U(\uy)|_{\Lie_{\mathcal{O}};2}^2 
			+ |\uy|^2 |\nabla_{\hat{N}} U(\uy)|_{\Lie_{\mathcal{O}};1}^2  \, d^3 \uy \bigg)^{1/2},
	\end{align}
	where $\hat{N} \eqdef \partial_r$ is the radial derivative.
	
	For all real numbers $t \geq 0$ and all $\ux \in \mathbb{R}^3$
	with $|\ux| \eqdef r \geq 1 + t/2,$ we have that
	\begin{align} \label{E:ExteriorRegionSobolev}
		|U(\ux)| & \lesssim (1 + s)^{-1}(1 + |q|)^{-1/2} \bigg( \int_{|\uy| \geq t/2 + 1}  |U(\uy)|_{\Lie_{\mathcal{O}};2}^2 
			+ \big(1 + \big||\bar{y}| - t\big|^2 \big) |\nabla_{\hat{N}} U(\uy)|_{\Lie_{\mathcal{O}};1}^2  \, d^3 \uy\bigg)^{1/2},
	\end{align}
	where $\hat{N} \eqdef \partial_r$ is the radial derivative.
\end{lemma}

\hfill $\qed$

Before proving Proposition \ref{P:GlobalSobolev}, we prove two additional technical lemmas.

\begin{lemma} \label{L:VectorfieldAlgebra} \cite[Corollary of Lemma 3.3]{dCsK1990}
	Let $\Far$ be a two-form, and let $\ualpha,$ $\alpha,$ $\rho,$ $\sigma$ be its null decomposition
	as defined in \eqref{E:ualphadef} - \eqref{E:sigmadef}. Then with $r \eqdef |\ux|, \ s \eqdef r + t, 
	\ q \eqdef r - t,$ the following pointwise inequality holds:
	
	\begin{align} \label{E:VectorfieldAlgebra}
		\sum_{k + l = 0}^M (1 + |q|)^{k} s^{l} \bigg\lbrace (1 + |q|) |\nabla_{\uL}^k \nabla_{L}^l 
			\ualpha|_{\Lie_{\mathcal{O}};M-k-l}
		+ s |\nabla_{\uL}^k \nabla_{L}^l \alpha|_{\Lie_{\mathcal{O}};M-k-l}
		+ s |\nabla_{\uL}^k \nabla_{L}^l \rho|_{\Lie_{\mathcal{O}};M-k-l}
		+ s |\nabla_{\uL}^k \nabla_{L}^l \sigma|_{\Lie_{\mathcal{O}};M-k-l}
			\bigg\rbrace \lesssim \Knorm \Far \Knorm_{\Lie_{\mathcal{Z}};M}.
	\end{align}
\end{lemma}

\begin{proof}
	Inequality \eqref{E:VectorfieldAlgebra} follows from Corollary \ref{C:LieRotationcommuteswithnulldecomp},
	\eqref{E:FaradayKNormsquaredintermsofEBPQ}, and \eqref{E:nablauLnablaLFarLierotaionsnormLieZnormcomparison}.
\end{proof}

\begin{lemma} \label{L:partialsnpartiaqpartialsnalpha}
	Let $\Far$ be a solution to the MBI system, and let $\ualpha,$ $\alpha,$ $\rho,$ $\sigma$ be its null decomposition. Then there 
	exists an $\epsilon > 0$ such that if $| \Far |_{\Lie_{\mathcal{Z}};\lfloor M/2 \rfloor} \leq \epsilon,$ then the 
	following pointwise inequality holds:

	\begin{align} \label{E:partialsnpartiaqpartialsnalpha}
		r \sum_{l=0}^M r^{l} |\nabla_{L}^l \alpha|_{\Lie_{\mathcal{O}};M-l}
			+ r^2 \sum_{l=0}^{M-1} r^{l} |\nabla_{\uL} \nabla_{L}^l \alpha|_{\Lie_{\mathcal{O}};M-l-1} 
			& \lesssim \Knorm \Far \Knorm_{\Lie_{\mathcal{Z}};M}.
	\end{align}
\end{lemma}

\begin{proof}
	To deduce that the first sum on the left-hand side of \eqref{E:partialsnpartiaqpartialsnalpha} is  $\lesssim \Knorm \Far 
	\Knorm_{\Lie_{\mathcal{Z}};M},$ we simply apply Lemma \ref{L:VectorfieldAlgebra}. 
	To estimate the second sum, we begin by recalling the equation \eqref{E:MBIalphanulldecomp} 
	satisfied by the null components of $\Far.$ We write the equation relative to an arbitrary coordinate system, instead of a 
	null frame:
	
	\begin{align} \label{E:MBIpartialqalphaagain}
	& \nabla_{\uL} \alpha_{\mu} - r^{-1} \alpha_{\mu} - \angn_{\mu} \rho 
		- \angepsilon_{\mu}^{\ \nu} \angn_{\nu} \sigma \\
	& \ \ - \frac{1}{4}\ell_{(MBI)}^{-2} \Big\lbrace (\nabla_{\uL}\Farinvariant_{(1)} - 2 \Farinvariant_{(2)} \nabla_{\uL} 
			\Farinvariant_{(2)})\alpha_{\mu}
			- 2(\angepsilon_{\mu}^{\ \nu} \angn_{\nu}\Farinvariant_{(1)} 
			- 2 \Farinvariant_{(2)} \angepsilon_{\mu}^{\ \nu}\angn_{\nu} \Farinvariant_{(2)})\sigma
			+ (\nabla_{L}\Farinvariant_{(1)} - 2 \Farinvariant_{(2)} \nabla_{L} \Farinvariant_{(2)}) \ualpha_{\mu} \Big\rbrace 
			\notag \\
	& \ \ - \frac{1}{4}\ell_{(MBI)}^{-2} \Farinvariant_{(2)} \Big\lbrace (\nabla_{\uL}\Farinvariant_{(1)} - 2 \Farinvariant_{(2)} 
			\nabla_{\uL} \Farinvariant_{(2)}) \angepsilon_{\mu}^{\ \nu} \alpha_{\nu}
			- 2 (\angepsilon_{\mu}^{\ \nu}\angn_{\nu} \Farinvariant_{(1)} 
			- 2 \Farinvariant_{(2)} \angepsilon_{\mu}^{\ \nu} \angn_{\nu} \Farinvariant_{(2)}) \rho
			- (\nabla_{L}\Farinvariant_{(1)} - 2 \Farinvariant_{(2)} \nabla_{L} \Farinvariant_{(2)})
			\angepsilon_{\mu}^{\ \nu}\ualpha_{\nu} \Big\rbrace  \notag \\
	& \ \ + \frac{1}{2}\Big\lbrace (\nabla_{\uL} \Farinvariant_{(2)}) \angepsilon_{\mu}^{\ \nu}\alpha_{\nu} 
			- 2(\angepsilon_{\mu}^{\ \nu} \angn_{\nu} \Farinvariant_{(2)}) \rho  
			- (\nabla_{L}\Farinvariant_{(2)}) \angepsilon_{\mu}^{\ \nu} \ualpha_{\nu} \Big\rbrace \notag \\
	& \ \ = 0. \notag
\end{align}	
We then differentiate each side of \eqref{E:MBIpartialqalphaagain} with $\nabla_{\uL}^k, \nabla_{L}^l,$ and 
$\Lie_{\mathcal{O}}^I,$ where $I$ is a rotational multi-index satisfying ${|I| \leq M - k - l - 1},$ multiply by $r^2 r^l (1 + |q|)^k,$ and use Lemma \ref{L:LLunderlinecommutewithLieO}, Lemma \ref{L:derivativesofanggandepsilon}, Corollary \ref{C:nablaLnablaunderlineLLieDerivativeIntrinsicCovariantDerivativeComparison} (to exchange the $\angn$ derivatives for
$r^{-1}-$ weighted rotational Lie derivatives), Lemma \ref{L:Liederivativeofnullformsaresumsofnullforms}, the fact that $\nabla_L r = -\nabla_{\uL} r = 1,$ the fact that $(1 + |q|) \lesssim r$ in the exterior region, and the assumption that $\Knorm \Far \Knorm_{\Lie_{\mathcal{Z}};\lfloor M/2 \rfloor}$ is sufficiently small to derive the
following inequality:

\begin{align} \label{E:alphaupgradedfirstinequality}
		r^2 r^l (1 + |q|)^k |\nabla_{\uL}^{k+1} \nabla_{L}^l \alpha |_{\Lie_{\mathcal{O}};M-k-l-1}  
		& \lesssim r \mathop{\sum_{0 \leq k' \leq k,}}_{0 \leq l' \leq l} 
			r^{l'} (1 + |q|)^{k'}  |\nabla_{\uL}^{k'} \nabla_{L}^{l'} \alpha|_{\Lie_{\mathcal{O}};M-k'-l'-1} \\
		& \ \ \ + r r^l (1 + |q|)^k |\nabla_{\uL}^k \nabla_{L}^l \rho|_{\Lie_{\mathcal{O}};M-k-l} + r r^l (1 + |q|)^k 
			|\nabla_{L}^l \sigma|_{\Lie_{\mathcal{O}};M-k-l} \notag \\ 
		& \ \ \ + r^2 r^l (1 + |q|)^k \sum nonlinear \ terms. \notag
\end{align}
After fully expanding via the Leibniz rule and using the smallness assumption on $| \Far |_{\Lie_{\mathcal{Z}};\lfloor M/2 \rfloor},$ it follows that up to order $1$ factors, each $nonlinear \ term$ on the right-hand side of \eqref{E:alphaupgradedfirstinequality} is of one of the following three types:

\begin{align}
	(i) & = \Big\lbrace \big| \mathcal{Q}_{(i)}\big(\Lie_{\mathcal{O}}^{I_1} \nabla_{\uL}^{k_1 + 1} \nabla_{L}^{l_1} \Far, 
		\Lie_{\mathcal{O}}^{I_2} \nabla_{\uL}^{k_2} \nabla_{L}^{l_2} \Far \big) \big| \Big\rbrace
		\big| \Lie_{\mathcal{O}}^{I_3} \nabla_{\uL}^{k_3} \nabla_{L}^{l_3} \big(\alpha, \sigma, \rho \big) \big|
		\mathop{\prod_{4 \leq a,b,c,}}_{0 < k_a + l_b + |I_c|} 
		\big| \Lie_{\mathcal{O}}^{I_c} \nabla_{\uL}^{k_a} \nabla_{L}^{l_b} \Far \big|, \label{E:i} \\
	(ii) & = \Big\lbrace \big| \mathcal{Q}_{(i)}\big(\Lie_{\mathcal{O}}^{I_1} \nabla_{\uL}^{k_1} \nabla_{L}^{l_1} \angn \Far, 
		\Lie_{\mathcal{O}}^{I_2} \nabla_{\uL}^{k_2} \nabla_{L}^{l_2} \Far \big) \big| \Big\rbrace
		\big| \Lie_{\mathcal{O}}^{I_3} \nabla_{\uL}^{k_3} \nabla_{L}^{l_3} \big(\alpha, \sigma, \rho \big) \big|
		\mathop{\prod_{4 \leq a,b,c,}}_{0 < k_a + l_b + |I_c|} 
		\big| \Lie_{\mathcal{O}}^{I_c} \nabla_{\uL}^{k_a} \nabla_{L}^{l_b} \Far \big|, \label{E:ii} \\
	(iii) & = \Big\lbrace \big| \mathcal{Q}_{(i)}\big(\Lie_{\mathcal{O}}^{I_1} \nabla_{\uL}^{k_1} \nabla_{L + 1}^{l_1} \Far, 
		\Lie_{\mathcal{O}}^{I_2} \nabla_{\uL}^{k_2} \nabla_{L}^{l_2} \Far \big) \big| \Big\rbrace
		\big| \Lie_{\mathcal{O}}^{I_3} \nabla_{\uL}^{k_3} \nabla_{L}^{l_3} \ualpha \big|
		\mathop{\prod_{4 \leq a,b,c,}}_{0 < k_a + l_b + |I_c|} 
		\big| \Lie_{\mathcal{O}}^{I_c} \nabla_{\uL}^{k_a} \nabla_{L}^{l_b} \Far \big|, \label{E:iii} 
\end{align}
where the $k_a$ are non-negative integers such that $k_1 + \cdots + k_k = k,$ the $l_b$ are non-negative integers such that
$l_1 + \cdots + l_l = l,$ the $I_c$ are rotational multi-indices such that $|I_1| + \cdots + |I_{M-k-l}| \leq M-k-l,$
and the $\mathcal{Q}_{(i)}(\cdot,\cdot)$ are null forms arising from the $\Farinvariant_{(i)}.$ We remark that the type $(i)$ terms arise from e.g. the derivatives of $\ell_{(MBI)}^{-2} (\nabla_{\uL}\Farinvariant_{(1)}) \alpha_{\mu}$ on the right-hand side of \eqref{E:MBIpartialqalphaagain}, the type $(ii)$ terms arise from e.g. the derivatives of $\ell_{(MBI)}^{-2} (\angepsilon_{\mu}^{\ \nu} \angn_{\nu}\Farinvariant_{(1)}) \sigma,$ and the type $(iii)$ terms arise 
from e.g. the derivatives of $\ell_{(MBI)}^{-2} (\nabla_{L}\Farinvariant_{(1)}) \ualpha_{\mu}.$ 

Each linear (in $\Far$) term on the right-hand side of \eqref{E:alphaupgradedfirstinequality} is manifestly bounded by the left-hand side of \eqref{E:VectorfieldAlgebra}. Therefore, in order to prove \eqref{E:partialsnpartiaqpartialsnalpha}, what remains to be shown is that the nonlinear terms of type $(i)-(iii)$ are each 
$\lesssim r^{-2} r^{-l} (1 + |q|)^{-k} \Knorm \Far \Knorm_{\Lie_{\mathcal{Z}};M}.$ 

For the type $(i)$ terms, we use \eqref{E:VectorfieldAlgebra} and the smallness assumption on $| \Far |_{\Lie_{\mathcal{Z}};\lfloor M/2 \rfloor}$ to deduce that 

\begin{align} \label{E:first3rdorderterm}
	|(i)| & \lesssim r^{-2} r^{-(l_1 + \cdots + l_l)} (1 + |q|)^{-(k_1 + \cdots + k_k)} 
		\Knorm \Far \Knorm_{\Lie_{\mathcal{Z}};M}
		= r^{-2} r^{-l} (1 + |q|)^{-k} \Knorm \Far \Knorm_{\Lie_{\mathcal{Z}};M}, 
\end{align}
where the $r^{-2}$ arises from the fact that at least $2$ of the factors on the right-hand side of \eqref{E:i}
involve derivatives of the more rapidly decaying terms $\alpha, \rho, \sigma;$
one of the factors is explicitly written, while the second arises from the fact that each $\mathcal{Q}_{(i)}$ is a null form.
The type $(ii)$ terms can be handled similarly to the type $(i)$ terms. In fact, we note that they have even better decay since they have an angular derivative $\angn \Far$ in place of one the $\nabla_{\uL};$ we do not make use of this fact.

To bound the type $(iii)$ terms, we first note that only one factor on the right-hand side of \eqref{E:ii} involves the fast-decaying terms $\alpha, \rho, \sigma;$ it arises from the $\mathcal{Q}_{(i)}.$  However, there is an additional power of $\nabla_{L}$ available to compensate. Therefore, we again use \eqref{E:VectorfieldAlgebra} to deduce that
\begin{align}
	|(iii)| & \lesssim r^{-1} r^{-(l_1 + \cdots + l_l + 1)} (1 + |q|)^{-(k_1 + \cdots + k_k)} 
		\Knorm \Far \Knorm_{\Lie_{\mathcal{Z}};M} = r^{-2} r^{-l} (1 + |q|)^{-k} \Knorm \Far 
		\Knorm_{\Lie_{\mathcal{Z}};M}. \notag
\end{align}

We now set $k=0$ and combine the estimates for the linear terms and the type $(i) - (iii)$ nonlinear terms, arriving at the estimate \eqref{E:partialsnpartiaqpartialsnalpha}. We remark that we will use the expressions for the terms $(i) - (iii)$ in the case $k \geq 1$ later in the article, during our proof of \eqref{E:GlobalSobolevalphaupgraded}.

\end{proof}

Armed with the previous estimates, we are now ready for the proof of the proposition.

\subsection{Proof of Proposition \ref{P:GlobalSobolev} (global Sobolev inequalities)}
Most of these estimates were proved as Theorems 3.1 and 3.2 of \cite{dCsK1990}. In particular, we do not repeat the proof of \eqref{E:Interiorestimates}. However, our proofs of \eqref{E:GlobalSobolevalpha} - \eqref{E:GlobalSobolevalphaupgraded} involve modifications of the arguments that take into account the special nonlinear structure of the MBI system. Therefore, we prove \eqref{E:GlobalSobolevualpha} - \eqref{E:GlobalSobolevalphaupgraded} in complete detail, and supply some additional details not contained in \cite{dCsK1990}.

The arguments we give concern the exterior region $\lbrace (t,\ux) \mid |\ux| \geq 1 + t/2 \rbrace.$ To begin, we square inequality \eqref{E:VectorfieldAlgebra} and integrate over the exterior region, thereby obtaining the following inequality:

\begin{align} \label{E:VectorfieldAlgebraExteriorRegionIntegrated}
	\int_{|\bar{y}| \geq 1 + t/2} \sum_{k+l=0}^M & (1 + \big||\bar{y}| - t\big|^2)^{2k} |\bar{y}|^{2l} \Big\lbrace 
		(1 + \big||\bar{y}| - t\big|^2) 
		|\nabla_{\uL}^k\nabla_L^l \ualpha(t,\bar{y})|_{\Lie_{\mathcal{O}};M-k-l}^2 + |\bar{y}|^2 |\nabla_{\uL}^k \nabla_L^l 
		\alpha(t,\bar{y})|_{\Lie_{\mathcal{O}};M-k-l}^2 \\
	& + |\bar{y}|^2 |\nabla_{\uL}^k \nabla_L^l \rho(t,\bar{y})|_{\Lie_{\mathcal{O}};M-k-l}^2
		+ |\bar{y}|^2 |\nabla_{\uL}^k \nabla_L^l \sigma(t,\bar{y})|_{\Lie_{\mathcal{O}};M-k-l}^2 \Big\rbrace \, d^3 \bar{y} 
		\lesssim  \Kintnorm \Far(t) \Kintnorm_{\Lie_{\mathcal{Z}};M}^2. \notag	
\end{align}

Now for any $k + l + m = 0, \cdots, M - 2,$ we define $U(t,\ux) \eqdef r^{m+l} (\sqrt{1 + q^2})^{m+1} \nabla_{\uL}^k\nabla_L^l 
\angn_{(m)} \ualpha$ \\
or ${U(t,\ux) \eqdef r^{l+m+1} (\sqrt{1 + q^2})^k \nabla_{\uL}^k\nabla_L^l  \angn_{(m)} (\alpha,\sigma,\rho)}.$ If $\hat{N}$ denotes the outward normal to the sphere $S_{r,t},$ then the fact that $\nabla_{\hat{N}} = \partial_r$ implies that $\nabla_{\hat{N}} \Big((1 + q^2)^{1/2}\Big) \lesssim 1.$ Using this estimate, Lemma \ref{L:LLunderlinecommutewithLieO}, Corollary \ref{C:nablaLnablaunderlineLLieDerivativeIntrinsicCovariantDerivativeComparison}, 
\eqref{E:VectorfieldAlgebraExteriorRegionIntegrated}, and the fact that $\nabla_{\hat{N}} = \frac{1}{2}(\nabla_L - \nabla_{\uL}),$ we arrive at the following inequality:

\begin{align} \label{E:uSobolevready}
	\int_{|\bar{y}| > 1 + t/2} \Big\lbrace |U(t,\bar{y})|_{\Lie_{\mathcal{O}};2}^2 
		+ \big(1 + \big||\bar{y}| - t\big|^2\big)|\nabla_{\hat{N}} 
		U(t,\bar{y})|_{\Lie_{\mathcal{O}};1}^2  \Big\rbrace \, d^3 \bar{y} 
		\lesssim \Kintnorm \Far(t) \Kintnorm_{\Lie_{\mathcal{Z}};M}^2.
\end{align}
From \eqref{E:ExteriorRegionSobolev} and \eqref{E:uSobolevready}, we conclude that in
the exterior region, we have that

\begin{align}
	|U(t,\ux)| \lesssim (1 + s)^{-1} (1 + |q|)^{-1/2} \Kintnorm \Far(t) \Kintnorm_{\Lie_{\mathcal{Z}};M},
\end{align}
which implies that

\begin{align}
	|\nabla_{\uL}^k\nabla_L^l \angn_{(m)} \ualpha(t,\ux)| 
		& \lesssim (1+s)^{-1-l-m} (1 + |q|)^{-3/2 - k} \Kintnorm \Far(t) \Kintnorm_{\Lie_{\mathcal{Z}};M}, 
		\label{E:ualphawrongorder}  \\
	|\nabla_{\uL}^k\nabla_L^l \angn_{(m)} \big(\alpha(t,\ux),\rho(t,\ux),\sigma(t,\ux)\big)| 
		& \lesssim (1+s)^{-2-l-m}(1 + |q|)^{-1/2 - k} \Kintnorm \Far(t) \Kintnorm_{\Lie_{\mathcal{Z}};M}. 
		\label{E:alphasigmarhowrongorder}
\end{align}
This proves \eqref{E:GlobalSobolevualpha} and \eqref{E:GlobalSobolevalpharhosigma}.

To prove \eqref{E:GlobalSobolevalphaupgraded}, we apply the estimates \eqref{E:GlobalSobolevualpha} - \eqref{E:GlobalSobolevalpharhosigma} to the right-hand side of 
\eqref{E:alphaupgradedfirstinequality}, using \eqref{E:i} - \eqref{E:iii} to estimate the nonlinear 
terms, and Corollary \ref{C:nablaLnablaunderlineLLieDerivativeIntrinsicCovariantDerivativeComparison} to translate
Lie derivative estimates into intrinsic (to the $S_{r,t}$) covariant derivative estimates.

To prove \eqref{E:GlobalSobolevalpha}, for each $l+m = 0,\cdots,M-2,$
we define $U(t,\ux) \eqdef r^{l+m+1} \nabla_L^l \angn_{(m)} \alpha.$ Using Corollary \ref{C:nablaLnablaunderlineLLieDerivativeIntrinsicCovariantDerivativeComparison} and arguing 
as above, we integrate the inequality \eqref{E:partialsnpartiaqpartialsnalpha} over the exterior region
to obtain the following inequality:

\begin{align} \label{E:partialsnpartiaqpartialsnalphaintegrated}
	\int_{|\bar{y}| \geq 1 + t/2} \Big\lbrace |\bar{y}|^2 \sum_{l=0}^M |\bar{y}|^{2l} |\nabla_L^l 
		\alpha(t,\bar{y})|_{\Lie_{\mathcal{O}};M-l}^2
		+ |\bar{y}|^4 \sum_{l=0}^{M-1} |\bar{y}|^{2l} |\nabla_{\uL} \nabla_L^l \alpha(t,\bar{y})|_{\Lie_{\mathcal{O}};M-l-1}^2 
		\Big\rbrace \, d^3\bar{y} \lesssim \Kintnorm \Far(t) \Kintnorm_{\Lie_{\mathcal{Z}};M}^2. 
\end{align}
Using Corollary \ref{C:nablaLnablaunderlineLLieDerivativeIntrinsicCovariantDerivativeComparison} and \eqref{E:partialsnpartiaqpartialsnalphaintegrated}, it then follows that

\begin{align}
	\int_{|\bar{y}| \geq 1 + 1/2} \bigg\lbrace |U(t,\bar{y})|_{\Lie_{\mathcal{O}};2}^2 + |\bar{y}|^2 |\nabla_{\hat{N}} 
		U(t,\bar{y})|_{\Lie_{\mathcal{O}};1}^2
	\bigg\rbrace \, d^3 \bar{y} \lesssim \Kintnorm \Far(t) \Kintnorm_{\Lie_{\mathcal{Z}};M}^2.
\end{align}
Therefore, using \eqref{E:alphaimprovementExteriorRegionSobolev}, we conclude that in the exterior region, we have that

\begin{align} \label{E:alphaimprovedinequalityforU}
	|U(t,\ux)| \lesssim r^{-3/2} \Kintnorm \Far(t) \Kintnorm_{\Lie_{\mathcal{Z}};M}.
\end{align}
From \eqref{E:alphaimprovedinequalityforU}, we conclude that the following inequality holds in the exterior region:

\begin{align} \label{E:partialsnalphawrongorder}
	|\nabla_L^l \angn_{(m)}  \alpha(t,\ux)| \lesssim r^{-5/2 - l - m} \Kintnorm \Far(t) \Kintnorm_{\Lie_{\mathcal{Z}};M}.
\end{align}
This proves inequality \eqref{E:GlobalSobolevalpha}.

\section{Energy Estimates for the MBI System} \label{S:EnergyEstimates}
\setcounter{equation}{0}

The goal of this section is to prove the most important estimate in the article:
an integral inequality for the norm $\Kintnorm \Far(t) \Kintnorm_{\Lie_{\mathcal{Z}};N}.$ The inequality, which 
is the conclusion of the next proposition, is the crux of the proof of Theorem \ref{T:GlobalExistence}.

\begin{proposition} \label{P:Energydifferentialinequality}
	Let $N \geq 3$ be an integer. Assume that $\Far$ is a classical solution to the MBI system \eqref{E:modifieddFis0summary} 
	- \eqref{E:HmodifieddMis0summary} existing on the slab $[0,T] \times \mathbb{R}^3.$ Then there exist constants $\epsilon' > 
	0$ and $C > 0$ such that if ${\sup_{t \in [0,T]}\Kintnorm \Far(t) \Kintnorm_{\Lie_{\mathcal{Z}};N} \leq \epsilon'},$ then the 
	following inequality holds for $t \in [0,T]:$

	\begin{align} \label{E:LieZintegralnormintegralinequality}
		\Kintnorm \Far(t) \Kintnorm_{\Lie_{\mathcal{Z}};N}^2 \leq 
			C \Big\lbrace \Kintnorm \Far(0) \Kintnorm_{\Lie_{\mathcal{Z}};N}^2  
			+ \int_{\tau = 0}^t \frac{1}{1 + \tau^2} \Kintnorm \Far(\tau) \Kintnorm_{\Lie_{\mathcal{Z}};N}^2 \, d \tau \Big\rbrace.
	\end{align}
	
\end{proposition}
\ \\

\noindent \hrulefill
\ \\

The proof of Proposition \ref{P:Energydifferentialinequality} will follow easily from the following lemma and
its corollary. 

\begin{lemma} \label{L:divergenceofdotJLinfinitybound}
	Let $N \geq 3$ be an integer. Assume that $\Far$ is a classical solution to the MBI system \eqref{E:modifieddFis0summary} 
	- \eqref{E:HmodifieddMis0summary} existing on the slab $[0,T] \times \mathbb{R}^3.$ For each $|I| \leq N,$ let 
	$\dot{J}_{\Far}^{\mu}[\Lie_{\mathcal{Z}}^I \Far] \eqdef - \Stress_{\ \nu}^{\mu} \overline{K}^{\nu}$
	be the energy current current \eqref{E:Jdotdef} constructed out of the variation 
	$\dot{\Far} \eqdef \Lie_{\mathcal{Z}}^I \Far$ and the background $\Far.$
	Then there exists a constant $\epsilon > 0$ such that if 
	$\Kintnorm \Far \Kintnorm_{\Lie_{\mathcal{Z}};\lfloor N/2 \rfloor +2} \leq \epsilon,$ then the following pointwise estimate 
	holds on $[0,T] \times \mathbb{R}^3:$
	
	\begin{align} \label{E:divergenceofdotJLinfinitybound}
		\Big|\nabla_{\mu} \big(\dot{J}_{\Far}^{\mu}[\Lie_{\mathcal{Z}}^I \Far(t,\ux)]\big) \Big| 
			\lesssim \frac{\mathcal{E}_{\lfloor N/2 \rfloor + 2}^2[\Far(t)]}{1 + s^2} 
			\Knorm \Far \Knorm_{\Lie_{\mathcal{Z}};N}^2. 
	\end{align}
	
\end{lemma}

\begin{corollary} \label{C:divergenceofdotJL1bound}
	Assume the hypotheses of Lemma \ref{L:divergenceofdotJLinfinitybound}. 
	Then there exists a constant $\epsilon > 0$ such that if $\Kintnorm \Far \Kintnorm_{\Lie_{\mathcal{Z}};\lfloor N/2 \rfloor 
	+2} \leq \epsilon,$	then the following estimate holds for $t \in [0,T]:$
	
	\begin{align} \label{E:divergenceofdotJL1bound}
		\int_{\mathbb{R}^3} \nabla_{\mu} \big(\dot{J}_{\Far}^{\mu}[\Lie_{\mathcal{Z}}^I \Far(t,\ux)]\big) \, d^3 \ux 
		\lesssim \frac{\mathcal{E}_{\lfloor N/2 \rfloor + 2}^2[\Far(t)]}{1 + t^2} 
		\Kintnorm \Far(t) \Kintnorm_{\Lie_{\mathcal{Z}};N}^2.
	\end{align}
\end{corollary}

\begin{proof}
	Corollary \ref{C:divergenceofdotJL1bound} follows from 
	integrating inequality \eqref{E:divergenceofdotJLinfinitybound} over $\Sigma_t.$
\end{proof}

We will now prove the proposition; we will subsequently provide a proof of the lemma.
\\

\noindent \textbf{Proof of Proposition \ref{P:Energydifferentialinequality}:} \\

Using the definition of $\mathcal{E}_N[\Far(t)],$ the fact that $\lfloor N/2 \rfloor + 2 \leq N,$
the divergence theorem, \eqref{E:EnergyNormEquivalence}, Corollary \ref{C:divergenceofdotJL1bound}, 
and the smallness assumption on $\Kintnorm \Far \Kintnorm_{\Lie_{\mathcal{Z}};\lfloor N/2 \rfloor +2},$
we have that
	
	\begin{align} \label{E:energydifferentialinequality}
		\frac{d}{dt} \big(\mathcal{E}_N^2[\Far(t)]\big) 
		& = \sum_{|I| \leq N} \int_{\mathbb{R}^3} \partial_{t} \Big(\dot{J}_{\Far}^0 [\Lie_{\mathcal{Z}}^I \Far(t,\ux)] \Big) \, 
			d^3 \ux
		  = \sum_{|I| \leq N} \int_{\mathbb{R}^3} \nabla_{\mu} \Big(\dot{J}_{\Far}^{\mu} [\Lie_{\mathcal{Z}}^I \Far(t,\ux)] 
		  	\Big) \, d^3 \ux \\
		 & \lesssim \frac{\mathcal{E}_{\lfloor N/2 \rfloor + 2}^2[\Far(t)]}{1 + t^2} 
		 		\Kintnorm \Far(t) \Kintnorm_{\Lie_{\mathcal{Z}};N}^2 
		 		\lesssim \frac{\mathcal{E}_N^2[\Far(t)]}{1 + t^2}. \notag 
	\end{align}
	Inequality \eqref{E:LieZintegralnormintegralinequality} now follows from integrating
	inequality \eqref{E:energydifferentialinequality} from $0$ to $t$ and using \eqref{E:EnergyNormEquivalence}
	again. \hfill $\qed$ \\
 
We now return to the proof of Lemma \ref{L:divergenceofdotJLinfinitybound}. \\

\noindent \textbf{Proof of Lemma \ref{L:divergenceofdotJLinfinitybound}:} \\
	
	First, we note that by the smallness assumption 
	$\Kintnorm \Far \Kintnorm_{\Lie_{\mathcal{Z}};\lfloor N/2 \rfloor +2} \leq \epsilon,$
	together with Corollary \ref{C:GlobalSobolev}, we have that
	
	\begin{align}
		\Knorm \Far \Knorm_{\Lie_{\mathcal{Z}};\lfloor N/2 \rfloor} 
		& \leq C \Kintnorm \Far \Kintnorm_{\Lie_{\mathcal{Z}};\lfloor N/2 \rfloor +2}
		\leq C \epsilon (1 + s)^{-1} (1 + |q|)^{-1/2}.
	\end{align}
	The above inequality is more than sufficient to guarantee that if $\epsilon$ is sufficiently small, then the hypotheses of 
	Proposition \ref{P:Equivalences} and of all of the lemmas and corollaries of Section \ref{S:NullFormEstimates}
	are satisfied; we will make use of these results in our argument below.
	
	By \eqref{E:MBIInhomogeneoustermsJvanish} - \eqref{E:MBIInhomogenoustermsI} and 
	\eqref{E:divJdot}, we have that
	
	\begin{align} \label{E:divergenceofJpointwisebound}
		\Big|\nabla_{\mu} \big(\dot{J}_{\Far}^{\mu}[\Lie_{\mathcal{Z}}^I \Far]\big) \Big|
		& \lesssim \Big| \Big\lbrace H_{\triangle}^{\mu \zeta \kappa \lambda} \nabla_{\mu} \Big(\Lie_{\mathcal{Z}}^I 
			\Far_{\kappa \lambda}\Big) -	\Liemod_{\mathcal{Z}}^I \Big(H_{\triangle}^{\mu \zeta \kappa \lambda} \nabla_{\mu} 
			\Far_{\kappa \lambda} \Big) \Big\rbrace \Lie_{\mathcal{Z}}^I \Far_{\zeta \nu} \overline{K}^{\nu} \Big| \\
		& \ \ \ + \Big|(\nabla_{\mu}H^{\mu \zeta \kappa \lambda}) \Lie_{\mathcal{Z}}^I \Far_{\kappa \lambda} \Lie_{\mathcal{Z}}^I 
			\Far_{\ \zeta}^{\nu} \overline{K}_{\nu}\Big|
			+ \Big| \overline{K}_{\nu}(\nabla^{\nu}H^{\zeta \beta \kappa \lambda}) 
			\Lie_{\mathcal{Z}}^I \Far_{\zeta \beta} \Lie_{\mathcal{Z}}^I \Far_{\kappa \lambda}\Big| \notag \\
		& \ \ \ +  \bigg| \bigg\lbrace
			\ell_{(MBI)}^{-2} \Far^{\kappa \lambda}\dot{\Far}_{\kappa \lambda} 
			\big( \Far^{\nu \zeta}\dot{\Far}_{\ \zeta}^{\mu} - \Far^{\mu \zeta} \dot{\Far}_{\ \zeta}^{\nu} \big) \notag \\
		& \ \ \ \ \ \ \ + \big(1 + \Farinvariant_{(2)}^2 \ell_{(MBI)}^{-2}\big)
			\Fardual^{\kappa \lambda} \dot{\Far}_{\kappa \lambda} \big( 
			\Fardual^{\nu \zeta}\dot{\Far}_{\ \zeta}^{\mu} - \Fardual^{\mu \zeta}\dot{\Far}_{\ \zeta}^{\nu} \big) \notag \\
		& \ \ \ \ \ \ \ + \Farinvariant_{(2)} \ell_{(MBI)}^{-2} \Fardual^{\kappa \lambda}\dot{\Far}_{\kappa \lambda}
				\big( \Far^{\mu \zeta}\dot{\Far}_{\ \zeta}^{\nu} - \Far^{\nu \zeta}\dot{\Far}_{\ \zeta}^{\mu} \big) \notag \\
		& \ \ \ \ \ \ \ + \Farinvariant_{(2)} \ell_{(MBI)}^{-2} \Fardual^{\kappa \lambda}\dot{\Far}_{\kappa \lambda}
				\big( \Far^{\mu \zeta}\dot{\Far}_{\ \zeta}^{\nu} - \Far^{\nu \zeta}\dot{\Far}_{\ \zeta}^{\mu} \big) \notag \\
		& \ \ \ \ \ \ \ + \Farinvariant_{(2)} \ell_{(MBI)}^{-2} \Far^{\kappa \lambda}\dot{\Far}_{\kappa \lambda}
			\big( \Fardual^{\mu \zeta}\dot{\Far}_{\ \zeta}^{\nu} - \Fardual^{\nu \zeta}\dot{\Far}_{\ \zeta}^{\mu} \big) \bigg\rbrace 
			\nabla_{\mu} \overline{K}_{\nu} \bigg|. \notag
	\end{align}
	Inequality \eqref{E:divergenceofdotJLinfinitybound} now follows from \eqref{E:divergenceofJpointwisebound}, 
	Corollary \ref{C:HtriangleCommutator}, Lemma \ref{L:DerivativeofHtriangleerrorterms}, Lemma 
	\ref{L:CanonicalStressMorawetzDerivativeTerm}, and Corollary \ref{C:GlobalSobolev}. 
	
	As an example, we describe the estimate of the first term on the right-hand side of \eqref{E:divergenceofJpointwisebound} in 
	more detail. To estimate this term, we first note that by Corollary \ref{C:HtriangleCommutator}, we have that
	
	\begin{align} \label{E:algebraicCauchySchwarzHtriangleCommutatoragain}
		\sum_{|I| \leq N}
			\Big| \Big\lbrace H_{\triangle}^{\mu \zeta \kappa \lambda} & \nabla_{\mu} \Big(\Lie_{\mathcal{Z}}^I \Far_{\kappa 	
			\lambda}\Big) -	\Liemod_{\mathcal{Z}}^I \Big(H_{\triangle}^{\mu \zeta \kappa \lambda} \nabla_{\mu} \Far_{\kappa 
			\lambda} \Big) \Big\rbrace (\Lie_{\mathcal{Z}}^I \Far_{\nu \zeta}) \overline{K}^{\nu} \Big| \\
		& \lesssim \bigg\lbrace \sum_{|J| \leq \lfloor N/2 \rfloor} (1 + s)^2|\Lie_{\mathcal{Z}}^J 
			\Far|_{\mathcal{L}\mathcal{U}}^2 + (1 + s)^2 |\Lie_{\mathcal{Z}}^J \Far|_{\mathcal{T}\mathcal{T}}^2 
			+ |\Lie_{\mathcal{Z}}^J \Far|^2 \bigg\rbrace \notag \\
		& \ \ \times \bigg\lbrace \sum_{|I| \leq N} (1 + s)^2|\Lie_{\mathcal{Z}}^I \Far|_{\mathcal{L}\mathcal{U}}^2 
			+ (1 + s)^2 |\Lie_{\mathcal{Z}}^I \Far|_{\mathcal{T}\mathcal{T}}^2 
			+ (1 + |q|)^2|\Lie_{\mathcal{Z}}^I \Far|^2 \bigg\rbrace. \notag 
		\end{align}
	By Corollary \ref{C:GlobalSobolev}, we have that 
	$\bigg| \sum_{|J| \leq \lfloor N/2 \rfloor} (1 + s)^2|\Lie_{\mathcal{Z}}^J 
	\Far|_{\mathcal{L}\mathcal{U}}^2 + (1 + s)^2 |\Lie_{\mathcal{Z}}^J \Far|_{\mathcal{T}\mathcal{T}}^2 
	+ |\Lie_{\mathcal{Z}}^J \Far|^2 \bigg| \lesssim (1 + s)^{-2} 
	\mathcal{E}_{\lfloor N/2 \rfloor + 2}^2[\Far],$ while by definition, 
	$\bigg| \sum_{|I| \leq N} (1 + s)^2|\Lie_{\mathcal{Z}}^I \Far|_{\mathcal{L}\mathcal{U}}^2 
	+ (1 + s)^2 |\Lie_{\mathcal{Z}}^I \Far|_{\mathcal{T}\mathcal{T}}^2 
	+ (1 + |q|)^2|\Lie_{\mathcal{Z}}^I \Far|^2  \bigg| \lesssim \Knorm \Far \Knorm_{\Lie_{\mathcal{Z}};N}^2.$
	It thus follows that the left-hand side of \eqref{E:algebraicCauchySchwarzHtriangleCommutatoragain}
	is bounded by the right-hand side of \eqref{E:divergenceofdotJLinfinitybound}. All other terms can be estimated in a 
	similar fashion, and we omit the details. \hfill $\qed$

\section{Local Existence and the Continuation Principle} \label{S:IVP}
\setcounter{equation}{0}

In this section, we briefly discuss the initial value problem for the MBI system. With the exception of 
the availability of Proposition \ref{P:LocalExistenceCurrent}, the material presented here is very standard. For the purposes of our global existence theorem, which is proved in Section \ref{S:GlobalExistence}, the most important fact presented is the continuation principle: it shows that a-priori control over the norm $\Kintnorm \Far(t) \Kintnorm_{\Lie_{\mathcal{Z}};3}$ is sufficient to deduce global existence.
\\

\noindent \hrulefill
\ \\

\begin{proposition} \label{P:LocalExistence}
	Let $N \geq 3$ be an integer, and let the pair of one-forms $(\mathring{\Magneticinduction}, \mathring{\Displacement})$ 
	be initial data that are tangent to the Cauchy hypersurface $\Sigma_0,$ that satisfy the constraints 
	\eqref{E:Dconstraint} - \eqref{E:Bconstraint}, and that satisfy 
	$\| (\mathring{\Magneticinduction}, \mathring{\Displacement}) \|_{H_1^N} < \infty.$ Here, $H_1^N$ is the weighted Sobolev 
	norm defined in \eqref{E:HNdeltanorm}. Then these data launch a unique classical solution $\Far$ to the MBI system 
	existing on non-trivial a maximal spacetime slab of the form $[0,T_{max}) \times \mathbb{R}^3.$ The solution has the 
	following regularity properties:
	
	\begin{subequations}
	\begin{align}
		\Far \in C^{N-2}\big([0,T_{max}) \times \mathbb{R}^3\big), \\
		(\Magneticinduction,\Displacement) \in C^{N-2}\big([0,T_{max}) \times \mathbb{R}^3\big) 
		\cap \bigcap_{k=0}^{k=N-2} C^k\big([0,T_{max}),H_1^{N-k} \big).
	\end{align}
	\end{subequations}
	
	Furthermore, either $T_{max} = \infty,$ or one of the following two breakdown scenarios must occur:
	
	\begin{enumerate}
		\item $\lim_{t \uparrow T_{\max}} \Kintnorm \Far(t) \Kintnorm_{\Lie_{\mathcal{Z}};N} = \infty$
		\item There exists a sequence $(t_n,\ux_n)$ with $t_n < T_{max}$ such that
		$\lim_{n \to \infty} \ell_{(MBI)}[\Far(t_n,\ux_n)] = 0,$
	\end{enumerate}
	where $\ell_{(MBI)} \eqdef \ell_{(MBI)}[\Far] = \big(1 + \Farinvariant_{(1)}[\Far] - \Farinvariant_{(2)}^2[\Far]\big)^{1/2}$ 
	is the function of $\Far$ defined in \eqref{E:ldef}, and $\Kintnorm \Far(t) \Kintnorm_{\Lie_{\mathcal{Z}};3}$ is the norm 
	defined in \eqref{E:MorawetzWeightedLieDerivativeIntegralNormN}.
	
\end{proposition}

\begin{remark}
	The classification of the two breakdown scenarios is known as a ``continuation principle.'' 
\end{remark}

Since Proposition \ref{P:LocalExistence} is rather standard, we don't provide a full proof, but instead refer to the reader to e.g. \cite[Ch. VI]{lH1997}, \cite{aM1984}, \cite{jSmS1998}, \cite{cS1995}, \cite{jS2008a}, and \cite[Ch. 16]{mT1997III} for the missing details concerning local existence, and e.g. \cite{jS2008b} for the ideas behind the continuation principle. The crucial point is the availability of Proposition \ref{P:LocalExistenceCurrent}, which can be used to derive $H_1^N$ estimates for solutions to the linearized MBI system, that is, the equations of variation. More specifically, in constructing the local solution of Proposition \ref{P:LocalExistence}, one typically uses an iteration argument or a contraction mapping argument. Both methods involve an analysis of solutions to the equations of variation\footnote{In the equations of variation, one can think of the background $\Far$ as the ``previous'' iterate, and $\dot{\Far}$ as the ``next'' one.}, and in particular, they require uniform estimates of their weighted Sobolev norms; these uniform estimates can be derived using energy currents and the ideas contained in the proof of Proposition \ref{P:Energydifferentialinequality}. In particular, suitable energy estimates for solutions to the equations of variation \eqref{E:EOVBianchi} - \eqref{E:EOVMBI} can be derived by using energy currents $\dot{J}_{local;\Far}^{\mu}[\dot{\Far}]$ defined by

\begin{align}
	\dot{J}_{local;\Far}^{\mu}[\dot{\Far}] \eqdef - (1 + s^2) \Stress_{\ \nu}^{\mu} \Vmult^{\nu},
\end{align}
where $\Stress_{\ \nu}^{\mu}$ is the canonical stress from \eqref{E:widetildeTHdef}, and $\Vmult^{\nu}$ is the vectorfield
defined in Proposition \ref{P:LocalExistenceCurrent}. Now by Proposition \ref{P:LocalExistenceCurrent}, if $\Far \in \mathfrak{K},$ where $\mathfrak{K}$ is a compact subset of the domain $\mathscr{H}$ of state space where the MBI theory is defined, then we have that

\begin{align}\label{E:KnormpointwiseboundedbyweidghtedsT00}
	\Knorm \dot{\Far} \Knorm^2 \leq (1 + s^2) |\dot{\Far}|^2 \leq C_{\mathfrak{K}} \dot{J}_{local;\Far}^0[\dot{\Far}].
\end{align}
By Remark \ref{R:BDHyperbolicity}, in terms of the state-space variables $(\Magneticinduction, \Displacement),$
this domain comprises the set finite values of $(\Magneticinduction, \Displacement).$ On the other hand, using the simple inequalities $(1 + q^2) |\dot{\Far}^2| \leq \Knorm \dot{\Far} \Knorm^2$ and $1 + s^2 \lesssim (1 + t^2)(1 + q^2),$ 
together with the fact that $|\Stress| \lesssim |\dot{\Far}|^2,$ we deduce that

\begin{align} \label{E:sweightedT00pointwisebounedbytweightedKnorm}
	\dot{J}_{local;\Far}^0[\dot{\Far}] \leq C^{-1}(1 + s^2) |\dot{\Far}|^2
	\leq C^{-1} (1 + t^2)(1 + q^2) |\dot{\Far}^2| \lesssim (1 + t^2) \Knorm \dot{\Far} \Knorm^2.
\end{align}

Consequently, if we define the energy $\mathcal{E}_{local;N}[\dot{\Far}(t)]$ by
\begin{align}
	\mathcal{E}_{local;N}^2[\dot{\Far}(t)] & \eqdef \sum_{|I| \leq N}
	\int_{\mathbb{R}^3} \dot{J}_{local;\Far}^0[\Lie_{\mathcal{Z}}^I \dot{\Far}] \, d^3 \ux,
\end{align}
then \eqref{E:KnormpointwiseboundedbyweidghtedsT00} and \eqref{E:sweightedT00pointwisebounedbytweightedKnorm} imply that

\begin{align} \label{E:KintnormweightedsT00comparision}
	\Kintnorm \dot{\Far}(t) \Kintnorm_{\Lie_{\mathcal{Z}};N} \lesssim \mathcal{E}_{local;N}[\dot{\Far}(t)] 
	\lesssim (1 + t) \Kintnorm \dot{\Far}(t) \Kintnorm_{\Lie_{\mathcal{Z}};N}.
\end{align}
We remark that the implicit constants in \eqref{E:KintnormweightedsT00comparision} depend on $\mathfrak{K}.$

We now illustrate the fundamental energy estimate that can be used to deduce the desired local existence result. We set $N=0$ for simplicity, and consider a solution $\dot{\Far}$ to the MBI equations of variation \eqref{E:EOVBianchi} - \eqref{E:EOVMBI}
with the initial data $(\mathring{\Magneticinduction}, \mathring{\Displacement}).$ Then using \eqref{E:Hstress}, \eqref{E:divergenceofwidetildeT}, \eqref{E:divergenceofTX}, \eqref{E:KnormpointwiseboundedbyweidghtedsT00}, \eqref{E:KintnormweightedsT00comparision}, the divergence theorem, and the Cauchy-Schwarz inequality for integrals, it follows 
(as in our proof of Proposition \ref{P:Energydifferentialinequality}) that

\begin{align} \label{E:FundamentalLocalExistenceDifferentialEnergyInequality}
	\frac{d}{dt}\Big(\mathcal{E}_{local}^2[\dot{\Far}(t)] \Big) 
	& \lesssim f(\mathfrak{K}; \| \Far(t) \|_{L^{\infty}}; \| \nabla \Far(t) \|_{L^{\infty}})
	\int_{\mathbb{R}^3} (1 + t^2 + |\ux|^2)|\dot{\Far}(t,\ux)|^2 
		+ (1 + t^2 + |\ux|^2)|\dot{\Far}|(|\mathfrak{J}(t,\ux)| + |\mathfrak{I}(t,\ux)|)\, d^3 \ux \\
	& \lesssim f(\mathfrak{K}; \| \Far(t) \|_{L^{\infty}}; \| \nabla \Far(t) \|_{L^{\infty}})
		\Big\lbrace \mathcal{E}_{local}^2[\dot{\Far}(t)] + \mathcal{E}_{local;N}[\dot{\Far}(t)] 
		\big\| |\mathfrak{J}(t,\ux)| + |\mathfrak{I}(t,\ux)| \big\|_{H_1^0} \Big\rbrace, \notag
\end{align}
where $H_1^0$ is the weighted Sobolev norm defined in \eqref{E:HNdeltanorm},
$\mathfrak{J}_{\lambda \mu \nu}, \mathfrak{I}^{\nu}$ are the inhomogeneous terms on the right-hand sides of
\eqref{E:EOVBianchi} - \eqref{E:EOVMBI}, and $f(\mathfrak{K};\cdot)$ can be chosen to be a positive, increasing, continuous function of its arguments. We remark similar inequalities can be deduced for $N \geq 1,$ and that the inhomogeneous terms 
would be measured using the $H_1^N$ norm. 

The availability of inequality \eqref{E:FundamentalLocalExistenceDifferentialEnergyInequality} for solutions to
the equations of variation is the fundamental reason that Proposition \ref{P:LocalExistence} holds. From
\eqref{E:FundamentalLocalExistenceDifferentialEnergyInequality}, Gronwall's inequality, and appropriate weighted Sobolev estimates for the inhomogeneous terms, it can be shown that $\mathcal{E}_{local}^2[\dot{\Far}(t)]$ can be uniformly bounded in terms of $\|(\mathring{\Magneticinduction}, \mathring{\Displacement}) \|_{H_1^N},$ if $t$ is sufficiently small. Similar inequalities can be deduced for the higher order energies $\mathcal{E}_{local;N}^2[\dot{\Far}(t)].$ As mentioned above, this is the main step in deducing local existence for the nonlinear equations; the remaining details can be found in the aforementioned references. 

There is one additional step in the proof of local existence that we will comment on, namely the issue of
showing that $\mathcal{E}_{local;N}^2[\dot{\Far}(0)]$ is uniformly bounded by $\|(\mathring{\Magneticinduction}, \mathring{\Displacement}) \|_{H_1^N}$ whenever $N \geq 3$ and $\dot{\Far}$ is a solution to the equations of variation. To accomplish this rather tedious step, one can first express the equations of variation in terms of $(\Electricfield, \Magneticinduction)$ and $(\dot{\Electricfield}, \dot{\Magneticinduction}),$ and inhomogeneous terms, where $(\Electricfield, \Magneticinduction)$ and $(\dot{\Electricfield}, \dot{\Magneticinduction})$ are the electromagnetic decompositions
of $\Far$ and $\dot{\Far}$ respectively. One would then use weighted Sobolev multiplication estimates, as in our proof of Lemma \ref{L:EBFarintegralnormequivalence}, to deduce that

\begin{align} \label{E:InherentNormFiniteImpliesWeightedFarNormFinite}
	\|(\mathring{\Magneticinduction}, \mathring{\Displacement}) \|_{H_1^N} < \infty	 
		\implies \Kintnorm \dot{\Far}(0) \Kintnorm_{\Lie_{\mathcal{Z}};N} \leq
		\widetilde{f}(\| (\mathring{\Magneticinduction}, \mathring{\Displacement}) \|_{H_1^N}),
\end{align}
where $\widetilde{f}$ can be chosen to be a positive, increasing, continuous function of its argument. By \eqref{E:KintnormweightedsT00comparision}, the desired uniform bound for $\mathcal{E}_{local;N}^2[\dot{\Far}(0)]$ then
follows from \eqref{E:InherentNormFiniteImpliesWeightedFarNormFinite}. To deduce \eqref{E:InherentNormFiniteImpliesWeightedFarNormFinite}, we have assumed that both $\Far$ and $\dot{\Far}$ have the same initial data, and that $\sum_{n = 0}^N (1 + r^2)^{n + 1} \Big(|\nabla_{(n)} \Electricfield|^2|_{\Sigma_0} + |\nabla_{(n)} \Magneticinduction|^2|_{\Sigma_0} \Big),$ the relevant Sobolev-Moser norm for $\Far$ during a proof of \eqref{E:InherentNormFiniteImpliesWeightedFarNormFinite},
can be bounded by a positive, increasing, continuous function of $\| (\mathring{\Magneticinduction}, \mathring{\Displacement}) \|_{H_1^N}.$ In practice, during an iteration scheme, this argument would need to be slightly modified; for technical reasons, typical iteration schemes involve a slightly different smoothing\footnote{The data are smoothed because for several reasons, one reason being that during the iteration process, it is convenient to work with classically differentiable functions, rather than distributions.}
%during the first stage of a local existence proof, it is convenient to derive an equation for the solution minus the data. This results in a loss of one order of differentiability of the solution (compared to the data) arising from inhomogeneous terms involving first derivatives of the data. This loss of a derivative can be recovered through a separate argument
of the initial data at each stage, so that the initial data change slightly from iterate to iterate. 

\hfill $\qed$

\section{Global Existence for the MBI System} \label{S:GlobalExistence}
\setcounter{equation}{0}

In this section, we provide a proof of our main theorem. The global existence aspect of our result will be an easy consequence of the energy inequality of Proposition \ref{P:Energydifferentialinequality} and the continuation principle of Proposition \ref{P:LocalExistence}, while the decay aspect will follow directly from Proposition \ref{P:GlobalSobolev}, which is the global Sobolev inequality.

\noindent \hrulefill
\ \\

\begin{theorem} \label{T:GlobalExistence}
	Let $N \geq 3$ be an integer, and let the pair of one-forms $(\mathring{\Magneticinduction}, \mathring{\Displacement})$ 
	be initial data that are tangent to the Cauchy hypersurface $\Sigma_0,$ and that satisfy the constraints 
	\eqref{E:Dconstraint} - \eqref{E:Bconstraint}. There exists an $\epsilon_0 > 0$ such that if 
	$\| (\mathring{\Magneticinduction}, \mathring{\Displacement}) \|_{H_1^N} \leq \epsilon_0,$ 
	then these data launch a unique classical solution $\Far$ to the Maxwell-Born-Infeld system 
	\eqref{E:modifieddFis0summary} - \eqref{E:HmodifieddMis0summary} existing on the spacetime slab $[0,\infty) \times \mathbb{R}^3.$ 
	Furthermore, there exists a constant $C_* > 0$ such that 
	\begin{align} \label{E:Faradaynormgloballysmall}
		\Kintnorm \Far(t) \Kintnorm_{\Lie_{\mathcal{Z}};N} \leq C_* \| (\mathring{\Magneticinduction}, \mathring{\Displacement}) 
		\|_{H_1^N}
	\end{align}
	holds for all $t \geq 0.$ Here, $\| \cdot \|_{H_1^N}$ 
	is the weighted Sobolev space defined in \eqref{E:HNdeltanorm}, while  $\Kintnorm \cdot \Kintnorm_{\Lie_{\mathcal{Z}};N}$ is 
	the weighted integral norm defined in \eqref{E:MorawetzWeightedLieDerivativeIntegralNormN}. In addition, the null components 
	$\ualpha, \alpha, \rho, \sigma$ of $\Far,$ which are defined in Section \ref{SS:NullComponents}, decay according to the rates 
	given by Proposition \ref{P:GlobalSobolev}.
\end{theorem}

\begin{proof}
	We will show that if $\| (\mathring{\Magneticinduction}, \mathring{\Displacement}) 
	\|_{H_1^N}$ is sufficiently small, then neither of the the two breakdown scenarios from Proposition 
	\ref{P:LocalExistence} occur. To this end, let $\epsilon' > 0$ be the small constant from the conclusion of Proposition 
	\ref{P:Energydifferentialinequality}. Choose a positive constant $\epsilon''$ such that $0 < \epsilon'' < \epsilon'$ and such 
	that $\Kintnorm \Far(t) \Kintnorm_{\Lie_{\mathcal{Z}};N} \leq \epsilon'' \implies 
	\inf_{\ux \in \mathbb{R}^3} \ell_{(MBI)}[\Far(t,\ux)] \geq 1/2,$ where $\ell_{(MBI)}$ is defined in \eqref{E:ldef};
	this is possible by Sobolev embedding. Define
	
	\begin{align}
		T_{max} = \sup \lbrace t \geq 0 \mid \mbox{The solution exists on} \ [0,t] \ \mbox{and} 
		\ \Kintnorm \Far(t) \Kintnorm_{\Lie_{\mathcal{Z}};N} \leq \epsilon'' \rbrace. 
	\end{align}
	By the local existence aspect of Proposition \ref{P:LocalExistence}, 
	we have that $T_{max} > 0$ if 
	$\| (\mathring{\Magneticinduction}, \mathring{\Displacement}) \|_{H_1^N}$ is sufficiently small. Applying Proposition 
	\ref{P:Energydifferentialinequality}, we conclude that the following inequality holds on $[0,T_{max}):$
	
	\begin{align} \label{E:firstnorminequality}
		\Kintnorm \Far(t) \Kintnorm_{\Lie_{\mathcal{Z}};N}^2 \leq 
			C \Big\lbrace \Kintnorm \Far(0) \Kintnorm_{\Lie_{\mathcal{Z}};N}^2 
			+ \int_{\tau = 0}^t \frac{1}{1 + \tau^2} \Kintnorm \Far(\tau) \Kintnorm_{\Lie_{\mathcal{Z}};N}^2 \, d \tau \Big\rbrace.
	\end{align}
	Applying Gronwall's inequality to \eqref{E:firstnorminequality}, and using Lemma \ref{L:EBFarintegralnormequivalence},
	we conclude that the following inequality holds on $[0,T_{max}):$
	
	\begin{align} \label{E:secondnorminequality}
		\Kintnorm \Far(t) \Kintnorm_{\Lie_{\mathcal{Z}};N}^2 \leq C \Kintnorm \Far(0) \Kintnorm_{\Lie_{\mathcal{Z}};N}^2
		\mbox{exp}\Big(\int_{\tau = 0}^{\infty} \frac{C}{1 + \tau^2} \Big) \leq C_*^2 
		\| (\mathring{\Magneticinduction}, \mathring{\Displacement}) \|_{H_1^N}^2.
	\end{align}
	
	Now if $C_* \| (\mathring{\Magneticinduction}, \mathring{\Displacement}) \|_{H_1^N} < \epsilon'',$ then 
	the continuation principle of Proposition \ref{P:LocalExistence} and inequality \eqref{E:secondnorminequality}
	together imply that $T_{max} = \infty.$ Furthermore, inequality \eqref{E:Faradaynormgloballysmall} is a direct 
	consequence of \eqref{E:secondnorminequality}.

\end{proof}

\section*{Acknowledgments}
I was supported by the Commission of the European Communities, ERC Grant Agreement No 208007. I was also funded in parts by NSF through grants DMS-0406951 and DMS-0807705. I would like to thank Mihalis Dafermos, Sergiu Klainerman, and Willie Wong for providing helpful suggestions. I offer special thanks to Michael Kiessling for introducing me to the MBI model and for explaining its intriguing features, and to Shadi Tahvildar-Zadeh for helping me delve into the ideas and methods of \cite{dC2000}.

\setcounter{section}{0}
\setcounter{subsection}{0}
\setcounter{subsubsection}{0}
\setcounter{equation}{0}
\setcounter{proposition}{0}

\renewcommand{\thesection}{\Alph{section}}
\renewcommand{\theequation}{\Alph{section}.\arabic{equation}}
\renewcommand{\theproposition}{\Alph{section}-\arabic{proposition}}
\renewcommand{\thecorollary}{\Alph{section}-\arabic{corollary}}
\renewcommand{\thedefinition}{\Alph{section}-\arabic{definition}}
\renewcommand{\thetheorem}{\Alph{section}-\arabic{theorem}}
\renewcommand{\theremark}{\Alph{section}-\arabic{remark}}
\renewcommand{\thelemma}{\Alph{section}-\arabic{lemma}}

\section{Appendix}

In this Appendix, we state the lemmas and corollaries that are used in Section \ref{SS:DataNorms} 
to relate the smallness condition on the inherent data $(\mathring{\Magneticinduction}, \mathring{\Displacement})$
to a smallness condition on $\Kintnorm \Far(0) \Kintnorm_{\Lie_{\mathcal{Z}};N}.$ The lemmas were essentially proved
as Lemmas 2.4 and 2.5 of \cite{yCBdC1981}, while the corollaries are easy (and non-optimal) consequences of the lemmas; we leave their proofs as exercises for the reader. Throughout the appendix, we abbreviate 
$C_{\delta}^{N} \eqdef C_{\delta}^{N}(\mathbb{R}^3),$ $H_{\delta}^{N} \eqdef H_{\delta}^{N}(\mathbb{R}^3),$ etc.
Furthermore, $\SigmafirstfundNabla$ denotes the Levi-Civita connection corresponding to the standard Euclidean metric $\Sigmafirstfund$ on $\mathbb{R}^3,$ and we equip $\mathbb{R}^3$ with standard rectangular coordinates $\ux.$

\begin{lemma} \label{L:SobolevEmbeddingHNdeltaCNprimedeltamprime} \cite[Lemma 2.4]{yCBdC1981}
	Let $N, N' \geq 0$ be integers, and let $\delta, \delta'$ be real numbers subject to the constraints
	$N' < N - 3/2$ and $\delta' < \delta + 3/2.$ Assume that $v \in H_{\delta}^N.$ Then $v \in C_{\delta'}^{N'},$ and
	\begin{align} \label{E:SobolevEmbeddingHNdeltaCNprimedeltamprime}
		\| v \|_{C_{\delta'}^{N'}} \leq C \| v \|_{H_{\delta}^N}.
	\end{align}
\end{lemma}

\hfill $\qed$

\begin{lemma} \label{L:WeightedSobolevSpaceMultiplicationProperties} \cite[slight extension of Lemma 2.5]{yCBdC1981}
	Let $N_1, \cdots, N_p \geq 0$ be integers, and let $\delta_1, \cdots, \delta_p$ be real numbers.
	Suppose that $v_i \in H_{\delta_i}^{N_i}$ for $i = 1, \cdots, p.$ Assume that the integer $N$ satisfies
	$0 \leq N \leq \min \lbrace N_1, \cdots, N_p \rbrace$ and ${N \leq \sum_{i=1}^p N_i - (p-1)3/2},$
	and that $\delta < \sum_{i=1}^p \delta_i + (p-1)3/2.$ Then
	
	\begin{align}
		\prod_{i=1}^p v_i \in H_{\delta}^N,
	\end{align}
	and the multiplication map
	
	\begin{align}
		H_{\delta_1}^{N_1} \times \cdots \times H_{\delta_p}^{N_p} & \rightarrow H_{\delta}^N, &&
		(v_1, \cdots, v_p) \rightarrow \prod_{i=1}^p v_i
	\end{align}
	is continuous.

\end{lemma}

\hfill $\qed$

\begin{corollary} \label{C:WeightedSobolevSpaceMultiplicationProperties}
	Let $N \geq 2$ be an integer, and let $\delta \geq 0.$ Assume that 
	$v_i \in H_{\delta}^N$ for $i = 1, \cdots,p,$
	and that $m_i \geq 0$ are integers
	satisfying $\sum_{i=1}^p m_i \leq N.$ Then
	
	\begin{align}
		(1 + |\ux|^2)^{(\delta + \sum_{i=1}^p m_i)/2} \prod_{j=1}^p \SigmafirstfundNabla_{(m_j)} v_j \in L^2
	\end{align} 
	and
	
	\begin{align}
		\big \| (1 + |\ux|^2)^{(\delta + \sum_{i=1}^p m_i)/2} \prod_{j=1}^p \SigmafirstfundNabla_{(m_j)} v_j \big \|_{L^2} 
		\lesssim \prod_{j=1}^p \|v_j \|_{H_{\delta}^N}.
	\end{align}
	
\end{corollary}

\hfill $\qed$

\begin{corollary} \label{C:CompositionProductHNdelta}
	Let $N \geq 2$ be an integer, and let $\delta \geq 0.$ Assume that $\mathfrak{K}$ is a compact set, that 
	$\mathfrak{F} \in C^N(\mathfrak{K})$ is a function, and that $v_1$ is a function satisfying
	$v_1(\mathbb{R}^3) \subset \mathfrak{K}.$ Assume further that $v_1, v_2 \in H_{\delta}^N.$ 
	Then $(F \circ v_1) v_2 \in H_{\delta}^N,$ and
	
	\begin{align} 
		\| (\mathfrak{F} \circ v_1) v_2 \|_{H_{\delta}^N} 
			& \leq  C(N) \Big\lbrace \| v_2 \|_{H_{\delta}^{N}} \sum_{j=0}^N |\mathfrak{F}^{(j)}|_{\mathfrak{K}} 
    		\| v_1 \|_{H_{\delta}^N}^j \Big\rbrace.
	\end{align}
	In the above inequality, $\mathfrak{F}^{(j)}$ denotes the array of all $j^{th}$ order partial derivatives of 
	$\mathfrak{F}$ with respect to its arguments, and $|\mathfrak{F}^{(j)}|_{\mathfrak{K}} \eqdef \sup_{v \in \mathfrak{K}} 
	|\mathfrak{F}^{(j)}(v)|.$
\end{corollary}

\hfill $\qed$

\bibliographystyle{amsalpha}
\bibliography{JBib}
\end{document}